\documentclass[11pt]{amsart}
\usepackage[letterpaper,top=1in,bottom=1in,left=1in,right=1in,marginparwidth=1.75cm]{geometry}
\usepackage[T1]{fontenc}
\usepackage[english]{babel}
\usepackage[round]{natbib}
\usepackage[dvipsnames]{xcolor}
\usepackage{amsmath, amssymb, amsthm,mathtools,dsfont}
\usepackage{bbm}
\usepackage{blkarray}
\usepackage{cancel}
\usepackage{enumitem}
\usepackage{float}
\usepackage{bm}
\usepackage{graphicx}
\usepackage{mathrsfs}
\usepackage{microtype}
\usepackage{ragged2e}
\usepackage{thmtools}
\usepackage{tikz}
\usepackage[Symbolsmallscale]{upgreek}
\usepackage{caption}
\usepackage{subcaption}
\usepackage[pagebackref]{hyperref}
\hypersetup{citecolor=purple, colorlinks=true, linkcolor=blue}
\usepackage{fancyhdr}
\usepackage[capitalize,noabbrev]{cleveref}
\usepackage[foot]{amsaddr}
\usepackage{algorithm}
\usepackage{algpseudocode}
\usepackage[most]{tcolorbox}

\DeclareMathOperator*{\argmax}{arg\,max}
\DeclareMathOperator*{\argmin}{arg\,min}

\makeatletter
\renewcommand\paragraph{\@startsection{paragraph}{4}%
  {0pt}%
  {1ex \@plus .2ex}%
  {-0.5em}%
  {\normalfont\bfseries}}
\makeatother

\makeatletter
\renewcommand\subsubsection{\@startsection{subsubsection}{3}
  \z@{.5\linespacing\@plus.7\linespacing}{-.5em}%
  {\normalfont\bfseries}}
\makeatother

\makeatletter
\def\l@paragraph{\@tocline{4}{0pt}{1pc}{7pc}{}}
\def\l@subparagraph{\@tocline{5}{0pt}{1pc}{7pc}{}}
\makeatother

\theoremstyle{plain}
\newtheorem{theorem}{Theorem}
\newtheorem{proposition}[theorem]{Proposition}
\newtheorem{lemma}[theorem]{Lemma}
\newtheorem{corollary}[theorem]{Corollary}

\newtheorem*{example1}{\textsc{Example 1}}
\newtheorem*{example1a}{\textsc{Example 1(a)}}
\newtheorem*{example2}{\textsc{Example 2}}

\theoremstyle{definition}

\newtheorem{assumption}[theorem]{Assumption}

\newtheorem{subassumption}{Assumption}[assumption]

\theoremstyle{remark}

\title{Tuning-Free Sampling via Optimization on the Space of Probability Measures}

\author{Louis Sharrock\textsuperscript{1}}
\address{\textsuperscript{1}Department of Statistical Science, University College London, UK}
\email{l.sharrock@ucl.ac.uk}

\author{Christopher Nemeth\textsuperscript{2}}
\address{\textsuperscript{2}School of Mathematical Sciences, Lancaster University, UK}
\email{c.nemeth@lancaster.ac.uk}

\begin{document}

\begin{abstract}
We introduce adaptive, tuning-free step size schedules for gradient-based 
sampling algorithms obtained as time-discretizations of Wasserstein gradient flows. The result is a suite of tuning-free sampling algorithms, including tuning-free variants of the unadjusted Langevin algorithm (ULA), stochastic gradient Langevin dynamics (SGLD), mean-field Langevin dynamics (MFLD), Stein variational gradient descent (SVGD), and variational gradient descent (VGD). More widely, our approach yields tuning-free algorithms for solving a broad class of stochastic optimization problems over the space of probability measures. Under mild assumptions (e.g., geodesic convexity and locally bounded stochastic gradients), we establish strong theoretical guarantees for our approach. In particular, we recover the convergence rate of optimally tuned versions of these algorithms up to logarithmic factors, in both nonsmooth and smooth settings. We then benchmark the performance of our methods against comparable existing approaches. Across a variety of tasks, our algorithms achieve similar performance to the optimal performance of existing algorithms, with no need to tune a step size parameter. 
\end{abstract}

\maketitle

\tableofcontents

\section{Introduction}
\label{sec:introduction}
Many tasks in computational statistics and machine learning can be viewed as optimization problems over a space of probability measures of the form
\begin{equation}
    \pi = \argmin_{\mu\in\mathcal{P}_2(\mathcal{X})} \mathcal{F}(\mu), \label{eq:wasserstein-optimization-intro}
\end{equation}
where $\mathcal{P}_2(\mathcal{X})$ is the set of Borel probability measures supported on $\mathcal{X}\subseteq\mathbb{R}^d$ with finite second moment, and $\mathcal{F}:\mathcal{P}_2(\mathcal{X})\rightarrow\mathbb{R}$ is a functional which maps probability measures to $\mathbb{R}$. Examples include sampling from a target probability measure \citep{wibisono2018sampling}, training mean-field neural networks \citep{suzuki2023convergence}, computing Wasserstein barycenters \citep{chewi2020gradient}, distributionally robust optimization \citep{kuhn2019wasserstein}, and marginal maximum likelihood estimation in latent variable models \citep{kuntz2023particle}. 

A general strategy for solving problems of this type is to simulate a time discretization of the \emph{gradient flow} of the objective functional over the space of probability measures, having equipped this space with a suitable metric \citep{ambrosio2008gradient}. Indeed, many well-known sampling algorithms can be interpreted through this lens. For instance, the unadjusted Langevin algorithm (ULA), a popular Markov chain Monte Carlo (MCMC) method, corresponds to the forward-flow discretization of the gradient flow of the Kullback-Leibler (KL) divergence with respect to the quadratic Wasserstein metric \citep{wibisono2018sampling,durmus2019analysis}. Similarly, the mean-field Langevin dynamics (MFLD), which arises in the analysis of mean-field neural networks, can be viewed as the forward-flow discretization of the Wasserstein gradient flow (WGF) of an entropy regularized objective functional \citep[e.g.,][]{suzuki2023convergence}. Meanwhile, Stein variational gradient descent (SVGD) can be interpreted as an explicit Euler discretization of the gradient flow of the KL divergence under a kernelized Wasserstein metric \citep{liu2017stein,duncan2019geometry}. 

A persistent challenge across all such methods is the selection of a suitable \emph{learning rate} or \emph{step size} parameter. This step size must be small enough to guarantee convergence to the minimizer of the optimization problem -- or a close approximation thereof -- while also sufficiently large to ensure convergence in a reasonable number of iterations. In principle, existing non-asymptotic convergence guarantees enable one to derive an optimal step size for certain algorithms (e.g., \citealp{korba2020nonasymptotic,salim2022convergence,sun2023convergence} for SVGD; \citealp{dalalyan2017further,dalalyan2017theoretical,durmus2017nonasymptotic,dalalyan2019user,durmus2019analysis,durmus2019high} for ULA). In practice, however, these optimal step sizes depend on properties of the unknown minimizer (see, e.g., Theorem 9 in \citealp{durmus2019analysis}; Corollary 6 in \citealp{korba2020nonasymptotic}), and thus provide limited practical utility.

To illustrate this point, let us consider ULA, a widely used sampling algorithm used to generate samples approximately distributed according to a target distribution $\pi$ on $\mathbb{R}^d$, which we assume has density $\pi(x)\propto e^{-U(x)}$ with respect to the Lebesgue measure.\footnote{In a slight abuse of notation, we use $\pi$ to denote both the target distribution and its density with respect to the Lebesgue measure.} Specifically, ULA defines a discrete-time, possibly nonhomogeneous Markov chain $(x_t)_{t\geq 0}$ according to the recursion
\begin{equation}
    x_{t+1} = x_t - \eta_t \nabla U(x_t) + \sqrt{2\eta_t}z_t, \quad\quad x_0\sim \mu_0 
\end{equation}
where $(z_t)_{t\geq 0}$ is an i.i.d. sequence of $d$-dimensional standard Gaussian random variables, and $(\eta_t)_{t\geq 0}$ is a positive sequence of step sizes. Under the assumptions that $\pi$ is log-concave (i.e., the potential $U$ is convex), and the step size $\eta_t:=\eta$ is constant for all $t$, then it is possible to establish the following non-asymptotic convergence rate (see Corollary \ref{corollary:average-bound-forward-flow})
\begin{equation}
        \mathrm{KL}(\bar{\mu}_T\|\pi) \leq \frac{1}{T}\left[\frac{W_2^2(\mu_{\frac{1}{2}},\pi)}{2\eta} + \frac{\eta}{2}\sum_{t=1}^{T} \int_{\mathbb{R}^d}\|\nabla U(x)\|^2\mathrm{d}\mu_t(x)\right], \label{eq:lmc-bound}
\end{equation}
where $\bar{\mu}_T=\frac{1}{T}\sum_{t=1}^T \mu_t$ denotes the average of the marginal distributions $(\mu_t)_{t\in[T]}$ of the algorithm iterates $(X_t)_{t\in[T]}$, where $[T]=1,\dots,T$; and where $\smash{\mu_{\frac{1}{2}}}$ denotes the marginal distribution of the algorithm after the first ``half-step'': $\smash{\mu_{\frac{1}{2}} = \mathrm{Law}(x_{\frac{1}{2}})}$, where $\smash{x_{\frac{1}{2}}=x_0 - \eta \nabla U(x_0)}$. By minimizing this upper bound with respect to the step size $\eta$, the \textit{optimal} fixed step size for ULA can easily be obtained as
\begin{equation}
    \eta_{*} = \frac{W_2(\mu_{\frac{1}{2}},\pi)}{\sqrt{\sum_{t=1}^{T} \int_{\mathbb{R}^d}\|\nabla U(x)\|^2\,\mathrm{d}\mu_t(x)}}. 
     \label{eq:optimal-average-bound-forward-flow-intro}
\end{equation}
Consequently, after substituting this step size back into the bound in \eqref{eq:lmc-bound}, the \textit{optimal} convergence rate for ULA (in the convex, nonsmooth setting) is given by
\begin{equation}
\mathrm{KL}(\bar{\mu}^{*}_T\|\pi) \leq \frac{1}{T} W_2(\mu_{\frac{1}{2}},\pi)\sqrt{\sum_{t=1}^{T} \int_{\mathbb{R}^d}\|\nabla U(x)\|^2\,\mathrm{d}\mu_t(x)}. \label{eq:optimal-rate-lmc-intro}
\end{equation}
Unfortunately, the optimal step size $\eta_{*}$ in \eqref{eq:optimal-average-bound-forward-flow-intro} cannot be computed in practice, since it relies on several unknown quantities: the Wasserstein distance $\smash{W_2(\mu_{\frac{1}{2}},\pi)}$; the sequence of gradients $\smash{(\nabla U(x_t))_{t\in[T]}}$; and the sequence of distributions $(\mu_t)_{t\in[T]}$. Thus, in practice, it is not possible to obtain the optimal upper bound in \eqref{eq:optimal-rate-lmc-intro}.

In reality, practitioners will often resort to a grid search, running the algorithm of choice (e.g., ULA) for multiple step sizes, and selecting the value which minimizes an appropriately chosen metric, such as the kernel Stein discrepancy (KSD) \citep{liu2016kernelized,chwialkowski2016kernel,gorham2019measuring}. This approach, however, is highly inefficient and does not guarantee an optimal choice of step size. Alternatively, methods based on \textit{coin betting} ideas from the parameter-free online learning literature \citep[e.g.,][]{orabona2014simultaneous,orabona2016coin,orabona2017training} have shown significant promise, but currently lack rigorous theoretical guarantees  \citep{sharrock2023coin,sharrock2023learning,sharrock2023tuning}. Finally, while there have been some efforts to semi-automate the design of effective step size schedules \citep{chen2016bridging,li2016preconditioned,kim2022stochastic,coullon2021efficient}, typically these approaches are still heavily dependent on an appropriate choice of certain hyperparameters, and have therefore not been widely adopted in practice. 

In this context, there remains a pressing need to develop effective and theoretically-grounded methods for solving optimization problems over the space of probability measures that do not require manual tuning of the step size parameter. This paper aims to fill this gap.

\subsection{Contributions}
 \label{sec:contributions}
 Our main contributions are summarized below:
 \begin{itemize}
     \item We introduce adaptive, tuning-free step size schedules for optimization algorithms obtained as time-discretizations of Wasserstein gradient flows (WGFs). In particular, we consider two time-discretizations of the WGF -- the forward Euler discretization and the forward-flow discretization -- and propose adaptive step size schedules tailored specifically to each case (see Section \ref{sec:adaptive-deterministic-step-size} and Section \ref{sec:adaptive-forward-euler}). 
     Collectively, we refer to our proposed step size schedule(s) as \textbf{\textsc{Fuse}} (\textbf{F}unctional \textbf{U}pper Bound \textbf{S}tep Size \textbf{E}stimator).
     \item We consider, as a special case, the forward-flow discretization of the WGF of the KL divergence, which corresponds to ULA. In this case, our proposed step size schedule reads 
    \begin{equation}
    \label{eq:dog-step-size-lmc-intro}
    \tcboxmath[colback=black!5,colframe=black!15,boxrule=0.4pt,
        arc=2pt,left=8pt, right=8pt, top=6pt, bottom=6pt]{\eta_t = \frac{\max\left[r_{\varepsilon},\max_{1\leq s\leq t}W_2(\mu_{\frac{1}{2}}, \mu_{s-\frac{1}{2}})\right]}{\left[{\sum_{s=1}^t \int ||\nabla U(x)||^2 \,\mathrm{d}\mu_s(x)}\right]^{\frac{1}{2}}}},
    \end{equation}
    where $r_{\varepsilon}>0$ is a small positive constant. Thus, the step size is equal to the maximum Wasserstein-2 distance between $\smash{\mu_{\frac{1}{2}}}$, the marginal distribution of the algorithm after the first half-step, and $\smash{(\mu_{s-\frac{1}{2}})_{1\leq s \leq t}}$, the marginal distributions of the algorithm after each subsequent half-step, divided by the square root of the sum of the squared $\smash{(L^2(\mu_s))_{1\leq s \leq t}}$ norms of the gradients of the potential. 
     \item We establish rigorous theoretical guarantees for the proposed step size schedules. In particular, under mild assumptions on the objective functional -- e.g., geodesic convexity and locally bounded (stochastic) gradients -- we recover the convergence rate of optimally tuned versions of the forward-flow discretization (see Sections \ref{sec:forward-flow-nonsmooth} - \ref{sec:forward-flow-smooth}) and the forward Euler discretization (see Sections \ref{sec:forward-euler-nonsmooth} - \ref{sec:forward-euler-smooth}) up to logarithmic factors, in both the smooth and nonsmooth settings. 
     \item Based on our theoretical results, we introduce practical, particle-based, step-size-free algorithms for optimization problems on the Wasserstein space, with a particular focus on the sampling problem. Our proposed algorithms include tuning-free variants of (parallel-chain) ULA (see Algorithm \ref{alg:ula}), in which case our practical step size schedule reads 
     \begin{equation}
     \tcboxmath[colback=black!5,colframe=black!15,boxrule=0.4pt,
        arc=2pt,left=8pt, right=8pt, top=6pt, bottom=6pt]{
      \hat{\eta}^n_t = \frac{\max\left[r_{\varepsilon},\max_{1\leq s\leq t} \big[\big(\frac{1}{n}\sum_{i=1}^n \|x_{\frac{1}{2}}^{i} - x_{s-\frac{1}{2}}^{i}\|^2\big)^{\frac{1}{2}} \big]\right]}{\left[\sum_{s=1}^t\left(\frac{1}{n}\sum_{i=1}^n  \|\nabla U(x_s^i)\|^2\right)\right]^{\frac{1}{2}}}
      },
      \end{equation}
      as well as its stochastic counterpart, stochastic gradient Langevin dynamics (SGLD) (see Algorithm \ref{alg:sgld}). We also obtain step-size-free variants of SVGD (see Algorithm \ref{alg:svgd}) and its mean-field extension, variational gradient descent (VGD); as well as mean-field Langevin dynamics (MFLD).
     \item We rigorously benchmark the performance of our methods. Over a wide range of tasks, our approach achieves comparable or superior performance to the \emph{optimal} performance of existing approaches, with no need to tune a step size parameter. In particular, our approach has a \emph{very} mild dependence on the initial movement size $r_{\varepsilon}$, and thus in practical terms can be considered entirely tuning-free.
 \end{itemize}
  
\subsection{Related Work}
Our work lies at the intersection of two distinct lines of research: parameter-free optimization methods on Euclidean spaces, and optimization over the space of probability measures. In this section, we provide a detailed survey of both of these areas.

\label{sec:related-work}
\subsubsection{Adaptive and Parameter-Free Optimization}
\label{sec:parameter-free-opt-euclidean}
There is a significant body of work devoted to both the development and analysis of \textit{adaptive} and \textit{parameter-free} (stochastic) optimization algorithms in Euclidean spaces. Below, we review some of the most prominent of these methods. 

\paragraph{Adaptive Methods}  Traditionally, the term \textit{adaptive} refers to algorithms which can adapt (in the sense of achieving the same or almost the same convergence guarantees) to some property of the problem at hand (e.g., subgradient bound, smoothness constant), without prior knowledge of this property. This meaning has somewhat muddied over time \citep[e.g.,][Section 4.3]{orabona2019modern}. In particular, \textit{adaptive} is now often used simply to refer to algorithms with (coordinate-wise) step size schedules which are a function of the past (stochastic) gradients; see, e.g., Section 1 in \citet{li2019convergence}. 

There are a number of different (coordinate-wise) adaptive methods. Typically, these methods can adapt (i.e., do not require knowledge of) the (sub)gradient bound or smoothness constant or strong convexity constant, but still require knowledge of the (initial) distance to the optimizer (or the domain diameter). 

Such methods include Adagrad -- introduced in the online learning community \citep{mcmahan2010adaptive,duchi2011adaptive} and analyzed in a number of subsequent works \citep[e.g.,][]{li2019convergence,ward2020adagrad,kavis2022high,attia2023sgd} -- and its variants, including Adagrad-Norm \citep{streeter2010less,faw2022power,kavis2022high}, RMSProp \citep{hinton2012neural}, Adam \citep{kingma2014adam,reddi2018convergence,chen2018convergence}, AdamW \citep{loshchilov2019decoupled}, AdaDelta \citep{zeiler2012adadelta}, AMSGrad \citep{reddi2018convergence,yang2024two}, and Shampoo \citep{gupta2018preconditioned}, amongst others. For some other related works, see also \citet{xu2012new,levy2018online,kavis2019unixgrad,luo2019adaptive,li2020high,malitsky2020adaptive,ene2021adaptive}.

\paragraph{Parameter-Free Methods} In this paper, we refer to an algorithm as \emph{parameter-free} if it does not require the specification of any tuning parameters. This follows other recent uses of this terminology \citep[see, e.g.][Section 2]{khaled2023dowg}, but differs from more precise definitions originating in the online learning literature \citep[e.g.,][Remark 1]{orabona2023normalized}. 

There are a number of traditional parameter-free methods for both deterministic and stochastic optimization. In the deterministic setting, line search \citep{beck2009fast,nesterov2015universal,grimmer2024optimal} achieves the optimal rate of convergence in smooth and nonsmooth settings with only an extra logarithmic factor. However, this approach is often expensive and outperformed in practice \citep[e.g.,][]{malitsky2020adaptive}. 

Meanwhile, Polyak's step size \citep[Section 5.3]{polyak1987introduction} recovers the optimal convergence rate up to a logarithmic factor for smooth, nonsmooth, and strongly convex optimization \citep{hazan2019revisiting}, but requires knowledge of the optimum function value. In combination with an adaptive re-estimation mechanism, this approach also achieves the optimal rate of convergence, up to a logarithmic factor, when $f$ is Lipschitz \citep{hazan2019revisiting}. The \emph{doubling trick} can also be used to render (stochastic) optimization algorithms parameter-free. This approach originates in the online learning literature \cite[e.g.,][]{cesa1997use,cesa2006prediction,hazan2007online,shalev2012online,streeter2012no}, and has since also been studied in the context of stochastic optimization \citep[e.g.,][]{hazan2014beyond}.

\paragraph{Parameter-Free Methods in Online Learning.} In the online learning literature, \textit{parameter-free} algorithms refer specifically to those which adapt to the unknown (initial) distance to the solution, but may still assume knowledge of other problem parameters \citep[e.g.,][Remark 1]{orabona2023normalized}. 

There are a number of algorithms which achieve such guarantees, including those based on coin betting \citep[e.g.,][]{orabona2014simultaneous,orabona2016coin,orabona2017training,chen2022better} and exponentiated gradients \citep{streeter2012no,orabona2013dimension}. Additional related works include \citet{chaudhuri2009parameter,mcmahan2014unconstrained,orabona2014simultaneous,cutkosky2018black,jun2019parameter,orabona2021parameter,orabona2019modern}.

In the deterministic setting, one can obtain optimization algorithms which do not require knowledge of \emph{any} tuning parameters (e.g., the distance to the optimum or the subgradient bound) by combining parameter-free online learning algorithms \citep[e.g.,][]{streeter2012no,orabona2016coin} with the normalization techniques described in \citet{levy2017online}. This is true for both Lipschitz functions \citep{orabona2021parameter} and smooth functions \citep{orabona2023normalized}.

\paragraph{Recent Parameter-Free Methods.} Recently, \citet{carmon2022making} proposed a novel bisection method which recovers the optimal rate of convergence for (stochastic) gradient descent in nonsmooth convex, smooth convex, and strongly convex settings up to a \emph{double}-logarithmic factor. This is better than any bound achievable using parameter-free regret bounds from the online learning literature, which are translated into stochastic convex optimization bounds via online-to-batch conversion \citep[e.g.,][]{hazan2016introduction}. However, this method requires resets (i.e., restarting the optimization procedure multiple times), which can be very expensive {in practice}, and disregards potentially useful information gained from previous runs. 

Building on this paper, several groups have proposed practical parameter-free optimization algorithms for both the deterministic \citep{defazio2023learning,khaled2023dowg,mishchenko2024prodigy} and stochastic \citep{ivgi2023dog} settings. \citet{defazio2023learning,mishchenko2024prodigy} proposed \emph{D-Adaptation} and \emph{Prodigy}, which have guarantees for the deterministic, nonsmooth setting, and achieve highly competitive results on modern machine learning problems. 

Most relevant to our work, \citet{ivgi2023dog} introduced \emph{Distance over Gradients} (DoG), which recovers the optimal convergence rate of SGD up to a logarithmic factor in high-probability in both the nonsmooth and smooth settings, assuming the iterates remain stable. A `tamed' version of their algorithm (T-DoG), whose step sizes are smaller by a logarithmic factor, can achieve a similar guarantee without an a priori assumption on iterate stability. Inspired by this approach, \citep{khaled2023dowg} proposed \emph{Distance over Weighted Gradients} (DoWG), a variant of DoG which uses adaptively weighted gradients. This algorithm matches the performance of gradient descent (in the deterministic setting) in both smooth and nonsmooth settings, up to an additional logarithmic factor. There have since been various extensions of these methods \citep[e.g.,][]{dodd2024learning,kreisler2024accelerated}.

\subsubsection{Optimization on the Space of Probability Measures}
\label{sec:lit-review-optimization-on-space-of-prob-measures}
There is also a vast literature regarding the use of gradient-based methods for solving optimization problems over the space of probability measures.

\paragraph{Methods} Undoubtedly, the most well-studied and widely used of these methods are \emph{Wasserstein gradient flows} (WGFs), which define a notion of gradient flow over the space of twice-integrable probability measures equipped with the quadratic Wasserstein metric \citep[e.g.,][]{ambrosio2008gradient,santambrogio2017euclidean}. More precisely, WGFs arise as the solutions to continuity equations driven by velocity fields that correspond to the direction of steepest descent of the chosen objective functional (e.g., relative entropy, potential energy) with respect to the Wasserstein geometry. Under mild conditions, WGFs admit a unique solution \citep[e.g.,][Theorem 11.1.4]{ambrosio2008gradient}, and decrease the value of the objective functional along the trajectory of the gradient flow. Thus, under additional structural assumptions (e.g., geodesic convexity), they converge to the infimum of the objective functional \citep[e.g.,][Chapter 11]{ambrosio2008gradient}. 

WGFs were first identified in the seminal work of \citet{jordan1998variational}, who showed that the Fokker-Planck equation describing the time evolution of the marginal law of the overdamped Langevin diffusion could be viewed as the WGF of $\mathrm{KL}(\cdot\|\pi)$, the KL divergence between the current distribution and the target distribution $\pi$. Building on this, \citet{otto2001geometry} developed a formal Riemannian calculus on the Wasserstein space, interpreting probability measures as points on an infinite-dimensional manifold and introducing a differential structure that enabled the use of tools from Riemannian geometry to analyze WGFs. This perspective made it possible to describe various diffusion-type PDEs -- such as the porous medium equation -- as gradient flows of natural energy functionals. Other significant early contributions include the work of \citet{villani2003topics,villani2008optimal} and \citet{ambrosio2008gradient}, who developed a rigorous theoretical framework for studying WGFs.

Subsequent work has focused on various aspects of WGFs, including theoretical generalizations; algorithmic improvements; and practical implementations. One key emphasis has been the development of effective time and space discretizations, which are necessary for the practical implementation of WGFs. Several different papers study different time-discretizations of WGFs, each leading to a distinct class of numerical algorithms. These include the backward Euler scheme, which results in the JKO or proximal point algorithm \citep{jordan1998variational}; the forward Euler scheme, which results in Wasserstein (sub)gradient descent \citep[e.g.,][]{guo2022online,wang2022projected}; the proximal gradient scheme, which results in the Wasserstein proximal gradient algorithm \citep{salim2020wasserstein}; and the forward-flow scheme which, when the objective functional coincides with $\mathrm{KL}(\cdot\|\pi)$, results in ULA \citep{wibisono2018sampling}. 

Spatial discretizations are equally essential for practical implementation of WGFs. Existing approaches include finite-element or finite-volume methods, which discretize the PDE underlying the WGF \citep[e.g.,][]{carrillo2015finite,benamou2016discretization,cances2020variational,li2020fisher,fu2023high}. Other popular methods include particle-based methods, which approximate measures via empirical distributions \citep[e.g.,][]{carrillo2019blob}. More recently, input convex neural networks (ICNNs) have been used to parametrize convex potentials in WGFs, offering a scalable, mesh-free method suitable for simulating high-dimensional flows \citep[e.g.,][]{mokrov2021large,alvarezmelis2022optimizing}.

Another line of work considers algorithmic extensions and modifications of the standard WGF. For example, various authors consider the use of metrics other than the quadratic Wasserstein metric, leading to a different geometric structure over the space of probability measures, and thus to different gradient flows. For example, use of a kernelized Wasserstein metric known as the \emph{Stein metric} leads to the Stein variational gradient flow \citep{liu2017stein,duncan2019geometry} and, after time-discretization, the popular SVGD algorithm \citep{liu2016stein}. Meanwhile, \citet{garbuno2020affine,garbuno2020interacting} consider the gradient flow induced by the \emph{Kalman-Wasserstein metric}, resulting in the affine-invariant Langevin dynamics (ALDI). Other authors have proposed the use of the \emph{Fisher-Rao metric}, yielding the so-called \textit{birth-death} dynamics \citep[e.g.,][]{lu2023birth}. Finally, \citet{lu2019accelerating,he2024regularized} consider combinations of different metrics; e.g., Wasserstein and Fisher-Rao \citep{lu2019accelerating}, or Wasserstein and Stein \citep{he2024regularized}.

Among other extensions of the standard WGF are so-called \textit{mirror} or \textit{mirrored} WGFs, inspired by mirror descent \citep[e.g.,][]{hsieh2018mirrored,chewi2020exponential,zhang2020wasserstein,ahn2021efficient,li2022mirror,sharrock2023learning}; accelerated WGFs, which are obtained by the introduction of an additional momentum or velocity variable \citep[e.g.,][]{wang2022accelerated,chen2025accelerating,stein2025accelerated}; and WGFs on submanifolds of the Wasserstein space, e.g., the Bures-Wasserstein space of Gaussian distributions \citep{chewi2020gradient,diao2023forward}. 

Despite many significant methodological developments, there remains a notable lack of work on tuning-free algorithms for optimization problems on the space of probability measures. While Euclidean optimization has seen considerable progress in this area (see Section \ref{sec:parameter-free-opt-euclidean}), analogous developments in the space of probability measures have remained sparse. There have been some recent efforts to extend parameter-free techniques such as \textit{coin betting} \citep[e.g.,][]{orabona2016coin} to the space of probability measures, with the introduction of so-called \emph{coin sampling} algorithms \citep{sharrock2023coin,sharrock2023learning,sharrock2023tuning}. Indeed, these algorithms have demonstrated comparable performance to that of \emph{optimally tuned} existing algorithms across a wide range of tasks -- yet they require no manual tuning of a step size or learning rate parameter. Despite their promise, however, this approach faces significant theoretical challenges. For example, due to the additional complexity introduced by the geometry of the space of probability measures, existing theoretical tools used to analyze the convergence properties of coin betting algorithms in Euclidean spaces cannot easily be applied \citep[see, e.g.,][Appendix A]{sharrock2023coin}. As a result, coin sampling algorithms still lack rigorous theoretical guarantees (e.g., non-asymptotic convergence rates), even under strong assumptions on the target measure (e.g., log-concavity).

\paragraph{Applications} Over the past two decades, WGFs have found application in a wide variety of problems across computational statistics, machine learning, and applied mathematics. Below, we survey two prominent examples, which we will later revisit in our numerical experiments.

\begin{example1}[\textsc{Sampling from a Target Probability Distribution}]
\label{example-1}
The task of sampling from a probability distribution $\pi$, with density $\pi(x) \propto e^{-U(x)}$ w.r.t. the Lebesgue measure, is of fundamental importance to Bayesian inference \citep{robert2004monte,stuart2010inverse,fearnhead2025scalable}, machine learning \citep{neal1996bayesian,andrieu2003introduction,welling2011bayesian,wilson2020bayesian}, molecular dynamics \citep{krauth2006statistical,lelievre2016partial,leimkuhler2016efficient}, and scientific computing \citep{mackay2003information,liu2009monte}. 

It is now well known that the sampling problem can be reformulated as an optimization problem over the space of measures \citep[e.g.,][]{jordan1998variational,wibisono2018sampling,durmus2019analysis}. In this setting, one views the target $\pi$ as the solution of
\begin{equation}
 \pi = \argmin_{\mu\in\mathcal{P}_2(\mathbb{R}^d)} \mathcal{F}(\mu), \label{eq:sampling-opt}
\end{equation} 
where $\mathcal{P}_2(\mathbb{R}^d)$ denotes the set of probability measures $\{\mu:\int_{\mathbb{R}^d}\|x\|^2\mu(\mathrm{d}x)<\infty\}$ with finite second moment, and $\mathcal{F}:\mathcal{P}_2(\mathbb{R}^d)\rightarrow\mathbb{R}$ is a functional which is uniquely minimized at $\pi$. Perhaps the most popular choice of functional in this context is $\mathcal{F}(\mu) =\mathrm{KL}(\mu \| \pi)$, the KL divergence of $\mu$ with respect to $\pi$, defined according to
\begin{equation}
\mathrm{KL}(\mu \| \pi) = \left\{ 
\begin{array}{lll} 
\displaystyle{\int_{\mathbb{R}^d} \log \left(\frac{\mathrm{d}\mu}{\mathrm{d}\pi}(x)\right)\mu(\mathrm{d}x)} & \hspace{-1mm}, \hspace{-1mm}& \mu \ll \pi, \\[3mm]
+\infty &\hspace{-1mm}, \hspace{-1mm}& \text{otherwise,} \\[1mm]
\end{array}
\right.
\label{eq:kl-divergence}
\end{equation}
where $\smash{\frac{\mathrm{d}\mu}{\mathrm{d}\pi}}$ denotes the Radon-Nikodym derivative of $\mu$ with respect to $\pi$. For this choice of $\mathcal{F}$, it is straightforward to verify that $\mathcal{F}(\mu)=\mathrm{KL}(\mu \| \pi)\geq 0$ for all $\mu\in\mathcal{P}_2(\mathbb{R}^d)$, with $\mathcal{F}(\mu)=\mathrm{KL}(\mu \| \pi)=0$ if and only if $\mu=\pi$. Thus, in particular, $\mathcal{F}(\mu)=\mathrm{KL}(\mu \| \pi)$ is indeed uniquely minimized at the target measure $\pi$. 

As noted earlier, many popular sampling algorithms can be viewed as discretizations of gradient flows of the KL divergence over the space of probability measures. These include classical MCMC algorithms such as ULA \citep[e.g.,][]{wibisono2018sampling,durmus2019analysis}, as well as more recent ParVI algorithms such as SVGD \citep{liu2016stein}, Laplacian adjusted Wasserstein gradient descent (LAWGD) \citep{chewi2020svgd}, and the various particle-based algorithms described in \citet{chen2018unified,liu2019understanding,liu2019understandingmcmc}. For different choices of the objective functional, one also obtains various other sampling algorithms, including the Wasserstein proximal gradient algorithm \citep{salim2020wasserstein}, kernel Stein discrepancy (KSD) descent \citep{korba2021kernel}, and mollified interaction energy descent (MIED) \citep{li2023sampling}. 

The connections between sampling and optimization have, in recent years, given rise to a flourishing literature. For a more comprehensive overview of this topic, we refer to any one of a number of recent survey papers \citep[e.g.,][]{chen2023gradient,chen2023sampling,trillos2023optimization,chewi2024statistical,chewi2024log,sharrock2025sampling}.
\end{example1}

\begin{example2}[\textsc{Training Neural Networks in the Mean-Field Regime}]
Over the past decade, neural networks have achieved remarkable success across a wide range of domains, including computer vision, natural language processing, and scientific computing \citep[e.g.,][]{lecun2015deep}. However, obtaining rigorous mathematical theory to explain their empirical success remains a significant challenge. In general, learning the parameters of a neural network involves solving a high-dimensional, highly non-convex optimization problem via stochastic gradient descent, for which convergence results are very hard to obtain \citep{danilova2022recent}.

Recently, significant theoretical progress has been made on this front by considering neural networks in a suitable scaling limit: the so-called \emph{mean-field regime} \citep[e.g.,][]{chizat2018global,mei2018mean,sirignano2020mean,sirignano2020meanfield,de2020quantitative,rotskoff2022trainability}. Specifically, consider the task of learning a neural network with $M$ trainable neurons, parameterized as 
\begin{equation}
    h_{\Theta}(z) := \frac{1}{M}\sum_{m=1}^M h_{\theta_m}(z),
\end{equation}
where each neuron $h_{\theta_{m}}(\cdot)$ depends on trainable parameters $\theta_m\in\mathbb{R}^d$, and where $\Theta = \{\theta_{m}\}_{m=1}^M$. For example, we might have $h_{\theta_m}(z) = \sigma(w_{m}^{\top}z + b_m)$, where $w_m\in\mathbb{R}^{d-1}$, $b_m\in\mathbb{R}$, and $\theta_m = (w_m,b_m)^{\top}$. Under suitable conditions, as the number of neurons $M\rightarrow\infty$, we obtain the \emph{mean-field limit}
\begin{equation}
    h_{\mu}(z) = \int h_{\theta}(z)\mu(\mathrm{d}\theta),
\end{equation}
where $\mu\in\mathcal{P}_2(\mathbb{R}^d)$ represents the probability distribution of the weights. Suppose now that we are given training data $\{(z_i,y_i)\}_{i=1}^n \in \mathbb{R}^{d-1}\times\mathbb{R}$. We can then define the empirical risk of $h_{\mu}$ as $U(\mu) = \frac{1}{n}\sum_{i=1}^n \ell (h_{\mu}(z_i),y_i)$ for a loss function $\ell:\mathbb{R}\times\mathbb{R}\rightarrow\mathbb{R}$. To train our neural network, we would like to minimize a regularized version of the empirical risk, viz
\begin{equation}
    \mathcal{F}(\mu) = \frac{1}{n}\sum_{i=1}^n \ell (h_{\mu}(z_i),y_i) + \lambda \int_{\mathbb{R}^d} r(\theta) \mu(\mathrm{d}\theta), \label{eq:nn-loss}
\end{equation}
where $\lambda\in\mathbb{R}_{+}$ is the regularization strength, and $r:\mathbb{R}^d\rightarrow\mathbb{R}$ is a regularization term, e.g., $r(\theta) = \|\theta\|^2$. In other words, training a mean-field neural network is equivalent to solving the following optimization problem over the space of probability measures:
\begin{equation}
    \mu_{\star} = \argmin_{\mu\in\mathcal{P}_2(\mathbb{R}^d)} \mathcal{F}(\mu),
\end{equation}
where $\mathcal{F}:\mathcal{P}_2(\mathbb{R}^d)\rightarrow \mathbb{R}$ is as defined in \eqref{eq:nn-loss}. By lifting the task of learning a neural network to an infinite-dimensional optimization problem on the space of probability measures, it is possible to leverage the convexity of the loss function to establish the global convergence of gradient-based optimization methods \citep[e.g.,][]{mei2018mean,chizat2018global,rotskoff2018parameters,araujo2019mean,mei2019mean,javanmard2020analysis,sirignano2020meanfield,sirignano2020mean,rotskoff2022trainability}.

One popular method for solving this optimization problem is the so-called \emph{mean-field Langevin dynamics} (MFLD). This can be viewed as the continuous-time limit of a noisy gradient descent algorithm on the infinite-width neural network, in which the additional noise encourages exploration and enables global convergence \citep[e.g.,][]{mei2018mean,hu2021meanfield}. Alternatively, these dynamics can be viewed as the WGF of an entropy-regularized version of the objective function in \eqref{eq:nn-loss} \citep[e.g.][]{mei2018mean,hu2021meanfield,chen2022uniform,chizat2022meanfield,nitanda2022convex,suzuki2023convergence}.

Based on this perspective, there has been significant progress towards understanding the theoretical properties of the MFLD. \citet{mei2018mean,hu2021meanfield,chizat2022meanfield,nitanda2022convex,chizat2022meanfield} all show exponential convergence of the MFLD to its unique invariant distribution, i.e., the minimizer of the entropy-regularized energy functional. Meanwhile, \citet{chen2022uniform,suzuki2023uniform} establish uniform-in-time propagation of chaos, i.e., global upper bounds for the distance between the empirical distribution of the interacting particle system and the distribution of the mean-field dynamics. \citet{nitanda2022convex} establish the convergence rate of a time-discretization of the MFLD. Finally, \citet{suzuki2023convergence} provide convergence rates for a space-time discretization of the MFLD. Many of these results have since been extended to the underdamped case, including exponential convergence to the unique invariant measure \citep{kazeykina2020ergodicity,chen2024uniform}, uniform-in-time propagation of chaos \citep{chen2024uniform,fu2023mean}, and convergence rates for a space-time discretization \citep{fu2023mean}.
\end{example2}

There are numerous other examples of optimization problems over the space of probability measures, which we do not have time to survey in detail here. We mention, amongst others: generalized and post-Bayesian inference \citep[e.g.,][]{knoblauch2022optimization,chazal2025computable,mclatchie2025predictively}, variational inference (VI) including mean-field VI \citep{yao2022mean,lacker2023independent,jiang2023algorithms} and Gaussian VI \citep{lambert2022variational,diao2023forward}, trajectory inference in stochastic differential equations (SDEs) \citep{chizat2022trajectory}, multi-objective optimization \citep{ren2024multi}, computing Wasserstein barycenters \citep{chewi2020gradient}, non-parametric maximum likelihood estimation \citep{yan2024learning}, mathematical modelling of transformers \citep{geshkovski2025mathematical}, reinforcement learning and policy optimization \citep{zhang2018policy}, distributionally robust optimization \citep[e.g.,][]{mohajerin2018data,kuhn2019wasserstein,rahimian2019distributionally,gao2023distributionally}, generative modelling \citep[e.g.,][]{arbel2019maximum,cheng2024convergence,haviv2024wasserstein,galashov2024deep}, sparse deconvolution \citep{chizat2022sparse}, lossy compression \citep[e.g.,][]{yang2023estimating}, and the modelling of granular media  \citep[e.g.,][]{otto2001geometry}. In principle, the methodology developed in this paper could be applied to any one of these domains.

\subsection{Paper Organization}
\label{sec:organization}
The remainder of this paper is organized as follows. In Section \ref{sec:preliminaries}, we recall some basic definitions relating to optimization on the Wasserstein space (e.g., the Wasserstein metric, geodesic convexity, subdifferential calculus in the Wasserstein space). In Section \ref{sec:optimization-wasserstein-space}, we introduce standard methods for performing optimization on this space (e.g., WGFs and their time-discretizations).  In Section \ref{sec:main-results}, we introduce our dynamic, tuning-free step size schedule, and present our main theoretical results. In Section \ref{sec:algorithms}, we propose practical, particle-based sampling algorithms based on our theoretical results. In Section \ref{sec:numerical-results}, we rigorously benchmark our approach against existing algorithms across several different tasks (e.g., sampling from a target probability distribution, training a mean-field neural network). Finally, in Section \ref{sec:conclusions}, we provide some concluding remarks, including possible directions for future work.

\section{Preliminaries}
\label{sec:preliminaries}
In this section, we recall some basic concepts relating to optimization on the Wasserstein space. These include the definitions of the Wasserstein metric, geodesic convexity, and subdifferential calculus on the Wasserstein space. For a more comprehensive coverage of these topics, we refer to the excellent books of \citet{ambrosio2008gradient} and \citet{villani2008optimal}.

\subsection{General Notation}
Let $\mathcal{P}_2(\mathbb{R}^d)$ denote the set of probability measures on $\mathbb{R}^d$ with finite second moment: $\int_{\mathbb{R}^d} \|x\|^2 \mu(\mathrm{d}x)<\infty$. In addition, let $\mathcal{P}_{2,\mathrm{ac}}(\mathbb{R}^d)$ denote the subset of $\mathcal{P}_{2}(\mathbb{R}^d)$ consisting of probability measures which are absolutely continuous with respect to the Lebesgue measure, $\mathcal{L}^d$. For any $\mu\in\mathcal{P}_{2}(\mathbb{R}^d)$, let $L^{2}(\mu)$ denote the set of measurable functions $f:\mathbb{R}^d\rightarrow\mathbb{R}^d$ such that $\int_{\mathbb{R}^d}\|f(x)\|^2\mu(\mathrm{d}x)<\infty$. We will write $\|\cdot\|_{L^2(\mu)}$ and $\langle\cdot,\cdot\rangle_{L^2(\mu)}$ to denote, respectively, the norm and the inner product of this space.

\subsection{The Wasserstein Space}
Given a probability measure $\mu\in\mathcal{P}_2(\mathbb{R}^d)$ and a measurable function $T:\mathbb{R}^d\rightarrow\mathbb{R}^d$, we write $T_{\#}\mu$ for the pushforward measure of $\mu$ under $T$, that is, the measure such that $T_{\#}\mu(B)=\mu(T^{-1}(B))$ for all Borel measurable $B\in\mathcal{B}(\mathbb{R}^d)$. For every $\mu,\nu\in\mathcal{P}_2(\mathbb{R}^d)$, let $\Gamma(\mu,\nu)$ be the set of couplings (or transport plans) between $\mu$ and $\nu$, defined as  
\begin{equation}
\Gamma(\mu,\nu) = \{\gamma\in\mathcal{P}_2(\mathbb{R}^d\times\mathbb{R}^d): Q^{1}_{\#}\gamma=\mu, Q^{2}_{\#}\gamma=\nu\},
\end{equation}
where $Q^{1}$ and $Q^{2}$ denote the projections onto the first and second components of $\mathbb{R}^d\times\mathbb{R}^d$. The Wasserstein $2$-distance between $\mu$ and $\nu$ is then defined according to
\begin{equation}
W_2^2(\mu,\nu) = \inf_{\gamma \in\Gamma(\mu,\nu)}  \int_{\mathbb{R}^d\times\mathbb{R}^d}\|x-y\|^2 \gamma(\mathrm{d}x,\mathrm{d}y). \label{eq:wasserstein}
\end{equation}
We denote by $\Gamma_{o}(\mu,\nu)$ the set of couplings which attain the minimum in \eqref{eq:wasserstein}; and refer to any $\gamma\in\Gamma_{o}(\mu,\nu)$ as an \emph{optimal coupling} or \emph{optimal transport plan} between $\mu$ and $\nu$. By \citet[Proposition 2.1]{villani2001limites}, the set $\Gamma_{o}(\mu,\nu)$ is always non-empty. That is, there always exists an optimal transport plan $\gamma\in\Gamma_{o}(\mu,\nu)$.

Under appropriate regularity conditions, there exists a \emph{unique} optimal coupling $\gamma_{*}\in\Gamma(\mu,\nu)$ \citep[e.g.,][Section 6.2.3]{ambrosio2008gradient}. In particular, suppose $\mu\in\mathcal{P}_{2,\mathrm{ac}}(\mathbb{R}^d)$ and $\nu\in\mathcal{P}_2(\mathbb{R}^d)$. Then there exists a unique optimal transport plan $\gamma_{*}\in\Gamma(\mu,\nu)$, i.e., $\Gamma_{o}(\mu,\nu) = \{\gamma_{*}\}$. Moreover, this optimal transport plan is given by
\begin{equation}
\gamma_{*} = \smash{(\boldsymbol{\mathrm{id}}\times \boldsymbol{t}_{\mu}^{\nu})_{\#}\mu},
\end{equation}
where $\boldsymbol{\mathrm{id}}:\mathbb{R}^d\rightarrow \mathbb{R}^d$ is the identity map, and $\boldsymbol{t}_{\mu}^{\nu}:\mathbb{R}^d\rightarrow\mathbb{R}^d$ is a measurable function known as the \emph{optimal transport map} (e.g., \citealp{brenier1991polar}; \citealt{villani2003topics}, Theorem 2.12; \citealt{ambrosio2008gradient}, Theorem 6.2.4; \citealp{gigli2011on}). It follows that $(\boldsymbol{t}_{\mu}^{\nu})_{\#}\mu = \nu$ and \begin{equation}
    W_{2}^2(\mu,\nu) = \int_{\mathbb{R}^d}\|x-y\|^2\gamma_{*}(\mathrm{d}x,\mathrm{d}y)= \int_{\mathbb{R}^d} \|x-\boldsymbol{t}_{\mu}^{\nu}(x)\|^2\mathrm{d}\mu(x).
\end{equation}

We also have a \textit{dynamic} formulation of the Wasserstein distance due to \citet{benamou2000computational}. In particular, suppose $\mu,\nu\in\mathcal{P}_{2}(\mathbb{R}^d)$. Then it holds that \citep[Proposition 1.1]{benamou2000computational}
\begin{equation}
    W_2^2(\mu,\nu) = \inf\left\{\int_0^1 \|v_t\|^2_{L^2(\mu_t)}\mathrm{d}t~\big|~ \partial_{t}\mu_t + \nabla \cdot(v_t\mu_t) = 0, \mu_0 = \mu, \mu_1 = \nu\right\}. \label{eq:benamou-brenier}
\end{equation}

The Wasserstein distance $W_2$ is a distance over $\mathcal{P}_2(\mathbb{R}^d)$. Thus $(\mathcal{P}_2(\mathbb{R}^d), W_2)$ is indeed a metric space of probability measures \citep[e.g.,][Theorem 8.3]{ambrosio2021lectures}, known as the Wasserstein space.

\paragraph{Riemannian Interpretation of the Wasserstein Space} 
Following \citet{otto2001geometry}, the Wasserstein space can be (formally) identified as an infinite-dimensional Riemannian manifold. In this geometric interpretation, one identifies the tangent space to $\mathcal{P}_2(\mathbb{R}^d)$ at $\mu$ as 
\begin{equation}
    \mathcal{T}_{\mu}\mathcal{P}_2(\mathbb{R}^d) = \overline{\{\nabla \psi: \psi\in C_c^{\infty}(\mathbb{R}^d)\}}^{L^2(\mu)} \subseteq L^2(\mu), \label{eq:tangent-space}
\end{equation}
where $C_c^{\infty}(\mathbb{R}^d)$ denotes the space of smooth, compactly supported scalar functions $\psi: \mathbb{R}^d \rightarrow \mathbb{R}$, and where $\smash{\overline{\{\cdot\}}^{L^2(\mu)}}$ denotes the $L^2(\mu)$ closure. This tangent space is endowed with the usual inner product from $L^2(\mu)$, namely, 
\begin{equation}
    \langle \nabla \psi_{1}, \nabla \psi_{2}\rangle_{\mu} = \int_{\mathbb{R}^d} \langle \nabla \psi_1(x), \nabla \psi_2(x)\rangle \mathrm{d}\mu(x).
\end{equation}

\subsection{Geodesic Convexity}
\label{sec:geo-convex}
Let $\mu,\nu\in\mathcal{P}_2(\mathbb{R}^d)$. We define a constant speed geodesic between $\mu$ and $\nu$ as a curve $(\mu^{\mu\rightarrow\nu}_t)_{t\in[0,1]}$ such that $\mu_0 = \mu$, $\mu_1 = \nu$, and $W_2(\mu_{s},\mu_{t}) = |t-s|W_2(\mu,\nu)$ for all $s,t\in[0,1]$. If $\boldsymbol{t}_{\mu}^{\nu}$ is the optimal transport map between $\mu$ and $\nu$, then a constant speed geodesic is given by \citep[e.g.,][Sec. 7.2]{ambrosio2008gradient}
\begin{equation}
\mu_{t}^{\mu\rightarrow\nu} = \left((1-t)\boldsymbol{\mathrm{id}} + t\, \boldsymbol{t}_{\mu}^{\nu}\right)_{\#}\mu.
\end{equation}
This constant speed geodesic is sometimes referred to as McCann's interpolant \citep{mccann1997convexity}. 
Let $\mathcal{F}:\mathcal{P}_2(\mathbb{R}^d)\rightarrow(-\infty,\infty]$. For $m\geq 0$, we say that $\mathcal{F}$ is $m$-geodesically convex, or $m$-displacement convex, if, for any $\mu,\nu\in\mathcal{P}_2(\mathbb{R}^d)$ and for all $t\in[0,1]$, 
\begin{equation}
\mathcal{F}(\mu_{t}^{\mu\rightarrow\nu}) \leq (1-t)\mathcal{F}(\mu) + t\mathcal{F}(\nu) - \frac{m}{2}t(1-t)W_2^2(\mu,\nu).
\end{equation}
In the case that this inequality holds for $m=0$, we will simply say that $\mathcal{F}$ is geodesically convex, or displacement convex.

\subsection{Subdifferential Calculus in the Wasserstein Space}
We are now ready to present the differential structure of the Wasserstein space. Throughout, we will consider a functional $\smash{\mathcal{F}:\mathcal{P}_2(\mathbb{R}^d)\rightarrow(-\infty,\infty]}$ with domain $\mathcal{D}(\mathcal{F})=\{\mu\in\mathcal{P}_{2,\mathrm{ac}}(\mathbb{R}^d):\mathcal{F}(\mu)<\infty\}$.\footnote{For simplicity, we here restrict the domain to $\mathcal{P}_{2,\mathrm{ac}}(\mathbb{R}^d)$, the set of measures which are absolutely continuous w.r.t. the Lebesgue measure, which ensures the existence of an optimal transport map $\boldsymbol{t}_{\mu}^{\nu}$ \cite[see, e.g.,][Section 10.1]{ambrosio2008gradient}. Many of these definitions, however, can be extended to the more general case (e.g., \citealp{ambrosio2008gradient}, Section 10.3; \citealp{lanzetti2022first}, Section 2.3).} Moreover, we will assume that the functional  $\mathcal{F}$ is proper and lower semicontinuous. That is, $\mathcal{D}(\mathcal{F})\neq \emptyset$, and $\mathcal{F}(\bar{\mu})\leq \liminf_{\mu\rightarrow\bar{\mu}}\mathcal{F}(\mu)$ for all $\bar{\mu}\in\mathcal{P}_2(\mathbb{R}^d)$.

Let $\mu\in\mathcal{D}(\mathcal{F})$. We say that a map $\xi\in L^2(\mu)$ belongs to the Fr\'{e}chet subdifferential of $\mathcal{F}$ at $\mu$, and write $\xi\in\partial\mathcal{F}(\mu)$ if
\begin{equation}
\mathcal{F}(\nu) - \mathcal{F}(\mu)\geq \int_{\mathbb{R}^d} \langle \xi(x), \boldsymbol{t}_{\mu}^{\nu}(x)-x\rangle \mathrm{d}\mu(x) + o(W_2(\mu,\nu))
\end{equation}
for any $\nu\in\mathcal{D}(\mathcal{F})$ \citep[e.g.,][Definition 10.1.1]{ambrosio2008gradient}. In this case, we say that $\xi$ is a Wasserstein subgradient of $\mathcal{F}$ at $\mu$. 

Suppose, in addition, that $\mathcal{F}$ is $m$-geodesically convex. Then $\xi\in L^2(\mu)$ belongs to the Fr\'{e}chet subdifferential $\partial\mathcal{F}(\mu)$ if and only if
\begin{equation}
\mathcal{F}(\nu) - \mathcal{F}(\mu) \geq \int_{\mathbb{R}^d}\langle \xi(x), \boldsymbol{t}_{\mu}^{\nu}(x) -x \rangle\mu(\mathrm{d}x) + \frac{m}{2}W_2^2(\mu,\nu)
\end{equation}
for any $\nu\in\mathcal{D}(\mathcal{F})$ \citep[e.g.,][Section 10.1.1]{ambrosio2008gradient}.
In this case, $\partial\mathcal{F}(\mu)$ has a unique element $\partial^{o}\mathcal{F}(\mu)$ known as the \emph{minimal selection} \citep[][Lemma 10.1.5]{ambrosio2008gradient}, which has the smallest norm in the following sense: $\smash{\partial^{o}\mathcal{F}(\mu) = \argmin_{\xi\in\partial\mathcal{F}(\mu)}\|\xi\|^2_{L^2(\mu)}}$. We will refer to the element of minimal norm as the Wasserstein gradient of $\mathcal{F}$ at $\mu$, and denote it by $\smash{\nabla_{W_2}\mathcal{F}(\mu)}$. 

Under mild regularity conditions, the Wasserstein gradient can be explicitly characterised in terms of the \emph{first variation} \citep[][Section 10.4.1]{ambrosio2008gradient}. In particular, suppose that $\mu\in\mathcal{D}(\mathcal{F})$, with density in $C^1(\mathbb{R}^d)$. Then the Wasserstein gradient is given by
\begin{equation}
\nabla_{W_2}\mathcal{F}(\mu) = \nabla \frac{\delta \mathcal{F}(\mu)}{\delta \mu}(x)~~~\text{for $\mu$-a.e. $x\in\mathbb{R}^d$}, \label{eq:wasserstein-first-var}
\end{equation}
where $\frac{\delta\mathcal{F}(\mu)}{\delta\mu}:\mathbb{R}^d\rightarrow\mathbb{R}$ denotes the first variation of $\mathcal{F}$ at $\mu$, that is, the unique (up to an additive constant) function such that
\begin{equation}
\lim_{\varepsilon\rightarrow0}\frac{1}{\varepsilon}\left( \mathcal{F}(\mu + \varepsilon \chi) - \mathcal{F}(\mu)\right) = \int_{\mathbb{R}^d} \frac{\delta \mathcal{F}(\mu)}{\delta \mu}(x)\chi(\mathrm{d}x), \label{eq:first-variation}
\end{equation}
for all signed measures $\chi$ such that $\mu+ \varepsilon\chi \in\mathcal{P}_{2}(\mathbb{R}^d)$ for sufficiently small $\varepsilon>0$ \citep[see, e.g.,][Lemma 10.4.1, Lemma 10.4.13]{ambrosio2008gradient}. It is clear that the Wasserstein gradient $\nabla_{W_2}\mathcal{F}(\mu)$ is an element of the tangent space $\mathcal{T}_{\mu}\mathcal{P}_2(\mathbb{R}^d)$.

We can also define the notion of the Hessian of a functional $\mathcal{F}:\mathcal{P}_2(\mathbb{R}^d)\rightarrow(-\infty,\infty]$ over the Wasserstein space. In particular, the \emph{Wasserstein Hessian} of $\mathcal{F}$ at $\mu$ is a bilinear form $\mathrm{Hess}_{W_2}\mathcal{F}(\mu):\mathcal{T}_{\mu}\mathcal{P}_2\times\mathcal{T}_{\mu}\mathcal{P}_2\rightarrow\mathbb{R}$ defined according to
\begin{equation}
    \mathrm{Hess}_{W_2}\mathcal{F}(\mu)[v,v] = \left.\frac{\mathrm{d}^2}{\mathrm{d}t^2}\right|_{t=0}\mathcal{F}(\mu_t),
\end{equation}
where $(\mu_t)_{t\in[0,1]}$ is a constant speed geodesic satisfying the initial conditions $\mu_0 = \mu$ and $\dot{\mu}_0 = v$ (e.g., \citealt{villani2008optimal}, Chapter 15).

\section{Optimization on the Wasserstein Space}
\label{sec:optimization-wasserstein-space}
We are interested in optimization problems over the Wasserstein space of probability measures of the form
\begin{equation}
    \pi = \argmin_{\mu\in\mathcal{P}_2(\mathbb{R}^d)} \mathcal{F}(\mu) ,\label{eq:wasserstein-optimization}
\end{equation}
where $\mathcal{F}:\mathcal{P}_2(\mathbb{R}^d)\rightarrow(-\infty,\infty]$ is a functional which maps probability measures to $\mathbb{R}$. In this section, we introduce the general form of the objective functionals studied in this paper, as well as their properties (e.g., geodesic convexity, Wasserstein gradients) (Section \ref{sec:objective-functions}). We then recall several standard gradient-based methods for solving optimization problems of this type (Section \ref{sec:methods}). 

\subsection{The Objective Function}
\label{sec:objective-functions}
In this paper, we will focus principally on objectives that can be expressed as the sum of the three fundamental functionals identified in Villani's well-known treatise on optimal transport \citep{villani2003topics}, viz
\begin{equation}
    \label{eq:functional}
    \mathcal{F}(\mu) = \underbrace{\int V(x)\mathrm{d}\mu(x)}_{\mathcal{V}(\mu)} + \underbrace{\frac{1}{2}\int_{\mathbb{R}^d\times\mathbb{R}^d}W(x-y)\mathrm{d}\mu(x)\mathrm{d}\mu(y)}_{\mathcal{W}(\mu)} + \underbrace{\int_{\mathbb{R}^d} H(\mu(x))\mathrm{d}x}_{\mathcal{H}(\mu)},
\end{equation}
where $\mathcal{V}:\mathcal{P}_2(\mathbb{R}^d)\rightarrow(-\infty,\infty]$ denotes the \emph{potential energy}, $\mathcal{W}:\mathcal{P}_2(\mathbb{R}^d)\rightarrow(-\infty,\infty]$ the \emph{interaction energy}, and $\mathcal{H}:\mathcal{P}_2(\mathbb{R}^d)\rightarrow(-\infty,\infty]$ the \emph{internal energy}.\footnote{Later, we will often treat the potential energy and the interaction energy separately from the internal energy. It will thus also be convenient to define $\mathcal{E}:\mathcal{P}_2(\mathbb{R}^d)\rightarrow(-\infty,\infty]$ as the sum of the potential energy and interaction energy, viz, $\mathcal{E}(\mu):= \mathcal{V}(\mu) + \mathcal{W}(\mu)$.} These three terms are defined as follows.

\subsubsection{The Potential Energy}  Let $V:\mathbb{R}^d\rightarrow(-\infty,\infty]$ be a proper, lower semicontinuous function (the \textit{potential function}), whose negative part satisfies a quadratic growth condition: for all $x\in\mathbb{R}^d$, there exist $A,B\in\mathbb{R}$ such that $V(x)\geq -A - B\|x\|^2$.  We define the potential energy as 
\begin{equation}
    \mathcal{V}(\mu) = \int_{\mathbb{R}^d} V(x)\mathrm{d}\mu(x).
\end{equation}
This functional encourages the distribution $\mu$ to concentrate in regions where the potential function is minimal. Under the assumptions stated above, the potential energy is both proper and lower semicontinuous \citep[e.g.,][Example 9.3.1]{ambrosio2008gradient}. If, in addition, the potential function is $m$-convex, then the potential energy is $m$-geodesically convex \citep[][Proposition 9.3.2]{ambrosio2008gradient}. Finally, for all $\mu\in\mathcal{P}_2(\mathbb{R}^d)$, the Wasserstein gradient of $\mathcal{V}$ at $\mu$ is given by \citep[e.g.,][Proposition 4.1.3, Section 10.4.7]{ambrosio2007gradient}
\begin{equation}
\label{eq:potential-minimal-selection}
    \nabla_{W_2} \mathcal{V}(\mu) (x) = \nabla V(x) \quad \quad \text{ for $\mu$-a.e. $x\in\mathbb{R}^d$},
\end{equation}
where, in a slight abuse of notation, we have used $\nabla V(x)$ to denote the minimal selection $\partial^{o}V(x):=\argmin\{\|g\|:g\in\partial V(x)\}$, where $\partial V(x)$ denotes the subdifferential of $V$ at $x$ \citep[e.g.,][]{rockafellar1970convex}.

\subsubsection{The Interaction Energy} Let $W:\mathbb{R}^d\rightarrow(-\infty,\infty]$ be a function (the \textit{interaction kernel}) whose negative part satisfies the usual quadratic growth condition; and satisfying $W(0)<+\infty$. We define the interaction energy as 
\begin{equation}
    \mathcal{W}(\mu) = \frac{1}{2} \int_{\mathbb{R}^d\times\mathbb{R}^d} W(x-y)\mathrm{d}\mu(x) \mathrm{d}\mu(y). \label{eq:interaction}
\end{equation}
This term is a form of potential energy associated to pairs of particles. It can be both attractive (e.g., $W(x-y) = \|x-y\|^2$) and repulsive (e.g., $W(x-y) = - \log \|x-y\|$). Once again, this functional is both proper and lower semicontinuous under our assumptions \citep[e.g.,][Example 9.3.4]{ambrosio2008gradient}. If, in addition, the interaction kernel $W$ is convex, then $\mathcal{W}$ is geodesically convex \citep[][Proposition 9.3.5]{ambrosio2008gradient}. Finally, suppose $W$ is a differentiable, even function, and satisfies a doubling condition, i.e., there exists $C_W>0$ such that, for all $x,y$, $W(x+y)\leq C_W(1+W(x)+W(y))$. Then, for each $\mu\in\mathcal{P}_{2}(\mathbb{R}^d)$, the Wasserstein gradient of $\mathcal{W}$ at $\mu$ is given by \citep[e.g.,][Section 10.4.7]{ambrosio2008gradient}
\begin{equation}
    \nabla_{W_2}\mathcal{W}(\mu)(x) = \nabla W\star \mu(x) := \int \nabla W(x-y)\mathrm{d}\mu(y).
\end{equation}

\subsubsection{The Internal Energy} Let $H:[0,\infty)\rightarrow(-\infty,\infty]$ be a proper, lower semicontinuous function which satisfies $H(0) = 0$ and $\smash{\liminf_{s\rightarrow0}\frac{H(s)}{s^{\alpha}}>-\infty}$ for some $\smash{\alpha>\frac{d}{d+2}}$. The internal energy is defined as 
\begin{equation}
    \mathcal{H}(\mu) = \left\{ \begin{array}{lll} \int_{\mathbb{R}^d} H(\mu(x))\mathrm{d}x &, & \text{if $\mu\in\mathcal{P}_{2,\mathrm{ac}}(\mathbb{R}^d)$} \\[1mm] +\infty & , & \text{otherwise},
    \end{array}
    \right.
\end{equation}
where, in a slight abuse of notation, we write $\mu(x)$ for the density of $\mu\in\mathcal{P}_{2,\mathrm{ac}}(\mathbb{R}^d)$. This term is repulsive, favouring distributions evenly spread throughout the domain. 

Once more, under the assumptions above, this functional is proper and lower semicontinuous \cite[e.g.,][Example 9.3.6]{ambrosio2008gradient}. Suppose, in addition, that the function $H$ is convex, and satisfies the condition: 
\begin{equation}
     \text{\emph{$s\mapsto s^{d}H(s^{-d})$ is convex and non-increasing in $(0,+\infty)$}}. \label{eq:internal-energy-convexity}
\end{equation}
Then it follows that the internal energy $\mathcal{H}$ is geodesically convex (e.g., \citealp{ambrosio2008gradient}, Proposition 9.3.9; \citealp{erbar2010heat}).
Suppose also that $H$ satisfies the \textit{doubling} condition, i.e., there exists $C_{H}>0$ such that for all $x,y$, $H(x+y) \leq C_{H}(1+H(x)+H(y))$. Then, for all $\mu\in\mathcal{P}_{2,\mathrm{ac}}(\mathbb{R}^d)$, the Wasserstein gradient of $\mathcal{H}$ at $\mu$ is defined by (\citealt{ambrosio2008gradient}, Theorem 10.4.6; \citealt{lanzetti2022first}, Proposition 2.25)
    \begin{equation}
    \label{eq:internal-minimal-selection}
        \nabla_{W_2}\mathcal{H}(\mu) = \frac{\nabla \left(\mu H'(\mu) - H(\mu)\right)}{\mu}.
    \end{equation}

\subsubsection{Extensions} More generally, it may be of interest to consider generic entropy-regularized objective functionals of the form
\begin{equation}
    \mathcal{F}(\mu) =  \mathcal{E}(\mu) + \mathrm{Ent}(\mu), \label{eq:entropy-regularized}
\end{equation} 
 where $\mathrm{Ent}(\mu) := \int \log [\mu(x)] \,\mathrm{d}\mu(x)$ denotes the (negative) entropy, and $\mathcal{E}:\mathcal{P}_2(\mathbb{R}^d)\rightarrow(-\infty,\infty]$ is a known functional, often of the form 
\begin{equation}
    \mathcal{E}(\mu) = \mathcal{E}_0(\mu) + \frac{\lambda}{2}\int_{\mathbb{R}^d}\|\cdot\|^2\mathrm{d}\mu(\cdot),
\end{equation}
for some convex functional $\mathcal{E}_0:\mathcal{P}_2(\mathbb{R}^d)\rightarrow(-\infty,\infty]$. This encompasses, for example, the objective functionals which arise in the analysis of mean-field neural networks \citep[e.g.][]{suzuki2023convergence,kook2024sampling}, as well as those encountered in post-Bayesian inference \citep[e.g.,][]{chazal2025computable,mclatchie2025predictively}. We leave a detailed theoretical treatment of this more general setting to future work. We will, however, revisit this case in our later discussion of practical algorithms (Section \ref{sec:algorithms}), as well as in our numerical experiments (Section \ref{sec:numerical-results}).

\begin{tcolorbox}[enhanced,
  colback=white,
  frame hidden,
  borderline north={0.5pt}{0pt}{black!70},
  borderline south={0.5pt}{0pt}{black!70},
  arc=2pt,             
  left=0pt,            
  right=0pt,           
  top=4pt,             
  bottom=4pt,           
  parbox=false,
  before skip=10pt, 
  after skip=15pt]
\begin{example1}[\textsc{Sampling from a Target Probability Distribution}]
Consider again the example introduced in Section \ref{sec:introduction}: $\smash{\mathcal{F}(\mu) = \mathrm{KL}(\mu\|\pi)}$, where $\smash{\pi\in\mathcal{P}_2(\mathbb{R}^d)}$ is a {probability distribution} on $\smash{\mathbb{R}^d}$, with density $\smash{\pi(x)\propto e^{-U(x)}}$ with respect to the Lebesgue measure.

This functional 
can be rewritten in the general form of the objective functional in \eqref{eq:functional} (see, e.g., \citealt{ambrosio2008gradient}, Remark 9.4.2; \citealt{ambrosio2021lectures}, Proposition 15.6). 
In particular, we have that
\begin{align}
\label{eq:kl-divergence-energy}
    \mathrm{KL}(\mu\|\pi) 
    = \left\{ \begin{array}{lll}  \int_{\mathbb{R}^d} U(x) \mathrm{d}\mu(x) + \int_{\mathbb{R}^d} \log \left[\mu(x)\right]\mu(x)\mathrm{d}x + \mathrm{const.} & , & \text{if $\mu \in\mathcal{P}_{2,\mathrm{ac}}(\mathbb{R}^d)$ }, \\ + \infty & , & \text{otherwise,} \end{array} \right. 
\end{align}
so that $\smash{\mathcal{V}(\mu) = \int U(x)\mathrm{d}\mu(x)}$, $\smash{\mathcal{W}(\mu)= 0}$, 
and $\smash{\mathcal{H}(\mu)=\int_{\mathbb{R}^d}\log [\mu(x)]\mu(x)\mathrm{d}x}$ if $\mu \in\mathcal{P}_{2,\mathrm{ac}}(\mathbb{R}^d)$, $\mathcal{H}(\mu)=+\infty$ otherwise.

Suppose that $\pi\propto e^{-U}$ is \textit{log-concave}, i.e., the potential $U:\mathbb{R}^d\rightarrow(-\infty,\infty]$ is convex. Suppose also that the negative part of $U$ satisfies a quadratic growth condition. It follows that $\mathcal{V}$ is proper, lower semicontinuous, and geodesically convex. In addition, $\mathcal{H}$ is proper, lower semicontinuous, and geodesically convex, since the function $H(s) = s\log s$ satisfies all of the necessary conditions \cite[e.g.,][pg 214]{ambrosio2008gradient}. Putting these results together, it follows that $\mathrm{KL}(\mu\|\pi) = \mathcal{V}(\mu) + \mathcal{W}(\mu) + \mathcal{H}(\mu)$ is proper, lower semicontinuous, and geodesically convex \citep[][Theorem 9.4.11]{ambrosio2008gradient}.

We can also compute its Wasserstein gradient. Using the results stated above, we have that $\smash{\nabla_{W_2} \mathcal{V}(\mu) = \nabla U = - \nabla \log \pi}$ and $\smash{\nabla_{W_2}\mathcal{W}(\mu) = 0}$.
In addition, $\smash{\nabla_{W_2} \mathcal{H}(\mu) = \nabla \mu / \mu = \nabla \log \mu}$ for $\mu\in\mathcal{P}_{2,\mathrm{ac}}(\mathbb{R}^d)$, noting that $H(s) = s\log s$ satisfies all of the conditions required to compute the $W_2$ gradient \citep[e.g.,][Section 10.4.3]{ambrosio2008gradient}. It follows that, for all $\mu\in\mathcal{P}_{2,\mathrm{ac}}(\mathbb{R}^d)$, we have
\begin{equation}
    \nabla_{W_2}\mathrm{KL}(\mu\|\pi) = \nabla_{W_2}\mathcal{V}(\mu) + \nabla_{W_2}\mathcal{W}(\mu) + \nabla_{W_2}\mathcal{H}(\mu) = \nabla U+\nabla \log \mu = \nabla \log \frac{\mu}{\pi}. \label{eq:kl-gradient}
\end{equation}
\end{example1}
\end{tcolorbox}

\subsection{Gradient Flows} 
\label{sec:methods}
We now turn our attention to \emph{how} to solve the optimization problem in \eqref{eq:wasserstein-optimization}. A natural solution is to consider a gradient flow of the objective functional over the space of probability measures. 

\subsubsection{Wasserstein Gradient Flow}
We will focus primarily on the \emph{Wasserstein gradient flow} (WGF) of $\mathcal{F}$, defined as the weak solution $\smash{\mu:[0,\infty)\rightarrow\mathcal{P}_2(\mathbb{R}^d)}$ of the continuity equation \citep[][Chapter 11]{ambrosio2008gradient}
\begin{equation}
    \frac{\partial \mu_t}{\partial t}  +\nabla \cdot \left(v_t\mu_t\right)=0~,~~~v_t \in -\partial \mathcal{F}(\mu_t). \label{eq:wasserstein-grad-flow}
\end{equation}
Under mild conditions, this equation admits a unique solution for any initial condition (e.g., Theorem 11.1.4, Theorem 11.2.1 in \citealp{ambrosio2008gradient}; Proposition 4.13 in \citealp{santambrogio2017euclidean}). In particular, this solution is dictated by the unique minimal selection of the subdifferential \citep[e.g.,][Theorem 5.3]{ambrosio2007gradient}. Thus, for almost all $t>0$, we have
\begin{equation}
    \frac{\partial \mu_t}{\partial t} +\nabla \cdot \left(v_t\mu_t\right)=0~,~~~v_t = -\nabla_{W_2}\mathcal{F}(\mu_t). \label{eq:wasserstein-grad-flow-minimal-norm}
\end{equation}
In addition, the function $t\mapsto \mathcal{F}(\mu_t)$ is decreasing and so, under additional constraints (e.g., $\mathcal{F}$ is geodesically convex), $\lim_{t\rightarrow\infty} \mathcal{F}(\mu_t) = \inf_{\mu\in\mathcal{P}_2(\mathbb{R}^d)}\mathcal{F}(\mu)$ \citep[][Chapter 11]{ambrosio2008gradient}.

We can also give a Lagrangian description of the WGF (e.g., Section 8 in \citealp{ambrosio2008gradient}; Section 4 in \citealp{santambrogio2015optimal}). In particular, suppose that we define $x=(x_t)_{t\geq 0}$ as the integral curve of the vector fields $(v_t)_{t\geq 0}$, viz
\begin{equation}
    \dot{x}_t = v_t(x_t)~,~~~v_t = -\nabla_{W_2}\mathcal{F}(\mu_t),\label{eq:wasserstein-grad-flow-lagrangian}
\end{equation}
with initial condition $x_0\sim\mu_0$. Then, if we let $(\mu_t)_{t\geq 0}$ denote the marginal law of $(x_t)_{t\geq 0}$, the continuity equation in \eqref{eq:wasserstein-grad-flow-minimal-norm} precisely describes the evolution of $(\mu_t)_{t\geq 0}$.

\begin{tcolorbox}[enhanced,
  colback=white,
  frame hidden,
  borderline north={0.5pt}{0pt}{black!70},
  borderline south={0.5pt}{0pt}{black!70},
  arc=2pt,             
  left=0pt,            
  right=0pt,           
  top=4pt,             
  bottom=4pt,           
  parbox=false,
  before skip=10pt, 
  after skip=15pt]
\begin{example1}[\textsc{Sampling from a Target Probability Distribution}]
Consider again our running example: $\mathcal{F}(\mu) = \mathrm{KL}(\mu\|\pi)$, with $\pi\propto e^{-U}$. Substituting the Wasserstein gradient of the KL divergence, cf. \eqref{eq:kl-gradient}, into the general form of the WGF, cf.  \eqref{eq:wasserstein-grad-flow-minimal-norm}, it follows that the WGF of the KL divergence is given by
\begin{align}
    \frac{\partial \mu_t}{\partial t} &= \nabla \cdot \Big( \mu_t \nabla \log \frac{\mu_t}{\pi}\Big) = \nabla \cdot\Big(\mu_t \Big[\nabla U + \nabla \log \mu_t\Big]\Big). \label{eq:wgf-kl-divergence}
\intertext{Meanwhile, substituting \eqref{eq:kl-gradient} into \eqref{eq:wasserstein-grad-flow-lagrangian}, the Lagrangian formulation of the WGF of the KL divergence is given by}
    \mathrm{d}x_t &= -\Big(\nabla \log \frac{\mu_t(x_t)}{\pi(x_t)} \Big)\mathrm{d}t = -\Big(\nabla U(x_t) + \nabla \log \mu_t(x_t)\Big) \mathrm{d}t, 
    \label{eq:wgf-kl-divergence-lagrangian}
\end{align}
where $\mu_t = \mathrm{Law}(x_t)$. This is sometimes referred to as the \textit{deterministic Langevin diffusion} \citep[e.g.,][]{maoutsa2020interacting,grumitt2022deterministic}.

We can also obtain an alternative representation of this WGF. In particular, using the identity $\nabla \log \mu_t = \nabla \mu_t / \mu_t$, the WGF in \eqref{eq:wgf-kl-divergence} can be rewritten as
\begin{align}
    \frac{\partial \mu_t}{\partial t}  &= \nabla \cdot (\mu_t \nabla U) + \Delta \mu_t. \label{eq:wgf-kl-divergence-langevin}
\end{align}
This is precisely the Fokker-Planck equation governing the evolution of the marginal law of the overdamped Langevin diffusion \citep[e.g.,][]{fruehwirth2024ergodicity}, viz, 
\begin{equation}
    \mathrm{d}x_t = -\nabla U(x_t)\mathrm{d}t + \sqrt{2}\mathrm{d}w_t, \label{eq:wgf-kl-divergence-langevin-lagrangian}
\end{equation}
where $w=(w_t)_{t\geq 0}$ is a standard $\mathbb{R}^d$-valued Brownian motion. Thus, the overdamped Langevin diffusion can be viewed as the WGF of the KL divergence \citep{jordan1998variational}.
\end{example1}
\end{tcolorbox}

\subsubsection{Alternative Gradient Flows} 
\label{sec:alternative-gradient-flows}
There are, in fact, several other gradient flows which one could use to solve the minimization problem in \eqref{eq:wasserstein-optimization}. These commonly arise as a result of a different choice of metric over the space of probability measures \citep[e.g.,][]{liu2017stein,garbuno2020interacting,garbuno2020affine,duncan2019geometry,lu2023birth}. Below, we survey several popular alternatives to the Wasserstein metric, and the corresponding gradient flows. 

We mention these since, in principle, our adaptive step size schedule(s) could also be used to obtain tuning-free versions of the algorithms obtained as time-discretizations of \emph{these} gradient flows. Indeed, given an appropriate modification of our existing assumptions (see Section \ref{sec:assumptions}) -- e.g., replacing geodesic convexity w.r.t. the Wasserstein geometry with geodesic convexity w.r.t. the new geometry -- our theoretical analysis would go through largely unchanged. We leave a more detailed treatment of this to future work.

\paragraph{The Fisher-Rao Metric}
The Fisher-Rao metric \citep[e.g.,][]{amari2016information,nihat2015information}, also known as the Hellinger distance \citep{hellinger1909neue} or Hellinger-Kakutani distance \citep{kakutani1948equivalence},\footnote{To be precise, the Fisher-Rao distance is a metric over the space $\mathcal{M}_{+}(\mathbb{R}^d)$ of positive measures on $\mathbb{R}^d$. In the case that $\mu,\nu$ are both probability measures, it coincides with the Hellinger distance.} is a metric over the space $\mathcal{M}_{+}(\mathbb{R}^d)$ of positive measures over $\mathbb{R}^d$, defined via
    \begin{equation}
        d_{\mathrm{FR}}^2(\mu,\nu) = \int_{\mathbb{R}^d}\left|\sqrt{\frac{\mathrm{d}\mu}{\mathrm{d}\lambda}} - \sqrt{\frac{\mathrm{d}\nu}{\mathrm{d}\lambda}}\right|^2 \mathrm{d}\lambda,
    \end{equation}
where $\lambda$ is any reference measure such that $\mu\ll \lambda$ and $\nu\ll \lambda$. Given a continuous and differentiable functional $\mathcal{F}:\mathcal{P}(\mathbb{R}^d)\rightarrow\mathbb{R}\cup\{+\infty\}$, the gradient with respect to the Fisher-Rao metric can be computed as \citep[e.g.,][Appendix A]{lu2019accelerating}
    \begin{equation}
        \nabla_{\mathrm{FR}}\mathcal{F}(\mu) = \frac{\delta \mathcal{F}(\mu)}{\delta \mu} - \int \frac{\delta \mathcal{F}(\mu)}{\delta \mu}(x) \mathrm{d}\mu(x).
    \end{equation}
The \emph{Fisher-Rao gradient flow} is then defined as the curve $\mu:[0,\infty)\rightarrow\mathcal{P}(\mathbb{R}^d)$ of probability measures which satisfy the \emph{reaction equation} \cite[e.g.,][Section 2]{lu2023birth}  
    \begin{equation}
    \label{eq:fisher-rao-gradient-flow}
        \frac{\partial \mu_t}{\partial t} = \Lambda_t \mu_t,\quad \Lambda_t = -\nabla_{\mathrm{FR}}\mathcal{F}(\mu_t).
    \end{equation}
This is a mean-field equation of \emph{birth-death type}, where the birth-death rate $\Lambda_t$ contains a nonlocal interaction term that allows global movement of mass. 

The convergence properties of the Fisher-Rao gradient flow of the KL divergence (e.g., exponential convergence rate, independent of the target measure) have recently been established by various authors \citep{domingo2023explicit,lu2023birth,chen2023sampling}. More recently,  \citet{carrillo2024fisher} study their convergence properties for a wider-class of energy functionals. See also \citet{chen2024efficient,maurais2024sampling} for some applications to sampling problems.

Other authors have considered a composite of the Wasserstein and Fisher-Rao geometries \citep[e.g.,][]{kondratyev2016new,chizat2018interpolating,liero2018optimal}. In particular, \citet{gallouet2017jko,lu2019accelerating,kondratyev2019spherical,yan2024learning,tan2024accelerate} all study gradient flows with respect to the Wasserstein-Fisher-Rao metric. Such gradient flows allow for both mass transportation and mass teleportation.

\paragraph{The Stein Metric}
An alternative choice is the Stein metric, first rigorously defined in \citet{duncan2019geometry}. To introduce this metric, we will first require some additional notation. 

Let $k:\mathbb{R}^d\times\mathbb{R}^d\rightarrow\mathbb{R}$ denote a continuous, symmetric, positive-definite kernel. Let $\mathcal{H}_k$ denote the reproducing kernel Hilbert space (RKHS) associated to the kernel $k$. Let $\langle \cdot, \cdot \rangle_{\mathcal{H}_k}$ denote the inner product on this space, and let $\|\cdot\|_{\mathcal{H}_k}$ the norm induced by this inner product. Let $\mathcal{H}_k^{d} = \mathcal{H}_k \times \cdots \times \mathcal{H}_k$ denote the $d$-fold Cartesian product of $\mathcal{H}_k$, i.e., the Hilbert space of vector fields $v = (v_1,\cdots,v_d)$ with norm $\smash{\|v\|_{\mathcal{H}_k^d}^2 = \sum_{i=1}^d \|v_i\|^2_{\mathcal{H}_k}}$. We can then define the Stein distance as \citep[Definition 15, Lemma 17.3]{duncan2019geometry}: 
    \begin{equation}
        d_k^2(\mu,\nu) = \inf \left\{\int_{0}^{1} \|v_t\|^2_{\mathcal{H}_k^{d}} \mathrm{d}t, ~~~\partial_{t}\mu_t 
 + \nabla \cdot (v_t\mu_t) = 0,~\mu_0 = \mu,~\mu_1 = \nu \right\}.
    \end{equation}
The gradient with respect to the Stein geometry can then be computed as \citep[Remark 10, Lemma 11]{duncan2019geometry}
\begin{equation}
    \nabla_{k}\mathcal{F}(\mu) = \mathcal{S}_{\mu} \nabla \frac{\delta \mathcal{F}}{\delta \mu}, \label{eq:stein-gradient}
\end{equation}
where $\mathcal{S}_{\mu}:L^2(\mu)\rightarrow\mathcal{H}_k^d$ denotes the linear operator given by $\mathcal{S}_{\mu} f = \int_{\mathbb{R}^d} k(x,\cdot)f(x) \mu(\mathrm{d}x)$. Finally, the \emph{Stein gradient flow} is defined as the curve of probability measures which satisfy
\begin{equation}
    \frac{\partial \mu_t}{\partial t} + \nabla \cdot (v_t\mu_t) = 0, \quad v_t = - \nabla_{k}\mathcal{F}(\mu_t). \label{eq:stein-gradient-flow}
\end{equation}
In the case where $\mathcal{F}(\mu)=\mathrm{KL}(\mu\|\pi)$, this is known as the \emph{Stein variational gradient flow} \citep[e.g.,][]{he2024regularized}. This algorithm (or more precisely, its space-time discretization) was first introduced in \citet{liu2016stein}, under the name \emph{Stein variational gradient descent} (SVGD). The geometric interpretation presented here was subsequently developed in \citet{liu2017stein,duncan2019geometry}.

Since its introduction, the theoretical properties of SVGD have been the subject of intense research, although the convergence theory remains largely underdeveloped. For a non-exhaustive list of significant theoretical contributions, we refer to \citet{liu2017stein,liu2018stein,lu2019scaling,chu2020equivalence,gorham2020stochastic,korba2020nonasymptotic,salim2022convergence,das2023provably,duncan2019geometry,liu2024towards,nusken2021stein,karimi2023stochastic,shi2022finite,sun2023convergence,carrillo2024stein,fujisawa2024convergence,he2024regularized,balasubramanian2024improved,carrillo2023convergence}. 

\paragraph{The Kalman-Wasserstein Metric} 
Finally, we mention the Kalman-Wasserstein metric, defined according to (e.g., \citealt{garbuno2020interacting}, Definition 6; \citealt{burger2024covariance}, Definition 1.1)
\begin{equation}
    W_{\mathcal{C}}^2(\mu,\nu) = \inf \left\{ \int_0^1 \left[ \int_{\mathbb{R}^d} \|v_t\|^2_{\mathcal{C}(\mu_t)} \,\mathrm{d}\mu_t(x) \right] \mathrm{d}t ~\big|~ \partial_t \mu_t + \nabla \cdot (v_t\mu_t) = 0,~\mu_0 = \mu,~\mu_1 = \nu\right\}
\end{equation}
where $\smash{\|v\|^2_{\mathcal{C}(\mu)} = \langle v, \mathcal{C}^{-1}(\mu)v\rangle}$, with $\langle\cdot,\cdot\rangle$ the usual Euclidean inner product on $\mathbb{R}^d$, and $\smash{\mathcal{C}(\mu)}$ denotes the covariance of the measure $\smash{\mu\in\mathcal{P}_2(\mathbb{R}^d)}$. For a suitably regular functional $\mathcal{F}:\mathcal{P}_2(\mathbb{R}^d)\rightarrow (-\infty,\infty]$, the gradient with respect to the Kalman-Wasserstein geometry can be computed as 
\begin{equation}
    \nabla_{\mathrm{KW}}\mathcal{F}(\mu) = \mathcal{C}(\mu) \nabla \frac{\delta \mathcal{F}(\mu)}{\delta \mu}. \label{eq:kalman-wasserstein-gradient}
\end{equation}
The \emph{Kalman Wasserstein gradient flow} is then defined as the curve $\mu:[0,\infty)\rightarrow\mathcal{P}_2(\mathbb{R}^d)$ of probability measures which satisfy the continuity equation
\begin{equation}
    \frac{\partial \mu_t}{\partial t} + \nabla \cdot (v_t\mu_t) = 0, \quad v_t = - \nabla_{\mathrm{KW}}\mathcal{F}(\mu_t). \label{eq:kalman-wasserstein-gradient-flow}
\end{equation}
The properties of this gradient flow, in the particular case $\mathcal{F}(\mu) =\mathrm{KL}(\mu\|\pi)$, are studied in several recent papers \citep{garbuno2020interacting,garbuno2020affine,carrillo2021wasserstein,burger2024covariance}; see also \citet[Section 4]{chen2023sampling}.

\subsection{Practical Algorithms}
\label{sec:practical-algorithms}
In order to obtain practical optimization algorithms, it will of course be necessary to discretize the continuous-time gradient flow defined in \eqref{eq:wasserstein-grad-flow} or, indeed, \eqref{eq:fisher-rao-gradient-flow}, \eqref{eq:stein-gradient-flow}, or \eqref{eq:kalman-wasserstein-gradient-flow}. This is the subject of Section \ref{sec:time-discretization}.

\subsubsection{Time Discretization}
\label{sec:time-discretization}
We will focus on two common time-discretizations of the continuous-time dynamics in \eqref{eq:wasserstein-grad-flow}.

\paragraph{Forward-Flow Discretization} 
The first of these is the so-called \emph{forward-flow discretization} \citep[e.g.,][Section 2.2.2]{wibisono2018sampling}. This is applicable when $\mathcal{F}:\mathcal{P}_2(\mathbb{R}^d)\rightarrow(-\infty,\infty]$ takes the general form in \eqref{eq:functional}, but now with $\mathcal{H}(\mu) = \mathrm{Ent}(\mu)$ given specifically by the negative entropy; see also \eqref{eq:entropy-regularized}. In its most general form, this discretization is defined as follows. Fix $\mu_0\in\mathcal{P}_2(\mathbb{R}^d)$. Then, for $t\geq 0$, update
\begin{align}
    \mu_{t+\frac{1}{2}} &= \left(\mathrm{id} - \eta_t \zeta_t\right)_{\#}\mu_{t}, \quad \zeta_t = \nabla_{W_2}\mathcal{E}(\mu_t) \label{eq:forward-flow-1} \\
    \mu_{t+1} &= \mathcal{N}(0,2\eta_t \mathbf{I}_d) \star \mu_{t+\frac{1}{2}} \label{eq:forward-flow-2} 
\end{align}
where $(\eta_t)_{t\geq 0}$ is a positive, non-increasing sequence known as the \emph{step size schedule} or \emph{learning-rate schedule}. Naturally, we can give a corresponding description of this time-discretization at the level of a single particle. In particular, suppose we initialize $x_0\sim\mu_0$, where $\mu_0\in\mathcal{P}_2(\mathbb{R}^d)$. Then, for $t\geq 0$, we update
\begin{align}
    x_{t+\frac{1}{2}} &= x_t - \eta_t\zeta_t(x_t), \quad \zeta_t = \nabla_{W_2} \mathcal{E}(\mu_t) \label{eq:forward-flow-lagrangian-1} \\
    x_{t+1} &= x_{t+\frac{1}{2}} + \sqrt{2\eta_{t}} z_t \label{eq:forward-flow-lagrangian-2}
\end{align}
where $\mu_t = \mathrm{Law}(x_t)$, and where $(z_t)_{t\in\mathbb{N}}$ are a sequence of i.i.d. random variables distributed according to $\mathcal{N}(0,\mathbf{I}_d)$. More compactly, we can write this as
\begin{align}
x_{t+1} &= x_t - \eta_t\zeta_t(x_t) + \sqrt{2\eta_t} z_t, \quad \zeta_t = \nabla_{W_2} \mathcal{E}(\mu_t). \label{eq:forward-flow-lagrangian}
\end{align}
This algorithm is sometimes referred to as the \emph{mean-field Langevin algorithm} \citep[e.g.,][]{hu2021meanfield,chizat2022meanfield,nitanda2022convex,suzuki2023convergence}.
~\\

\paragraph{Forward Euler Discretization} The second time-discretization we consider is the forward or explicit Euler discretization. This yields the so-called \emph{Wasserstein subgradient descent} algorithm \citep{guo2022online,wang2022projected}, defined as follows. Fix $\mu_0\in\mathcal{P}_2(\mathbb{R}^d)$. Then, for $t\geq 0$, update 
\begin{equation}
    \mu_{t+1} = (\mathrm{id} - \eta_t \xi_t)_{\#}\mu_t,\quad \xi_t = \nabla_{W_2}\mathcal{F}(\mu_t). \label{eq:wasserstein-sub-grad-descent}
\end{equation}
Similar to before, we can also give a Lagrangian description of the Wasserstein subgradient descent algorithm. Suppose we initialize $x_0\sim\mu_0$, where $\mu_0\in\mathcal{P}_2(\mathbb{R}^d)$ as above. For $t\geq 0$, we update 
\begin{equation}
    x_{t+1} = x_t - \eta_t \xi_t(x_t),\quad \xi_t= \nabla_{W_2}\mathcal{F}\mathcal(\mu_t). \label{eq:wasserstein-sub-grad-descent-lagrange}
\end{equation}
Then, if we write $(\mu_t)_{t\in\mathbb{N}_0}$ for the marginal law of $(x_t)_{t\in\mathbb{N}_0}$, then the update equation in  \eqref{eq:wasserstein-sub-grad-descent} precisely describes the evolution of $(\mu_t)_{t\in\mathbb{N}_0}$.

\paragraph{Other Time Discretizations} There are, of course, other possible time-discretizations of the WGF in \eqref{eq:wasserstein-grad-flow}, though we do not consider these any further in this paper. These include a backward Euler discretization, which corresponds to the minimizing movement scheme (MMS) \citep[Definition 2.0.6]{ambrosio2008gradient} or Jordan-Kinderlehrer-Otto (JKO) algorithm \citep{jordan1998variational}; and a forward-backward discretization, which yields the Wasserstein proximal gradient algorithm \citep{salim2020wasserstein,zhu2025convergence}.

\begin{tcolorbox}[enhanced,
  colback=white,
  frame hidden,
  borderline north={0.5pt}{0pt}{black!70},
  borderline south={0.5pt}{0pt}{black!70},
  arc=2pt,             
  left=0pt,            
  right=0pt,           
  top=4pt,             
  bottom=4pt,           
  parbox=false,
  before skip=10pt, 
  after skip=15pt]
\begin{example1}[\textsc{Sampling from a Target Probability Distribution}]
Consider again our running example: $\mathcal{F}(\mu) = \mathrm{KL}(\mu\|\pi)$, with $\pi\propto e^{-U}$. Substituting \eqref{eq:kl-gradient} into \eqref{eq:forward-flow-lagrangian}, the forward-flow discretization of the WGF of the KL divergence is given by
    \begin{equation}
         \hspace{18mm} x_{t+1} = x_{t} - \eta_t \nabla U(x_t) + \sqrt{2\eta_t} z_t, \quad z_t\sim \mathcal{N}(0,\mathbf{I}_{d}). 
        \label{eq:lmc}
     \end{equation}
This is the well-known Langevin Monte Carlo (LMC) or unadjusted Langevin algorithm (ULA) \citep[e.g.,][]{dalalyan2019user,durmus2019analysis}. We note that this algorithm can also be obtained by considering the Euler-Maruyama discretization of the WGF of KL divergence, as defined by the overdamped Langevin diffusion in \eqref{eq:wgf-kl-divergence-langevin-lagrangian}.

Meanwhile, substituting \eqref{eq:kl-gradient} into \eqref{eq:wasserstein-sub-grad-descent-lagrange}, the Wasserstein subgradient descent algorithm for the KL divergence writes as
    \begin{align}
    x_{t+1}&= x_t - \eta_t\left[\nabla U(x_t) + \nabla \log {\mu_t(x_t)} \right]. \label{eq:deterministic-lmc}  
    \end{align}
This algorithm is sometimes referred to as \textit{deterministic Langevin Monte Carlo} \citep[e.g.,][]{grumitt2022deterministic}.
\end{example1}
\end{tcolorbox}

\subsubsection{Constrained Domains}
\label{sec:constrained-domains}
In certain applications, we may in fact be interested in solving a constrained version of the optimization problem in \eqref{eq:wasserstein-optimization}, viz
\begin{equation}
    \pi = \argmin_{\mu\in\mathcal{P}_2(\mathcal{X})} \mathcal{F}(\mu),
\end{equation}
where $\mathcal{X}\subseteq\mathbb{R}^d$ is a closed, convex set and where $\mathcal{P}_{2}(\mathcal{X})$ is the set of Borel probability measures supported on $\mathcal{X}\subseteq\mathbb{R}^d$ with finite second moment. This includes the unconstrained case $\mathcal{X}=\mathbb{R}^d$ as an important special case. 

In this setting, we can consider \emph{projected} versions of the previous algorithms, cf.  \eqref{eq:forward-flow-1} - \eqref{eq:forward-flow-2} and \eqref{eq:wasserstein-sub-grad-descent}, which ensure that the iterates remains in the feasible set. To be precise, after each iteration we can make use of the Wasserstein projection operator $\mathrm{Proj}_{\mathcal{P}_2(\mathcal{X})}:\mathcal{P}_2(\mathbb{R}^d)\rightarrow\mathcal{P}_2(\mathcal{X})$, which defines the projection onto the set of probability measures with support contained in $\mathcal{X}$, viz 
\begin{equation}
    \mathrm{Proj}_{\mathcal{P}_2(\mathcal{X})}(\mu) := \argmin_{\bar{\mu}\in\mathcal{P}_2(\mathcal{X})}W_2(\mu,\bar{\mu}).
\end{equation}
This projection operator is well defined, and coincides with the pushforward of $\mu$ under the Euclidean projection operator \citep[e.g.,][Proposition 3.1]{lanzetti2023stochastic}. That is, for all $\mu\in\mathcal{P}_2(\mathbb{R}^d)$, $\smash{\mathrm{Proj}_{\mathcal{P}_2(\mathcal{X})}(\mu) = \big(\mathrm{Proj}_{\mathcal{X}}\left[\cdot\right]\big)_{\#}\mu}$, where $\smash{\mathrm{Proj}_{\mathcal{X}}:\mathbb{R}^d\rightarrow\mathcal{X}}$ denotes the standard Euclidean projection onto $\mathcal{X}$. Naturally, the Wasserstein projection operator reduces to the identity operator in the unconstrained case where $\mathcal{X} = \mathbb{R}^d$.

\paragraph{Forward-Flow Discretization} We can now define a projected version of the forward-flow discretization of the WGF, as defined by \eqref{eq:forward-flow-1} - \eqref{eq:forward-flow-2}. Once more, let $\mu_0\in\mathcal{P}_2(\mathcal{X})$. Then, instead of \eqref{eq:forward-flow-1} - \eqref{eq:forward-flow-2}, we now iterate
\begin{align}
\mu_{t+\frac{1}{2}} &= \left(\mathrm{id} - \eta_t\zeta_t\right)_{\#}\mu_{t} , \quad \zeta_t = \nabla_{W_2}\mathcal{E}(\mu_t)\label{eq:projected-forward-flow-1} \\
    \mu_{t+1} &= \mathrm{Proj}_{\mathcal{P}_2(\mathcal{X})}\Big[\mathcal{N}(0,2\eta_t \mathbf{I}_d) \star \mu_{t+\frac{1}{2}}\Big]. \label{eq:projected-forward-flow-2} 
\end{align}
Due to the result in \citet[Proposition 3.1]{lanzetti2023stochastic}, we also have a simple Lagrangian description of this algorithm. Let $x_0\sim \mu_0$, with $\mu_0\in\mathcal{P}_2(\mathcal{X})$. Then, for $t\geq 0$, we update
\begin{align}
    x_{t+\frac{1}{2}} &= x_t - \eta_t\zeta_t(x_t), \quad \zeta_t = \nabla_{W_2} \mathcal{E}(\mu_t) \label{eq:projected-forward-flow-lagrangian-1}\\
    x_{t+1} &= \mathrm{Proj}_{\mathcal{X}}\Big[x_{t+\frac{1}{2}} + \sqrt{2\eta_{t}} z_t\Big] \label{eq:projected-forward-flow-lagrangian-2}
\end{align}

\paragraph{Forward Euler Discretization} We can provide a similar generalization of the Wasserstein subgradient descent algorithm. Let $\mu_0\in\mathcal{P}_2(\mathcal{X})$. Then, for $t\geq 0$, update
\begin{align}
    \mu_{t+1} &= \mathrm{Proj}_{\mathcal{P}_2(\mathcal{X})}\Big[\left(\mathrm{id} - \eta_t\xi_t\right)_{\#}\mu_t\Big], \quad \xi_t = \nabla_{W_2}\mathcal{F}(\mu_t). \label{eq:projected-wasserstein-subgrad-descent}
\end{align}
Alternatively, at the level of a single particle: let $x_0\sim \mu_0$, for some $\mu_0\in\mathcal{P}_2(\mathcal{X})$. Then, for $t\geq 0$, update
\begin{align}
    x_{t+1} &= \mathrm{Proj}_{\mathcal{X}}\Big[x_{t} - \eta_t \xi_t(x_t)\Big], \quad \xi_t = \nabla_{W_2}\mathcal{F}(\mu_t).
\end{align}

\subsubsection{Stochastic Gradients} 
\label{sec:stochastic-gradients}
In many cases of interest, we may only have access to (unbiased) stochastic estimates of the Wasserstein (sub)gradients of our objective functional \citep[e.g.,][]{welling2011bayesian,ma2015complete,bardenet2017markov}. In such cases, we can consider stochastic (e.g., mini-batch) versions of our previous algorithms.

\paragraph{Forward-Flow Discretization} Suppose that there exists a measurable space $(\mathsf{B},\mathcal{B})$, a probability measure $\nu$ on $(\mathsf{B},\mathcal{B})$ and, for each $\mu\in\mathcal{P}_2(\mathbb{R}^d)$, a measurable function $\hat{\zeta}(\mu):\mathbb{R}^d\times\mathsf{B}\rightarrow\mathbb{R}^d$ such that, for all $x\in\mathbb{R}^d$,
\begin{equation}
    \mathbb{E}_{b\sim \nu}\left[\hat{\zeta}(\mu)(x,b)\right]:= \int_{\mathsf{Z}} \hat{\zeta}(\mu)(x,b)\mathrm{d}\nu(b) \in \partial \mathcal{E} (\mu)(x). \label{eq:stochastic-gradient-forward-flow-def}
\end{equation}
Thus, $\hat{\zeta}(\mu)(\cdot,b):\mathbb{R}^d\rightarrow\mathbb{R}^d$ represents an unbiased estimate of the subgradient of the sum of the potential and interaction energies, viz 
$\partial\mathcal{E} (\mu):\mathbb{R}^d\rightarrow\mathbb{R}^d$. We can now define a stochastic version of the (projected) forward-flow discretization of the WGF, as defined by \eqref{eq:projected-forward-flow-1} - \eqref{eq:projected-forward-flow-2} or \eqref{eq:projected-forward-flow-lagrangian-1} - \eqref{eq:projected-forward-flow-lagrangian-2}. Let $(b_t)_{t\in\mathbb{N}}$ denote a sequence of i.i.d. random variables distributed according to $\nu$. Let $\mu_0\in\mathcal{P}_2(\mathbb{R}^d)$. Then, for $t\geq 0$, we now update
\begin{align}
\mu_{t+\frac{1}{2}} &= \big(\mathrm{id} - \eta_t\hat{\zeta_t}\big)_{\#}\mu_{t} , \quad \hat{\zeta}_t := \hat{\zeta_t}(\mu_t)(\cdot, b_t) \label{eq:stochastic-forward-flow-1} \\
    \mu_{t+1} &= \mathrm{Proj}_{\mathcal{P}_2(\mathcal{X})}\left[\mathcal{N}(0,2\eta_t \mathbf{I}_d) \star \mu_{t+\frac{1}{2}}\right]. \label{eq:stochastic-forward-flow-2} 
\end{align}
Naturally, we also have a Lagrangian description of these updates. In particular, let $x_0\sim \mu_0$. Then, for $t\geq 0$, we define
\begin{align}
    x_{t+\frac{1}{2}} &= x_t - \eta_t\hat{\zeta}_t(x_t), \quad \hat{\zeta}_t =  \hat{\zeta}(\mu_t)(\cdot,b_t) \label{eq:stochastic-forward-flow-lagrangian-1}\\
    x_{t+1} &= \mathrm{Proj}_{\mathcal{X}}\Big[x_{t+\frac{1}{2}} + \sqrt{2\eta_{t}} z_t\Big]. \label{eq:stochastic-forward-flow-lagrangian-2}
\end{align}

\paragraph{Forward Euler Discretization} Alternatively, suppose that there exists a measurable space $(\mathsf{B},\mathcal{B})$, a probability measure $\nu$ on $(\mathsf{B},\mathcal{B})$ and, for each $\mu\in\mathcal{P}_2(\mathbb{R}^d)$, a measurable function $\hat{\xi}(\mu):\mathbb{R}^d\times\mathsf{B}\rightarrow\mathbb{R}^d$ such that, for all $x\in\mathbb{R}^d$,
\begin{equation}
    \mathbb{E}_{b\sim \nu}\left[\hat{\xi}(\mu)(x,b)\right]:= \int_{\mathsf{Z}} \hat{\xi}(\mu)(x,b)\mathrm{d}\nu(b) \in \partial \mathcal{F}(\mu)(x). \label{eq:stochastic-gradient-forward-euler-def}
\end{equation}
Thus, $\hat{\xi}(\mu)(\cdot,b):\mathbb{R}^d\rightarrow\mathbb{R}^d$ is an unbiased estimate of the subgradient $\partial\mathcal{F}(\mu):\mathbb{R}^d\rightarrow\mathbb{R}^d$. 
Similar to above, we can then define a more general version of \eqref{eq:wasserstein-sub-grad-descent} and \eqref{eq:wasserstein-sub-grad-descent-lagrange}, which utilizes stochastic estimates of the subgradients $(\xi_t)_{t\geq 0}$. 
Similar to above, let $(b_t)_{t\in\mathbb{N}}$ denote a sequence of i.i.d. random variables distributed according to $\nu$. Let $\mu_0\in\mathcal{P}_2(\mathbb{R}^d)$. Then, for $t\geq 0$, we define 
\begin{equation}
    \mu_{t+1} = \mathrm{Proj}_{\mathcal{P}_2(\mathcal{X})}\Big[\big(\mathrm{id} - \eta_t \hat{\xi}_t\big)_{\#}\Big]\mu_t,\quad \hat{\xi}_t = \hat{\xi}(\mu_t)(\cdot,b_t). \label{eq:wasserstein-stoc-sub-grad-descent}
\end{equation}
Similar to before, at the level of a single particle, this corresponds to the following update equation. Let $x_0\sim \mu_0$. Then, for $t\geq 0$, set
\begin{equation}
    x_{t+1} = \mathrm{Proj}_{\mathcal{X}}\big[x_t - \eta_t \hat{\xi}_t(x_t)\big],\quad \hat{\xi}_t = \hat{\xi}(\mu_t)(\cdot,b_t). \label{eq:WGF-stoc}
\end{equation}

\begin{tcolorbox}[enhanced,
  colback=white,
  frame hidden,
  borderline north={0.5pt}{0pt}{black!70},
  borderline south={0.5pt}{0pt}{black!70},
  arc=2pt,             
  left=0pt,            
  right=0pt,           
  top=4pt,             
  bottom=4pt,           
  parbox=false,
  before skip=10pt, 
  after skip=15pt]
\begin{example1}[\textsc{Sampling from a Target Probability Distribution}]
Consider again our running example: $\mathcal{F}(\mu) = \mathrm{KL}(\mu\|\pi)$, for some target distribution $\pi\propto e^{-U}$. Suppose now that $\pi$ is a Bayesian posterior. In particular, let $x\in\mathbb{R}^d$ denote a parameter of interest. Let $\pi_0(x)$ denote a prior distribution for this parameter, and $\mathcal{L}(y|x)$ the likelihood of an observation $y$ given the parameter $x$. Given a set of observations $\{y_i\}_{i=1}^n$, the posterior distribution is given by Bayes' theorem as $\smash{\pi(x)\propto e^{-U(x)}}$, where $\smash{U(x) = -[\ln \pi_0(x) +\sum_{i=1}^N \ln \mathcal{L}(y_i|x)}]$. We can thus write the objective functional as
\begin{align}
    \mathcal{F}(\mu) = - \textstyle \int \left[ \ln \pi_0(x) + \sum_{i=1}^n  \ln \mathcal{L}(y_i|x) \right]\mathrm{d}\mu(x) + \int \log \mu(x) \mathrm{d}\mu(x).
\end{align}
Meanwhile, from \eqref{eq:kl-gradient}, the Wasserstein gradient of our objective functional can now be expressed as
\begin{equation}
    \nabla_{W_2}\mathcal{F}(\mu)(x) 
    = \textstyle \underbrace{-\big[ \textstyle \nabla\ln \pi_0(x) + \sum_{i=1}^N \nabla \ln \mathcal{L}(y_i|x)\big]}_{\nabla_{W_2}\mathcal{E}(\mu)(x)} + \underbrace{\textstyle \vphantom{\sum_{i=1}^N}\nabla \log \mu(x)}_{\nabla_{W_2}\mathcal{H}(\mu)(x)}.
\end{equation}
In the ``big-data'' setting, where $N$ is very large, it may be prohibitively costly to evaluate this gradient over the full data set. In this case, we can form an unbiased estimate of the true gradient by subsampling a mini-batch $\Omega \subseteq \{1,\dots,N\}$ of the data uniformly at random, and defining  
\begin{equation}
    \widehat{\xi(\mu)}(x) = \textstyle \underbrace{ -\big[\textstyle \nabla \ln \pi_0(x) + \frac{N}{|\Omega|}\sum_{i\in\Omega}\nabla \ln \mathcal{L}(y_i|x)\big]}_{\hat{\zeta}(\mu)(x)} + \underbrace{\textstyle \vphantom{\frac{N}{|\Omega|}\sum_{i\in\Omega}}\nabla \log \mu(x)}_{\nabla_{W_2}\mathcal{H}(\mu)(x)}.
\end{equation}
Let $(\Omega_t)_{t\in\mathbb{N}}$ denote a sequence of independent mini-batches, each sampled uniformly at random. The stochastic version of the forward-flow discretization of the WGF, cf. \eqref{eq:stochastic-forward-flow-lagrangian-1} - \eqref{eq:stochastic-forward-flow-lagrangian-2}, is given by
\begin{equation}
    x_{t+1} = \textstyle \mathrm{Proj}_{\mathcal{X}}\left[x_{t} + \eta_t \left(\nabla \ln \pi_0(x_t) + \frac{N}{|\Omega_t|} \sum_{i\in\Omega_t} \nabla \ln\mathcal{L}(y_i|x_t)\right) + \sqrt{2\eta_t} z_t\right],
\end{equation}
where $(z_t)_{t\geq 0}$ are a sequence of i.i.d. random variables distributed as $\mathcal{N}(0,\mathbf{I}_d)$. This is the (projected version of) the popular stochastic (sub)gradient Langevin dynamics (SGLD) algorithm \citep[e.g.,][]{welling2011bayesian,durmus2019analysis}. Meanwhile, the Wasserstein stochastic subgradient descent algorithm, cf. \eqref{eq:WGF-stoc}, reads
\begin{equation}
    x_{t+1} = \textstyle \mathrm{Proj}_{\mathcal{X}}\left[x_{t} + \eta_t \left(\nabla \ln \pi_0(x_t) + \frac{N}{|\Omega_t|} \sum_{i\in\Omega_t} \nabla \ln\mathcal{L}(y_i|x_t) - \nabla \log \mu_t(x_t)\right)\right].
\end{equation}
This is a stochastic gradient version of the (projected version of) the \textit{deterministic Langevin Monte Carlo} algorithm, defined in \eqref{eq:deterministic-lmc}.
\end{example1}
\end{tcolorbox}

\section{Main Results}
\label{sec:main-results}
In this section, we present our main results. To be specific, we analyze the convergence of the forward-flow and the forward Euler discretizations of the WGF, as well as their stochastic counterparts. We first obtain non-asymptotic convergence rates in the case of a constant step size, in both nonsmooth and smooth settings. We then introduce an adaptive, tuning-free step size schedule inspired by the approach in \citet{ivgi2023dog}. Finally, we establish that the resulting algorithms achieve the optimal rates of convergence up to logarithmic factors. 

\subsection{Assumptions}
\label{sec:assumptions}
Throughout this section, we will impose the following assumptions. 

\begin{assumption}
    \label{assumption:lsc}
    The functional $\mathcal{F}:\mathcal{P}_2(\mathcal{X})\rightarrow(-\infty,\infty]$ is proper and lower semicontinuous. That is, $\mathcal{D}(\mathcal{F})=\{\mu\in\mathcal{P}_2(\mathcal{X}):\mathcal{F}(\mu)<\infty\}\neq \emptyset$ and $\mathcal{F}(\bar{\mu})\leq \liminf_{\mu\rightarrow\bar{\mu}}\mathcal{F}(\mu)$ for all $\bar{\mu}\in\mathcal{P}_2(\mathcal{X})$.
\end{assumption}

This is a standard technical condition in the literature on WGFs \citep[e.g.,][Section 11]{ambrosio2008gradient} and, more generally, optimization on Wasserstein spaces \citep[e.g.,][]{lanzetti2022first}. This is a mild assumption, being satisfied by most of the usual functionals that we will encounter in practice. For example, the squared Wasserstein distance is proper and lower semicontinuous \citep[][Corollary 2.19]{lanzetti2022first}; 
the potential energy $\mathcal{V}$ is proper and lower semicontinuous if the potential functional $V$ is \citep[][Example 9.3.1]{ambrosio2008gradient}; the interaction energy $\mathcal{W}$ is proper and lower semicontinuous if the interaction kernel $W$ is lower semicontinuous and satisfies $W(x,x)<\infty$ for some $x\in\mathcal{X}$ \citep[][Example 9.3.4]{ambrosio2008gradient}; and the internal energy $\mathcal{H}$ is proper and lower semicontinuous if the function $H$ is a proper, lower semicontinuous, convex function, which has superlinear growth at infinity which satisfies $H(0)=0$ \citep[][Example 9.3.6]{ambrosio2008gradient}.

\begin{assumption}
\label{assumption:geo-convex}
    The functional $\mathcal{F}:\mathcal{P}_2(\mathcal{X})\rightarrow(-\infty,\infty]$ is geodesically convex. That is, for any $\mu,\nu\in\mathcal{P}_2(\mathcal{X})$, there exists a constant speed geodesic $(\lambda_{\eta}^{\mu\rightarrow\nu})_{\eta\in[0,1]}$ between $\mu$ and $\nu$ such that, for all $\eta\in[0,1]$, 
\begin{equation}
\mathcal{F}(\lambda_{\eta}^{\mu\rightarrow\nu}) \leq (1-\eta)\mathcal{F}(\mu) + \eta\mathcal{F}(\nu).
\end{equation}
\end{assumption}

This is a standard assumption in both the classical (i.e., Euclidean) optimization literature \citep[e.g.,][]{bubeck2015convex}; and the literature on WGFs \citep[e.g.,][]{ambrosio2008gradient}. In the context of sampling as an optimization problem, where $\mathcal{F}(\mu) = \mathrm{KL}(\mu\|\pi)$, it is satisfied when the target density $\pi(x)\propto e^{-U(x)}$ is log-concave (see Section \ref{sec:objective-functions}). This is a rather typical assumption in the sampling literature \citep[e.g.,][]{durmus2019analysis,chewi2024log}. This being said, it is worth noting that the convergence of WGF-related sampling algorithms (e.g., ULA) can be established under somewhat weaker assumptions such as the log-Sobolev inequality or Poincar\'{e} inequality \citep[e.g.,][]{bakry2014analysis,chewi2024analysis}. We leave to future work the extension of our results to this setting.

\begin{tcolorbox}[enhanced,
  colback=white,
  frame hidden,
  borderline north={0.5pt}{0pt}{black!70},
  borderline south={0.5pt}{0pt}{black!70},
  arc=2pt,             
  left=0pt,            
  right=0pt,           
  top=4pt,             
  bottom=4pt,           
  parbox=false,
  before skip=10pt, 
  after skip=15pt]
\begin{example1}[\textsc{Sampling from a Target Probability Distribution}]
Consider again our running example: $\mathcal{F}(\mu) = \mathrm{KL}(\mu\|\pi)$, for some target distribution $\pi\propto e^{-U}$. Suppose that $U$ is lower semicontinuous, convex, and with negative part satisfying a quadratic growth condition. In this case, both Assumption \ref{assumption:lsc} and Assumption \ref{assumption:geo-convex} are satisfied (see Section \ref{sec:objective-functions}). That is, the KL divergence is proper and lower semicontinuous \citep[e.g.,][Example 2.29]{lanzetti2022first}, and geodesically convex \citep[e.g.,][Theorem 9.4.11]{ambrosio2008gradient}.
\end{example1}
\end{tcolorbox}

\subsubsection{Nonsmooth Functionals}
In the nonsmooth setting, we will impose different assumptions depending on whether we are considering the forward-flow discretization of the WGF or the forward Euler discretization of the WGF. To be specific, we will require one of the following.

\refstepcounter{assumption}

\begin{subassumption}[Forward-Flow]
\label{assumption:bounded-grads-forward-flow}
There exists a continuous function $g:\mathcal{P}_2(\mathcal{X})\rightarrow\mathbb{R}_{+}$ such that $\|\hat{\zeta}(\mu)(\cdot,b)\|_{L^2(\mu)}\leq g(\mu)$ for almost all $b\in\mathsf{B}$, where $\hat{\zeta}(\mu):\mathbb{R}^d\times \mathsf{B}\rightarrow\mathbb{R}^d$ is the stochastic gradient oracle defined in \eqref{eq:stochastic-gradient-forward-flow-def}.
\end{subassumption}

\begin{subassumption}[Forward Euler]
\label{assumption:bounded-grads}
There exists a positive constant $0<G<\infty$ such that $\|\hat{\xi}(\mu_t)(\cdot,b_t)\|_{L^2(\mu_t)}<G$ for all $t\in[T]$, where $\hat{\xi}(\mu):\mathbb{R}^d\times \mathsf{B}\rightarrow\mathbb{R}^d$ is the stochastic gradient oracle defined in \eqref{eq:stochastic-gradient-forward-euler-def}.
\end{subassumption}

Assumption \ref{assumption:bounded-grads-forward-flow} stipulates that the stochastic (sub)gradients of the potential energy $\mathcal{V}:\mathcal{P}_2(\mathcal{X})\rightarrow(-\infty,\infty]$ 
and the interaction energy $\smash{\mathcal{W}:\mathcal{P}_2(\mathcal{X})\rightarrow(-\infty,\infty]}$. 
are pointwise bounded. 
In the case that the objective is the KL divergence, it reduces to the assumption that the stochastic (sub)gradient of the potential energy is pointwise bounded.
A stronger version of this condition (uniform boundedness) has previously appeared in the analysis of stochastic (sub)gradient Langevin dynamics (e.g., \citealt{durmus2019analysis}, Assumption A3; \citealt{habring2024subgradient}, Assumption A).

Assumption \ref{assumption:bounded-grads} is significantly stronger than Assumption \ref{assumption:bounded-grads-forward-flow}, requiring that the stochastic (sub)gradients of the objective functional $\mathcal{F}:\mathcal{P}_2(\mathcal{X})\rightarrow(-\infty,\infty]$ are uniformly bounded along the path of the algorithm iterates. 
In particular, it requires boundedness of the stochastic (sub)gradient of the internal energy
(e.g., the negative entropy), which can be problematic even in relatively simple cases \citep[e.g.,][]{xu2024forward}. Nonetheless, this is a somewhat standard assumption in existing analyses of the (stochastic) Wasserstein gradient descent algorithm, and variants thereof (e.g., \citealp{korba2020nonasymptotic}, Assumption A3; \citealp{guo2022online}, Remark 9; \citealp{lanzetti2023stochastic}, Assumption 3.2). 

More generally, we note that the Euclidean analogues of these assumptions are standard in the literature on (stochastic) subgradient descent \citep[e.g.,][]{zinkevich2003online,moulines2011non,shamir2013stochastic,garrigos2023handbook,zamani2023exact}, including in the analysis of parameter-free algorithms \citep[e.g.,][]{orabona2017training,cutkosky2018black,jun2019parameter,ivgi2023dog}.

\begin{tcolorbox}[enhanced,
  colback=white,
  frame hidden,
  borderline north={0.5pt}{0pt}{black!70},
  borderline south={0.5pt}{0pt}{black!70},
  arc=2pt,             
  left=0pt,            
  right=0pt,           
  top=4pt,             
  bottom=4pt,           
  parbox=false,
  before skip=10pt, 
  after skip=15pt]
\begin{example1}[\textsc{Sampling from a Target Probability Distribution}]
Suppose now that the potential function decomposes as $U(x) = \frac{1}{n}\sum_{i=1}^n U_i(x)$. Using the results in Section \ref{sec:objective-functions}, the Wasserstein gradient is
\begin{equation}
    \nabla_{W_2}\mathcal{E}(\mu)(x) = \frac{1}{n}\sum_{i=1}^n \nabla U_i(x), \quad \nabla_{W_2}\mathcal{F}(\mu)(x) = \frac{1}{n}\sum_{i=1}^n \nabla U_i(x) + \nabla \log \mu(x). 
\end{equation}
We can then form unbiased stochastic estimates of these two quantities by randomly sampling an index $b\in\mathsf{B}:=\{1,\dots,N\}$, and setting 
\begin{equation}
    \hat{\zeta}(\mu)(x,b) = \nabla U_{b}(x), \quad \hat{\xi}(\mu)(x,b) = \nabla U_b(x) + \nabla \log \mu(x).
\end{equation}

Suppose that $U_{b}$ is $G_b$-Lipschitz for each $b\in\{1,\dots,n\}$, we have $\|\nabla U_{b}(x)\|\leq G_b$ for all $x\in\mathbb{R}^d$ \citep[e.g.,][Proposition 16.20]{bauschke2017convex}. It then follows straightforwardly that $\|\nabla U_b\|_{L^2(\mu)}\leq G_{b}$ for all $\mu\in\mathcal{P}_2(\mathcal{X})$.

For Assumption \ref{assumption:bounded-grads-forward-flow}, we just require that $\|\hat{\zeta}(\mu)(\cdot,b)\|_{L^2(\mu)}\leq g(\mu)$ for some continuous function $g:\mathcal{P}_2(\mathcal{X})\rightarrow\mathbb{R}_{+}$, for all $b\in\mathsf{B}$. This assumption is clearly satisfied, defining for example $g(\mu) = G$ for all $\mu\in\mathcal{P}_2(\mathcal{X})$, where $G:=\max_{b\in\mathsf{B}}G_b$.

For Assumption \ref{assumption:bounded-grads}, we require a bound on $\smash{\|\hat{\xi}(\mu_t)(\cdot,b_t)\|_{L^2(\mu_t)}}$. Using Minkowski's inequality, we have $\smash{\|\hat{\xi}(\mu_t)\|_{L^2(\mu_t)} \leq \|\nabla U_{b}\|_{L^2(\mu_t)} + \|\nabla \log \mu_t\|_{L^2(\mu_t)}}$, so it is sufficient to bound each term separately. Arguing as above, we can bound the first term by $G:=\max_{b\in\mathsf{B}}G_b$. Meanwhile, bounding the second term is equivalent to bounding the Fisher information, since $\smash{I(\mu_t):= \|\nabla \log \mu_t\|^2_{L^2(\mu_t)}}$. This is a somewhat strong assumption, which can be hard to verify in practice, and may fail to hold even for relatively simple examples \citep[e.g.,][]{xu2024forward}. Nonetheless, similar assumptions have appeared elsewhere in the sampling literature \citep[e.g.,][Assumption A3]{korba2020nonasymptotic}.
\end{example1}
\end{tcolorbox}

\subsubsection{Smooth Functionals}
In the smooth case, we will also require different assumptions for the forward-flow discretization and the forward Euler discretization. In particular, we impose the following.

\refstepcounter{assumption}

\begin{subassumption}[Forward-Flow]
    \label{assumption:smooth-forward-flow}
    The functional $\mathcal{E}:\mathcal{P}_2(\mathcal{X})\rightarrow(-\infty,\infty]$ is $L$-geodesically-smooth. That is, there exists a constant $0<L<\infty$ such that, for any $v\in\mathcal{T}_{\mu}\mathcal{P}_2(\mathcal{X})$, it holds that
    \begin{equation}
        \mathrm{Hess}_{W_2}\mathcal{E}(\mu)[v,v] \leq L \|v\|_{L^2(\mu)}^2 .
    \end{equation}
\end{subassumption}

\begin{subassumption}[Forward Euler]
    \label{assumption:smooth}
    The functional $\mathcal{F}:\mathcal{P}_2(\mathcal{X})\rightarrow(-\infty,\infty]$ is $L$-geodesically-smooth. That is, there exists a constant $0<L<\infty$ such that, for any $v\in\mathcal{T}_{\mu}\mathcal{P}_2(\mathcal{X})$, it holds that
    \begin{equation}
        \mathrm{Hess}_{W_2}\mathcal{F}(\mu)[v,v] \leq L \|v\|_{L^2(\mu)}^2 .
    \end{equation}
\end{subassumption}

Smoothness assumptions are common in the optimization literature \citep[e.g.,][]{boyd2004convex,nesterov2018lectures,garrigos2023handbook}, including in the analysis of parameter-free algorithms \citep[e.g.,][]{ivgi2023dog,khaled2023dowg}. In the Euclidean setting,  it is well known that smoothness can improve the non-asymptotic convergence rate of (stochastic) gradient descent from $\smash{\mathcal{O}(\frac{1}{\sqrt{T}})}$ \citep[e.g.,][Theorem 3.2.2]{nesterov2018lectures} to $ \smash{\mathcal{O}(\frac{1}{T})}$ \citep[e.g.,][Theorem 2.1.14]{nesterov2018lectures}. The same is also true for (stochastic) gradient descent on Riemannian manifolds \citep[e.g.,][]{zhang2016first}.

In many cases, smoothness assumptions are stated in terms of a Lipschitz condition on the gradient of the objective function, rather than a bound on the (operator) norm of its Hessian. In particular, in the Euclidean case, the function $f:\mathbb{R}^d\rightarrow(-\infty,\infty]$ is said to be $L$-smooth if $\|\nabla f(x) - \nabla f(y)\|\leq L\|x-y\|$ for all $x,y\in\mathbb{R}^d$. This is strictly weaker than the assumption that $\|\nabla^2 f(x)\|_{\mathrm{op}}\leq L$ for all $x\in\mathbb{R}^d$, as it also defines a notion of smoothness for functions that are not twice differentiable. The two conditions are equivalent, however, when the function $f$ is twice differentiable \citep[e.g.,][Lemma 2.26]{garrigos2023handbook}.

In our setting, one can similarly characterise the smoothness of the objective functional in terms of its Wasserstein gradients. In particular, the functional $\mathcal{F}:\mathcal{P}_2(\mathbb{R}^d)\rightarrow(-\infty,\infty]$ is $L$-smooth iff, for any $\mu\in\mathcal{P}_{2,\mathrm{ac}}(\mathbb{R}^d)$ and $\nu\in\mathcal{P}_2(\mathbb{R}^d)$, it holds that
\begin{equation}
    \|\nabla_{W_2}\mathcal{F}(\nu) \circ \boldsymbol{t}_{\mu}^{\nu} - \nabla_{W_2}\mathcal{F}(\mu)\|_{L^2(\mu)} \leq L\|t_{\mu}^{\nu} - \mathrm{id}\|_{L^2(\mu)},
    \label{eq:l-smoothness-grad}     
\end{equation}
where $\boldsymbol{t}_{\mu}^{\nu}$ is the optimal transport map between $\mu$ and $\nu$ \citep[e.g.,][Corollary 10.47]{boumal2023introduction}. Note that, unlike the Euclidean case, we do not directly compare the gradients $\nabla_{W_2}\mathcal{F}(\nu)$ and $\nabla_{W_2}\mathcal{F}(\mu)$, since they live in different (tangent) spaces $\mathcal{T}_{\nu}\mathcal{P}_2(\mathbb{R}^d)\subseteq L^2(\nu)$ and $\mathcal{T}_{\mu}\mathcal{P}_2(\mathbb{R}^d)\subseteq L^2(\mu)$.

Regarding optimization over the space of probability measures, smoothness assumptions have appeared in various works (e.g., \citealp{durmus2019analysis} for ULA; \citealp{korba2020nonasymptotic} for SVGD; \citealp{salim2020wasserstein} for Wasserstein proximal gradient). Typically, however, such assumptions are stated in terms of, e.g., the potential $V:\mathcal{X}\rightarrow(-\infty,\infty]$ or the interaction kernel $W:\mathcal{X}\rightarrow(-\infty,\infty]$, rather than on the potential \emph{energy} $\mathcal{V}:\mathcal{P}_2(\mathcal{X})\rightarrow(-\infty,\infty]$ or the interaction \emph{energy} $\mathcal{W}:\mathcal{P}_2(\mathcal{X})\rightarrow(-\infty,\infty]$. We find the latter more convenient for our analysis. It is, however, straightforward to obtain sufficient conditions which imply smoothness in the sense of our assumptions. For example, if the potential $V:\mathbb{R}^d\rightarrow(-\infty,\infty]$ is $L$-smooth (in the Euclidean sense), it is straightforward to show that the potential energy $\mathcal{V}:\mathcal{P}_2(\mathcal{X})\rightarrow(-\infty,\infty]$ is $L$-smooth (in the Wasserstein sense).

We conclude this commentary with the remark that Assumption \ref{assumption:smooth-forward-flow}, which is used in the analysis of the {forward-flow} discretization of the WGF, is significantly weaker than Assumption \ref{assumption:smooth}, which is used in the analysis of the {forward Euler} discretization. Indeed, the latter additionally requires smoothness of the internal energy $\mathcal{H}:\mathcal{P}_2(\mathcal{X})\rightarrow(-\infty,\infty]$, which is known \emph{not} to hold in many important cases, including the (negative) entropy \citep[e.g.,][Section 6.2.2]{chewi2024statistical}.

\begin{tcolorbox}[enhanced,
  colback=white,
  frame hidden,
  borderline north={0.5pt}{0pt}{black!70},
  borderline south={0.5pt}{0pt}{black!70},
  arc=2pt,             
  left=0pt,            
  right=0pt,           
  top=4pt,             
  bottom=4pt,           
  parbox=false,
  before skip=10pt, 
  after skip=15pt]
\begin{example1}[\textsc{Sampling from a Target Probability Distribution}]
Suppose now that the potential function $U:\mathbb{R}^d\rightarrow(-\infty,\infty]$ is $L$-smooth. This is rather a standard assumption used in the analysis of ULA \citep[e.g.,][Assumption A2]{durmus2019analysis}. 

In this case, it is straightforward to show that Assumption \ref{assumption:smooth-forward-flow} is satisfied. Indeed, working from the RHS of the inequality in \eqref{eq:l-smoothness-grad}, we have  
\begin{align}
    \|\nabla_{W_2} \mathcal{E}(\nu) \circ \boldsymbol{t}_{\mu}^{\nu} - \nabla_{W_2}\mathcal{E}(\mu) \|_{L^2(\mu)} & =\|\nabla_{W_2} \mathcal{V}(\nu) \circ \boldsymbol{t}_{\mu}^{\nu} - \nabla_{W_2}\mathcal{V}(\mu) \|_{L^2(\mu)} \\
    &= \big[\textstyle \int \| \nabla U\left(t_{\mu}^{\nu}(x)\right) - \nabla U(x)\|^2 \mathrm{d}\mu(x)\big]^{\frac{1}{2}} 
    \\
    &\leq \left[\textstyle \int L^2 \|t_{\mu}^{\nu}(x) - x\|^2 \mathrm{d}\mu(x) \right]^{\frac{1}{2}} 
    \\[2mm]
    &= L \|t_{\mu}^{\nu} - \mathrm{id}\|_{L^2(\mu)},
\end{align}
where the first equality follows from the fact that $\mathcal{W}(\mu)=0\implies \mathcal{E}(\mu) = \mathcal{V}(\mu) + \mathcal{W}(\mu) = \mathcal{V}(\mu)$; the second equality from the definition of the Wasserstein gradient; and the third inequality from the $L$-smoothness of the potential.

On the other hand, Assumption \ref{assumption:smooth} is \emph{not} satisfied. Indeed, working from the definition, it is possible to show that the Wasserstein Hessian of the (negative) entropy is given by \citep{otto2005eulerian,villani2008optimal}
\begin{equation}
    \mathrm{Hess}_{W_2}\mathcal{\mathcal{H}
    }(\mu)[v,v] = \int \|\nabla v(x) - x\|^2_{\mathrm{HS}}\mathrm{d}\mu(x).
\end{equation}
It follows that there is no constant $C>0$ such that $\mathrm{Hess}_{W_2}\mathcal{H}(\mu)[v,v]\leq C\int \|v(x)\|^2 \mathrm{d}\mu(x)$ for all $v\in\mathcal{T}_{\mu}\mathcal{P}_2(\mathcal{X})$ \citep[e.g.,][]{diao2023forward,chewi2024statistical}.

\end{example1}
\end{tcolorbox}

\subsection{Forward-Flow Discretization, Nonsmooth Setting}
\label{sec:forward-flow-nonsmooth}
We are now ready to analyze the convergence of \eqref{eq:forward-flow-1} - \eqref{eq:forward-flow-2}, the forward-flow discretization of the WGF in \eqref{eq:wasserstein-grad-flow}. We begin in the nonsmooth setting. 

\subsubsection{Deterministic Case: Constant Step Size}
We start with the case where the step size in \eqref{eq:forward-flow-1} - \eqref{eq:forward-flow-2} is constant, that is, $\eta_t=\eta$ for all $t\geq 0$. We will analyze the convergence of the average iterate defined according to the following recursion:
    \begin{align}
        \bar{\mu}_1 &= \mu_1, \label{eq:bar-mu-t-forward-flow-1-first-def} \\
        \bar{\mu}_{t+1} &= \left[\left(1-\frac{1}{t+1}\right) \mathrm{id} + \frac{1}{t+1} \boldsymbol{t}_{\bar{\mu}_{t}}^{\mu_{t+1}}\right]_{\#}\bar{\mu}_{t} \quad \text{for $t\geq 1$.} \label{eq:bar-mu-t-forward-flow-2-first-def}
    \end{align}

\begin{lemma}
\label{lemma:average}
    Suppose that Assumption \ref{assumption:lsc} and \ref{assumption:geo-convex} hold. Let $\bar{\mu}_T$ be the average iterate defined according to \eqref{eq:bar-mu-t-forward-flow-1-first-def} - \eqref{eq:bar-mu-t-forward-flow-2-first-def}. Then, for all $T\geq 1$, it holds that 
    \begin{equation}
    \mathcal{F}(\bar{\mu}_{T}) - \mathcal{F}(\pi) \leq \frac{1}{T}\left[\sum_{t=1}^{T} \mathcal{F}(\mu_t) - \sum_{t=1}^{T}\mathcal{F}(\pi)\right]. \label{eq:jensens-inequality-average-iterate}
    \end{equation}
\end{lemma}

\begin{proof}
    See Appendix \ref{app:additional-proofs-forward-flow-non-smooth-deterministic-constant}. 
\end{proof}

\begin{lemma}
    \label{lemma:evi-forward-flow-transport}
    Suppose that Assumption \ref{assumption:lsc} and \ref{assumption:geo-convex} hold. Let $(\mu_t)_{t\geq 0}$ denote the sequence of measures defined by \eqref{eq:forward-flow-1} - \eqref{eq:forward-flow-2}. Suppose that $\eta_t = \eta>0$ for all $t\geq 0$. Then, for any $t\geq 1$, and for any $\pi\in\mathcal{P}_2(\mathbb{R}^d)$, we have 
    \begin{equation}
        \mathcal{E}(\mu_t) - \mathcal{E}(\pi) \leq \frac{W_2^2(\mu_{t},\pi) - W_2^2(\mu_{t+\frac{1}{2}}, \pi)}{2\eta} + \frac{\eta}{2} \int_{\mathbb{R}^d} \|\zeta_t(x)\|^2\,\mathrm{d}\mu_t(x).
        \label{eq:evi-forward-flow-transport}
    \end{equation}
\end{lemma}

\begin{proof}
    See Appendix \ref{app:additional-proofs-forward-flow-non-smooth-deterministic-constant}.
\end{proof}

\begin{lemma}
    \label{lemma:evi-forward-flow-heat}
    Let $(\mu_t)_{t\geq 0}$ denote the sequence of measures defined by \eqref{eq:forward-flow-1} - \eqref{eq:forward-flow-2}. Suppose that $\eta_t = \eta>0$ for all $t\geq 0$. Then, for any $t\geq 1$, and for any $\pi\in\mathcal{P}_2(\mathbb{R}^d)$, we have 
    \begin{equation}
        \mathcal{H}(\mu_t) - \mathcal{H}(\pi) \leq \frac{W_2^2(\mu_{t-\frac{1}{2}},\pi) - W_2^2(\mu_t,\pi)}{2\eta}.
        \label{eq:evi-forward-flow-heat}
    \end{equation}
\end{lemma}

\begin{proof}
    See \citet[][Lemma 5]{durmus2019analysis}.
\end{proof}

\begin{corollary}
\label{corollary:evi-forward-flow}
    Suppose that Assumption \ref{assumption:lsc} and \ref{assumption:geo-convex} hold. Let $(\mu_t)_{t\geq 0}$ denote the sequence of measures defined by \eqref{eq:forward-flow-1} - \eqref{eq:forward-flow-2}. Suppose that $\eta_t = \eta>0$ for all $t\geq 0$. Then, for any $t\geq 1$, and for any $\pi\in\mathcal{P}_2(\mathbb{R}^d)$, we have 
    \begin{equation}
        \mathcal{F}(\mu_t) - \mathcal{F}(\pi) \leq \frac{W_2^2(\mu_{t-\frac{1}{2}},\pi) - W_2^2(\mu_{t+\frac{1}{2}}, \pi)}{2\eta} + \frac{\eta}{2} \int_{\mathbb{R}^d} \|\zeta_t(x)\|^2\,\mathrm{d}\mu_t(x).
        \label{eq:evi-forward-flow}
    \end{equation}
\end{corollary}

\begin{proof}
The result is an immediate consequence of Lemma \ref{lemma:evi-forward-flow-transport} and Lemma \ref{lemma:evi-forward-flow-heat}, making use of the decomposition $\smash{\mathcal{F}(\mu_t) - \mathcal{F}(\pi) = [\mathcal{E}(\mu_t) - \mathcal{E}(\pi)] + [\mathcal{H}(\mu_t)-\mathcal{H}(\pi)]}$.
\end{proof}

\begin{proposition}
\label{prop:regret-bound-forward-flow}
    Suppose that Assumption \ref{assumption:lsc} and \ref{assumption:geo-convex} hold. Let $(\mu_t)_{t\geq 0}$ denote the sequence of measures defined by \eqref{eq:forward-flow-1} - \eqref{eq:forward-flow-2}. 
    Suppose that $\eta_t=\eta>0$ for all $t\geq 0$. Then, for any $T\geq 1$, and for any $\pi\in\mathcal{P}_2(\mathbb{R}^d)$, 
    \begin{equation}
        \sum_{t=1}^{T} \mathcal{F}(\mu_t) - \sum_{t=1}^{T} \mathcal{F}(\pi) \leq \frac{W_2^2(\mu_{\frac{1}{2}},\pi) - W_2^2(\mu_{T+\frac{1}{2}},\pi)}{2\eta} + \frac{\eta}{2}\sum_{t=1}^{T} \int_{\mathbb{R}^d}\|\zeta_t(x)\|^2\,\mathrm{d}\mu_t(x).
    \end{equation}
\end{proposition}

\begin{proof}
    This proposition is an immediate consequence of Corollary \ref{corollary:evi-forward-flow}. In particular, summing \eqref{eq:evi-forward-flow} over $t\in[T]$, and cancelling like terms in the telescoping sum, the result follows.
\end{proof}

\begin{corollary}
\label{corollary:average-bound-forward-flow}
    Suppose that Assumption \ref{assumption:lsc} and \ref{assumption:geo-convex} hold. Let $(\mu_t)_{t\geq 0}$ denote the sequence of measures defined by \eqref{eq:forward-flow-1} - \eqref{eq:forward-flow-2}. Suppose that $\eta_t=\eta>0$ for all $t\geq 0$. Let $\bar{\mu}_{T}$ be the average iterate defined in \eqref{eq:bar-mu-t-forward-flow-1-first-def} - \eqref{eq:bar-mu-t-forward-flow-2-first-def}. Then, for all $\pi\in\mathcal{P}_2(\mathbb{R}^d)$, 
    \begin{equation}
        \mathcal{F}(\bar{\mu}_{T}) - \mathcal{F}(\pi) \leq \frac{1}{T}\left[\frac{W_2^2(\mu_{\frac{1}{2}},\pi)}{2\eta} + \frac{\eta}{2}\sum_{t=1}^{T} \int_{\mathbb{R}^d}\|\zeta_t(x)\|^2\,\mathrm{d}\mu_t(x)\right].\label{eq:average-bound-forward-flow}
    \end{equation}
\end{corollary}

\begin{proof}
    The result is an immediate consequence of Lemma \ref{lemma:average} and Proposition \ref{prop:regret-bound-forward-flow}.
\end{proof}

\begin{theorem}
\label{thm:constant-step-forward-flow}
    Suppose that Assumption \ref{assumption:lsc} and \ref{assumption:geo-convex} hold. Let $(\mu_t)_{t\geq 0}$ denote the sequence of measures defined by \eqref{eq:forward-flow-1} - \eqref{eq:forward-flow-2}. Suppose that $\eta_t=\eta>0$ for all $t\geq 0$. Let $\bar{\mu}_T$ be the average iterate defined in \eqref{eq:bar-mu-t-forward-flow-1-first-def} - \eqref{eq:bar-mu-t-forward-flow-2-first-def}. Then, for all $\pi\in\mathcal{P}_2(\mathbb{R}^d)$, the upper bound in \eqref{eq:average-bound-forward-flow} is minimized when
    \begin{equation}
        \eta = \frac{W_2(\mu_{\frac{1}{2}},\pi)}{\sqrt{\sum_{t=1}^{T} \int_{\mathbb{R}^d}\|\zeta_t(x)\|^2\,\mathrm{d}\mu_t(x)}}. \label{eq:optimal-lr-forward-flow}
    \end{equation}
    Moreover, this choice of $\eta$ guarantees
    \begin{equation}
        \mathcal{F}(\bar{\mu}_T) - \mathcal{F}(\pi) \leq \frac{1}{T} W_2(\mu_{\frac{1}{2}},\pi)\sqrt{\sum_{t=1}^{T} \int_{\mathbb{R}^d}\|\zeta_t(x)\|^2\,\mathrm{d}\mu_t(x)}. \label{eq:optimal-average-bound-forward-flow}
    \end{equation}
\end{theorem}

\begin{proof}
   This result follows straightforwardly from Corollary \ref{corollary:average-bound-forward-flow}. In particular, differentiating the RHS of \eqref{eq:average-bound-forward-flow} w.r.t. $\eta$, and setting equal to zero, gives $\smash{0=-\eta^{-2}W_2^2(\mu_{\frac{1}{2}},\pi) + \sum_{t=1}^{T} \int_{\mathbb{R}^d} \|\zeta_t(x)\|^2\,\mathrm{d}\mu_t(x)}$. Solving for $\eta$ gives \eqref{eq:optimal-lr-forward-flow}, and substituting back into \eqref{eq:average-bound-forward-flow} gives \eqref{eq:optimal-average-bound-forward-flow}.
\end{proof}

\begin{tcolorbox}[enhanced,
  colback=white,
  frame hidden,
  borderline north={0.5pt}{0pt}{black!70},
  borderline south={0.5pt}{0pt}{black!70},
  arc=2pt,             
  left=0pt,            
  right=0pt,           
  top=4pt,             
  bottom=4pt,           
  parbox=false,
  before skip=10pt, 
  after skip=15pt]
\begin{example1}[\textsc{Sampling from a Target Probability Distribution}]
Consider again our running example: $\mathcal{F}(\mu) = \mathrm{KL}(\mu\|\pi)$, with $\pi\propto e^{-U}$. In this case, Corollary \ref{corollary:average-bound-forward-flow} yields the convergence rate
\begin{equation}
    \mathrm{KL}(\bar{\mu}_T\|\pi) \leq \frac{1}{T}\left[ \frac{W_2^2(\mu_{\frac{1}{2}},\pi)}{2\eta} + \frac{\eta}{2} \sum_{t=1}^T \int_{\mathbb{R}^d}\|\nabla U(x)\|^2\,\mathrm{d}\mu_t(x)\right] \label{eq:kl-constant-step-size-bound}
\end{equation}
for the forward-flow discretization of the WGF (i.e., ULA). Meanwhile, Corollary \ref{thm:constant-step-forward-flow} gives the theoretically optimal step size for this algorithm as
\begin{equation}
    \eta = \frac{W_2(\mu_{\frac{1}{2}},\pi)}{\sqrt{\sum_{t=1}^T \int_{\mathbb{R}^d}\|\nabla U(x)\|^2\,\mathrm{d}\mu_t(x)}}. 
    \label{eq:optimal-step-size-ula}
\end{equation}
Substituting this step size back into \eqref{eq:kl-constant-step-size-bound}, the optimal convergence rate for ULA (with a constant step size) is given by 
\begin{equation}
    \mathrm{KL}(\bar{\mu}_T\|\pi) \leq \frac{1}{T} W_2(\mu_{\frac{1}{2}},\pi)\sqrt{\sum_{t=1}^{T} \int_{\mathbb{R}^d}\|\nabla U(x)\|^2\,\mathrm{d}\mu_t(x)}. \label{eq:optimal-average-bound-forward-flow-kl}
\end{equation}
Suppose that the potential is Lipschitz. That is, there exists $0<G<\infty$ such that $|U(x) - U(y)|\leq G\|x-y\|$ for all $x,y\in\mathcal{X}$. It follows that the (sub)gradients of the potential are uniformly bounded, i.e., $\|\nabla U(x)\|\leq G$ for all $x\in\mathcal{X}$ \citep[e.g.,][Proposition 16.20]{bauschke2017convex}. Thus, in this case, writing also $D_{\frac{1}{2},\pi} = W_2(\mu_{\frac{1}{2}},\pi)$, the bound in \eqref{eq:optimal-average-bound-forward-flow-kl} simplifies further as 
\begin{equation}
    \mathrm{KL}(\bar{\mu}_T\|\pi) \leq \frac{GD_{\frac{1}{2},\pi}}{\sqrt{T}}. \label{eq:optimal-average-bound-forward-flow-kl-bounded-grad}
\end{equation}

\paragraph{The Gaussian Case}
In general, one cannot directly compute the optimal step size in \eqref{eq:optimal-step-size-ula}, even in hindsight, since both the Wasserstein distance $\smash{W_2(\mu_{\frac{1}{2}},\pi)}$
and the expectations $\int_{\mathbb{R}^d} [\cdot]\,\mathrm{d}\mu_t(x)$ are intractable. 

An exception to this is when $\mu_0 =\mathcal{N}(m_0,\Sigma_0)$ and $\pi = \mathcal{N}(m_{\pi},\Sigma_{\pi})$ are both Gaussian. In this case, $(\mu_t)_{t\geq 0}$ remain Gaussian for all times \citep[e.g.,][Example 2]{wibisono2018sampling}. In particular, we have that $\mu_{t} = \mathcal{N}(m_t,\Sigma_t)$ and $\smash{\mu_{t+\frac{1}{2}} = \mathcal{N}(m_{t+\frac{1}{2}},\Sigma_{t+\frac{1}{2}})}$ for all $t\geq 0$, where 
\begin{alignat}{2}
\label{eq:ula-gaussian-mean}
    m_{t+\frac{1}{2}} &= m_t - \eta \Sigma_{\pi}^{-1}(m_t - m_{\pi}) \quad \quad &&m_{t+1} = m_{t+\frac{1}{2}}, \\
    \Sigma_{t+\frac{1}{2}} &= (\mathbf{I}_{d}-\eta\Sigma_{\pi}^{-1})\Sigma_t(\mathbf{I}_d - \eta \Sigma_{\pi}^{-1})^{\top} \quad \quad &&\Sigma_{t+1} = \Sigma_{t+\frac{1}{2}} + 2\eta\mathbf{I}_d.
\label{eq:ula-gaussian-var}
\end{alignat}
Using these quantities, and the general formula for the Wasserstein distance between two Gaussians \citep[e.g.,][]{olkin1982distance}, we can compute
\begin{align}
\label{eq:w2-lmc-gaussian-v1}
    W_2^2(\mu_{\frac{1}{2}},\pi) 
    &=\|m_{\frac{1}{2}} -m_{\pi}\|^2 + \mathrm{Tr}\big(\Sigma_{\frac{1}{2}} + \Sigma_{\pi} - 2(\Sigma_{\frac{1}{2}}^{\frac{1}{2}}\Sigma_{\pi}\Sigma_{\frac{1}{2}}^{\frac{1}{2}})^{\frac{1}{2}}\big) \\
    \int ||\nabla U(x)||^2 \,\mathrm{d}\mu_t(x) &=  \int \left[ \|\Sigma_{\pi}^{-1}(x-m_{\pi})\|^2\right] \,\mathrm{d}\mu_t(x) = \mathrm{Tr}(\Sigma_{\pi}^{-2}\Sigma_t) + \|\Sigma_{\pi}^{-1}\left(m_t - m_{\pi}\right)\|^2. \label{eq:grad-lmc-gaussian-v1}
\end{align}
It follows, substituting \eqref{eq:w2-lmc-gaussian-v1} - \eqref{eq:grad-lmc-gaussian-v1} into \eqref{eq:optimal-step-size-ula}, that in the Gaussian case the optimal step size is given by
\begin{equation}
    \eta = \frac{\Big[ \|m_{\frac{1}{2}} -m_{\pi}\|^2 + \mathrm{Tr}\big(\Sigma_{\frac{1}{2}} + \Sigma_{\pi} - 2(\Sigma_{\frac{1}{2}}^{\frac{1}{2}}\Sigma_{\pi}\Sigma_{\frac{1}{2}}^{\frac{1}{2}})^{\frac{1}{2}}\big)\Big]^{\frac{1}{2}}}{\sqrt{\sum_{t=1}^T \left(\mathrm{Tr}(\Sigma_{\pi}^{-2}\Sigma_t) + \|\Sigma_{\pi}^{-1}\left(m_t - m_{\pi}\right)\|^2\right)}}.
    \label{eq:ideal-step-size-lmc-gaussian}
\end{equation}
\end{example1}
\end{tcolorbox}

\paragraph{Remark} \emph{In the case that the functional $\mathcal{F}:\mathcal{P}_2(\mathcal{X})\rightarrow(-\infty,\infty]$ is \emph{convex}, and not just \emph{geodesically convex}, the bounds in Corollary \ref{corollary:average-bound-forward-flow} and Theorem \ref{thm:constant-step-forward-flow} also hold for the uniform mixture $\smash{\bar{\mu}_T = \frac{1}{T}\sum_{t=1}^T \mu_t}$ by Jensen's inequality. This is the case, for instance, when $\mathcal{F} = \mathrm{KL}(\cdot\|\pi)$ (see, e.g., \citealp{cover2006elements}, Theorem 2.7.2; \citealp{van2014renyi}, Theorem 11).}

\subsubsection{Deterministic Case: Adaptive Step Size}
\label{sec:adaptive-deterministic-step-size}
In practice, we cannot compute the optimal step size in Theorem \ref{thm:constant-step-forward-flow}, since both the Wasserstein distance $\smash{W_2(\mu_{\frac{1}{2}},\pi)}$ and the sequence of expectations $\int_{\mathbb{R}^d} \|\zeta_t(x)\|^2\,\mathrm{d}\mu_t(x)$ are intractable. Inspired by the approach introduced in \citet{ivgi2023dog}, we thus introduce an adaptive sequence of step sizes $\smash{(\eta_t)_{t\geq 1}}$ defined according to
\begin{equation}
\tcboxmath[colback=black!5,colframe=black!15,boxrule=0.4pt, arc=2pt, left=8pt, right=8pt, top=6pt, bottom=6pt]{\eta_t = \frac{\max\left[r_{\varepsilon},\max_{1\leq s\leq t}W_2(\mu_{\frac{1}{2}}, \mu_{s-\frac{1}{2}})\right]}{\sqrt{\sum_{s=1}^t \int ||\zeta_s(x)||^2 \,\mathrm{d}\mu_s(x)}}} ,
\label{eq:dog-lr-forward-flow}
\end{equation}
where $r_\varepsilon>0$ is some small initial value which ensures that the algorithm takes a step in the first iteration. 

Thus, the step size in the $\smash{t^{\text{th}}}$ iteration is equal to the maximum of the Wasserstein distances between the law $\smash{\mu_{\frac{1}{2}}}$ and the subsequent ``half-step'' laws $\smash{(\mu_{t-\frac{1}{2}})_{1\leq s \leq t}}$ generated by the forward-flow discretization of the WGF, scaled by the square root of the cumulative squared $L_2$ norms of the sum of the Wasserstein (sub)gradients of the potential energy and the interaction energy. This is a natural proxy to the oracle step size in \eqref{eq:optimal-lr-forward-flow}.

Based on its connection with the functional upper bound in Corollary \ref{corollary:average-bound-forward-flow}, we will refer to our proposed step size schedule as 
\begin{equation}
\tcboxmath[colback=black!5,colframe=black!15,boxrule=0.4pt, arc=2pt, left=8pt, right=8pt, top=6pt, bottom=6pt]{
\text{\textbf{\textsc{Fuse}}: \textbf{F}unctional \textbf{U}pper-Bound \textbf{S}tep-Size \textbf{E}stimator.}
}
\end{equation}
\vspace{-1.5mm}

\paragraph{Convergence Analysis Assuming Bounded Iterates}
We now study the convergence of the forward-flow discretization of the WGF in \eqref{eq:forward-flow-1} - \eqref{eq:forward-flow-2} when using the adaptive step size schedule defined in \eqref{eq:dog-lr-forward-flow}. 

In order to proceed with our analysis, it will be useful to introduce some new notation. In particular, for $t\geq 1$, let us define
{\allowdisplaybreaks
\begin{alignat}{3}
    r_t &= W_2(\mu_{\frac{1}{2}},\mu_{t-\frac{1}{2}}),\quad &&\bar{r}_t = \max_{1\leq s \leq t}\left[r_s,r_\varepsilon\right],\quad &&G_t = \sum_{s=1}^t \int_{\mathbb{R}^d} \|\zeta_s(x)\|^2\,\mathrm{d}\mu_s(x) \label{eq:defs1-forward-flow} \\
    r'_t &= W_2(\mu_{\frac{1}{2}},\mu_t), \quad &&\bar{r}'_t = \max_{1\leq s\leq t}\left[r'_s,r_{\varepsilon}\right]  \label{eq:defs1-1-forward-flow} \\[-2mm]
    d_t &= W_2(\mu_{t-\frac{1}{2}},\pi),\quad &&\bar{d}_t = \max_{1\leq s \leq t} d_s
    ,\quad &&\bar{g}_t = \max_{1\leq s \leq t} \|\zeta_s\|_{L^2(\mu_s)}:= \left[\int_{\mathbb{R}^d}\|\zeta_s(x)\|^2\,\mathrm{d}\mu_s(x)\right]^{\frac{1}{2}}. \label{eq:defs2-forward-flow}
\end{alignat}
}

We will now analyze the convergence of the weighted average iterate defined according to the following recursion:
\begin{align}
    \tilde{\mu}_1 &= \mu_1 \label{eq:tilde-mu-t-1-forward-flow} \\
     \tilde{\mu}_{t+1} &= \left[\left(1-\frac{\bar{r}_t}{\sum_{s=1}^{t} \bar{r}_s }\right) \mathrm{id} + \frac{\bar{r}_t}{\sum_{s=1}^{t} \bar{r}_s } \boldsymbol{t}_{\tilde{\mu}_{t}}^{\mu_{t+1}}\right]_{\#}\tilde{\mu}_{t},\quad t\geq 1. \label{eq:tilde-mu-t-2-forward-flow}
\end{align}

\begin{lemma}
\label{lemma:tilde-mu-t-bound-1-forward-flow}
    Suppose that Assumption 
    \ref{assumption:lsc} and \ref{assumption:geo-convex} hold. Let $\tilde{\mu}_T$ be the average iterate defined in \eqref{eq:tilde-mu-t-1-forward-flow} - \eqref{eq:tilde-mu-t-2-forward-flow}. Then, for all $T\geq 1$, it holds that
    \begin{align}
        \mathcal{F}(\tilde{\mu}_{T}) - \mathcal{F}(\pi) &\leq \frac{1}{\sum_{t=1}^{T}\bar{r}_t} \sum_{t=1}^{T}\bar{r}_t \left[\mathcal{F}(\mu_t) - \mathcal{F}(\pi)\right]. \label{eq:average-ineq-forward-flow}
    \end{align}
\end{lemma}

\begin{proof}
    See Appendix \ref{app:additional-proofs-forward-flow-non-smooth-deterministic-adaptive}.
\end{proof}

\begin{lemma}
\label{lemma:tilde-mu-t-bound-2-forward-flow}
    Suppose that Assumption \ref{assumption:lsc} and \ref{assumption:geo-convex} hold. Let $(\mu_t)_{t\geq 0}$ denote the sequence of measures defined by \eqref{eq:forward-flow-1} - \eqref{eq:forward-flow-2}, with $(\eta_t)_{t\geq 0}$ defined as in \eqref{eq:dog-lr-forward-flow}. Then, for all $T\geq 1$,
    \begin{align}
        \sum_{t=1}^{T}\bar{r}_t\left[\mathcal{F}(\mu_t) - \mathcal{F}(\pi)\right] &\leq \sum_{t=1}^T \bar{r}_t \left[\frac{W_2^2(\mu_{t-\frac{1}{2}},\pi) - W_2^2(\mu_{t+\frac{1}{2}}, \pi)}{2\eta_t} + \frac{\eta_t}{2} \int_{\mathbb{R}^d} \|\zeta_t(x)\|^2\,\mathrm{d}\mu_t(x)\right].
    \end{align}
\end{lemma}

\begin{proof}
    The result follows immediately upon taking a weighted sum of the inequality in Corollary \ref{corollary:evi-forward-flow}, now setting $\eta=\eta_t$.
\end{proof}

\begin{lemma}
\label{lemma:tilde-mu-t-bound-2-forward-flow-v2}
    Suppose that Assumption \ref{assumption:lsc} and \ref{assumption:geo-convex} hold. Let $(\mu_t)_{t\geq 0}$ denote the sequence of measures defined by \eqref{eq:forward-flow-1} - \eqref{eq:forward-flow-2}, with $(\eta_t)_{t\geq 0}$ defined as in \eqref{eq:dog-lr-forward-flow}. Then, for all $T\geq 1$,
    \begin{equation}
    \sum_{t=1}^T \bar{r}_t \left[\frac{W_2^2(\mu_{t-\frac{1}{2}},\pi) - W_2^2(\mu_{t+\frac{1}{2}}, \pi)}{2\eta_t} + \frac{\eta_t}{2} \int_{\mathbb{R}^d} \|\zeta_t(x)\|^2\,\mathrm{d}\mu_t(x)\right] \leq \bar{r}_{T+1} (2\bar{d}_{T+1}+ \bar{r}_{T+1}) \sqrt{G_{T}}. \label{eq:tilde-mu-t-bound-2-forward-flow}
    \end{equation}
\end{lemma}

\begin{proof}
    See Appendix \ref{app:additional-proofs-forward-flow-non-smooth-deterministic-adaptive}.
\end{proof}

\begin{proposition}
\label{prop:dog-wgd-bound-1-forward-flow}
    Suppose that Assumption 
    \ref{assumption:lsc} and \ref{assumption:geo-convex} hold. Let $(\mu_t)_{t\geq 0}$ be the sequence of measures defined according to \eqref{eq:forward-flow-1} - \eqref{eq:forward-flow-2}, with $(\eta_t)_{t\geq 1}$ defined as in \eqref{eq:dog-lr-forward-flow}. Let $(\tilde{\mu}_t)_{t\geq 0}$ be the sequence of measures defined according to \eqref{eq:tilde-mu-t-1-forward-flow} - \eqref{eq:tilde-mu-t-2-forward-flow}. Then, for all $t\leq T$, we have 
    \begin{equation}
    \label{eq:dog-wgd-bound-1-forward-flow}
        \mathcal{F}(\tilde{\mu}_t) - \mathcal{F}(\pi) =  \mathcal{O}\left(\frac{(d_1 + \bar{r}_{t+1})\sqrt{G_{t}}}{\sum_{s=1}^{t}\bar{r}_s/\bar{r}_{t+1}}\right).
    \end{equation}
\end{proposition}

\begin{proof}
    The result follows as an immediate consequence of Lemmas \ref{lemma:tilde-mu-t-bound-1-forward-flow} - \ref{lemma:tilde-mu-t-bound-2-forward-flow-v2}, using also the fact that $\bar{d}_{t+1} \leq d_1 + \bar{r}_{t+1}$.
\end{proof}

\begin{corollary}
\label{corr:dog-wgd-bound-2-forward-flow}
    Suppose that Assumption 
    \ref{assumption:lsc} and \ref{assumption:geo-convex} hold. Let $(\mu_t)_{t\geq 0}$ be the sequence of measures defined according to \eqref{eq:forward-flow-1} - \eqref{eq:forward-flow-2}, with $(\eta_t)_{t\geq 1}$ defined as in \eqref{eq:dog-lr-forward-flow}. Let $(\tilde{\mu}_t)_{t\geq 0}$ be the sequence of measures defined according to \eqref{eq:tilde-mu-t-1-forward-flow} - \eqref{eq:tilde-mu-t-2-forward-flow}. Let $D\geq d_1$, and define 
    \begin{equation}
        G_D = \sup_{\mu\in\mathcal{P}_2(\mathcal{X}_{D})} \|\zeta(\mu)\|_{L^2(\mu)}
    \end{equation}
    where $\smash{\mathcal{P}_2(\mathcal{X}_{D}) = \{\mu\in\mathcal{P}_2(\mathcal{X}):W_2(\mu,\mu_{\frac{1}{2}})\leq D\}}$. In addition, let $\smash{\tau\in\argmax_{1\leq t\leq T} \sum_{s=1}^{t} \bar{r}_s /\bar{r}_{t+1}}$. Then, on the event $\{\bar{r}_{T+1} \leq D\}\cup\{\bar{r}'_{T+1} \leq D\}$, it holds that
    \begin{equation}
        \mathcal{F}(\tilde{\mu}_{\tau}) - \mathcal{F}(\pi) = \mathcal{O}\left(\frac{DG_D}{\sqrt{T}}\log_{+}\left(\frac{D}{r_{\varepsilon}}\right)\right).
    \end{equation}
\end{corollary}

\begin{proof}
    See Appendix \ref{app:additional-proofs-forward-flow-non-smooth-deterministic-adaptive}.
\end{proof}

\paragraph{Uniform Averaging} In fact, it is also possible to derive guarantees similar to those obtained in Proposition \ref{prop:dog-wgd-bound-1-forward-flow} and Corollary \ref{corr:dog-wgd-bound-2-forward-flow} for the unweighted average $\bar{\mu}_T$ defined in \eqref{eq:bar-mu-t-forward-flow-1-first-def} - \eqref{eq:bar-mu-t-forward-flow-2-first-def}. 

\begin{proposition}
\label{prop:dog-wgd-bound-1-forward-flow-uniform}
    Suppose that Assumption 
    \ref{assumption:lsc} and \ref{assumption:geo-convex} hold. Let $(\mu_t)_{t\geq 0}$ be the sequence of measures defined according to \eqref{eq:forward-flow-1} - \eqref{eq:forward-flow-2}, with $(\eta_t)_{t\geq 1}$ defined as in \eqref{eq:dog-lr-forward-flow}. Suppose there exists $0<\eta_{\mathrm{max}}<\infty$ such that $\eta_t\leq \eta_{\mathrm{max}}$ for all $1\leq t\leq T$. Let $(\bar{\mu}_t)_{t\geq 0}$ be the sequence of measures defined according to \eqref{eq:bar-mu-t-forward-flow-1-first-def} - \eqref{eq:bar-mu-t-forward-flow-2-first-def}.  Then, for all $T\geq 1$, we have
    \begin{equation}
    \label{eq:dog-wgd-bound-1-forward-flow-uniform}
        \mathcal{F}(\bar{\mu}_T) - \mathcal{F}(\pi) =  \mathcal{O}\left(\frac{(d_1\log_{+}\frac{\bar{r}_{T+1}}{r_{\varepsilon}} + \bar{r}_{T+1})\sqrt{G_{T}}}{T}\right).
    \end{equation}
\end{proposition}

\begin{proof}
    See Appendix \ref{app:additional-proofs-forward-flow-non-smooth-deterministic-adaptive}.
\end{proof}

\begin{corollary}
\label{corr:dog-wgd-bound-2-forward-flow-uniform}
    Suppose that Assumption 
    \ref{assumption:lsc} and \ref{assumption:geo-convex} hold. Let $(\mu_t)_{t\geq 0}$ be the sequence of measures defined according to \eqref{eq:forward-flow-1} - \eqref{eq:forward-flow-2}, with $(\eta_t)_{t\geq 1}$ defined as in \eqref{eq:dog-lr-forward-flow}. Suppose there exists $0<\eta_{\mathrm{max}}<\infty$ such that $\eta_t\leq \eta_{\mathrm{max}}$ for all $1\leq t\leq T$.  Let $(\bar{\mu}_t)_{t\geq 0}$ be the sequence of measures defined according to \eqref{eq:bar-mu-t-forward-flow-1-first-def} - \eqref{eq:bar-mu-t-forward-flow-2-first-def}. Let $D\geq d_1$, and define 
    \begin{equation}
        G_D = \sup_{\mu\in\mathcal{P}_2(\mathcal{X}_{D})} \|\zeta(\mu)\|_{L^2(\mu)}
    \end{equation}
    where $\smash{\mathcal{P}_2(\mathcal{X}_{D}) = \{\mu\in\mathcal{P}_2(\mathcal{X}):W_2(\mu,\mu_{\frac{1}{2}})\leq D\}}$. Then, on the event $\{\bar{r}_{T+1} > D\}\cup\{\bar{r}'_{T+1} \leq D\}$, it holds that
    \begin{equation}
        \mathcal{F}(\bar{\mu}_{T}) - \mathcal{F}(\pi) = \mathcal{O}\left(\frac{DG_D}{\sqrt{T}}\log_{+}\left(\frac{D}{r_{\varepsilon}}\right)\right).
    \end{equation}
\end{corollary}

\begin{proof}
See Appendix \ref{app:additional-proofs-forward-flow-non-smooth-deterministic-adaptive}.
\end{proof}

\begin{tcolorbox}[enhanced,
  colback=white,
  frame hidden,
  borderline north={0.5pt}{0pt}{black!70},
  borderline south={0.5pt}{0pt}{black!70},
  arc=2pt,             
  left=0pt,            
  right=0pt,           
  top=4pt,             
  bottom=4pt,           
  parbox=false,
  before skip=10pt, 
  after skip=15pt]
\begin{example1}[\textsc{Sampling from a Target Probability Distribution}]
Recall again our running example: $\mathcal{F}(\mu) = \mathrm{KL}(\mu\|\pi)$, with $\pi\propto e^{-U}$. In this case, the \textsc{Fuse} step size schedule $(\eta_t)_{t\geq 1}$ for the forward-flow discretization of the WGF (i.e., ULA) is given by
\begin{equation}
    \eta_t = \frac{\max\left[r_{\varepsilon}, \max_{1\leq s\leq t}W_2(\mu_{\frac{1}{2}}, \mu_{s-\frac{1}{2}})\right]}{\sqrt{\sum_{s=1}^t \int_{\mathbb{R}^d}\|\nabla U(x)\|^2\,\mathrm{d}\mu_s(x)}}.
    \label{eq:dog-step-size-lmc}
\end{equation}
Due to Corollary \ref{corr:dog-wgd-bound-2-forward-flow} and Corollary \ref{corr:dog-wgd-bound-2-forward-flow-uniform}, we then have the following guarantees. For any $D\geq W_2(\mu_{\frac{1}{2}},\pi)$, on the event $\{\bar{r}_{T+1} \leq D\}\cup \{\bar{r}_{T+1}'>D\}$, it holds that
\begin{align}
    \mathrm{KL}\left(\tilde{\mu}_{\tau}\|\pi\right) &= \mathcal{O}\left(\frac{DG_D}{\sqrt{T}}\log_{+}\left(\frac{D}{r_{\varepsilon}}\right)\right)
\intertext{and}
    \mathrm{KL}\left(\bar{\mu}_{T}\|\pi\right) &= \mathcal{O}\left(\frac{DG_D}{\sqrt{T}}\log_{+}\left(\frac{D}{r_{\varepsilon}}\right)\right)
\end{align}
where $\tilde{\mu}_T$ is the weighted average defined in \eqref{eq:tilde-mu-t-1-forward-flow} - \eqref{eq:tilde-mu-t-2-forward-flow}, $\bar{\mu}_T$ is the uniform average defined in \eqref{eq:bar-mu-t-forward-flow-1-first-def} - \eqref{eq:bar-mu-t-forward-flow-2-first-def}, and where 
\begin{equation}
    G_D = \sup_{\mu\in\mathcal{P}_2(\mathcal{X}_{D})}  [\int \| \nabla U(x)\|^2\mathrm{d}\mu(x)]^{\frac{1}{2}}.
\end{equation}
Thus, in particular, our algorithm matches the convergence rate of optimally tuned ULA, cf. \eqref{eq:optimal-average-bound-forward-flow-kl-bounded-grad}, up to an additional log-factor.

\vspace{2mm}
\paragraph{The Gaussian Case}
In general, it is still not possible to compute  $(\eta_t)_{t\geq 1}$ directly,  since both the Wasserstein distance  $\smash{W_2(\mu_{\frac{1}{2}},\mu_{i-\frac{1}{2}})}$
and the expectations $\int_{\mathbb{R}^d} [\cdot]\,\mathrm{d}\mu_s(x)$ are unknown. In the case where $\mu_0 =\mathcal{N}(m_0,\Sigma_0)$ and $\pi = \mathcal{N}(m_{\pi},\Sigma_{\pi})$ are both Gaussian, however, we can compute everything in closed form. In particular, following our previous calculations, cf. \eqref{eq:ula-gaussian-mean} - \eqref{eq:grad-lmc-gaussian-v1}, we have
\begin{equation}
    \eta_t = \frac{\max_{1\leq s\leq t}\left[ \sqrt{\|m_{\frac{1}{2}} -m_{s-\frac{1}{2}}\|^2 + \mathrm{Tr}\big(\Sigma_{\frac{1}{2}} + \Sigma_{s-\frac{1}{2}} - 2(\Sigma_{\frac{1}{2}}^{\frac{1}{2}}\Sigma_{s-\frac{1}{2}}\Sigma_{\frac{1}{2}}^{\frac{1}{2}})^{\frac{1}{2}}\big)},r_{\varepsilon}\right]}{\sqrt{\sum_{s=1}^t \left(\mathrm{Tr}(\Sigma_{\pi}^{-2}\Sigma_s) + \|\Sigma_{\pi}^{-1}\left(m_s - m_{\pi}\right)\|^2\right)}}.
    \label{eq:dog-step-size-lmc-gaussian}
\end{equation}

\end{example1}
\end{tcolorbox}

\paragraph{Iterate Stability Bound.} The results in the previous section are immediately useful when $\mathcal{X}\subseteq\mathbb{R}^d$ and hence $\mathcal{P}_2(\mathcal{X})\subseteq\mathcal{P}_2(\mathbb{R}^d)$ is bounded, and thus the iterates $(\mu_t)_{t\geq 0}$ generated by the forward-flow discretization of the WGF are also bounded. In practice, the iterates remain stable even when this is not the case. Nonetheless, it remains useful to define a variant of the dynamic step size formula $(\eta_t)_{t\geq 1}$ in \eqref{eq:dog-lr-forward-flow} which ensures that the iterates $(\mu_t)_{t\geq 0}$ remain bounded. To this end, we now define 
\begin{equation}
    \eta_t = \frac{\max\Big[r_{\varepsilon},\max_{1\leq s\leq t} W_2(\mu_{\frac{1}{2}},\mu_{s-\frac{1}{2}})\Big]}{\sqrt{8^4\log_{+}^2 (1+ t\bar{g}_t^2/\bar{g}_1^2)(\sum_{s=1}^{t-1} \int_{\mathbb{R}^d}\|\zeta_s(x)\|^2\,\mathrm{d}\mu_s(x) + 16\bar{g}_t^2)}}:= \frac{\bar{r}_t}{\sqrt{G'_{t}}} \label{eq:tamed-dog-lr-forward-flow}
\end{equation}
where we have adopted the convention that $\sum_{i=1}^{0}[\cdot] = 0$, we recall that $\bar{g}_t = \max_{1\leq s \leq t} \|\zeta_s\|_{L^2(\mu_s)}$, and in the second equality we have defined 
\begin{equation}
    G'_t = 8^4\log_{+}^2 \left(1+ \frac{t\bar{g}_t^2}{\bar{g}_1^2}\right)\left(G_{t-1} + 16\bar{g}_t^2\right).
\end{equation}
In this case, we have the following iterate stability guarantee.

\begin{proposition}
\label{prop:iterate-stability-forward-flow}
    Suppose that Assumption 
    \ref{assumption:lsc} and \ref{assumption:geo-convex} hold. Let $(\mu_t)_{t\geq 0}$ be the sequence of measures defined according to \eqref{eq:forward-flow-1} - \eqref{eq:forward-flow-2}, with $(\eta_t)_{t\geq 1}$ defined as in \eqref{eq:tamed-dog-lr-forward-flow}. Let $(\tilde{\mu}_t)_{t\geq 0}$ be the sequence of measures defined according to \eqref{eq:tilde-mu-t-1-forward-flow} - \eqref{eq:tilde-mu-t-2-forward-flow}. In addition, let $r_{\varepsilon}\leq 3d_1$. Then, for all $T\geq 1$, $\bar{r}_T\leq 3d_1$.
\end{proposition}

\begin{proof}
    See Appendix \ref{app:additional-proofs-forward-flow-non-smooth-deterministic-adaptive}.
\end{proof}

\subsubsection{Stochastic Case: Adaptive Step Size}
\label{sec:stochastic-theory}
We now turn our attention to the stochastic variant of the forward-flow discretization of the WGF, cf. \eqref{eq:stochastic-forward-flow-1} - \eqref{eq:stochastic-forward-flow-2}, or equivalently \eqref{eq:stochastic-forward-flow-lagrangian-1} - \eqref{eq:stochastic-forward-flow-lagrangian-2}. In this case, the adaptive step size is given by
\begin{equation}
\tcboxmath[colback=black!5,colframe=black!15,boxrule=0.4pt, arc=2pt, left=8pt, right=8pt, top=6pt, bottom=6pt]{
    \eta_t = \frac{\max\left[r_{\varepsilon},\max_{1\leq s\leq t} W_2(\mu_{\frac{1}{2}},\mu_{s-\frac{1}{2}})\right]}{\sqrt{\sum_{s=1}^{t} \int_{\mathbb{R}^d}\|\hat{\zeta}_s(x)\|^2\,\mathrm{d}\mu_s(x)}}
}. \label{eq:dog-lr-stoc-forward-flow}
\end{equation} 
Similar to before, cf. \eqref{eq:defs1-forward-flow} - \eqref{eq:defs2-forward-flow}, it will be convenient to define some additional notation. In particular, we will frequently make use of
\begin{alignat}{3}
    &r_t = W_2(\mu_{\frac{1}{2}},\mu_{t-\frac{1}{2}}),\quad &&\bar{r}_t = \max_{1\leq s \leq t}\left[r_s,r_\varepsilon\right],\quad&&G_t= \sum_{s=1}^t \int_{\mathbb{R}^d} \|\hat{\zeta}_s(x)\|^2\,\mathrm{d}\mu_s(x),\label{eq:defs1-stoc-forward-flow} \\
    &r'_t = W_2(\mu_{\frac{1}{2}},\mu_{t}),\quad &&\bar{r}'_t = \max_{1\leq s \leq t}\left[r'_s,r_\varepsilon\right],\label{eq:defs1-1-stoc-forward-flow} \\[1mm]
    &d_t = W_2(\mu_{t-\frac{1}{2}},\pi),\quad &&\bar{d}_t = \max_{1\leq s \leq t} d_s
    ,\quad &&\bar{g}_t = \max_{1\leq s \leq t} g(\mu_s). \label{eq:defs2-stoc-forward-flow}
\end{alignat}
In addition, for our stochastic analysis, we will require an additional definition:
\begin{equation}
    \theta_{t,\delta}:= \log \left(\frac{60\log (6t)}{\delta}\right).
\end{equation}
We begin with a result which decomposes the error into two terms: a weighted regret term and a noise term. 

\begin{lemma}
\label{lemma:tilde-mu-t-bound-1-stoch-forward-flow}
    Suppose that Assumption 
    \ref{assumption:lsc} and \ref{assumption:geo-convex} hold. Let $(\mu_t)_{t\geq 0}$ be the sequence of measures defined according to \eqref{eq:stochastic-forward-flow-1} - \eqref{eq:stochastic-forward-flow-2}, with $(\eta_t)_{t\geq 1}$ defined as in \eqref{eq:dog-lr-stoc-forward-flow}. Let $(\tilde{\mu}_t)_{t\geq 0}$ be the sequence of measures defined according to \eqref{eq:tilde-mu-t-1-forward-flow} - \eqref{eq:tilde-mu-t-2-forward-flow}. In addition, let  $\delta_t(x):=\hat{\zeta}_t(x) - {\zeta}_t(x)$. Then, for all $T\geq 1$, 
    \begin{align}
    \sum_{t=1}^{T}\bar{r}_t\left[\mathcal{F}(\mu_t) - \mathcal{F}(\pi)\right] &\leq \sum_{t=1}^T \bar{r}_t  \bigg[\frac{W_2^2(\mu_{t-\frac{1}{2}},\pi) - W_2^2(\mu_{t+\frac{1}{2}},\pi)}{2\eta_t} + \frac{\eta_t}{2} \int_{\mathbb{R}^d} \|\hat{\zeta}_t(x)\|^2\,\mathrm{d}\mu_t(x)\bigg] \\
        &- \sum_{t=1}^{T}\bar{r}_t\bigg[\int_{\mathbb{R}^d}\langle {\delta}_t(x), x- \boldsymbol{t}_{\mu_t}^{\pi}(x) \rangle \,\mathrm{d}\mu_t(x)\bigg]. \label{eq:decomp-stoc-forward-flow} 
    \end{align}
\end{lemma}

\begin{proof}
    See Appendix \ref{app:additional-proofs-forward-flow-non-smooth-stochastic-adaptive}.
\end{proof}

\begin{lemma}
\label{lemma:tilde-mu-t-bound-2-stoch-forward-flow}
Suppose that Assumption 
    \ref{assumption:lsc} and \ref{assumption:geo-convex} hold. Let $(\mu_t)_{t\geq 0}$ be the sequence of measures defined according to \eqref{eq:stochastic-forward-flow-1} - \eqref{eq:stochastic-forward-flow-2}, with $(\eta_t)_{t\geq 1}$ defined as in \eqref{eq:dog-lr-stoc-forward-flow}. Let $(\tilde{\mu}_t)_{t\geq 0}$ be the sequence of measures defined according to \eqref{eq:tilde-mu-t-1-forward-flow} - \eqref{eq:tilde-mu-t-2-forward-flow}. In addition, let  $\delta_t(x):=\hat{\zeta}_t(x) - {\zeta}_t(x)$. Then, for all $T\geq 1$, 
    \begin{align}
    \sum_{t=1}^T \bar{r}_t  \bigg[\frac{W_2^2(\mu_{t-\frac{1}{2}},\pi) - W_2^2(\mu_{t+\frac{1}{2}},\pi)}{2\eta_t} + \frac{\eta_t}{2} \int_{\mathbb{R}^d} \|\hat{\zeta}_t(x)\|^2\,\mathrm{d}\mu_t(x)\bigg]\leq \bar{r}_{T+1} (2\bar{d}_{T+1}+ \bar{r}_{T+1}) \sqrt{G_{T}}\quad \text{a.s.}
    \end{align}
\end{lemma}

\begin{proof}
    See Appendix \ref{app:additional-proofs-forward-flow-non-smooth-stochastic-adaptive}.
\end{proof}

\begin{lemma}
\label{lemma:tilde-mu-t-bound-3-stoch-forward-flow}
    Suppose that Assumption \ref{assumption:bounded-grads-forward-flow} holds. Let $(\mu_t)_{t\geq 0}$ be the sequence of measures defined according to \eqref{eq:stochastic-forward-flow-1} - \eqref{eq:stochastic-forward-flow-2}, with $(\eta_t)_{t\geq 1}$ defined as in \eqref{eq:dog-lr-stoc-forward-flow}. Let $(\tilde{\mu}_t)_{t\geq 0}$ be the sequence of measures defined according to \eqref{eq:tilde-mu-t-1-forward-flow} - \eqref{eq:tilde-mu-t-2-forward-flow}. Then, for all $0<\delta<1$, $G>0$, and $T\geq 1$, 
    \begin{align}
        \mathbb{P}\bigg(\exists t \leq T : \bigg|\sum_{k=1}^{t} \bar{r}_k \int_{\mathbb{R}^d} \langle \delta_k(x),x - \boldsymbol{t}_{\mu_k}^{\pi}(x)\rangle\mathrm{d}\mu_k(x) \bigg| \geq 8 \bar{r}_{t}\bar{d}_{t}\sqrt{\theta_{t,\delta} G_{t} + \theta_{t,\delta}^2G^2}\bigg) \leq \delta + \mathbb{P}\big(\bar{g}_{T}> G\big).
    \end{align}
\end{lemma}

\begin{proof}
    See Appendix \ref{app:additional-proofs-forward-flow-non-smooth-stochastic-adaptive}.
\end{proof}

\begin{proposition}
\label{prop:dog-wgd-bound-1-stoc-forward-flow}
    Suppose that Assumptions 
    \ref{assumption:lsc}, \ref{assumption:geo-convex}, and \ref{assumption:bounded-grads-forward-flow} hold. Let $(\mu_t)_{t\geq 0}$ be the sequence of measures defined according to \eqref{eq:stochastic-forward-flow-1} - \eqref{eq:stochastic-forward-flow-2}, with $(\eta_t)_{t\geq 1}$ defined as in \eqref{eq:dog-lr-stoc-forward-flow}. Let $(\tilde{\mu}_t)_{t\geq 0}$ be the sequence of measures defined according to \eqref{eq:tilde-mu-t-1-forward-flow} - \eqref{eq:tilde-mu-t-2-forward-flow}. Then, for all $0<\delta <1$, $G>0$, and for all $t\leq T$, we have with probability at least $1 - \delta - \mathbb{P}(\bar{g}_T>G)$, 
    \begin{equation}
    \label{eq:dog-wgd-bound-1-stoc-forward-flow}
        \mathcal{F}(\tilde{\mu}_t) - \mathcal{F}(\pi) =  \mathcal{O}\left(\frac{(d_1 + \bar{r}_{t+1})\sqrt{G_{t} + G_{t}\theta_{t,\delta} + \theta_{t,\delta}^2 G^2}}{\sum_{s=1}^{t}\bar{r}_s/\bar{r}_{t+1}}\right).
    \end{equation}
\end{proposition}

\begin{proof}
    The result follows as an immediate consequence of Lemma \ref{lemma:tilde-mu-t-bound-1-stoch-forward-flow}, Lemma \ref{lemma:tilde-mu-t-bound-2-stoch-forward-flow}, and Lemma \ref{lemma:tilde-mu-t-bound-3-stoch-forward-flow}, using also the facts that $\bar{d}_t\leq \bar{d}_{t+1}$, $\bar{r}_{t}\leq \bar{r}_{t+1}$, and $\bar{d}_{t+1} \leq d_1 + \bar{r}_{t+1}$.
\end{proof}

\begin{corollary}
    Suppose that Assumption 
    \ref{assumption:lsc} and \ref{assumption:geo-convex} hold. Let $(\mu_t)_{t\geq 0}$ be the sequence of measures defined according to \eqref{eq:stochastic-forward-flow-1} - \eqref{eq:stochastic-forward-flow-2}, with $(\eta_t)_{t\geq 1}$ defined as in \eqref{eq:dog-lr-stoc-forward-flow}. Let $(\tilde{\mu}_t)_{t\geq 0}$ be the sequence of measures defined according to \eqref{eq:tilde-mu-t-1-forward-flow} - \eqref{eq:tilde-mu-t-2-forward-flow}. Let $D\geq d_1$, and define 
    \begin{equation}
        G_D = \sup_{\mu\in\mathcal{P}_2(\mathcal{X}_{D})} \|{g}(\mu)\|_{L^2(\mu)}
    \end{equation}
    where $\smash{\mathcal{P}_2(\mathcal{X}_{D}) = \left\{\mu\in\mathcal{P}_2(\mathcal{X}):W_2(\mu,\mu_{\frac{1}{2}})\leq D\right\}}$. In addition, let $\smash{\tau\in\argmax_{1\leq t\leq T} \sum_{s=1}^{t} \bar{r}_s /\bar{r}_{t+1}}$. Then, with probability at least $1 - \delta - \mathbb{P}(\{\bar{r}_{T+1} > D\}\cup\{\bar{r}'_{T+1}>D\})$, it holds that
    \begin{equation}
        \mathcal{F}(\tilde{\mu}_{\tau}) - \mathcal{F}(\pi) = \mathcal{O}\left(\frac{D\sqrt{G_{\tau}\theta_{\tau,\delta} + G_D^2 \theta_{\tau,\delta}^2}}{T}\log_{+}\left(\frac{D}{r_{\varepsilon}}\right)\right)=\mathcal{O}\left(\frac{DG_D}{\sqrt{T}}\theta_{\tau,\delta}\log_{+}\left(\frac{D}{r_{\varepsilon}}\right)\right).
    \end{equation}
\end{corollary}

\begin{proof}
    The result follows from Proposition \ref{prop:dog-wgd-bound-1-stoc-forward-flow} in the same way that Corollary \ref{corr:dog-wgd-bound-2-forward-flow} follows from Proposition \ref{prop:dog-wgd-bound-1-forward-flow}.
\end{proof}

\paragraph{Uniform Averaging} Similar to the deterministic case, we can also derive guarantees for the uniform average iterate $\bar{\mu}_t$, as defined in \eqref{eq:bar-mu-t-forward-flow-1-first-def} - \eqref{eq:bar-mu-t-forward-flow-2-first-def}.

\begin{proposition}
\label{prop:dog-wgd-bound-1-stoc-forward-flow-uniform}
    Suppose that Assumptions 
    \ref{assumption:lsc}, \ref{assumption:geo-convex}, and \ref{assumption:bounded-grads-forward-flow} hold. Let $(\mu_t)_{t\geq 0}$ be the sequence of measures defined according to \eqref{eq:stochastic-forward-flow-1} - \eqref{eq:stochastic-forward-flow-2}, with $(\eta_t)_{t\geq 1}$ defined as in \eqref{eq:dog-lr-stoc-forward-flow}. Let $(\bar{\mu}_t)_{t\geq 0}$ be the sequence of measures defined according to \eqref{eq:bar-mu-t-forward-flow-1-first-def} - \eqref{eq:bar-mu-t-forward-flow-2-first-def}. Then, for all $0<\delta <1$, $G>0$, and for all $T\geq 1$, we have with probability at least $1 - \delta - \mathbb{P}(\bar{g}_T>G)$, 
    \begin{equation}
    \label{eq:dog-wgd-bound-1-stoc-forward-flow-uniform}
        \mathcal{F}(\bar{\mu}_T) - \mathcal{F}(\pi) =  \mathcal{O}\left(\frac{(d_1\log_{+}\frac{\bar{r}_{T+1}}{r_{\varepsilon}} + \bar{r}_{T+1})\sqrt{G_{T} + G_{T}\theta_{T,\delta} + \theta_{T,\delta}^2 G^2}}{T}\right).
    \end{equation}
\end{proposition}

\begin{proof}
    The result follows as an immediate consequence of Lemma \ref{lemma:tilde-mu-t-bound-1-stoch-forward-flow}, Lemma \ref{lemma:tilde-mu-t-bound-2-stoch-forward-flow}, and Lemma \ref{lemma:tilde-mu-t-bound-3-stoch-forward-flow}, using also the fact that $\bar{d}_t \leq d_1 + \bar{r}_t$.
\end{proof}

\begin{corollary}
    Suppose that Assumption 
    \ref{assumption:lsc} and \ref{assumption:geo-convex} hold. Let $(\mu_t)_{t\geq 0}$ be the sequence of measures defined according to \eqref{eq:stochastic-forward-flow-1} - \eqref{eq:stochastic-forward-flow-2}, with $(\eta_t)_{t\geq 1}$ defined as in \eqref{eq:dog-lr-stoc-forward-flow}. Let $(\tilde{\mu}_t)_{t\geq 0}$ be the sequence of measures defined according to \eqref{eq:tilde-mu-t-1-forward-flow} - \eqref{eq:tilde-mu-t-2-forward-flow}. Let $D\geq d_1$, and define 
    \begin{equation}
        G_D = \sup_{\mu\in\mathcal{P}_2(\mathcal{X}_{D})} \|{g}(\mu)\|_{L^2(\mu)},
    \end{equation}
    where $\smash{\mathcal{P}_2(\mathcal{X}_{D}) = \{\mu\in\mathcal{P}_2(\mathcal{X}):W_2(\mu,\mu_{\frac{1}{2}})\leq D\}}$. Then, with probability at least $1 - \delta - \mathbb{P}(\{\bar{r}_{T+1} > D\}\cup\{\bar{r}'_{T+1}>D\})$, it holds that
    \begin{equation}
        \mathcal{F}(\bar{\mu}_{T}) - \mathcal{F}(\pi) = \mathcal{O}\left(\frac{D\sqrt{G_{T}\theta_{T,\delta} + G_D^2 \theta_{T,\delta}^2}}{T}\log_{+}\left(\frac{D}{r_{\varepsilon}}\right)\right)=\mathcal{O}\left(\frac{DG_D}{\sqrt{T}}\theta_{T,\delta}\log_{+}\left(\frac{D}{r_{\varepsilon}}\right)\right).
    \end{equation}
\end{corollary}

\begin{proof}
    The result follows from Proposition \ref{prop:dog-wgd-bound-1-stoc-forward-flow-uniform} in the same way that Corollary \ref{corr:dog-wgd-bound-2-forward-flow-uniform} follows from Proposition \ref{prop:dog-wgd-bound-1-forward-flow-uniform}.
\end{proof}

\subsection{Forward-Flow Discretization, Smooth Setting}
\label{sec:forward-flow-smooth}
We now turn our attention to the smooth setting. Interestingly, unlike the Euclidean case, the (optimal) convergence rate will remain of order $\smash{\mathcal{O}(T^{-\frac{1}{2}})}$.

\subsubsection{Deterministic Case: Constant Step Size}
Similar to before, we begin our analysis in the case where the step size in \eqref{eq:forward-flow-1} - \eqref{eq:forward-flow-2} is constant, that is, $\eta_t=\eta$ for all $t\geq 0$. 

\begin{lemma}
\label{lemma:g-descent-smooth}
Suppose that Assumptions \ref{assumption:lsc}, \ref{assumption:geo-convex}, and \ref{assumption:smooth-forward-flow} hold. Let $(\mu_t)_{t\geq 0}$ denote the sequence of measures defined by \eqref{eq:forward-flow-1} - \eqref{eq:forward-flow-2}. Suppose that $\eta_t = \eta$ for all $t\geq 0$, with $\eta\in(0,\frac{1}{L}]$. Then, for any $t\geq 1$, we have 
\begin{equation}
\label{eq:g-descent-smooth}
\mathcal{E}(\mu_{t+\frac12}) \leq 
\mathcal{E}(\mu_t)-\eta\left(1-\tfrac{L\eta}{2}\right)\int_{\mathbb{R}^d}\|\zeta_t(x)\|^2\,\mathrm d\mu_t(x).
\end{equation}
\end{lemma}

\begin{proof}
    See Appendix \ref{app:additional-proofs-forward-flow-smooth-deterministic-constant}.
\end{proof}

\begin{lemma}
\label{lem:heat-step-tight}
Suppose that Assumptions \ref{assumption:lsc}, \ref{assumption:geo-convex}, and \ref{assumption:smooth-forward-flow} hold. Let $(\mu_t)_{t\geq 0}$ denote the sequence of measures defined by \eqref{eq:forward-flow-1} - \eqref{eq:forward-flow-2}. Suppose that $\eta_t = \eta>0$ for all $t\geq 0$. Then, for any $t\geq 1$, we have
\begin{equation}
\mathcal{E}(\mu_{t+1}) \leq \mathcal{E}(\mu_{t+\frac12}) + L\eta d.
\end{equation}
\end{lemma}

\begin{proof}
    See Appendix \ref{app:additional-proofs-forward-flow-smooth-deterministic-constant}.
\end{proof}

\begin{corollary}
\label{corollary:g-step-descent-smooth}
Suppose that Assumptions \ref{assumption:lsc}, \ref{assumption:geo-convex}, and \ref{assumption:smooth-forward-flow} hold. Let $(\mu_t)_{t\geq 0}$ denote the sequence of measures defined by \eqref{eq:forward-flow-1} - \eqref{eq:forward-flow-2}. Suppose that $\eta_t = \eta$ for all $t\geq 0$, with $\eta\in(0,\frac{1}{L}]$. Then, for any $t\geq 1$, we have 
\begin{equation}
    \mathcal{E}(\mu_{t+1}) \leq \mathcal{E}(\mu_t) -\eta\left(1-\tfrac{L\eta}{2}\right)\int_{\mathbb{R}^d}\|\zeta_t(x)\|^2\,\mathrm d\mu_t(x) + L\eta d.
\end{equation}
\end{corollary}
\begin{proof}
    The result is an immediate corollary of Lemma \ref{lemma:g-descent-smooth} and Lemma \ref{lem:heat-step-tight}.
\end{proof}

\begin{corollary}
\label{corollary:evi-g}
Suppose that Assumptions \ref{assumption:lsc}, \ref{assumption:geo-convex}, and \ref{assumption:smooth-forward-flow} hold. Let $(\mu_t)_{t\geq 0}$ denote the sequence of measures defined by \eqref{eq:forward-flow-1} - \eqref{eq:forward-flow-2}. Suppose that $\eta_t = \eta$ for all $t\geq 0$, with $\eta\in(0,\frac{1}{L}]$. Then, for any $t\geq 1$, and for any $\pi\in\mathcal{P}_2(\mathbb{R}^d)$, we have
\begin{equation}
    \mathcal{E}(\mu_{t+1}) -  \mathcal{E}(\pi) \leq  \frac{W_2^2(\mu_t,\pi)-W_2^2(\mu_{t+\frac12},\pi)}{2\eta}
-\frac{\eta}{2}\,(1-L\eta)\int_{\mathbb{R}^d}\!\|\zeta_t(x)\|^2\,\mathrm{d}\mu_t(x) + L\eta d.
\end{equation}
\end{corollary}

\begin{proof}
    The result follows from Corollary \ref{corollary:evi-forward-flow} and Corollary \ref{corollary:g-step-descent-smooth},  based on the decomposition $\mathcal{E}(\mu_{t+1}) - \mathcal{E}(\pi) = [\mathcal{E}(\mu_{t+1}) - \mathcal{E}(\mu_t)]+[\mathcal{E}(\mu_{t}) - \mathcal{E}(\pi)]$.
\end{proof}

\begin{corollary}
\label{corollary:f-pi-bound-heat-flow-smooth}
Suppose that Assumptions \ref{assumption:lsc}, \ref{assumption:geo-convex}, and \ref{assumption:smooth-forward-flow} hold. Let $(\mu_t)_{t\geq 0}$ denote the sequence of measures defined by \eqref{eq:forward-flow-1} - \eqref{eq:forward-flow-2}. Suppose that $\eta_t = \eta$ for all $t\geq 0$, with $\eta\in(0,\frac{1}{L}]$. Then, for any $t\geq 1$, and for any $\pi\in\mathcal{P}_2(\mathbb{R}^d)$, we have
\begin{equation}
    \mathcal{F}(\mu_{t+1}) - \mathcal{F}(\pi) \leq \frac{W_2^2(\mu_t,\pi)- W_2^2(\mu_{t+1},\pi)}{2\eta}
-\frac{\eta}{2}\,(1-L\eta)\int_{\mathbb{R}^d}\!\|\zeta_t(x)\|^2\,\mathrm{d}\mu_t(x) + L\eta d. \label{eq:f-pi-bound-heat-flow-smooth-full}
\end{equation}
In fact, since $\eta\in(0,\frac{1}{L}]$, the second term is non-positive, and so
\begin{equation}
    \mathcal{F}(\mu_{t+1}) - \mathcal{F}(\pi) \leq \frac{W_2^2(\mu_t,\pi)- W_2^2(\mu_{t+1},\pi)}{2\eta}
 + L\eta d.
 \label{eq:f-pi-bound-heat-flow-smooth}
\end{equation}
\end{corollary}

\begin{proof}
    The corollary is a direct consequence of Lemma \ref{lemma:evi-forward-flow-heat} and Corollary \ref{corollary:evi-g}, using the fact that $\mathcal{F}(\mu_{t+1})-\mathcal{F}(\pi) = [\mathcal{E}(\mu_{t+1}) - \mathcal{E}(\pi)] + [{\mathcal{H}(\mu_{t+1}) - \mathcal{H}(\pi)}]$.
\end{proof}

\begin{corollary}
\label{corollary:sum-f-pi-bound-heat-flow-smooth}
    Suppose that Assumptions \ref{assumption:lsc}, \ref{assumption:geo-convex}, and \ref{assumption:smooth-forward-flow} hold. Let $(\mu_t)_{t\geq 0}$ denote the sequence of measures defined by \eqref{eq:forward-flow-1} - \eqref{eq:forward-flow-2}. Suppose that $\eta_t = \eta$ for all $t\geq 0$, with $\eta\in(0,\frac{1}{L}]$. Let $\bar{\mu}_{T}$ be the average iterate defined in \eqref{eq:bar-mu-t-forward-flow-1-first-def} - \eqref{eq:bar-mu-t-forward-flow-2-first-def}. Then, for all $\pi\in\mathcal{P}_2(\mathbb{R}^d)$, 
    \begin{equation}
        \sum_{t=1}^T \mathcal{F}(\mu_{t}) - \sum_{t=1}^T  \mathcal{F}(\pi) \leq \frac{W_2^2(\mu_{0},\pi)- W_2^2(\mu_{T},\pi)}{2\eta} + LT\eta d.
        \label{eq:sum-f-pi-bound-heat-flow-smooth}
    \end{equation}
\end{corollary}

\begin{proof}
    The result follows immediately from Corollary \ref{corollary:f-pi-bound-heat-flow-smooth}, after summing \eqref{eq:f-pi-bound-heat-flow-smooth} over $t=0,\dots,T-1$; cancelling like terms in the telescoping sum; and changing indices. 
\end{proof}

\begin{corollary}
\label{corollary:average-bound-forward-flow-smooth}
    Suppose that Assumptions \ref{assumption:lsc}, \ref{assumption:geo-convex}, and \ref{assumption:smooth-forward-flow} hold. Let $(\mu_t)_{t\geq 0}$ denote the sequence of measures defined by \eqref{eq:forward-flow-1} - \eqref{eq:forward-flow-2}. Suppose that $\eta_t = \eta$ for all $t\geq 0$, with $\eta\in(0,\frac{1}{L}]$. Let $\bar{\mu}_{T}$ be the average iterate defined in \eqref{eq:bar-mu-t-forward-flow-1-first-def} - \eqref{eq:bar-mu-t-forward-flow-2-first-def}. Then, for all $\pi\in\mathcal{P}_2(\mathbb{R}^d)$, 
    \begin{equation}
    \label{eq:average-bound-forward-flow-smooth}
        \mathcal{F}(\bar{\mu}_T) - \mathcal{F}(\pi) \leq \frac{1}{T}\left[\frac{W_2^2(\mu_{0},\pi)}{2\eta} + LT\eta d\right].
    \end{equation}
\end{corollary}

\begin{proof}
    The result follows trivially from Lemma \ref{lemma:average} and Corollary \ref{corollary:sum-f-pi-bound-heat-flow-smooth}.
\end{proof}

\begin{theorem}
\label{thm:constant-step-forward-flow-smooth}
    Suppose that Assumption \ref{assumption:lsc} and \ref{assumption:geo-convex}, and \ref{assumption:smooth-forward-flow} hold. Let $(\mu_t)_{t\geq 0}$ denote the sequence of measures defined by \eqref{eq:forward-flow-1} - \eqref{eq:forward-flow-2}. Suppose that $\eta_t=\eta$ for all $t\geq 0$, with $\eta\in(0,\frac{1}{L}]$. Let $\bar{\mu}_T$ be the average iterate defined in \eqref{eq:bar-mu-t-forward-flow-1-first-def} - \eqref{eq:bar-mu-t-forward-flow-2-first-def}. Then, for all $\pi\in\mathcal{P}_2(\mathbb{R}^d)$, the upper bound in \eqref{eq:average-bound-forward-flow-smooth} is minimized when
    \vspace{-3mm}
    \begin{equation}
        \eta = \frac{W_2(\mu_{0},\pi)}{\sqrt{2LTd}}, \label{eq:optimal-lr-forward-flow-smooth}
    \end{equation}
    Moreover, provided this choice of $\eta$ satisfies $\eta\leq \frac{1}{L}$, it guarantees
    \begin{equation}
        \mathcal{F}(\bar{\mu}_T) - \mathcal{F}(\pi) \leq \frac{1}{\sqrt{T}} W_2(\mu_{0},\pi)\sqrt{2Ld}. \label{eq:optimal-average-bound-forward-flow-smooth}
    \end{equation}
\end{theorem}

\begin{proof}
   The result follows straightforwardly from Corollary \ref{corollary:average-bound-forward-flow-smooth}. Namely, differentiating the RHS of \eqref{eq:average-bound-forward-flow-smooth} w.r.t. $\eta$, and setting equal to zero, gives $\smash{0=-\frac{1}{2}\eta^{-2}W_2^2(\mu_{0},\pi) + LTd}$.
   Solving for $\eta$ gives \eqref{eq:optimal-lr-forward-flow-smooth}, and substituting back into \eqref{eq:average-bound-forward-flow-smooth} gives \eqref{eq:optimal-average-bound-forward-flow-smooth}.
\end{proof}

\begin{tcolorbox}[enhanced,
  colback=white,
  frame hidden,
  borderline north={0.5pt}{0pt}{black!70},
  borderline south={0.5pt}{0pt}{black!70},
  arc=2pt,             
  left=0pt,            
  right=0pt,           
  top=4pt,             
  bottom=4pt,           
  parbox=false,
  before skip=10pt, 
  after skip=15pt]
\begin{example1}[\textsc{Sampling from a Target Probability Distribution}]
Let us return once again to our running example: $\mathcal{F}(\mu) = \mathrm{KL}(\mu\|\pi)$, with $\pi\propto e^{-U}$. In this case, Corollary \ref{corollary:average-bound-forward-flow-smooth} yields the convergence rate
\begin{equation}
    \mathrm{KL}(\bar{\mu}_T\|\pi) \leq \frac{1}{T}\left[ \frac{W_2^2(\mu_{0},\pi)}{2\eta} + LT\eta d\right]. \label{eq:kl-constant-step-size-bound-smooth}
\end{equation}
From Theorem \ref{thm:constant-step-forward-flow-smooth}, the \textit{optimal} step size is given by $\eta = \frac{W_2(\mu_{0},\pi)}{\sqrt{2LTd}}$. This yields the optimal convergence rate for ULA in the smooth setting, now using the notation $D_{0,\pi} = W_2(\mu_0,\pi)$, as
\begin{equation}
    \mathrm{KL}(\bar{\mu}_T\|\pi) \leq \frac{D_{0,\pi} \sqrt{2Ld} }{\sqrt{T}} . \label{eq:optimal-average-bound-forward-flow-kl-smooth}
\end{equation}
It is interesting to compare this to the \textit{optimal} convergence rate of ULA in the nonsmooth setting, which we recall from \eqref{eq:optimal-average-bound-forward-flow-kl} was given by
\begin{equation}
    \mathrm{KL}(\bar{\mu}_T\|\pi) \leq \frac{GD_{\frac{1}{2},\pi}}{\sqrt{T}}. \label{eq:optimal-average-bound-forward-flow-kl-recall}
\end{equation}
In particular, the convergence rates in both the smooth setting and the nonsmooth setting have the same $\mathcal{O}(\frac{1}{\sqrt{T}})$ dependence on the number of iterations; only the constants differ.

\paragraph{Remark}
In the smooth case, the optimal step size and convergence rate for ULA are quantitatively different from those for (Euclidean) gradient descent. In the Euclidean case, the optimal choice of step size is $\eta=\frac{1}{L}$, which guarantees a convergence rate of $\mathcal{O}(\frac{1}{T})$ \citep[e.g.,][Theorem 2.1.4]{nesterov2018lectures}. Meanwhile, ULA can only attain a rate of $\mathcal{O}(\frac{1}{\sqrt{T}})$. This is due to the additional bias term $L\eta d$, which arises due to the heat step in the forward-flow discretization of the WGF.
\end{example1}
\end{tcolorbox}

\subsubsection{Deterministic Case: Adaptive Step Size} 
We now analyze convergence in the case where the step size $(\eta_t)_{t\geq 1}$ is defined adaptively by \eqref{eq:dog-lr-forward-flow}. On this occasion, we will analyze the convergence of the uniform average, as defined in \eqref{eq:bar-mu-t-forward-flow-1-first-def} - \eqref{eq:bar-mu-t-forward-flow-2-first-def}.

\begin{lemma}
\label{lemma:smooth-grad-bound-improved}
Suppose that Assumptions \ref{assumption:lsc}, \ref{assumption:geo-convex}, and \ref{assumption:smooth-forward-flow} hold.\footnote{Strictly speaking, we require a slightly stronger version of Assumption \ref{assumption:smooth-forward-flow}. In particular, we now assume that for all $\mu\in\mathcal{P}_2(\mathcal{X})$, the function $E_{\mu}:\mathbb{R}^d\rightarrow\mathbb{R}$, defined $E_{\mu}(x) = V(x) + \int_{\mathbb{R}^d} W(x-y)\mathrm{d}\mu(y)$, is $L$-smooth in the Euclidean sense. This is strictly stronger than our existing assumption, i.e., that $\mathcal{E}(\mu) = \int_{\mathbb{R}^d} E_\mu(x)\mathrm{d}\mu(x) $ is $L$-smooth in the Wasserstein sense.} Let $(\mu_t)_{t\geq 0}$ denote the sequence of measures defined by \eqref{eq:forward-flow-1} - \eqref{eq:forward-flow-2}, the forward-flow discretization of the WGF in \eqref{eq:wasserstein-grad-flow}.  
Then 
\begin{align}
    \sum_{t=1}^T \int_{\mathbb{R}^d}\|\zeta_t(x)\|^2\,\mathrm{d}\mu_t(x) &\leq \frac{8L}{3} \sum_{t=1}^T \left(\mathcal{F}(\mu_t) - \mathcal{F}(\pi)\right) + 2LdT. \label{eq:smooth-grad-bound-1-improved}
\end{align}
\end{lemma}

\begin{proof}
    See Appendix \ref{app:additional-proofs-forward-flow-smooth-deterministic-adaptive}.
\end{proof}

\begin{proposition}
\label{prop:smooth-adaptive-convergence-rate-improved}
Suppose that Assumptions \ref{assumption:lsc}, \ref{assumption:geo-convex}, and \ref{assumption:smooth-forward-flow} hold. Let $(\mu_t)_{t\geq 0}$ denote the sequence of measures defined by \eqref{eq:forward-flow-1} - \eqref{eq:forward-flow-2}, 
with $(\eta_t)_{t\geq 1}$ defined as in \eqref{eq:dog-lr-forward-flow}. Let $(\bar{\mu}_t)_{t\geq 0}$ be the sequence of measures defined according to \eqref{eq:bar-mu-t-forward-flow-1-first-def} - \eqref{eq:bar-mu-t-forward-flow-2-first-def}. Then, for all $T\geq 1$, we have 
\begin{equation}
\mathcal{F}(\bar{\mu}_T) - \mathcal{F}(\pi) = \mathcal{O}\left(\frac{\frac{8}{3}L\left(d_1\log_{+}\frac{\bar{r}_{T+1}}{r_{\varepsilon}} + \bar{r}_{T+1}\right)^2}{T} + \frac{\left(d_1\log_{+}\frac{\bar{r}_{T+1}}{r_{\varepsilon}} + \bar{r}_{T+1}\right)\sqrt{2Ld}}{\sqrt{T}}\right).
\end{equation}
\end{proposition}

\begin{proof}
    See Appendix \ref{app:additional-proofs-forward-flow-smooth-deterministic-adaptive}.
\end{proof}

\begin{corollary}
\label{corollary:smooth-dog-convergence}
    Suppose that Assumption 
    \ref{assumption:lsc}, \ref{assumption:geo-convex}, and \ref{assumption:smooth-forward-flow} hold. Let $(\mu_t)_{t\geq 0}$ denote the sequence of measures defined by \eqref{eq:forward-flow-1} - \eqref{eq:forward-flow-2}, 
    with $(\eta_t)_{t\geq 1}$ defined as in \eqref{eq:dog-lr-forward-flow}. Let $(\bar{\mu}_t)_{t\geq 0}$ be the sequence of measures defined according to \eqref{eq:bar-mu-t-forward-flow-1-first-def} - \eqref{eq:bar-mu-t-forward-flow-2-first-def}. Let $D\geq d_1$. Then, on the event $\{\bar{r}_{T+1}\leq D\}$ it holds that
    \begin{equation}
    \label{eq:smooth-dog-convergence-improved}
        \mathcal{F}(\bar{\mu}_{T}) - \mathcal{F}(\pi) =\mathcal{O}\left(\frac{D\sqrt{2Ld}}{\sqrt{T}}\log_{+}\left(\frac{D}{r_{\varepsilon}}\right)\right).
    \end{equation}
\end{corollary}

\begin{proof}
    See Appendix \ref{app:additional-proofs-forward-flow-smooth-deterministic-adaptive}. 
\end{proof}

\begin{tcolorbox}[enhanced,
  colback=white,
  frame hidden,
  borderline north={0.5pt}{0pt}{black!70},
  borderline south={0.5pt}{0pt}{black!70},
  arc=2pt,             
  left=0pt,            
  right=0pt,           
  top=4pt,             
  bottom=4pt,           
  parbox=false,
  before skip=10pt, 
  after skip=15pt]
\begin{example1}[\textsc{Sampling from a Target Probability Distribution}]
Recall once again our running example: $\mathcal{F}(\mu) = \mathrm{KL}(\mu\|\pi)$, with $\pi\propto e^{-U}$. 

Due to Proposition \ref{prop:smooth-adaptive-convergence-rate-improved} and Corollary \ref{corollary:smooth-dog-convergence}, we have the following guarantees for the forward-flow discretization of the WGF (i.e., ULA) with step size schedule $(\eta_t)_{t\geq 1}$, in the smooth setting. For any $\smash{D\geq W_2(\mu_{\frac{1}{2}},\pi)}$, it holds on the event $\smash{\{\bar{r}_{T+1} \leq D)}$ that
\begin{equation}
    \mathrm{KL}\left(\bar{\mu}_{T}|\pi\right)= \mathcal{O}\left(\frac{D \sqrt{2Ld}}{\sqrt{T}}\log_{+}\left(\frac{D}{r_{\varepsilon}}\right)\right),
\end{equation}
where $\bar{\mu}_T$ is the uniform average defined in \eqref{eq:bar-mu-t-forward-flow-1-first-def} - \eqref{eq:bar-mu-t-forward-flow-2-first-def}. Thus, also in the smooth setting, our algorithm essentially matches the convergence rate of optimally tuned ULA, cf. \eqref{eq:optimal-average-bound-forward-flow-kl-smooth}, up to an additional logartihmic factor. 
\end{example1}
\end{tcolorbox}

\section{Practical Algorithms}
\label{sec:algorithms}
In this section, we introduce practical, particle-based sampling algorithms based on the theoretical results obtained in Section \ref{sec:main-results}. For simplicity, our exposition focuses on the unconstrained setting. We emphasize, however, that our algorithms can easily be generalized to constrained setting according to the extension discussed in Section \ref{sec:constrained-domains}.

\subsection{The Forward-Flow Discretization}
We first present algorithms based on the forward-flow discretization of the WGF (see Sections \ref{sec:forward-flow-nonsmooth} - \ref{sec:forward-flow-smooth}). Recall, from \eqref{eq:forward-flow-1} - \eqref{eq:forward-flow-2}, the general form of the algorithm. Let $\smash{x_0\sim \mu_0}$, where $\smash{\mu_0\in\mathcal{P}_2(\mathbb{R}^d)}$. Then, for $t\geq 0$, writing $\smash{\mu_t = \mathrm{Law}(x_t)}$ and $\smash{\mu_{t+\frac{1}{2}} = \mathrm{Law}(x_{t+\frac{1}{2}})}$, update
\begin{alignat}{2}
    x_{t+\frac{1}{2}} &= x_t - \eta_t \nabla_{W_2}\mathcal{E}(\mu_t)(x_t) \label{eq:forward-flow-lagrangian-1-recall} \\
    x_{t+1} &= x_{t+\frac{1}{2}} + \sqrt{2\eta_{t}} z_t, \label{eq:forward-flow-lagrangian-2-recall}
    \intertext{
    where $z_t\stackrel{\mathrm{i.i.d.}}\sim \mathcal{N}(0,\mathbf{I}_d)$, and where, from \eqref{eq:dog-lr-forward-flow}, the sequence of step sizes $(\eta_t)_{t\geq 0}$ is defined according to
    }
    \eta_t &= \frac{\max\left[r_\varepsilon, \max_{1\leq s\leq t} W_2(\mu_{\frac{1}{2}},\mu_{s-\frac{1}{2}})\right]}{\sqrt{\sum_{s=1}^{t} \int_{\mathbb{R}^d}\|\nabla_{W_2}\mathcal{E}(\mu_s)(x)\|^2\,\mathrm{d}\mu_s(x)}}. \label{eq:dog-lr-forward-flow-recall}
\end{alignat}

\subsubsection{Particle-Based Approximation} In general, the measures $\smash{(\mu_t)_{t\geq 0}}$ and $\smash{(\mu_{t+\frac{1}{2}})_{t\geq 0}}$ are unknown, and thus we cannot implement these dynamics or compute the step size schedule directly. Instead, we will approximate the dynamics and the step size schedule using a system of interacting particles. Let $\smash{x_0^{i}\stackrel{\mathrm{i.i.d.}}{\sim}\mu_0}$, with $\mu_0\in\mathcal{P}_2(\mathbb{R}^d)$. For $t\geq 0$, let $
\smash{\hat{\mu}_t^n = \frac{1}{n}\sum_{j=1}^n \delta_{x_t^{j}}}$ and $\hat{\mu}_{t+\frac{1}{2}}^n = \frac{1}{n}\sum_{j=1}^n \delta_{x_{t+\frac{1}{2}}^j}$. We will then approximate the mean-field dynamics in \eqref{eq:forward-flow-lagrangian-1-recall} - \eqref{eq:forward-flow-lagrangian-2-recall} by
\begin{alignat}{2}
    x_{t+\frac{1}{2}}^{i} &= x_t^{i} - \hat{\eta}^n_t \nabla_{W_2}\mathcal{E}(\hat{\mu}^n_t)(x_t^i) \label{eq:forward-flow-lagrangian-1-particle} \\
    x_{t+1}^{i} &= x_{t+\frac{1}{2}}^{i} + \sqrt{2\hat{\eta}^n_{t}} z_t^{i}, \label{eq:forward-flow-lagrangian-2-particle}
    \intertext{
    where $z_t^{i}\stackrel{\mathrm{i.i.d.}}\sim \mathcal{N}(0,\mathbf{I}_d)$, and the step size schedule $(\hat{\eta}_t^n)_{t\geq 0}$ is now given by 
    }
    \hat{\eta}^n_t &= \frac{\max\left[r_\varepsilon,\max_{1\leq s\leq t} W_2(\hat{\mu}^n_{\frac{1}{2}},\hat{\mu}^n_{s-\frac{1}{2}})\right]}{\sqrt{\sum_{s=1}^{t} \int_{\mathbb{R}^d}\|\nabla_{W_2}\mathcal{E}(\hat{\mu}^n_s)(x)\|^2\mathrm{d}\hat{\mu}^n_s(x)}}. \label{eq:dog-lr-forward-flow-particle}
\end{alignat}
 This is the standard particle-based approximation of the mean-field Langevin algorithm \citep[e.g.,][]{suzuki2023convergence,suzuki2023uniform}, coupled with our dynamic step-size formula. In practice, the cost of computing the $W_2$ distance between two discrete distributions scales as $\mathcal{O}(n^3)$, which quickly becomes infeasible even for moderate numbers of particles. We will thus in fact use a more scalable alternative given by 
\begin{equation}
    \label{eq:practical-forward-flow-step}
\tcboxmath[colback=black!5,colframe=black!15,boxrule=0.4pt,arc=2pt,left=8pt, right=8pt, top=6pt, bottom=6pt]{
    \hat{\eta}_t^n = \frac{\max\Big[r_\varepsilon, \max_{1\leq s\leq t} \left(\frac{1}{n}\sum_{i=1}^n \|x_{\frac{1}{2}}^i - x_{s-\frac{1}{2}}^{i}\|^2\right)^{\frac{1}{2}}\Big]}{\sqrt{\sum_{s=1}^{t} \int_{\mathbb{R}^d}\|\nabla_{W_2}\mathcal{E}(\hat{\mu}^n_s) (x)\|^2\mathrm{d}\hat{\mu}^n_s(x)}}
},
\end{equation}
which replaces the $W_2$ distance in \eqref{eq:dog-lr-forward-flow-particle} by the identity-coupling proxy, and thus defines an upper bound for \eqref{eq:dog-lr-forward-flow-particle}. Together, \eqref{eq:forward-flow-lagrangian-1-particle} - \eqref{eq:forward-flow-lagrangian-2-particle} and \eqref{eq:practical-forward-flow-step} complete the specification of our tuning-free version of the forward-flow discretization of the WGF (i.e., the mean-field Langevin algorithm).

\subsubsection{Examples} For different choices of the objective functional, we can now obtain tuning-free variants of several existing algorithms.

\begin{example1}[\textsc{Sampling from a Target Probability Distribution}]
Suppose $\mathcal{F}(\mu) = \mathrm{KL}(\mu\|\pi)$, for some target probability distribution $\pi\propto e^{-U}$. In this case, the forward-flow discretization of the WGF, cf. \eqref{eq:forward-flow-lagrangian-1-particle} - \eqref{eq:forward-flow-lagrangian-2-particle}, corresponds to ULA (see Section \ref{sec:time-discretization}). Meanwhile, our step size schedule, cf. \eqref{eq:practical-forward-flow-step}, is given by 
\begin{equation}
    \hat{\eta}^n_t = \frac{\max\big[r_{\varepsilon},\max_{1\leq s\leq t} \big[\big(\frac{1}{n}\sum_{i=1}^n \|x_\frac{1}{2}^{i} - x_{s-\frac{1}{2}}^{i}\|^2\big)^{\frac{1}{2}} \big]\big]}{\left[\sum_{s=1}^t\left(\frac{1}{n}\sum_{i=1}^n  \|\nabla U(x_s^i)\|^2\right)\right]^{\frac{1}{2}}}.
\end{equation}
where $(x_t^i)_{t\geq 0}^{i\in [n]}$ denote a collection of (weakly) interacting ULA chains. We will refer to this approach, summarized in Algorithm \ref{alg:ula}, as ULA $\times$ \textsc{Fuse}.
\end{example1}

\begin{algorithm}[t]
\caption{ULA x \textsc{Fuse}}
\label{alg:ula}
\begin{algorithmic}[1]
\Require Target $\pi \propto e^{-U}$; particle number $n$; small initialization parameter $r_{\varepsilon}$,
\State Initialize parameters $x_0^{1},\dots,x_0^{n}$; step size $\eta_0 = r_{\varepsilon}$.
\For{$t=0,1,\dots,T-1$}
  \If{$t\geq 1$}
  \State Update step size:
  \begin{equation}
  \hat{\eta}^n_t = \frac{\max\big[r_{\varepsilon},\max_{1\leq s\leq t} \big[\big(\frac{1}{n}\sum_{i=1}^n \|x_{\frac{1}{2}}^{i} - x_{s-\frac{1}{2}}^{i}\|^2\big)^{\frac{1}{2}} \big]\big]}{\left[\sum_{s=1}^t\left(\frac{1}{n}\sum_{i=1}^n  \|\nabla U(x_s^i)\|^2\right)\right]^{\frac{1}{2}}}.
  \end{equation}
  \EndIf
  \State Update parameters:
  \begin{align}
    x_{t+\frac{1}{2}}^{i}&=x_t^{i} - \hat{\eta}^n_t\, \nabla U(x_t^{i}) \\
    x_{t+1}^{i} &= x_{t+\frac{1}{2}}^{i} + \sqrt{2\,\hat{\eta}^n_t}z_t^{i}, \quad \quad z_t^{i}\sim \mathcal{N}(0,\mathbf{I}_d).
  \end{align}
\EndFor
\State \textbf{Output:} final iterates $(x_T^{i})_{i=1}^n$ or averaged iterates $(\bar{x}_T^{i})_{i=1}^n = \big(\frac{1}{T}\sum_{t=1}^{T}x_t^{i}\big)_{i=1}^n$.
\end{algorithmic}
\vspace{0.25em}
\end{algorithm}

\begin{example1a}[\textsc{Sampling from a Target Probability Distribution via Entropy-Regularized KSD Minimization}] Suppose now that  $\mathcal{F}(\mu) = \frac{1}{2}\mathrm{KSD}^2(\mu|\pi) + \mathrm{Ent}(\mu)$, where $\mathrm{KSD}(\mu|\pi)$ denotes the kernel Stein discrepancy (KSD) \citep[e.g.,][]{liu2016kernelized,chwialkowski2016kernel}, 
 w.r.t. some target probability distribution $\pi\propto e^{-U}$. The squared KSD can be written as an interaction energy, cf. \eqref{eq:interaction}, and thus its Wasserstein gradient can be computed using standard results (see Section \ref{sec:objective-functions}). In particular, we have that \citep[see, e.g.,][]{korba2021kernel}
\begin{equation}
    \frac{1}{2}\mathrm{KSD}^2(\mu|\pi) = \frac{1}{2}\int_{\mathbb{R}^d}\int_{\mathbb{R}^d}  k_{\pi}(x,y) \,\mathrm{d}\mu(x)\,\mathrm{d}\mu(y),
\end{equation}
from which it follows that 
\begin{equation}
    \nabla_{W_2}\left[\frac{1}{2}\mathrm{KSD}^2(\mu|\pi)\right] = \int_{\mathbb{R}^d} \nabla_{2}k_{\pi}(x,\cdot)\,\mathrm{d}\mu(x), \label{eq:ksd-grad}
\end{equation}
where $k_{\pi}:\mathbb{R}^d\times\mathbb{R}^d\rightarrow\mathbb{R}$ denotes the Stein kernel, defined in terms of the target score $s(x) = \nabla \log \pi(x)$ and a base kernel $k(x,y)$ according to $k_{\pi}(x,y)  = s(x)^{\top}s(y) k(x,y) + s(x)^{\top} \nabla_{2}k(x,y) + \nabla_{1}k(x,y)^{\top} s(y) + \nabla_{2}^{\top}\nabla_{1} k(x,y)$. In this case, the forward-flow discretization of the WGF, cf. \eqref{eq:forward-flow-lagrangian-1-particle} - \eqref{eq:forward-flow-lagrangian-2-particle}, yields a Langevin version of KSD descent \citep{korba2021kernel}, namely, 
\begin{alignat}{2}
    x_{t+\frac{1}{2}}^{i} &= x_t^{i} - \hat{\eta}^n_t \frac{1}{n}\sum_{j=1}^n \nabla_{2} k_{\pi}(x_t^j, x_t^{i}) \label{eq:ksd-1} \\
    x_{t+1}^{i} &= x_{t+\frac{1}{2}}^{i} + \sqrt{2\hat{\eta}^n_{t}} z_t^{i}, \label{eq:ksd-2}
\end{alignat}
where our adaptive step size, cf. \eqref{eq:practical-forward-flow-step}, is now given by
\begin{equation}
    \hat{\eta}^n_t = \frac{\max\big[r_{\varepsilon},\max_{1\leq s\leq t} \big[\big(\frac{1}{n}\sum_{i=1}^n \|x_{\frac{1}{2}}^{i} - x_{s-\frac{1}{2}}^{i}\|^2\big)^{\frac{1}{2}} \big]\big]}{[\sum_{s=1}^{t} \frac{1}{n}\sum_{i=1}^n\|\frac{1}{n}\sum_{j=1}^n \nabla_{2} k_{\pi}(x_s^j, x_s^{i})\|^2]^{\frac{1}{2}}}.
\end{equation}
\end{example1a}

\begin{example2}[\textsc{Training a Mean-Field Neural Network}] Suppose $\mathcal{F}(\mu) = \mathcal{E}(\mu) + \mathrm{Ent}(\mu)$, where $\mathcal{E}(\mu) = \mathcal{E}_0(\mu) + \int r(x) \mathrm{d}\mu(x)$; $\mathcal{E}_0(\mu) = \frac{\lambda_1}{n}\sum_{k=1}^n \ell(y_k,h_{\mu}(z_k))$ is the (scaled) empirical risk of the mean-field neural network $h_{\mu}(z) = \int h_{x}(z)\mathrm{d}\mu(x)$ defined by a neuron $h_{x}(z)$ with parameter $x\in\mathbb{R}^d$, given training data $(z_i,y_i)_{i=1}^n\in\mathbb{R}^{d-1}\times\mathbb{R}$ and convex loss function $\ell:\mathbb{R}\times\mathbb{R}\rightarrow\mathbb{R}$; and $r:\mathbb{R}^d\rightarrow\mathbb{R}$ is a regularizer. For simplicity, let us suppose that $\ell(y,y') = \frac{1}{2}(y-y')^2$ and that $r(x) = 0$. Putting everything together, we thus have 
\begin{equation}
    \mathcal{E}(\mu) = \frac{\lambda_1}{2n}\sum_{k=1}^n (y_k - h_{\mu}(z_k))^2
\end{equation}
which implies, following the calculations in  \citet[][Section 3.1]{hu2021meanfield}, that 
\begin{equation}
    \nabla_{W_2}\mathcal{E}(\mu)(x) = -\frac{\lambda_1}{n}\sum_{k=1}^n \Big(y_k - \int h_{x}(z_k)\mathrm{d}\mu(x)\Big)\nabla_{x}h_x(z_k).
\end{equation}
In this case, the forward-flow discretization of the WGF, cf. \eqref{eq:forward-flow-lagrangian-1-particle} - \eqref{eq:forward-flow-lagrangian-2-particle}, corresponds to an example of the mean-field Langevin algorithm \citep[e.g.,][]{suzuki2023convergence}. In particular, we have 
\begin{align}
    x_{t+\frac{1}{2}}^{i} &= x_t^{i} - \hat{\eta}^n_t \Big[-\frac{\lambda_1}{n}\sum_{k=1}^n \Big(y_k - \frac{1}{n}\sum_{j=1}^n h_{x_t^{j}}(z_k)\Big)\nabla_{x}h_{x_t^{i}}(z_k) \Big]  
    \label{eq:forward-flow-lagrangian-1-particle-mfnn} \\
    x_{t+1}^{i} &= x_{t+\frac{1}{2}}^{i} + \sqrt{2\hat{\eta}^n_{t}} z_t^{i}. \label{eq:forward-flow-lagrangian-2-particle-mfnn}
\end{align}
Our adaptive step size rule, cf. \eqref{eq:practical-forward-flow-step}, results in a tuning-free version of this algorithm, and in this case is given by 
\begin{equation}
\hat{\eta}_t^n = \frac{\max[r_{\varepsilon},\max_{1\leq s\leq t} [(\frac{1}{n}\sum_{i=1}^n \|x_{\frac{1}{2}}^i - x_{s-\frac{1}{2}}^{i}\|^2)^{\frac{1}{2}}]]}{[\sum_{s=1}^{t} \frac{1}{n}\sum_{i=1}^n\|\frac{\lambda_1}{n}\sum_{k=1}^n (y_k - \frac{1}{n}\sum_{j=1}^n h_{x_s^{j}}(z_k))\nabla_{x}h_{x_s^{i}}(z_k) \|^2]^{\frac{1}{2}}}. 
\end{equation}
\end{example2}

\begin{algorithm}[b]
\caption{SGLD x \textsc{Fuse}}
\label{alg:sgld}
\begin{algorithmic}[1]
\Require
Dataset $\{y_i\}_{i=1}^N$; log-likelihood $\log p(y\mid x)$; log-prior $\log p(x)$; particle number $n$; small initialization parameter $r_{\varepsilon}$,
\State Initialize parameters $x_0^{i},\dots,x_0^{n}$; step size $\eta_0 = r_{\varepsilon}$.
\For{$t=0,1,\dots,T-1$}
  \State Sample mini-batch $\Omega_t \subset \{1,\dots,N\}$.
  \For{$i=1,\dots,n$}
  \State Compute gradient estimates:
  \[
    \hat{\zeta}_t^{i} = \frac{N}{|\Omega_t|}\sum_{j \in \Omega_t} \nabla_{x_t^{i}} \log p\!\left(y_j \mid x_t^{i}\right) + \nabla_{x_t^{i}} \log p(x_t^{i}).
  \]
  \EndFor
  \If{$t\geq 1$}
  \State Update step size:
  \begin{equation}
  \hat{\eta}^n_t = \frac{\max\big[r_{\varepsilon},\max_{1\leq s\leq t} \big[\big(\frac{1}{n}\sum_{i=1}^n \|x_{\frac{1}{2}}^{i} - x_{s-\frac{1}{2}}^{i}\|^2\big)^{\frac{1}{2}} \big]\big]}{\left[\sum_{s=1}^t\left(\frac{1}{n}\sum_{i=1}^n  \| \hat{\zeta}_s^i\|^2\right)\right]^{\frac{1}{2}}}.
  \end{equation}
  \EndIf
  \State Update parameters:
  \begin{align}
    x_{t+\frac{1}{2}}^{i}&=x_t^{i} - \hat{\eta}^n_t\, \hat{\zeta}_t^i \\
    x_{t+1}^{i} &= x_{t+\frac{1}{2}}^{i} + \sqrt{2\,\hat{\eta}^n_t}z_t^{i}, \quad \quad z_t^{i}\sim \mathcal{N}(0,\mathbf{I}_d).
  \end{align}
\EndFor
\State \textbf{Output:} final iterate $(x_T^{i})_{i=1}^n$ or averaged iterates $(\bar{x}_T^{i})_{i=1}^n = \big(\frac{1}{T}\sum_{t=1}^{T}x_t^{i}\big)_{i=1}^n$.
\end{algorithmic}
\vspace{0.25em}
\end{algorithm}

\subsubsection{Discussion} We conclude this section with some brief remarks regarding the algorithms introduced thus far.

\paragraph{Extensions} Naturally, one can also consider stochastic gradient (i.e., mini-batch) variants of these algorithms, in which gradients are replaced by (unbiased) stochastic estimates (see Section \ref{sec:stochastic-gradients}). Indeed, this is justified by our theoretical analysis in Section \ref{sec:stochastic-theory}. In \textsc{Example 1}, this leads to a step-size-free version of stochastic-gradient Langevin dynamics (SGLD) \citep{welling2011bayesian}. We will refer to this algorithm, summarized in Algorithm \ref{alg:sgld}, as SGLD x \textsc{Fuse}.

\paragraph{Computational Considerations} It is worth noting that ULA x \textsc{Fuse} (and SGLD x \textsc{Fuse}) differ from their classical formulation \citep[e.g.,][]{neal1992bayesian,roberts1996exponential,welling2011bayesian} in the sense that they involve a collection of (interacting) particles, rather than just a single particle. To be specific, in the classical setting, one would simulate a single trajectory $(x_t)_{t\geq 0}$, and then estimate expectations with respect to the target measure as $\smash{\mathbb{E}_{x\sim\pi}[f(x)] \approx \frac{1}{T}\sum_{t=1}^{T} f(x_t)}$. On the other hand, we simulate multiple trajectories $(x_t^{i})_{t\geq 0}^{i=1:n}$, and approximate expectations as $\smash{\mathbb{E}_{x\sim \pi}[f(x)] \approx \frac{1}{N}\sum_{i=1}^N f(x_t^{i})}$. This is sometimes referred to in the literature as \textit{parallel-chain ULA} or \textit{parallel-chain SGLD} \citep[e.g.,][]{futami2020accelerating,futami2021accelerated}. 

One can, of course, consider hybrids of these approaches, in which expectations with respect to the target measure are approximated by averaging over the trajectories of multiple particles, viz, $\smash{\mathbb{E}_{x\sim \pi}[f(x)] \approx \frac{1}{N}\sum_{i=1}^N[\frac{1}{T}\sum_{t=1}^T f(x_t^{i})]}$. This approach is particularly useful for ULA x \textsc{Fuse} since one may require a large number of samples to obtain accurate estimates of expectations with respect to the target measure, but only a (relatively) small number of particles to obtain a performant estimate of the \textsc{Fuse} step size schedule $(\hat{\eta}_t^n)_{t\geq 0}$. In this situation, one can also leverage parallel computing environments (e.g., GPUs or arrays of CPUs) for more efficient computation, as independent sets of trajectories can be computed in parallel.

\paragraph{Limitations} It is worth emphasizing that our theoretical results are only \emph{directly} applicable to \textsc{Example 1}. In particular, geodesic convexity does not hold either for the squared KSD \citep[][Section 4]{korba2021kernel}, or the objective functionals encountered in the training of mean-field neural networks \citep[e.g.,][Example 3.3]{lascu2024linear}. Nonetheless, our adaptive step size schedule is still applicable in these cases, with strong empirical performance (see Section \ref{sec:numerical-results}).

\subsection{The Forward Euler Discretization}
\label{sec:practical-forward-euler}
We now turn our attention to the forward Euler discretization of the WGF (see Appendices \ref{sec:forward-euler-nonsmooth} - \ref{sec:forward-euler-smooth}). In the interest of space, and given its relevance to the sampling problem, we will now focus only on the case where the internal energy in the objective functional coincides with the (negative) entropy: $\mathcal{F}(\mu) = \mathcal{E}(\mu) + \mathrm{Ent}(\mu)$. In this case, the forward Euler discretization of the WGF is defined as follows, cf. \eqref{eq:wasserstein-sub-grad-descent-lagrange}. Let $x_0\sim \mu_0$, for some $\mu_0\in\mathcal{P}_2(\mathbb{R}^d)$. Then, for $t\geq 0$, writing $\mu_t = \mathrm{Law}(x_t)$, update
\begin{align}
    x_{t+1} = x_t - \eta_t 
    \left(\nabla_{W_2} \mathcal{E}(\mu_t)(x_t)  + \nabla 
    \log \mu_t(x_t)\right)
    \label{eq:forward-euler-recall}
\intertext{where, from \eqref{eq:dog-lr}, the \textsc{Fuse} step size schedule is given by}
    \eta_t = \frac{\max\left[r_{\varepsilon},\max_{0\leq s\leq t} W_2(\mu_0,\mu_s)\right]}{\sqrt{\sum_{s=0}^{t} \int_{\mathbb{R}^d}\|\nabla_{W_2}\mathcal{E}(\mu_s)(x) + \nabla \log \mu_s(x)\|^2\,\mathrm{d}\mu_s(x)}}.
    \label{eq:dog-lr-forward-euler-recall}
\end{align}

\subsubsection{Particle-Based Approximation} In general, we cannot implement these dynamics directly, since both $(\mu_t)_{t\geq 0}$ and $(\nabla \log \mu_t)_{t\geq 0}$, are unknown. We would like, as before, to approximate the dynamics using a system of interacting particles. In this case, however, we have the additional difficulty that $\nabla \log \mu$ is only well defined when $\mu\in\mathcal{P}_{2,\mathrm{ac}}(\mathbb{R}^d)$. In particular, it is not well defined when $\mu\in\mathcal{P}_2(\mathbb{R}^d)\setminus\mathcal{P}_{2,\mathrm{ac}}(\mathbb{R}^d)$ is a discrete measure.

\paragraph{Kernel-Based Approximation} To circumvent this, we will use a kernel-based approximation, following the approach popularised in \citet{liu2016stein}, and recently extended to the mean-field setting by \citet{chazal2025computable}. In particular, we will replace the Wasserstein gradient $\nabla_{W_2}\mathcal{F}(\mu):=[\nabla_{W_2}\mathcal{E}(\mu) + \nabla \log \mu]$ in \eqref{eq:forward-euler-recall} - \eqref{eq:dog-lr-forward-euler-recall} by its kernel smoothed approximation, viz 
\begin{alignat}{2}
    \mathcal{S}_{\mu}\nabla_{W_2}\mathcal{F}(\mu) 
    &= \mathcal{S}_{\mu} \Big[\nabla_{W_2} \mathcal{E}(\mu) + \nabla \log \mu \Big] \\
    &=\int_{\mathbb{R}^d} \Big[k(x,\cdot) \left(\nabla_{W_2}\mathcal{E}(\mu)(x) + \nabla \log \mu(x)\right)\Big]\mathrm{d}\mu(x) \\
    &=\int_{\mathbb{R}^d} \Big[k(x,\cdot) \nabla_{W_2}\mathcal{E}(\mu)(x) - \nabla_{1}k(x,\cdot)\Big]\mathrm{d}\mu(x).
\end{alignat} 
$\smash{\mathcal{S}_{\mu}:L^2(\mu)\rightarrow\mathcal{H}_k^d}$ is the integral operator defined in Section \ref{sec:alternative-gradient-flows}, and where the final line follows via integration by parts (assuming mild boundary conditions).

\paragraph{Particle-Based Approximation} We can now construct a particle-based approximation for the algorithm defined in \eqref{eq:forward-euler-recall} - \eqref{eq:dog-lr-forward-euler-recall}. Similar to before, let $\smash{x_0^{i}\stackrel{\mathrm{i.i.d.}}{\sim} \mu_0}$, with $\mu_0\in\mathcal{P}_2(\mathbb{R}^d)$. Let $\hat{\mu}_t^n = \frac{1}{n}\sum_{j=1}^n \delta_{x_t^j}$. Then, for $t\geq 0$, we update
\begin{align}
    x_{t+1}^i = x_t^i - \hat{\eta}_t^n \int_{\mathbb{R}^d} \bigg(k(x,x_t^i) \nabla_{W_2}\mathcal{E}(\hat{\mu}_t^n)(x) - \nabla_{1}  k(x,x_t^i)\bigg)\,\mathrm{d}\hat{\mu}_t^n(x),
    \label{eq:forward-euler-particle}
\end{align}
where the sequence of step sizes $(\hat{\eta}_t^n)_{t\geq 0}$ is given by 
\begin{align}
\tcboxmath[colback=black!5,colframe=black!15,boxrule=0.4pt,arc=2pt,left=8pt, right=8pt, top=6pt, bottom=6pt]{
    \hat{\eta}^n_t = \frac{\max\left[r_\varepsilon, \max_{0\leq s\leq t} \left(\frac{1}{n}\sum_{i=1}^n \|x_{0}^i - x_{s}^{i}\|^2\right)^{\frac{1}{2}}\right]}{\sqrt{\sum_{s=0}^{t} \int_{\mathbb{R}^d}\| \int \big(k(y,x) \nabla_{W_2}\mathcal{E}(\hat{\mu}_s^n)(y)  - \nabla_{1}k(y,x)\big) \mathrm{d}\hat{\mu}^n_s(y)\|^2\,\mathrm{d}\hat{\mu}^n_s(x)}}
}.
    \label{eq:dog-lr-forward-euler-particle}
\end{align}
Together, \eqref{eq:forward-euler-particle} and \eqref{eq:dog-lr-forward-euler-particle} define the general form of our tuning-free version of the (kernelized) forward Euler discretization of the WGF. This corresponds to the recently introduced \emph{variational gradient descent} algorithm \citep[][Section 3.2.2.2]{chazal2025computable}, combined with our adaptive step size schedule. 

\paragraph{Remark} \emph{It is worth noting that other particle-based, kernel-based approximations of the forward Euler dynamics (and our step size schedule) are also possible. For example, the so-called \textit{blob method}, originating in \citet{carrillo2019blob}, approximates the entropy $\mathcal{H}(\mu) = \int \log[\mu(x)]\mu(x)\mathrm{d}x$ as $\smash{\mathcal{H}_{k}(\mu) = \int \log [k\star\mu(x)] \mu(x)\mathrm{d}x}$, for some kernel $k$. This method is more generally applicable than the one we have described above; in particular, it can be applied for other internal energies satisfying the standard assumptions, e.g. $\smash{\mathcal{H}(\mu) = \int \frac{\mu^m(x)}{m-1} \mathrm{d}x}$, not just the (negative) entropy. We refer to \citet{chen2018unified,liu2019understanding} for a further discussion of this approach in the context of Bayesian posterior sampling.}

\subsubsection{Examples} Similar to before, for different choices of the objective functional, we can now obtain tuning-free analogues of several existing algorithms.

\begin{example1}[\textsc{Sampling from a Target Probability Distribution}] Let $\mathcal{F}(\mu) = \mathrm{KL}(\mu\|\pi)$, for some target distribution $\pi\propto e^{-U}$.  In this case, the kernelized forward Euler discretization of the WGF, cf. \eqref{eq:forward-euler-particle}, corresponds to the popular SVGD algorithm, viz 
    \begin{equation}
    x_{t+1}^i = x_t^i - \hat{\eta}_t^n\frac{1}{n} \sum_{j=1}^n  \bigg(k(x_t^j,x_t^i) \nabla U(x_t^j) - \nabla_{1}  k(x_t^j,x_t^i)\bigg).
\end{equation}
Meanwhile, our adaptive step size rule, cf. \eqref{eq:dog-lr-forward-euler-particle}, yields a tuning-free variant of this algorithm (see Algorithm \ref{alg:svgd}). Specifically, we have 
\begin{equation}
    \hat{\eta}^n_t = \frac{\max\big[r_{\varepsilon},\max_{0\leq s\leq t} \big[\left(\frac{1}{n}\sum_{i=1}^n \|x_{0}^{i} - x_{s}^{i}\|^2\right)^{\frac{1}{2}} \big]\big]}{\big(\sum_{s=0}^t\big[\frac{1}{n}\sum_{i=1}^n  \big\|\frac{1}{n}\sum_{j=1}^n \big[k(x_s^{j},x_s^{i}) \nabla U(x_s^{j}) - \nabla_{1}k(x_s^j, x_s^{i})\big]\big\|^2\big]\big)^{\frac{1}{2}}}.
\end{equation}
\end{example1}

\begin{algorithm}[t]
\caption{SVGD x \textsc{Fuse}}
\label{alg:svgd}
\begin{algorithmic}[1]
\Require Target $\pi\propto e^{-U}$; 
particle number $n$; small initialization parameter $r_{\varepsilon}$, kernel $k:\mathbb{R}^d\times\mathbb{R}^d\rightarrow\mathbb{R}$.
\State Initialize parameters $x_0^{i},\dots,x_0^{n}$; step size $\eta_0 = r_{\varepsilon}$.
\For{$t=0,1,\dots,T-1$}
  \State Update step size:
  \[
  \hat{\eta}_t^n = \frac{\max\big[r_{\varepsilon},\max_{0\leq s\leq t} \big[\left(\frac{1}{n}\sum_{i=1}^n \|x_{0}^{i} - x_{s}^{i}\|^2\right)^{\frac{1}{2}} \big]\big]}{\big(\sum_{s=0}^t\big[\frac{1}{n}\sum_{i=1}^n  \big\|\frac{1}{n}\sum_{j=1}^n \big[k(x_s^{j},x_s^{i}) \nabla U(x_s^{j}) - \nabla_{1}k(x_s^j, x_s^{i})\big]\big\|^2\big]\big)^{\frac{1}{2}}}
  \]
  \State Update parameters:
  \begin{align}
    x_{t+1}^{i} &= x_{t}^{i} - \hat{\eta}^n_t \frac{1}{n}\sum_{j=1}^n \left[k(x_t^{j},x_t^{i}) \nabla U(x_t^j) - \nabla_{x_t^{j}}k(x_t^j, x_t^{i})\right]
  \end{align}
\EndFor
\State \textbf{Output:} final iterate $(x_T^{i})_{i=1}^n$ or averaged iterates $(\bar{x}_T^{i})_{i=1}^n = \left(\frac{1}{T}\sum_{t=1}^{T}x_t^{i}\right)_{i=1}^n$
\end{algorithmic}

\vspace{0.25em}
\end{algorithm}

\begin{example1a}[\textsc{Sampling from a Target Probability Distribution via KSD Minimization}] Suppose now that $\mathcal{F}(\mu) = \frac{1}{2}\mathrm{KSD}^2(\mu|\pi)$, for some target probability distribution $\pi\propto e^{-U}$. In this case, the objective functional is well-defined for discrete measures, and a kernel approximation is not required. We can thus directly approximate the forward Euler discretization of the WGF, cf. \eqref{eq:forward-euler-recall}, using a system of interacting particles. In particular, substituting the Wasserstein gradient from \eqref{eq:ksd-grad}, we have the update equation 
\begin{align}
    x_{t+1}^{i} &= x_t^{i} - \hat{\eta}^n_t \frac{1}{n}\sum_{j=1}^n \nabla_{2} k_{\pi}(x_t^j, x_t^{i})
\end{align}
where now the \textsc{Fuse} step size schedule is given by
\begin{align}
\hat{\eta}^n_t = \frac{\max\big[r_{\varepsilon},\max_{0\leq s\leq t} \big[\left(\frac{1}{n}\sum_{i=1}^n \|x_{0}^{i} - x_{s}^{i}\|^2\right)^{\frac{1}{2}} \big]\big]}{\big(\sum_{s=0}^t\big[\frac{1}{n}\sum_{i=1}^n  \big\| \frac{1}{n}\sum_{j=1}^n \nabla_{2} k_{\pi}(x_s^j, x_s^{i})\big\|^2\big]\big)^{\frac{1}{2}}}. \label{eq:ksd}
\end{align}
This algorithm can be viewed as a step-size-free version of KSD descent \citep{korba2021kernel}.\footnote{Strictly speaking, this is actually a tuning-free version of \emph{a variant of} KSD descent. In particular, the original KSD descent update equation reads $\smash{x_{t+1}^{i} = x_t^{i} - \hat{\eta}_t^n \frac{1}{n^2}\sum_{j=1}^n \nabla_{2} k_{\pi}(x_t^{j},x_t^{i})}$, i.e., the gradient is scaled by $\smash{\frac{1}{n^2}}$ rather than $\smash{\frac{1}{n}}$, as in \eqref{eq:ksd}. This difference arises due to the fact that the original update equation was derived by substituting an empirical measure into the objective functional itself, rather than its Wasserstein gradient. This defines an objective function $\smash{F:(\mathbb{R}^d)^n\rightarrow\mathbb{R}}$, given by $\smash{F([x^i]_{i=1}^n) = \frac{1}{2N^2}\sum_{i,j=1}^nk_{\pi}(x^i,x^j)}$, over the finite-dimensional Euclidean space $\smash{(\mathbb{R}^d)^n}$. This objective function can be optimized directly via (Euclidean) gradient descent, which yields the original update equation. We refer to \citet[][Section 3.2]{korba2021kernel} for further details.} 
\end{example1a}

\begin{example2}[\textsc{Training a Mean-Field Neural Network}]
Let $\mathcal{F}(\mu) = \mathcal{E}(\mu) + \mathrm{Ent}(\mu)$, where $\mathcal{E}(\mu) = \mathcal{E}_0(\mu) + \int r(x) \mathrm{d}\mu(x)$, with all terms as defined in the previous section. 
In this case, the kernelized forward Euler discretization of the WGF, cf. \eqref{eq:forward-euler-particle}, corresponds to a particular instance of the recently introduced \textit{variational gradient descent} \citep[][Section 3.2.2.2]{chazal2025computable}, given by 
\begin{equation}
    x_{t+1}^i = x_t^i - \frac{1}{n}\sum_{j=1}^n \hat{\eta}_t^n \bigg(k(x_t^j,x_t^i) \left[-\frac{\lambda_1}{n}\sum_{k=1}^n \left(y_k - \frac{1}{n}\sum_{l=1}^n h_{x_t^l}(z_k)\right)\nabla_{x}h_{x_t^j}(z_k)\right]  - \nabla_{1}  k(x_t^j,x_t^i)\bigg)
\end{equation}
Once again, our adaptive step size rule, cf. \eqref{eq:dog-lr-forward-euler-particle}, leads to a tuning-free version of this algorithm. In this case, it reads as
{\small
\begin{equation}
    \hat{\eta}^n_t = \frac{\max\big[r_{\varepsilon},\max_{0\leq s\leq t} \big[\left(\frac{1}{n}\sum_{i=1}^n \|x_{0}^{i} - x_{s}^{i}\|^2\right)^{\frac{1}{2}} \big]\big]}{\big(\sum_{s=0}^t\big[\frac{1}{n}\sum_{i=1}^n  \big\|\frac{1}{n}\sum_{j=1}^n \big[k(x_s^{j},x_s^{i}) [\frac{\lambda_1}{n}\sum_{k=1}^n (y_k - \frac{1}{n}\sum_{l=1}^n h_{x_s^l}(z_k))\nabla_{x}h_{x_s^j}(z_k)] + \nabla_{1}k(x_s^j, x_s^{i})\big]\big\|^2\big]\big)^{\frac{1}{2}}}.
\end{equation}
}
\end{example2}

\section{Numerical Results}
\label{sec:numerical-results}
In this section, we benchmark the performance of \textsc{Fuse} across several experiments involving both synthetic and real data. We consider examples which satisfy the assumptions of the previous sections, as well as examples which violate one or more of these assumptions (e.g., geodesic convexity). We performed all experiments using a MacBook Pro 16" (2021) laptop with Apple M1 Pro chip and 16GB of RAM. 

\subsection{Example 1: Sampling via KL Divergence Minimization}
We first consider the task of minimizing $\mathcal{F}(\mu) = \mathrm{KL}(\mu\|\pi)$, for some target probability distribution $\pi\propto e^{-U}$.

\subsubsection{Gaussian Target} 
\label{sec:gaussian}
We first consider an example where the target $\pi$ is given by a (unimodal) Gaussian distribution, namely, 
\begin{equation}
    \pi(x)=\mathcal{N}(x|m_{\pi},\Sigma_{\pi}) \propto 
    \exp\Big(-\frac{1}{2}(x-m_{\pi})^{\top}\Sigma_{\pi}^{-1}(x-m_{\pi})\Big).
\end{equation}
In this case, it is possible to generate true samples from the target, and thus compute a number of different performance metrics. This target is also (strongly) log-concave, meaning the KL divergence is (strongly) geodesically convex, and our assumptions are satisfied. 

\paragraph{Baseline Comparison} In Figure \ref{fig:1}, we compare the performance of ULA and ULA x \textsc{Fuse} for a 10-dimensional Gaussian target distribution $\pi=\mathcal{N}(m_{\pi},\Sigma_{\pi})$, with $\Sigma_{\pi} = \mathrm{diag}(\boldsymbol{\sigma}^2_{\pi})$. We measure performance using a particle-based approximation of the KL divergence.\footnote{Our approximation is computed as follows. First, we compute the empirical mean $\hat{m}_t^n$ and covariance $\hat{\Sigma}_t^n$ of the particles $(x_t^{i})_{i=1}^n$. We then approximate $\mu_t\approx \bar{\mu}_t^n$, where $\bar{\mu}_t^n = \mathcal{N}(\hat{m}_t^n,\hat{\Sigma}_t^n)$ is a Gaussian distribution with the mean and covariance just computed. This is justified by the fact that, if $\mu_0=\mathcal{N}(m_0,\Sigma_0)$ and $\pi = \mathcal{N}(m_{\pi},\Sigma_{\pi})$ are both Gaussian, then $\mu_t$ remains Gaussian for all $t\geq 0$ (see above). Finally, we approximate $\mathrm{KL}(\mu_t\|\pi)\approx \mathrm{KL}(\bar{\mu}_t^n\|\pi)$, which we can compute using the standard formula for the KL divergence between two Gaussians \citep[e.g.,][]{pardo2018statistical}.} We report results averaged over 10 random seeds: for each experiment, we draw $m_{\pi}$ uniformly at random in $[-2,2]$, and $\boldsymbol{\sigma}^2_{\pi}$ uniformly at random in $[1,5]$. In each case, we initialize using a standard Gaussian: $\mu_0 = \mathcal{N}(\mathbf{0},\mathbf{I}_d)$. We run ULA for a grid of uniformly logarithmically spaced step sizes $\eta\in[10^{-6}, 10^{0}]$, and ULA x \textsc{Fuse} for the same grid of values of $r_{\varepsilon}$. For ULA x \textsc{Fuse}, we report results for both the practical, particle-based step size schedule (orange) and also the ideal, mean-field step size schedule (green).\footnote{Under the assumption that the initial distribution $\mu_0 = \mathcal{N}(m_0,\Sigma_0)$ is also Gaussian, the WGF $(\mu_t)_{t\geq 0}$ is Gaussian for all $t\geq 0$ \citep[e.g.,][]{wibisono2018sampling}. It is thus possible to compute the theoretically optimal step size schedule, cf., \eqref{eq:dog-lr-forward-flow-recall}, exactly, since all quantities are available in closed form (see Section \ref{sec:adaptive-deterministic-step-size}). In particular, we do not have to resort to a particle approximation, cf. \eqref{eq:dog-lr-forward-flow-particle}, or a coarse upper bound for the Wasserstein distance, cf. \eqref{eq:practical-forward-flow-step}.
}

\begin{figure}[t]
  \centering
  \begin{subfigure}{0.485\linewidth}
    \includegraphics[width=\linewidth]{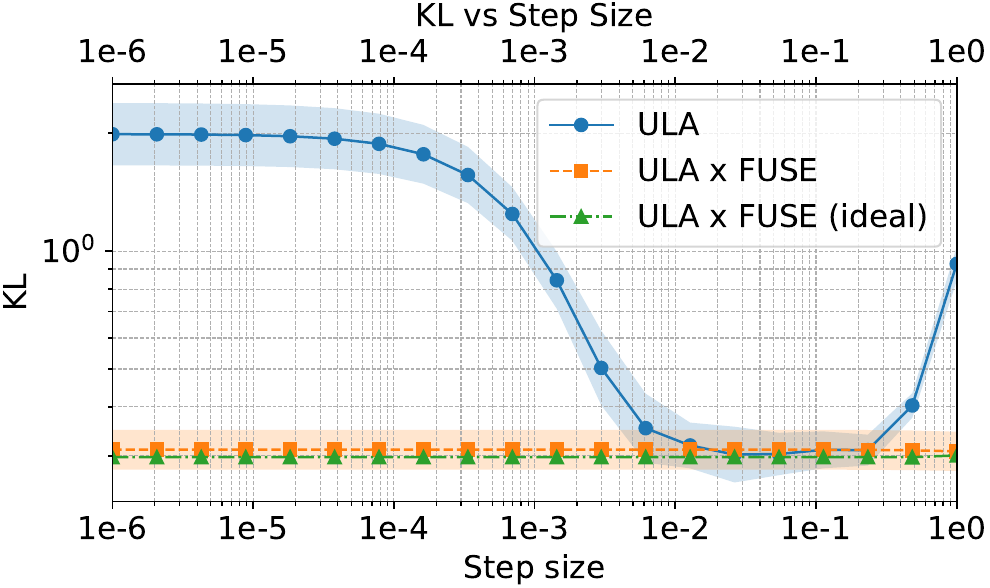}
    \caption{\textbf{KL Divergence vs Step Size}: $\mathrm{KL}(\mu_T\|\pi)$ as a function of the fixed step size $\eta$ (ULA) or the parameter $r_{\varepsilon}$ (ULA x \textsc{Fuse}), after $T=500$ iterations.}
    \label{fig:1a}
  \end{subfigure}\hfill
  \begin{subfigure}{0.485\linewidth}
    \includegraphics[width=\linewidth]{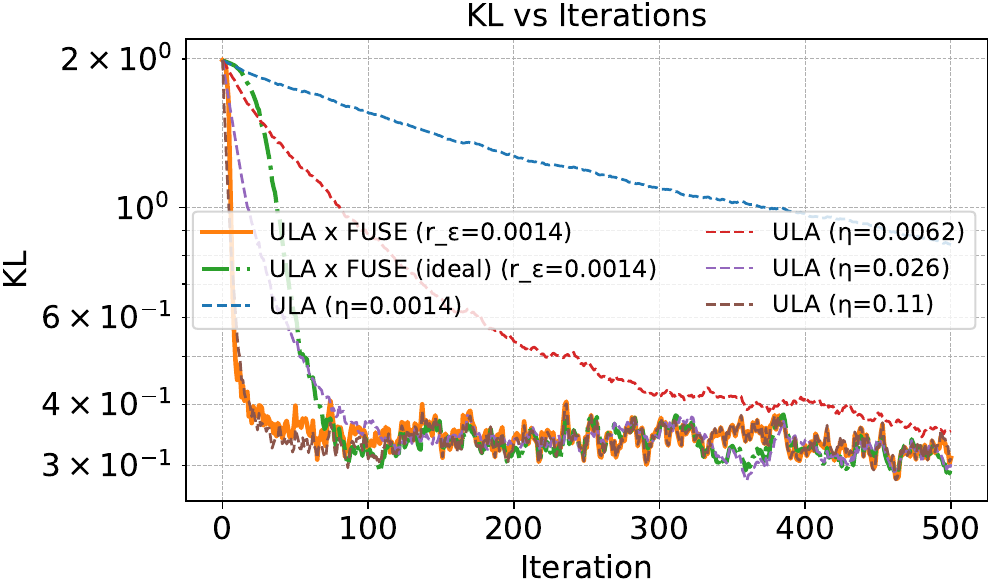}
    \caption{\textbf{KL Divergence vs Iterations}: $\mathrm{KL}({\mu}_T\|\pi)$ as a function of $t\in[0,500]$, for different values of the step size $\eta$ (ULA) and a single value of $r_{\varepsilon}$ (ULA x \textsc{Fuse}).}
    \label{fig:1b}
  \end{subfigure}
  \caption{\textbf{A comparison of ULA and ULA x \textsc{Fuse} for a 10-dimensional Gaussian target.}}
  \label{fig:1}
\end{figure}

The empirical results closely reflect our theoretical findings. In particular, the performance of ULA x \textsc{Fuse} essentially matches the performance of optimally tuned constant step size ULA. Our results also illustrate the robustness of ULA x \textsc{Fuse} to the choice of the initial movement parameter $r_{\varepsilon}$. Indeed, for any choice of $r_{\varepsilon}$ within the reported range of $[ 10^{-6}, 10^{0}]$, the performance of ULA x \textsc{Fuse} is essentially identical. 

It is also noteworthy that the ideal, {mean-field} ULA x \textsc{Fuse} step size schedule (green) slightly outperforms the practical, {particle-based} ULA x \textsc{Fuse} step size schedule (orange). This should not be surprising: recall that the ideal ULA x \textsc{Fuse} step size was originally derived as a proxy for an optimal fixed ULA step size, as measured by an upper bound for $\mathrm{KL}(\bar{\mu}_T\|\pi)$ (see Section \ref{sec:adaptive-deterministic-step-size}). In a similar way, the practical, particle-based ULA x \textsc{Fuse} step size can be viewed as a proxy for an optimal ULA fixed step size, but now as measured by a looser upper bound for $\smash{\mathrm{KL}(\bar{\mu}_T\|\pi)}$, in which the $W_2$ distance is replaced by the corresponding $L^2$ distance. In principle, one could replicate all of the results in Sections \ref{sec:forward-flow-nonsmooth} - \ref{sec:forward-flow-smooth} for this step size, but now with any $W_2$ distance replaced by the corresponding $L^2$ distance. Thus, in particular, this step size schedule attains the same $\smash{\mathcal{O}(\frac{1}{\sqrt{T}})}$ convergence rate as the \textit{ideal} step size schedule (in terms of the iteration number), but now with a slightly worse dependence on the problem diameter. 

\begin{figure}[b]
  \centering
  \includegraphics[width=\linewidth]{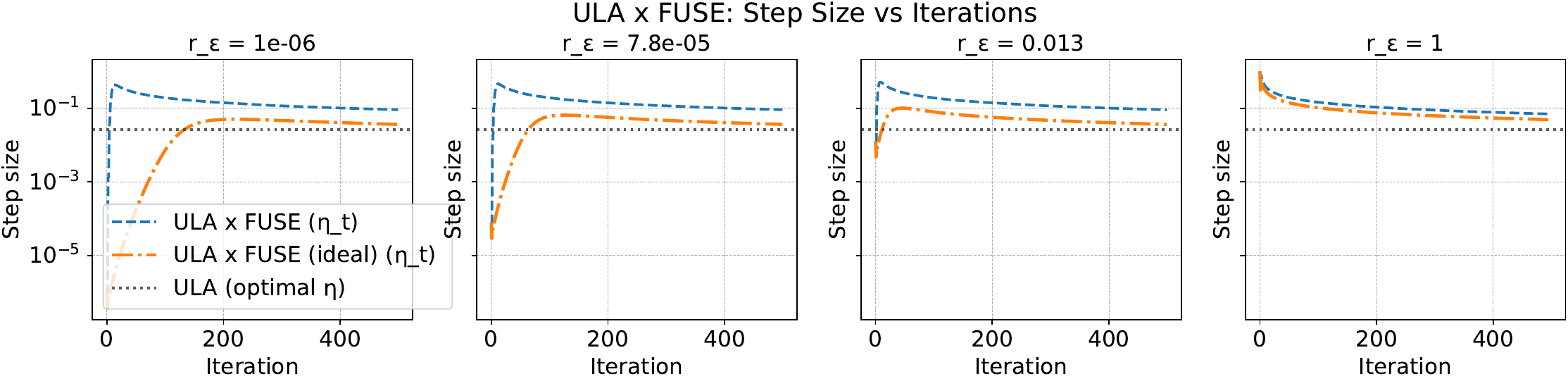}
  \caption{\textbf{A comparison of the {ideal}, mean-field and the {practical}, particle-based ULA x \textsc{Fuse} step size schedule for the 10-dimensional Gaussian target.} We plot the ideal, {mean-field} step size schedule $(\eta_t)_{t\geq0}$ (orange) and the practical, {particle-based} step size schedule $(\hat{\eta}_t^n)_{t\geq 0}$ (blue) for $t\in[0,500]$, for four values of the initial movement parameter $r_{\varepsilon}\in[1\times 10^{-6}, 1\times 10^{0}]$. We also plot the \textit{optimal} fixed ULA step size, defined here as the step size for which ULA achieves the lowest value of $\mathrm{KL}(\mu_{T}\|\pi)$ at $T=500$.}
  \label{fig:2}
\end{figure}

\paragraph{Particle Approximation} In Figure \ref{fig:2}, we further investigate the behavior of the {ideal}, mean-field ULA x \textsc{Fuse} step size schedule $(\eta_t)_{t\geq 0}$, as well as the practical, particle-based step size schedule $(\hat{\eta}_t^n)_{t\geq 0}$, for several different values of the initial movement parameter $r_{\varepsilon}$. For comparison, we also plot the {optimal} fixed ULA step size, with optimality here determined (retrospectively) as the step size for which ULA achieves the lowest value of $\mathrm{KL}(\mu_{T}\|\pi)$ after $T=500$ iterations. 

We make several observations. First, regardless of the value of $r_{\varepsilon}$, the ULA x \textsc{Fuse} step sizes converge (from above) to the fixed optimal step size. This should not be surprising: the ULA x \textsc{Fuse} step size is an upper bound for the theoretically optimal fixed step size, as discussed in Section \ref{sec:adaptive-deterministic-step-size}. Second, after an initial decrease, the ULA x \textsc{Fuse} step size decreases as a function of the iteration number. This is encouraging: as indicated by the theory, the larger the iteration budget, the smaller one should take the step size. Finally, the practical, particle-based step size schedule (uniformly) upper bounds the {ideal}, mean-field step size schedule. Once again, this is to be expected: the approximate step size schedule relies on substituting an upper bound for the Wasserstein distance in the numerator of the ideal step size schedule, cf. Section \ref{sec:algorithms}. Interestingly, the disparity between the two step size schedules is significantly diminished for larger values of $r_{\varepsilon}$, since in this case $r_{\varepsilon}$ dominates the numerator (particularly for smaller values of $t$). 

\paragraph{Robustness to the Initial Distance} In Figure \ref{fig:3}, we continue to benchmark the performance of ULA x \textsc{Fuse}. In particular, we compare the performance of ULA and ULA x \textsc{Fuse} as we vary the distance from the initial distribution $\mu_0 = \mathcal{N}(0,\mathbf{I}_d)$ to the target distribution $\pi = \mathcal{N}(m_{\pi},\mathbf{I}_d)$, for $m_{\pi}\in[0.2,0.5,1.0,2.0]$. This figure illustrates the robustness of ULA x \textsc{Fuse} to the initial distance, a hyper-parameter which can have a significant impact on the optimal (range of) fixed ULA step size(s). We should emphasize that this effect is particularly acute for smaller iteration budgets, and may diminish as the iteration budget increases. 

\begin{figure}[t]
  \centering
  \includegraphics[width=\linewidth]{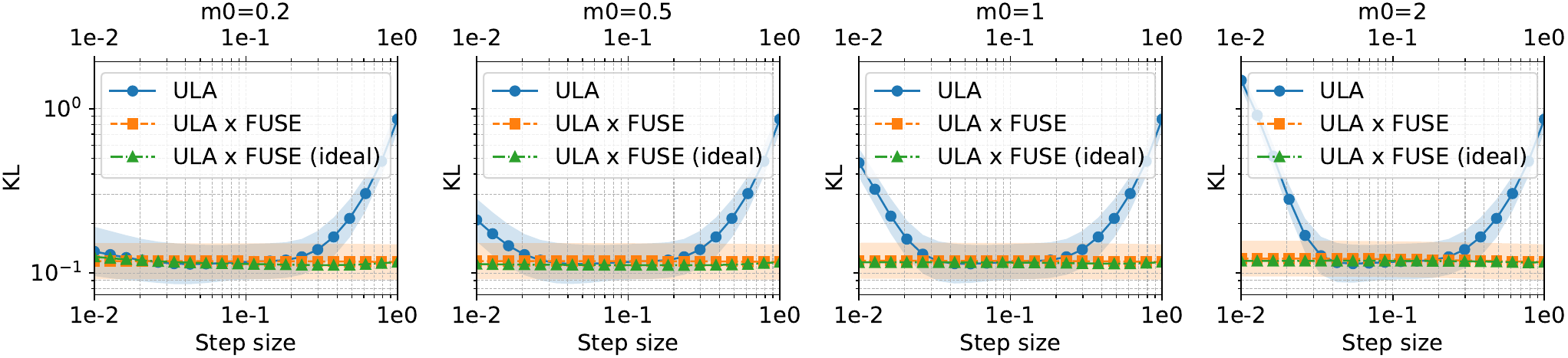}
  \caption{\textbf{A comparison of the ULA and ULA x \textsc{Fuse} for different distances between the initial distribution and the 10-dimensional Gaussian target distribution.} We plot $\mathrm{KL}(\mu_{T}\|\pi)$ as function of the step size $\eta$ (ULA) or the initial movement parameter $r_{\varepsilon}$ (ULA x \textsc{Fuse}) after $T=100$ iterations. In all panels, the initial distribution is $\mu_0 = \mathcal{N}(0,\mathbf{I}_d)$, while the target distribution is given by $\pi = \mathcal{N}(m_{\pi},\mathbf{I}_d)$, for $m_{\pi}\in[0.2,0.5,1.0,2.0]$.}
  \label{fig:3}
\end{figure}

\paragraph{Robustness to Anisotropy} In Figure \ref{fig:4}, we compare the performance of ULA and ULA x \textsc{Fuse} for Gaussian targets with increasing levels of anisotropy, i.e., increasing values of the condition number $\kappa = \frac{\lambda_{\mathrm{max}}}{\lambda_{\mathrm{min}}}$, where $\lambda_{\mathrm{max}}$ and $\lambda_{\mathrm{min}}$ denote the maximum and minimum eigenvalues of the covariance $\Sigma_{\pi}$ of the target distribution. As this figure makes evident, the optimal fixed ULA step size is highly dependent on the condition number. Meanwhile, ULA x \textsc{Fuse} is robust to the level of anisotropy. In all cases, it essentially matches the performance of an optimally tuned ULA, for any reasonable value of the initial movement parameter $r_{\varepsilon}$. Interestingly, these results also indicate that ULA is increasingly sensitive to the choice of step size as the anisotropy increases. That is, the range of step sizes for which ULA achieves near-optimal performance tightens as the anisotropy increases (compare, e.g., the ``U-shape'' of the left-hand panel in Figure \ref{fig:4} to the ``V-shape'' of the right-hand panel). This provides further motivation for the use of an automated step size scheduler like ULA x \textsc{Fuse}, which is agnostic to the level of anisotropy.

\begin{figure}[t]
  \centering
  \includegraphics[width=\linewidth]{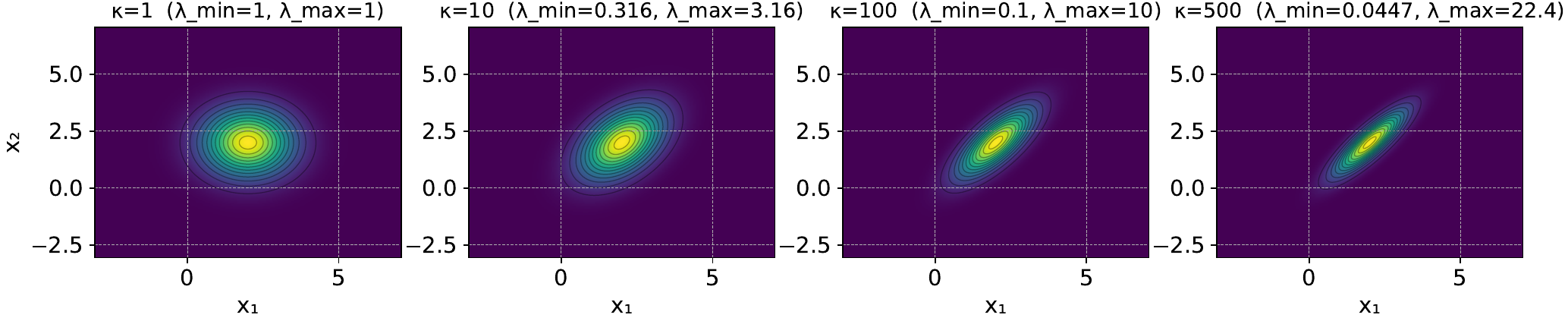}
  \vspace{3mm}
  \includegraphics[width=\linewidth]{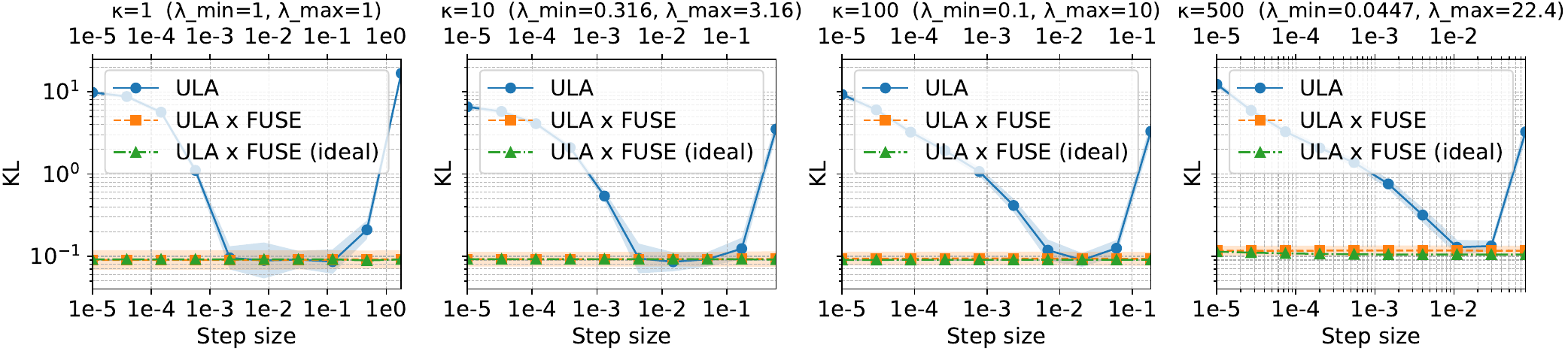}
  \caption{\textbf{A comparison of ULA and ULA x \textsc{Fuse} for 5-dimensional Gaussian target distributions with different levels of anisotropy.} We plot $\mathrm{KL}(\mu_{T}\|\pi)$ as a function of the step size $\eta$ (ULA) or the initial movement parameter $r_{\varepsilon}$ (ULA x \textsc{Fuse}) after $T=1000$ iterations. The initial distribution is $\mu_0 = \mathcal{N}(0,\mathbf{I}_d)$. The target distribution is given by $\pi = \mathcal{N}(m_{\pi},\mathbf{I}_d)$, with $m_{\pi} = (2,\dots,2)^{\top}$, and $\Sigma_{\pi}$ having condition numbers $\kappa:=\frac{\lambda_{\mathrm{max}}}{\lambda_{\mathrm{min}}}\in[1,10,100,500]$.}
  \label{fig:4}
\end{figure}

\paragraph{Remark}
\emph{
    Another effective method for dealing with anisotropic target distributions is to use a preconditioned version of the overdamped Langevin dynamics
    \citep[e.g.,][]{garbuno2020affine,carrillo2021wasserstein,burger2024covariance}. 
    For a particular choice of preconditioning, i.e., the covariance of the current distribution, the resulting dynamical system corresponds to the gradient flow of the KL divergence with respect to the so-called \emph{Kalman-Wasserstein} metric (see Section \ref{sec:alternative-gradient-flows}). 
    This contrasts with the standard overdamped Langevin dynamics, which represents the gradient flow of the KL divergence with respect to the Wasserstein metric \citep{jordan1998variational}. 
    After space-time discretization, the preconditioned variant yields an algorithm known as \emph{affine invariant Langevin dynamics} (ALDI) \cite{garbuno2020affine}.
    In principle, one could use a version of our step size schedule within this algorithm, with Wasserstein gradients replaced by Kalman-Wasserstein gradients. Moreover, replacing our existing assumptions (see Section \ref{sec:assumptions}) with the appropriate analogues (e.g., geodesic convexity in the Kalman-Wasserstein space), one could replicate our earlier theoretical results (see Section \ref{sec:main-results}).
}

\paragraph{SVGD vs SVGD x \textsc{Fuse}} We now turn our attention to SVGD and SVGD x \textsc{Fuse}. In the interest of space, we keep this comparison brief, noting that the observations made in the previous section (e.g., robustness to initial distance, robustness to anisotropy) also hold here. In Figure \ref{fig:5}, we compare the performance of these two algorithms, using the same experimental configuration as in Figure \ref{fig:1}. We now use the KSD (with IMQ kernel) as a performance metric, which only requires samples from the current distribution and the score of the target distribution \citep[e.g.,][]{gorham2017measuring}. The results are similar to those obtained for ULA x \textsc{Fuse}: the performance of SVGD x \textsc{Fuse} is similar to the performance of optimally tuned SVGD, for any choice of $r_{\varepsilon}$.

\begin{figure}[b]
  \centering
  \begin{subfigure}{0.485\linewidth}
    \includegraphics[width=\linewidth]{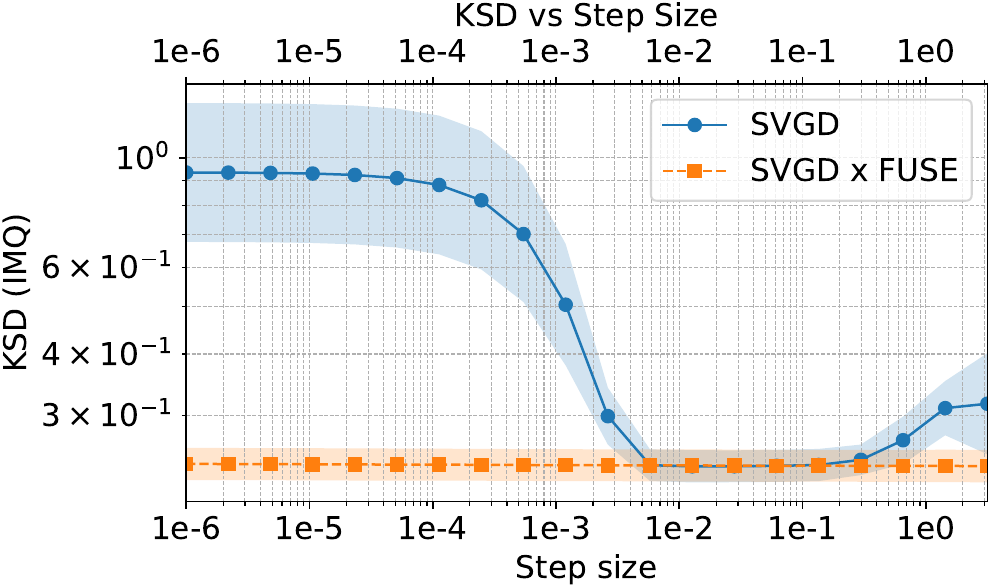}
    \caption{\textbf{KSD vs Step Size}: $\mathrm{KSD}(\mu_T\mid\pi)$ as a function of the fixed step size $\eta$ (SVGD) or the parameter $r_{\varepsilon}$ (SVGD x \textsc{Fuse}), after $T=500$ iterations.}
    \label{fig:5a}
  \end{subfigure}\hfill
  \begin{subfigure}{0.485\linewidth}
    \includegraphics[width=\linewidth]{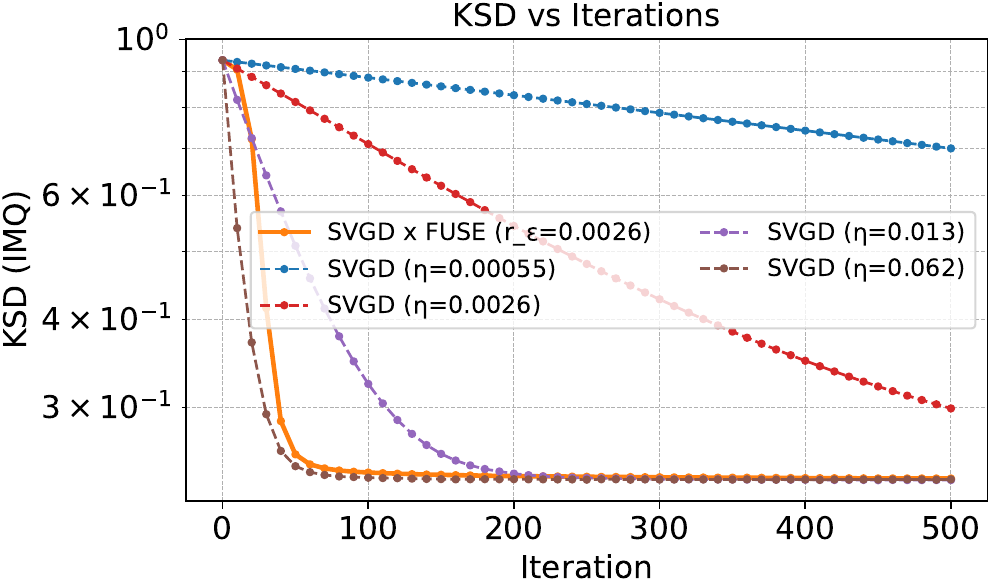}
    \caption{\textbf{KSD vs Iterations}: $\mathrm{KSD}(\mu_t\mid\pi)$ as a function of $t\in[0,500]$, for different values of the step size $\eta$ (SVGD) and a single value of $r_{\varepsilon}$ (SVGD x \textsc{Fuse}).}
    \label{fig:5b}
  \end{subfigure}
  \caption{\textbf{A comparison of SVGD and SVGD x \textsc{Fuse} for a 10-dimensional Gaussian target.}}
  \label{fig:5}
\end{figure}

\subsubsection{Bayesian Logistic Regression}
We next consider a Bayesian logistic regression problem \citep[e.g.,][]{held2006bayesian}. Consider i.i.d. observations $\mathcal{D}=(x_i,y_i)_{i=1}^N$, where $(y_i)_{i=1}^N$ are binary response variables, and $(x_i)_{i=1}^N$ are $\mathbb{R}^d$-valued features. The likelihood is then given by
\begin{equation}
    p(y_i\mid x_i,\beta) = p_i^{y_i}(1-p_i)^{1-y_i}, \quad \quad p_i = \frac{e^{\beta^{\top}x_i}}{1+e^{\beta^{\top}x_i}}. \label{eq:blr-likelihood}
\end{equation}

\paragraph{Gaussian Prior, Synthetic Data} In our first example, we consider a Gaussian prior for the regression coefficients $\boldsymbol{\beta}=(\beta_0,\beta_1,\dots,\beta_p)^{\top}$, where $\beta_0$ denotes the intercept. In particular, we assume that $\beta_j \sim \mathcal{N}(0,\lambda^{-1})$ for $j=1,\dots,p$, fixing $\lambda=5.0$ in our experiments. Meanwhile, the intercept is left unpenalized; formally, $\beta_0\sim \mathcal{N}(0,\lambda_0^{-1})$, with $\lambda_0=0$. 

We generate synthetic data as follows. We generate a ``true'' regression coefficient $\smash{\boldsymbol{\beta}=(\beta_0,\dots,\beta_p)^{\top}}$, setting $\beta_0=0$ and $\smash{\beta_j = s \frac{v_j}{||v||_2}}$ for $j=2,\dots,p$, where $\smash{v_j\sim \mathcal{N}(0,1)}$ and $s>0$ is a parameter which determines the {signal strength}. This normalization ensures that $||\beta||_{2} = s$. We generate feature vectors according to $\smash{x_i \sim \mathcal{N}(0,\Sigma_{\rho})}$ for $i=1,\dots,N$, where $\smash{(\Sigma_{\rho})_{jk}=\rho^{|j-k|}}$, and $\smash{\rho\geq 0}$ controls the correlation. We compute scaled logits $\smash{z_i = \tau^{-1}\boldsymbol{\beta}^{\top}x_i}$, where $\tau>0$ is a temperature parameter which controls class separability, with class probabilities given by $\smash{p_i = {e^{z_i}}/({1+e^{z_i}})}$. We then sample observations $y_i\sim \mathrm{Bernoulli}(p_i)$, for $i=1,\dots,N$. Finally, to introduce mild stochasticity, we independently flip each label with (small) probability $\pi_{\mathrm{flip}}$. In our experiments, we set $d=8$ and $N=500$. We set the other hyperparameters as $s=3.0$, $\rho=0.2$, $\tau=0.6$, $\pi_{\mathrm{flip}}=0.01$. The resulting model is non-separable but moderately well-conditioned, with a controlled signal-to-noise ratio determined $s=3.0$ and $\tau=0.6$, and a weak correlation structure $\rho=0.2$ between features.

We run ULA and ULA x \textsc{Fuse} using $n=100$ particles and for $T=500$ iterations. The step size $\eta$ (ULA) and the initial movement parameter $r_{\varepsilon}$ (ULA x \textsc{Fuse}) are varied over a logarithmic grid $[10^{-5}, 10^{-1}]$. We repeat all experiments over 10 random seeds. In order to assess the performance of each sampler, we approximate the posterior mean of four functionals $f_i:\mathbb{R}^d\rightarrow\mathbb{R}$, using the particles output at the final iteration. In particular, we consider the functionals
\begin{equation}
f_1:\boldsymbol{\beta}\mapsto\beta_0,\quad\quad f_2:\boldsymbol{\beta}\rightarrow \beta_1, \quad\quad f_3:\boldsymbol{\beta}\mapsto ||\boldsymbol{\beta}||_1, \quad \quad f_4:\boldsymbol{\beta}\rightarrow ||\boldsymbol{\beta}||_2^2.
\end{equation}
We also compute empirical 95\% credible intervals for each of the regression coefficients. We compare these values to ``ground-truth'' values, which we obtain via a long-run ($T=10^6$) ULA reference chain with small fixed step size ($\eta=10^{-4}$).

Our results are displayed in Figure \ref{fig:6} and Figure \ref{fig:7}. In Figure \ref{fig:6} we compare the posterior mean estimates obtained via ULA and ULA x \textsc{Fuse}, with posterior mean estimates obtained via a long-run ULA chain, which we treat as a proxy for the ground truth. In each box plot, we summarize the distribution of posterior mean estimates obtained over 10 independent seeds, corresponding to 10 distinct synthetic datasets. For the purpose of illustration, we report results for ULA using two sub-optimal choices of the fixed step size: $\eta=10^{-5}$ and $\eta=10^{-1}$. To ensure a fair comparison, we report results for ULA x \textsc{Fuse} using the \emph{same values} of the initial movement parameter: $r_{\varepsilon}=10^{-5}$ and $r_{\varepsilon}=10^{-1}$. Across all of the considered functionals, ULA x \textsc{Fuse} produces posterior mean estimates that are substantially closer to the ground truth. While standard ULA can perform well in this example, given a well-tuned step size, its performance significantly deteriorates when the step size is either overly conservative (top) or overly aggressive (bottom). In contrast, ULA x \textsc{Fuse} achieves comparable performance across a wide range of $r_{\varepsilon}$ values, indicating that there is no need to manually tune this parameter.

These observations are further supported by the results in Figure \ref{fig:7}, in which we plot the empirical 95\% credible intervals obtained using ULA and ULA x \textsc{Fuse}. In this case, rather than considering a single fixed step size $\eta$ (or initial movement parameter $r_0$), we aggregate results across multiple logarithmically spaced step sizes (or initial movement parameters) in the range $[10^{-5}, 10^{-1}]$. The result is the same: across all coefficients, the ULA x \textsc{Fuse} intervals much more closely resemble those obtained using the true posterior.

\begin{figure}[H]
  \centering
    \begin{subfigure}{0.239\linewidth}
    \includegraphics[width=\linewidth]{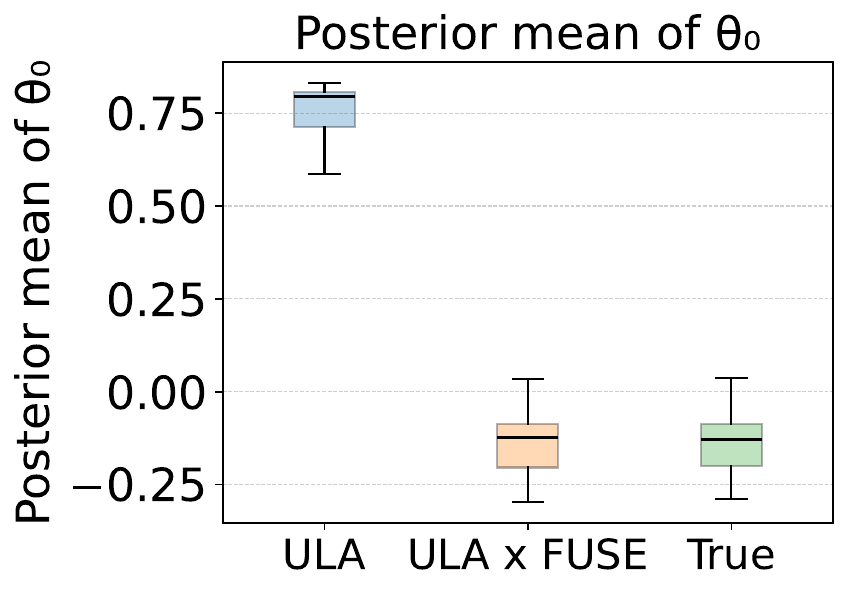}
    \caption{$\mathbb{E}[\beta_0]$, $\eta=10^{-5}$.}
    \label{fig:6a}
  \end{subfigure}
  \hfill
  \begin{subfigure}{0.24\linewidth}
    \includegraphics[width=\linewidth]{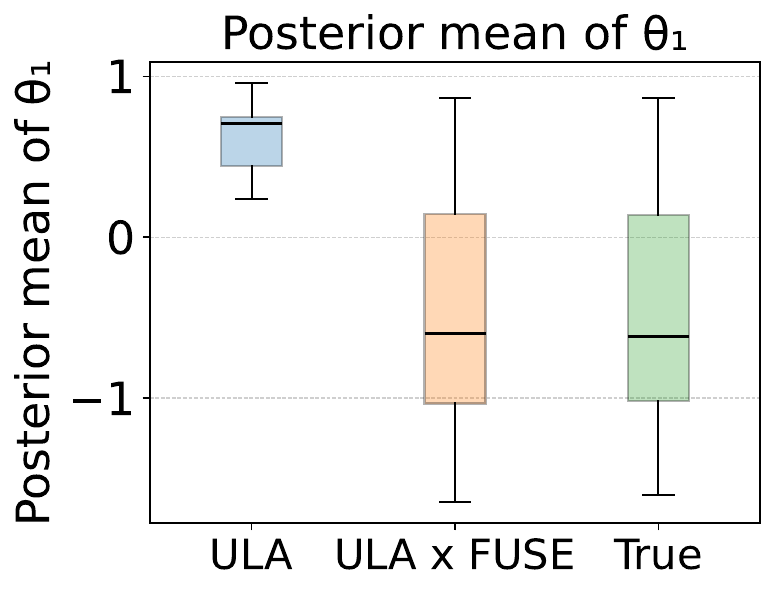}
    \caption{$\mathbb{E}[\beta_1]$, $\eta=10^{-5}$.}
    \label{fig:6b}
  \end{subfigure}
  \hfill
  \begin{subfigure}{0.223\linewidth}
    \includegraphics[width=\linewidth]{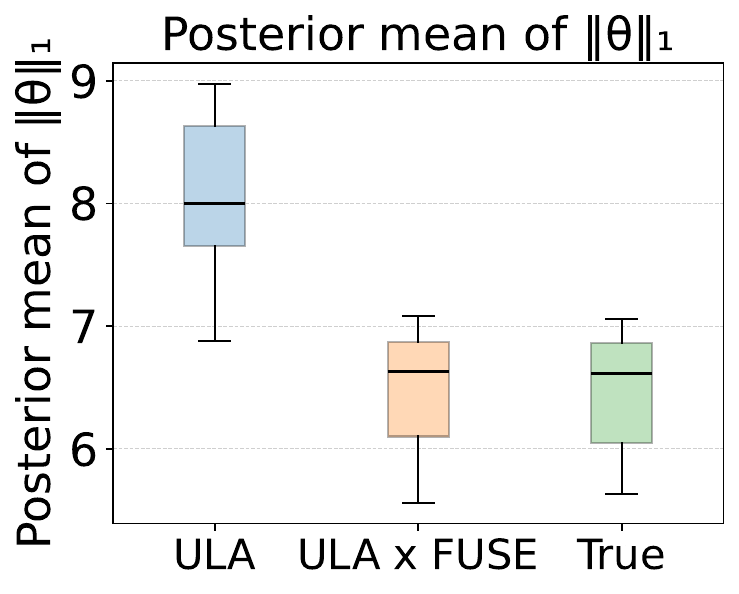}
    \caption{$\mathbb{E}[||\boldsymbol{\beta}||_1]$, $\eta=10^{-5}$.}
    \label{fig:6c}
  \end{subfigure}
  \hfill
  \begin{subfigure}{0.23\linewidth}
    \includegraphics[width=\linewidth]{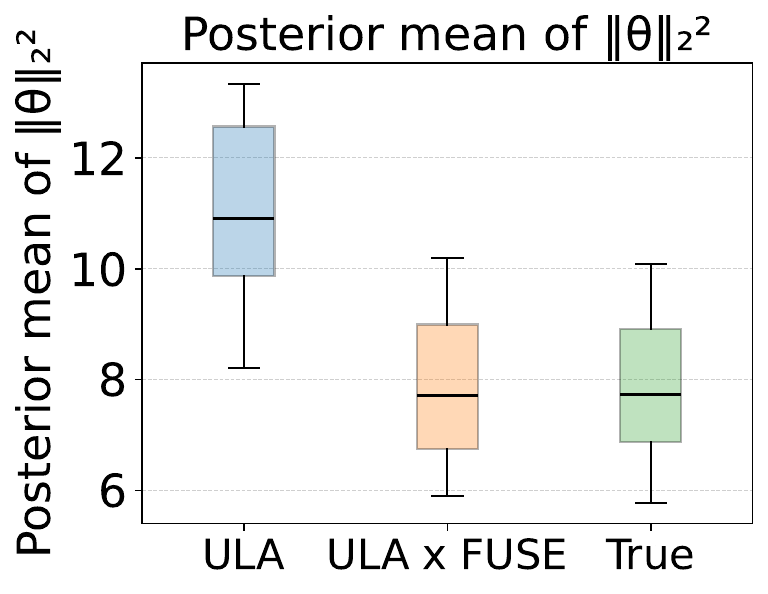}
    \caption{$\mathbb{E}[||\boldsymbol{\beta}||^2_2]$, $\eta=10^{-5}$.}
    \label{fig:6d}
  \end{subfigure} 
  \\[2mm]
  \begin{subfigure}{0.24\linewidth}
    \includegraphics[width=\linewidth]{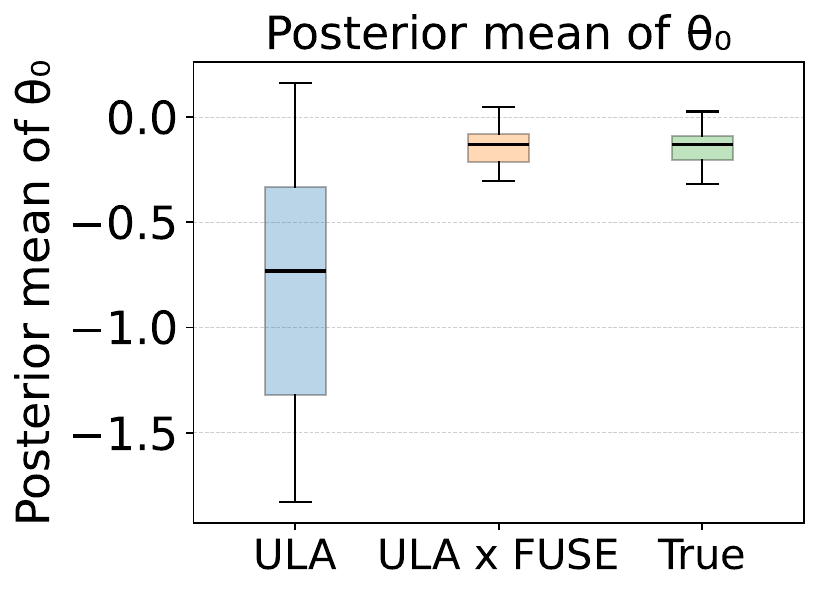}
    \caption{$\mathbb{E}[\beta_0]$, $\eta=10^{-1}$.}
    \label{fig:6a_}
  \end{subfigure}
  \hfill
  \begin{subfigure}{0.23\linewidth}
    \includegraphics[width=\linewidth]{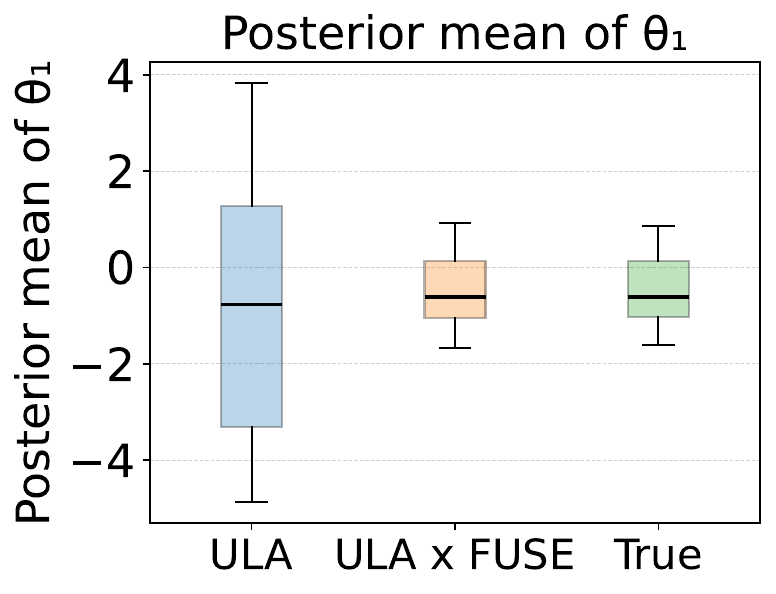}
    \caption{$\mathbb{E}[\beta_1]$, $\eta=10^{-1}$.}
    \label{fig:6b_}
  \end{subfigure}
  \hfill
  \begin{subfigure}{0.229\linewidth}
    \includegraphics[width=\linewidth]{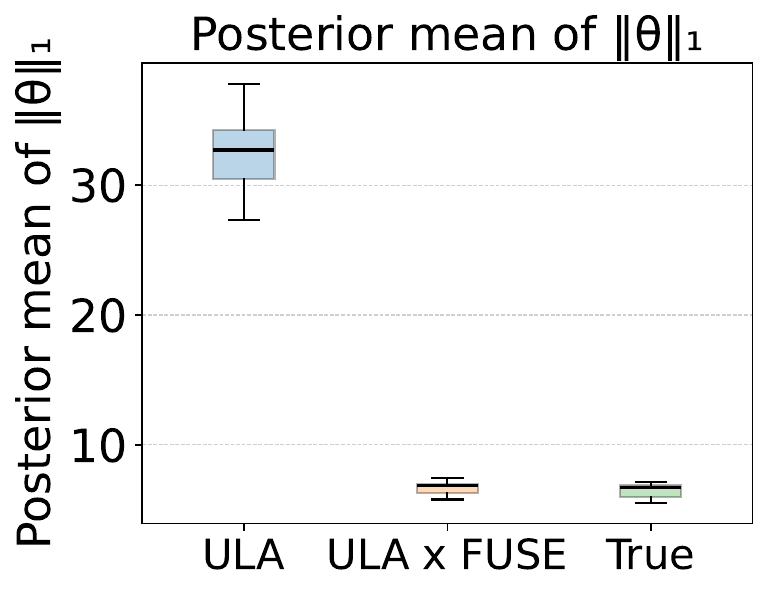}
    \caption{$\mathbb{E}[||\boldsymbol{\beta}||_1]$, $\eta=10^{-1}$.}
    \label{fig:6c_}
  \end{subfigure}
  \hfill
  \begin{subfigure}{0.232\linewidth}
    \includegraphics[width=\linewidth]{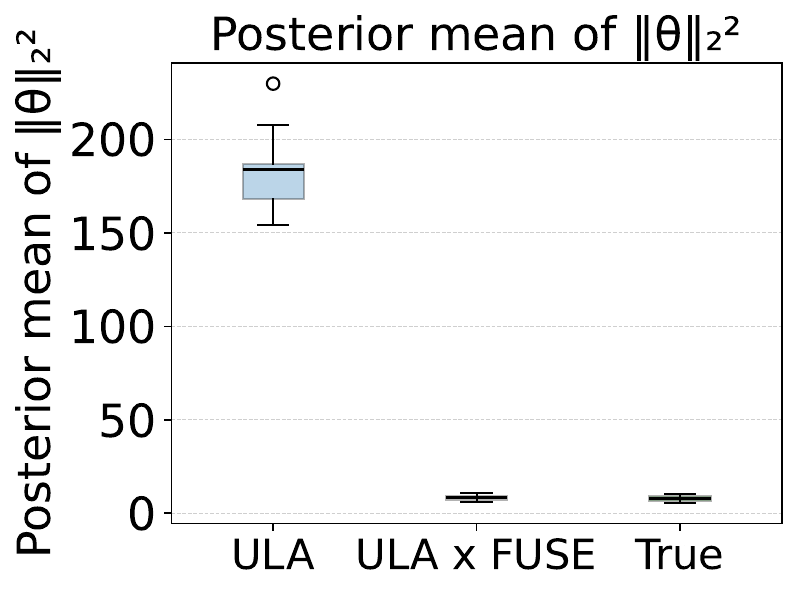}
    \caption{$\mathbb{E}[||\boldsymbol{\beta}||^2_2]$, $\eta=10^{-1}$.}
    \label{fig:6d_}
  \end{subfigure}
  \caption{\textbf{A comparison of ULA and ULA x \textsc{Fuse} for Bayesian logistic regression on a synthetic dataset.} We compare estimates of the posterior mean of four different functionals using the samples generated by ULA (with a sub-optimal choice of step size) and ULA x \textsc{Fuse} (with the same value of the initial movement parameter), after $T=500$ iterations. We also plot the values of these posterior expectations computed using samples from a very long-run ULA chain ($T=1\times 10^6$) with a small fixed step size ($\eta=1\times 10^{-4}$), which we treat as a proxy for the ground truth.}
  \label{fig:6}
\end{figure}

\begin{figure}[H]
  \centering
  \begin{subfigure}{0.46\linewidth}
    \includegraphics[width=\linewidth]{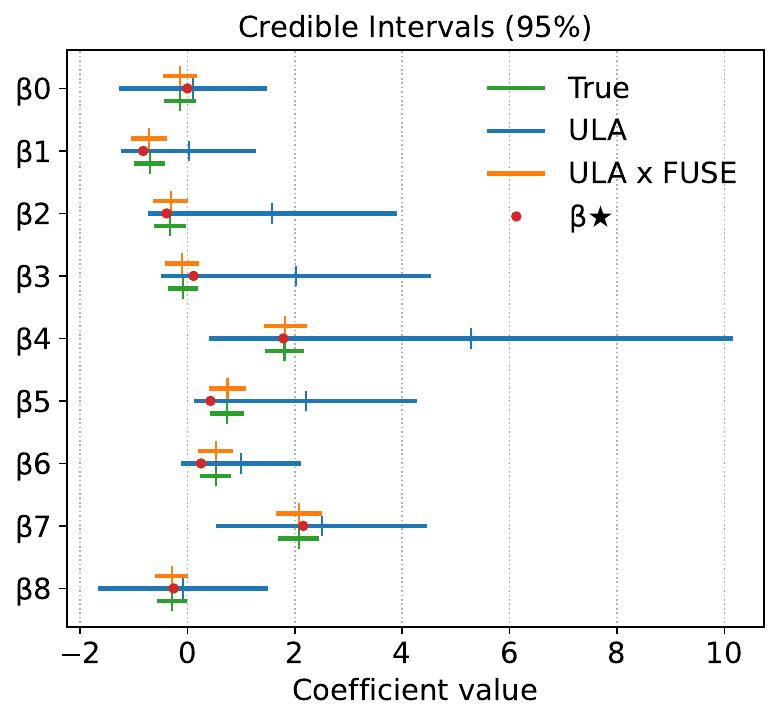}
    \caption{\textbf{Credible Intervals}.}
    \label{fig:7a}
  \end{subfigure}\hfill
  \begin{subfigure}{0.48\linewidth}
    \includegraphics[width=\linewidth]{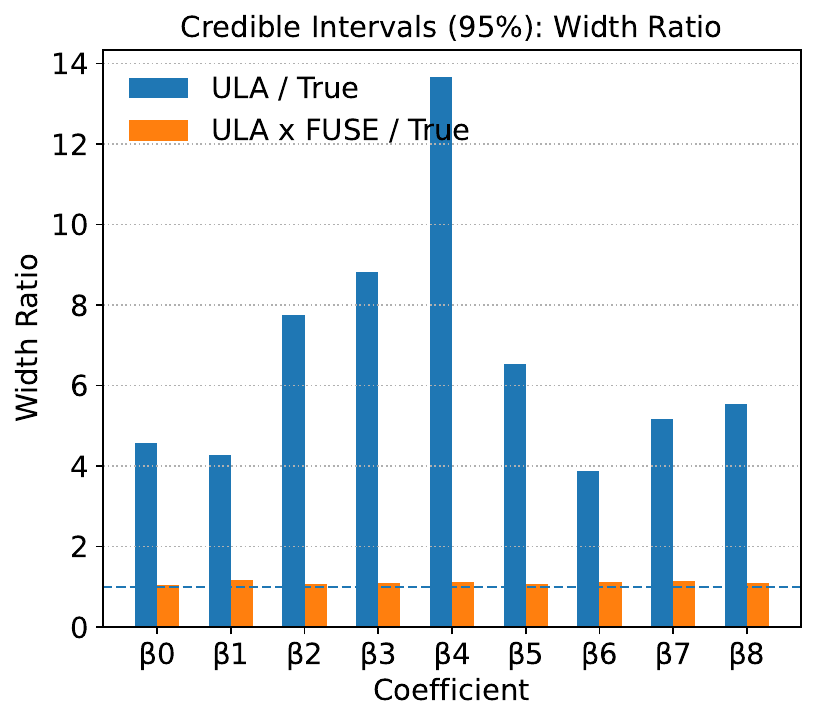}
    \caption{\textbf{Credible Intervals: Width Ratio}.}
    \label{fig:7b}
  \end{subfigure}
  \caption{\textbf{A comparison of ULA and ULA x \textsc{Fuse} for Bayesian logistic regression on a synthetic dataset.} We compare the 95\% credible intervals generated by ULA (using a range of step sizes) and ULA x \textsc{Fuse} (using the same range of initial movement parameters) after $T=500$ iterations. We also plot the 95\% credible intervals computed using samples from a very long-run ULA chain ($T=1\times 10^6$) with a small fixed step size ($\eta=1\times 10^{-4}$), which we treat as a proxy for the ground truth.}
  \label{fig:7}
\end{figure}

\paragraph{Hierarchical Prior, Real Data} We next consider a hierarchical specification for the prior over the regression coefficients $\beta = (\beta_1,\dots,\beta_d)^{\top}$, following \citet{mackay1995probable,gershman2012nonparametric,liu2016stein}. In particular, we assume that $p(\beta|\alpha) = \mathcal{N}(\beta|0,\alpha^{-1}\mathbf{I}_d)$, $p(\alpha) = \Gamma(\alpha |a,b)$, where $a,b\in\mathbb{R}_{+}$ are hyper-parameters. We will assume throughout that $a=1,b=0.01$. We are then interested in sampling from the posterior distribution $\pi({\theta}):=p({\theta}| \mathcal{D})$, where ${\theta}=(\alpha,\beta)^{\top}$. We now test our algorithms on the Covertype dataset \citep{covertype1998}, which consists of 581,012 data points and 54 features. Due to the size of the dataset, it is necessary to use stochastic gradients, which we compute using mini-batches of size 100. To evaluate the performance of each algorithm, we randomly partition the data into a training dataset (80\%) and a testing dataset (20\%). For each method, we use $n=20$ particles, $T=1000$ iterations, and report results averaged over 10 random runs.

The results for SGLD x \textsc{Fuse} and SVGD x \textsc{Fuse} are shown in Figures \ref{fig:8} and \ref{fig:9}, respectively. In particular, we plot the predictive (test) accuracy as a function of the step size (left-hand panel) and the iteration number (right-hand panel). The results are similar to before. In particular, the predictive performance of SGLD x \textsc{Fuse} and SVGD x \textsc{Fuse} is highly robust to the choice of the initial movement parameter $r_{\varepsilon}$, and essentially matches the optimal predictive performance of SGLD and SVGD with a fixed step size.

\begin{figure}[H]
\vspace{-1mm}
  \centering
  \begin{subfigure}{0.475\linewidth}
    \includegraphics[width=\linewidth]{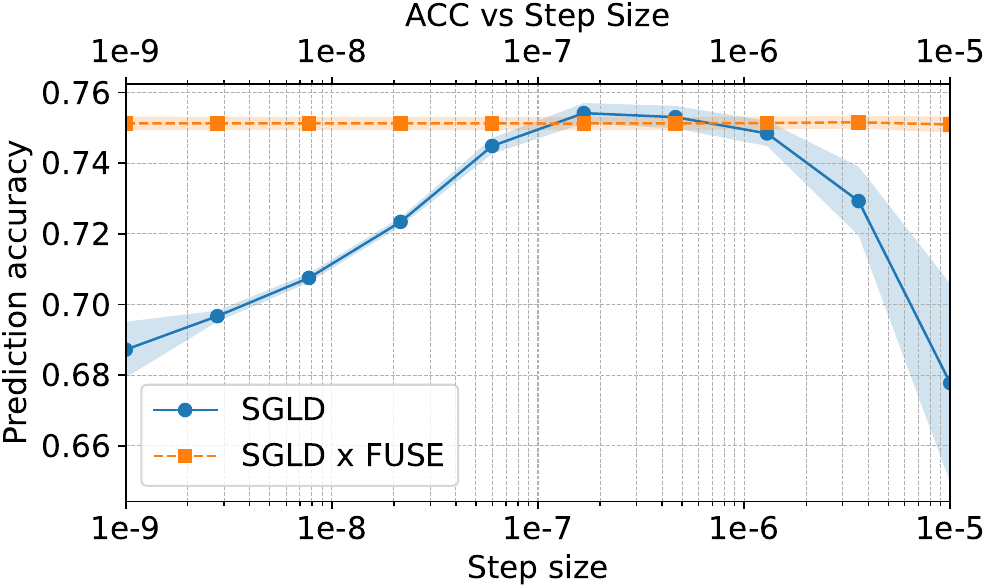}
    \caption{\textbf{Predictive Accuracy vs Step Size}: Predictive accuracy as a function of the fixed step size $\eta$ (SGLD) or the parameter $r_{\varepsilon}$ (SGLD x \textsc{Fuse}), after $T=500$ iterations.}
    \label{fig:8a}
  \end{subfigure}\hfill
  \begin{subfigure}{0.475\linewidth}
    \includegraphics[width=\linewidth]{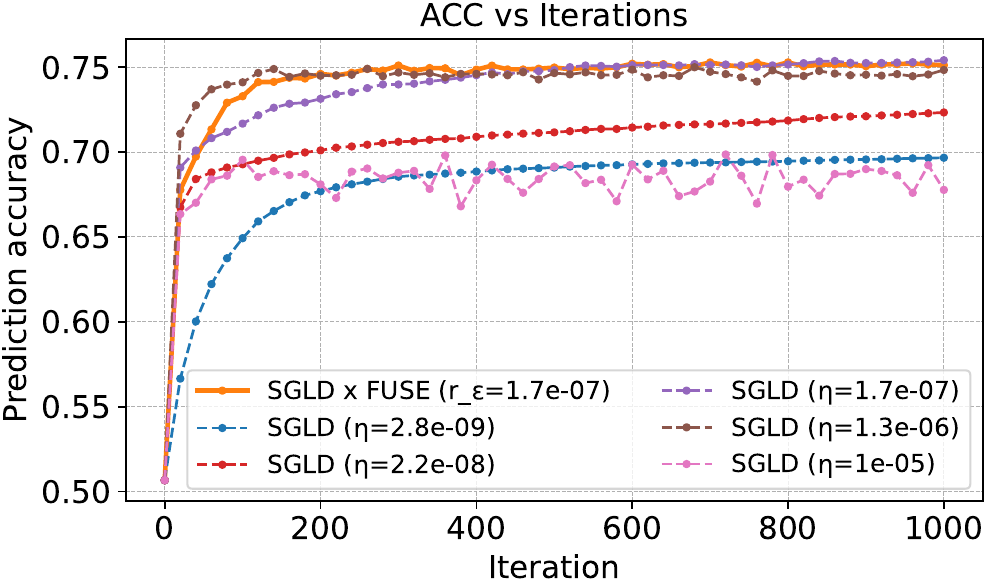}
    \caption{\textbf{Predictive Accuracy vs Iterations}: Predictive accuracy as a function of $t\in[0,500]$, for different values of the step size $\eta$ (SGLD) and a single value of $r_{\varepsilon}$ (SGLD x \textsc{Fuse}).}
    \label{fig:8b}
  \end{subfigure}
  \caption{\textbf{A comparison of SGLD and SGLD x \textsc{Fuse} for Bayesian logistic regression on the Covertype dataset.}}
  \label{fig:8}
\end{figure}

\begin{figure}[H]
\vspace{-5mm}
  \centering
  \begin{subfigure}{0.475\linewidth}
    \includegraphics[width=\linewidth]{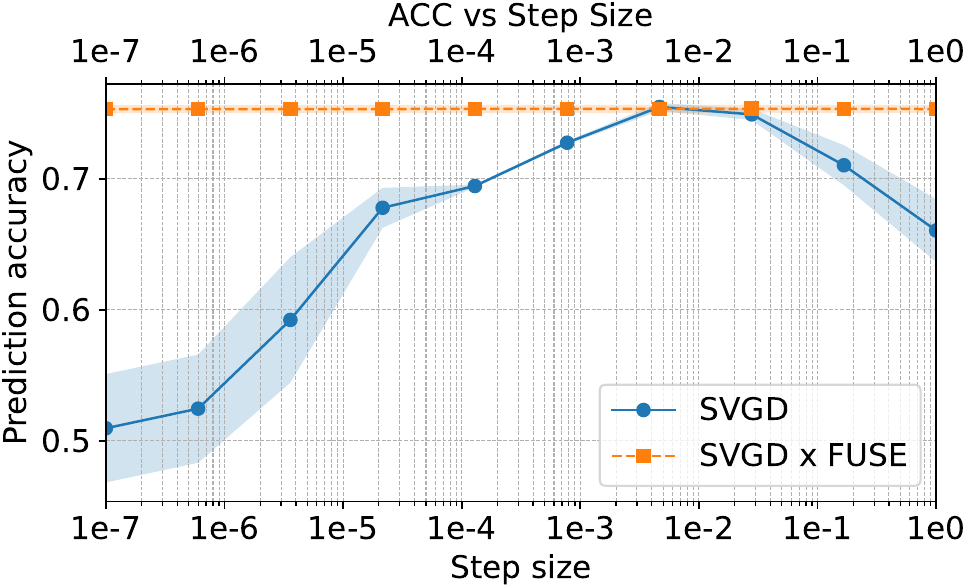}
    \caption{\textbf{Predictive Accuracy vs Step Size}: Predictive accuracy as a function of the fixed step size $\eta$ (SVGD) or the parameter $r_{\varepsilon}$ (SVGD x \textsc{Fuse}), after $T=500$ iterations.}
    \label{fig:9a}
  \end{subfigure}\hfill
  \begin{subfigure}{0.475\linewidth}
    \includegraphics[width=\linewidth]{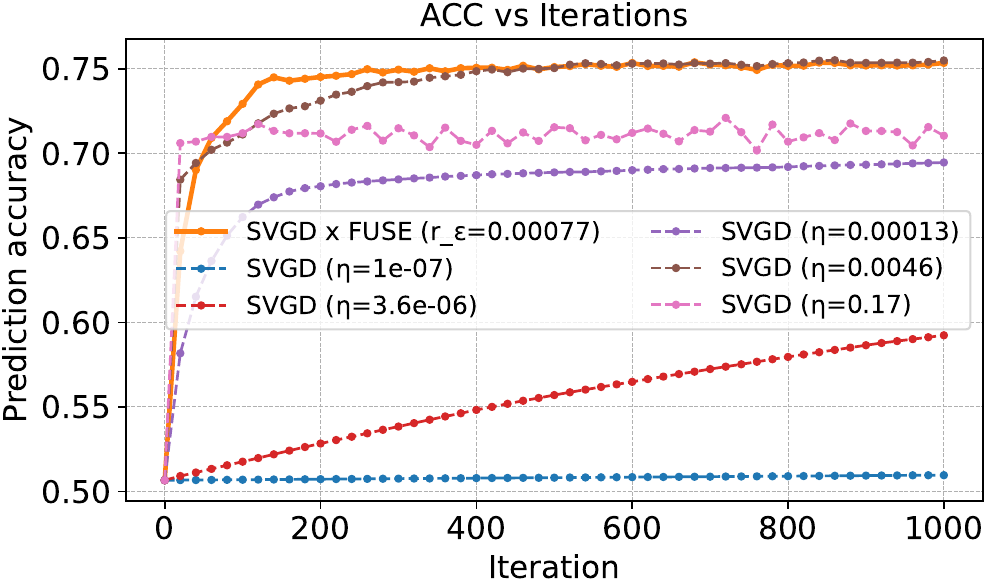}
    \caption{\textbf{Predictive Accuracy vs Iterations}: Predictive accuracy as a function of $t\in[0,500]$, for different values of the step size $\eta$ (SVGD) and a single value of $r_{\varepsilon}$ (SVGD x \textsc{Fuse}).}
    \label{fig:9b}
  \end{subfigure}
  \caption{\textbf{A comparison of SVGD and SVGD x \textsc{Fuse} for Bayesian logistic regression on the Covertype dataset.}}
  \label{fig:9}
\vspace{-7mm}
\end{figure}

We continue to investigate the performance of SGLD x \textsc{Fuse} in Figure \ref{fig:10}, where we perform a sensitivity analysis with respect to the mini-batch size used to compute the stochastic gradients. To reduce the computational burden, we consider a randomly selected subset of the Covertype dataset with $20,000$ data points. We then run each algorithm using $n=20$ particles, $T=1000$ iterations, and mini-batch sizes $[10,100,1000]$. Across all mini-batch sizes, SGLD exhibits the characteristic U-shaped performance profile: small step sizes lead to slow exploration and under-dispersed samples, while large step sizes cause numerical instability and divergence of the chain. In contrast, SGLD x \textsc{Fuse} maintains near-optimal predictive accuracy across several orders of magnitude in the initial movement parameter.

Comparing the three panels in Figure \ref{fig:10}, it is clear that the mini-batch size has a relatively pronounced effect on the performance of SGLD \citep{dubey2016variance,baker2019control,putcha2023preferential}. For small mini-batches (left-hand panel), the stochastic gradient noise dominates, and the performance of SGLD degrades sharply except in a narrow range of step sizes. SGLD x \textsc{Fuse}, however, maintains stable accuracy despite this increased gradient variance, suggesting that the adaptive update balances the stochasticity by automatically reducing the effective step size when gradients become noisy. As the mini-batch size increases (center panel, right-hand panel), both methods improve in stability, but SGLD x \textsc{Fuse} consistently matches or exceeds the best performance of SGLD, with no need for tuning. Notably, for the largest batch size, SGLD x \textsc{Fuse} achieves near-optimal accuracy across the full range of tested step sizes, whereas SGLD still requires careful step size selection. More broadly, our results indicate that the optimal fixed step size for SGLD is dependent on the mini-batch size. SGLD x \textsc{Fuse} removes the need to (re)-tune this parameter across different mini-batch configurations.

\begin{figure}[t]
  \centering
  \includegraphics[width=\linewidth]{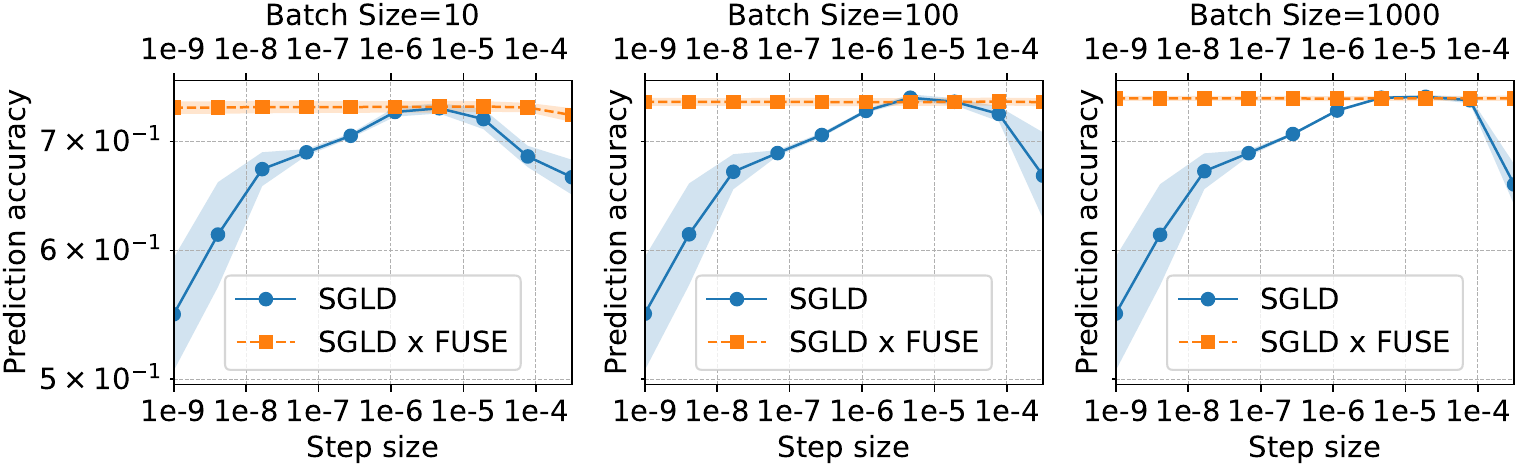}
  \caption{\textbf{A comparison of SGLD and SGLD x \textsc{Fuse} for different batch sizes, for Bayesian logistic regression on a subset of the Covertype dataset.} We plot the predictive accuracy as a function of the step size $\eta$ (SGLD) or the initial movement parameter $r_{\varepsilon}$ (SGLD x \textsc{Fuse}) after $T=1000$ iterations. We use mini-batch gradients with mini-batch sizes $[10,100,1000]$.}
  \label{fig:10}
\end{figure}

\subsection{Example 2: Training a Mean-Field Neural Network} We now consider the task of training a mean-field neural network, i.e., minimizing a functional of the form  $\mathcal{F}(\mu) = \mathcal{E}(\mu) + \mathrm{Ent}(\mu)$, where $\mathcal{E}(\mu) = \mathcal{E}_0(\mu) + \int r(x) \mathrm{d}\mu(x)$; $\mathcal{E}_0(\mu) = \frac{\lambda_1}{n}\sum_{k=1}^n \ell(y_k,h_{\mu}(z_k))$ is the (scaled) empirical risk of the mean-field neural network $h_{\mu}(z) = \int h_{x}(z)\mathrm{d}\mu(x)$ defined by a neural network $h_{x}(z)$ with parameter $x\in\mathbb{R}^d$, given training data $(z_i,y_i)_{i=1}^n\in\mathbb{R}^{d-1}\times\mathbb{R}$ and convex loss function $\ell:\mathbb{R}\times\mathbb{R}\rightarrow\mathbb{R}$; and $r:\mathbb{R}^d\rightarrow\mathbb{R}$ is a regularizer. 

We test our algorithms using an example similar to the one described in \citet[][Section 3.1]{chazal2025computable}. To be specific, we consider a univariate regression task defined as follows. We first generate synthetic data $\smash{\{(z_i,y_i)\}_{i=1}^N}$ of size $N=300$ by sampling each covariate $\smash{z_i\stackrel{\mathrm{i.i.d.}}{\sim}\mathcal{U}(0,1)}$, and then sampling $y_i|z_i \sim \mathcal{N}(\cdot|f(z_i), \sigma^2)$, with mean $f(z) =3\tanh(3z+\frac{1}{2})$ and $\sigma=0.1$. We use the squared error loss $\ell(y,y') = (y-y')^2$, ignore the regularizer so that $r(x) = 0$, and let $h_x(z) = w_2 \cdot \tanh(w_1 z + b_1) + b_2$ be a two-layer neural network with parameter $x=(w_1,b_1,w_2,b_2)^{\top} \in \mathbb{R}^4$. Finally, we set $\lambda_1 = 300$. We run each algorithm using $n=100$ particles for $T=1000$ iterations, initialized from $\mu_0=\mathcal{N}(0,\mathbf{I}_d)$. We assess the performance of each method using the test mean squared error (MSE), which we compute using an independent data set of size $N=300$, generated in the same way as the training data. We report results averaged over 10 random seeds.

The results, displayed in Figures \ref{fig:11} - \ref{fig:12}, are consistent with those observed in the previous experiments. MFLD x \textsc{Fuse} and VGD x \textsc{Fuse} maintain a consistently low test error across a broad range of step sizes, while the performance of MFLD and VGD deteriorates rapidly as the step size increases or decreases outside a narrow range. In particular, if the step size is chosen too small, MFLD and VGD do not converge within the given number of iterations, and if the step size is chosen too large, the bias incurred is substantial. These results highlight that our proposed tuning-free framework generalizes to models where our assumptions do not necessarily hold (e.g., non-geodesically convex), maintaining their efficacy in complex targets with pronounced multimodality.

\begin{figure}[t]
  \centering
  \begin{subfigure}{0.485\linewidth}
    \includegraphics[width=\linewidth]{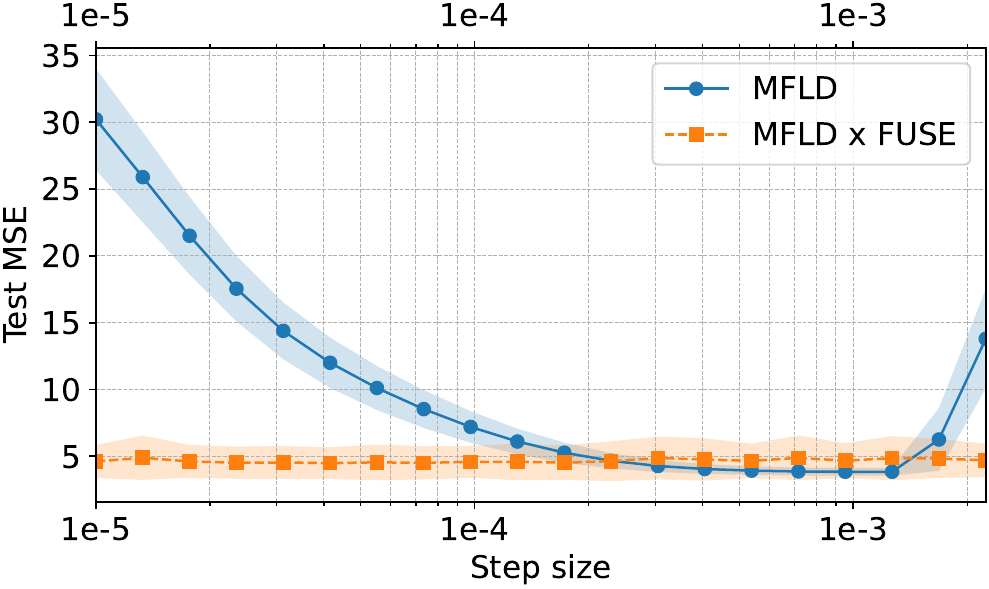}
    \caption{\textbf{MFLD vs MFLD x \textsc{Fuse}}.}
    \label{fig:11a}
  \end{subfigure}\hfill
  \begin{subfigure}{0.485\linewidth}
    \includegraphics[width=\linewidth]{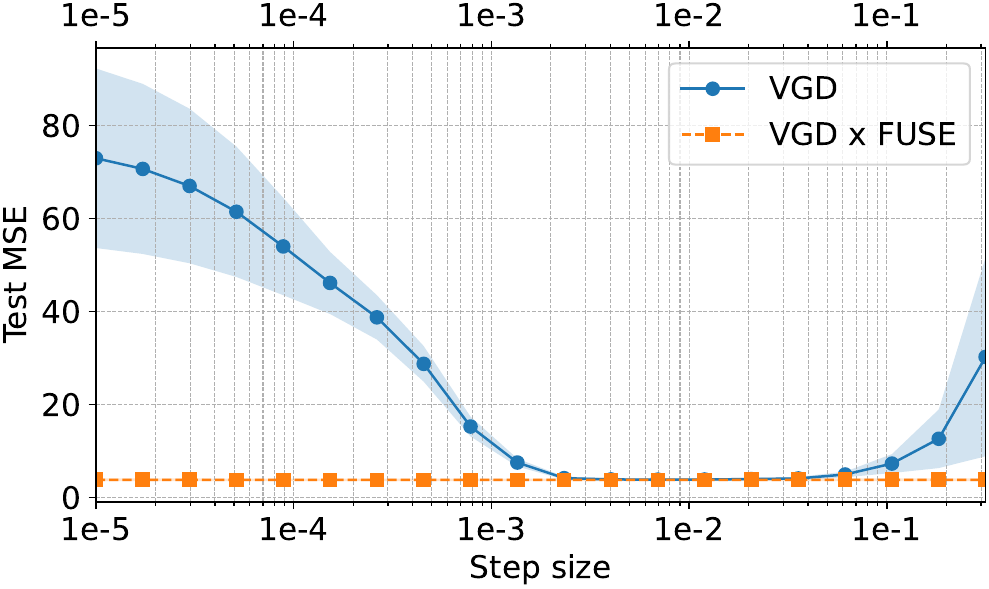}
    \caption{\textbf{VGD vs VGD x \textsc{Fuse}}.}
    \label{fig:11b}
  \end{subfigure}
  \caption{\textbf{A comparison of MFLD, MFLD x \textsc{Fuse} (left) and VGD, VGD x \textsc{Fuse} (right) for a mean-field neural network.} We plot the prediction accuracy vs the step size (MFLD, VGD) or the initial movement parameter (MFLD x \textsc{Fuse}, VGD x \textsc{Fuse}) after $T=1000$ iterations.}
  \label{fig:11}
\vspace{-2mm}
\end{figure}

\begin{figure}[b]
\vspace{-2mm}
  \centering
  \includegraphics[width=.9\linewidth]{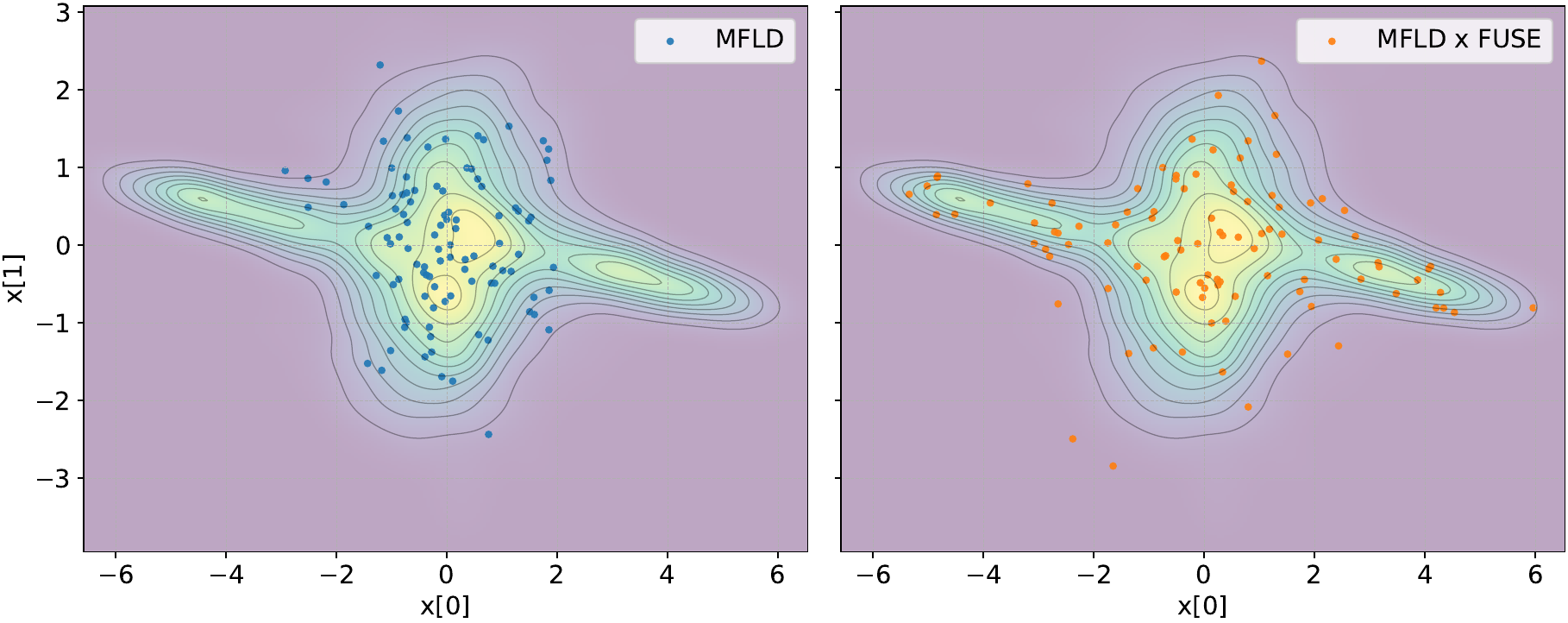}
  \caption{\textbf{A comparison of MFLD and MFLD x \textsc{Fuse} for the mean-field neural network.} We plot a set of 100 representative particles obtained via MFLD with $\eta=1\times 10^{-5}$ (left) and MFLD x \textsc{Fuse} with $r_{\varepsilon} = 1\times 10^{-5}$ (right) after $T=1000$ iterations.}
  \label{fig:12}
\vspace{-2mm}
\end{figure}

\section{Conclusions}
\label{sec:conclusions}

\subsection{Summary} In this paper, we have introduced \textsc{Fuse}: an adaptive, tuning-free step size schedule for (stochastic) optimization algorithms obtained as time-discretizations of Wasserstein gradient flows. We considered algorithms based on two popular time-discretizations -- namely, the forward Euler discretization and the forward-flow discretization -- and obtained rigorous theoretical guarantees in both cases. In particular, assuming geodesic convexity of the objective functional, we established that \textsc{Fuse} recovers the convergence rate of the optimal fixed step size up to a logarithmic factor, in both nonsmooth and smooth settings. Based on our theoretical results, we then introduced a number of step-size-free, gradient-based sampling algorithms, including variants of ULA, MFLD, SVGD, and VGD, as well as their stochastic (mini-batch) counterparts. 

In principle, the approach proposed in this paper could be used to obtain step-size-free algorithms for a broad class of optimization problems over the space of probability measures. Our primary focus in this paper, however, was on the development of tuning-free sampling algorithms. This was largely motivated by the growing demand for sampling methods which not only avoid the need for time-consuming manual tuning but are also applicable in modern ``big-data'' settings, where computing the likelihood (or its gradient) over the full dataset is often infeasible. 

There is, of course, a long history of work focused on the development of automated or semi-automated sampling methods. For example, adaptive MCMC, which iteratively tunes the parameter(s) of an MCMC proposal distribution based on the history of the chain, dates back several decades \citep{haario1999,haario2005componentwise,andrieu2006,roberts2007coupling,andrieu2008tutorial,pasarica2010adaptively,atchade2011adaptive}, and now encompasses a large family of different algorithms \citep[e.g.,][]{marshall2012adaptive,vihola2012robust,latuszynski2013adaptive,miasojedow2013adaptive,mohamed2013adaptive,leimkuhler2018ensemble,pompe2020framework,hoffman2021adaptive,karppinen2021conditional,gabrie2022adaptive,titsias2023optimal,cabezas2024markovian,tian2024adaptive,wang2024reinforcement,wang2025harnessing}.

Such approaches, however, are typically designed for exact MCMC algorithms which involve a Metropolis-Hastings acceptance step (e.g., RMW, MALA, HMC).  In particular, classical adaptive methods typically target an {optimal} acceptance rate derived from high-dimensional scaling limits \citep{gelman1997weak,roberts1998optimal,roberts2001optimal,bedard2008efficient,bedard2008optimal,beskos2013optimal}, maximize the expected squared jump distance \citep{pasarica2010adaptively,titsias2023optimal}, or consider some variant of these objectives \citep{levy2017generalizing,titsias2019gradient,campbell2021gradient,cannella2022semi,dharamshi2023sampling}. These methods are thus ill-suited to modern ``big-data'' settings, owing to the prohibitive cost of computing the Metropolis-Hastings acceptance ratio over the full dataset \citep{bardenet2017markov,prado2024metropolis}. Unlike these approaches, our methodology can be applied to automate the tuning of approximate, unadjusted MCMC algorithms (e.g., ULA), which, when combined with subsampling, yield stochastic gradient MCMC methods (e.g., SGLD) which are highly scalable to modern, large-scale datasets.

\subsection{Future Work} We conclude by briefly highlighting a number of interesting directions for future work. From a theoretical perspective, our results for both time-discretizations were stated in terms of the ideal, {mean-field} \textsc{Fuse} step size schedules, defined in \eqref{eq:dog-lr-forward-flow-recall} and \eqref{eq:dog-lr-forward-euler-recall}, respectively. However, this sequence of step sizes can only be computed directly when both the initial distribution and the target distribution are Gaussian. Naturally, it would be of interest to analyze the theoretical properties of the finite-particle step size schedule which we use in practice, cf. \eqref{eq:dog-lr-forward-flow-particle} and \eqref{eq:dog-lr-forward-euler-particle}. For this, the results obtained in \citet{fournier2015rate} may be a good starting point.

Our theoretical results also concerned the mean-field limit of the time-discretizations themselves. In the sampling context, whereby the objective functional is the KL divergence, the forward-flow discretization corresponds to ULA, and can be implemented directly. On the other hand, the forward Euler discretization does not result in an implementable algorithm, since it depends on the gradient of the log density of the current distribution. In fact, the forward Euler discretization cannot even be (directly) approximated using a collection of interacting particles, since the entropy is ill-defined for discrete measures. Thus, in practice, we form a finite-particle approximation to a kernelized version of both the original update equation and our proposed step size schedule (see Section \ref{sec:practical-forward-euler}). This is another gap between theory and practice which it would be of interest to close.

Finally, all of our theoretical results relied on the assumption that the objective functional (e.g., the KL divergence) was geodesically convex. In the sampling context, this is equivalent to the assumption that the target distribution is log-concave. While this assumption is by no means uncommon in the sampling literature \citep[e.g.,][]{durmus2019analysis}, it does exclude many targets encountered in practice. Given this, it would be of great interest to extend our theoretical analysis to more general settings, perhaps leveraging other ideas from the optimization literature \citep[e.g.,][]{ahn2025general}.

Regarding methodology, there are also several interesting extensions of our approach. One is to consider variants of our step size schedule (e.g., a sum of weighted gradients), following similar developments in the optimization literature \citep[e.g.,][]{defazio2023learning,khaled2023dowg,mishchenko2024prodigy}. In another direction, it would be interesting to develop tuning-free variants of momentum-enriched sampling algorithms such as (stochastic gradient) Hamiltonian Monte Carlo \citep[e.g.,][]{chen2014stochastic,nemeth2021stochastic}, which are well known to improve the convergence of first-order methods such as (stochastic gradient) Langevin dynamics. Here, once again, it may be possible to take inspiration from the optimization literature \citep[e.g.,][]{kreisler2024accelerated}.


\section*{Acknowledgements}
CN was supported by the UK Research and Innovation (UKRI) Engineering and Physical Sciences Research Council (EPSRC), grant number EP/V022636/1. CN acknowledges further support from the EPSRC ProbAI Hub, grant number EP/Y028783/1.


\bibliographystyle{apalike}
\bibliography{references}

\begin{thebibliography}{}

\bibitem[Ahn and Chewi, 2021]{ahn2021efficient}
Ahn, K. and Chewi, S. (2021).
\newblock {Efficient constrained sampling via the mirror-{Langevin} algorithm}.
\newblock In {\em Proceedings of the 35th Annual Conference on Neural
  Information Processing Systems (NeurIPS 2021)}, Virtual.

\bibitem[Ahn et~al., 2025]{ahn2025general}
Ahn, K., Magakyan, G., and Cutkosky, A. (2025).
\newblock General framework for online-to-nonconvex conversion: schedule-free
  {SGD} is also effective for nonconvex optimization.
\newblock In {\em Proceedings of the 42nd International Conference on Machine
  Learning (ICML 2025)}, Vancouver, Canada.

\bibitem[Alvarez-Melis et~al., 2022]{alvarezmelis2022optimizing}
Alvarez-Melis, D., Schiff, Y., and Mroueh, Y. (2022).
\newblock {Optimizing Functionals on the Space of Probabilities with Input
  Convex Neural Networks}.
\newblock {\em Transactions on Machine Learning Research}.

\bibitem[Amari, 2015]{amari2016information}
Amari, S.-I. (2015).
\newblock {\em Information Geometry and its Applications}.
\newblock Springer Tokyo.

\bibitem[Ambrosio et~al., 2021]{ambrosio2021lectures}
Ambrosio, L., Brué, E., and Semola, D. (2021).
\newblock {\em Lectures on Optimal Transport}.
\newblock Springer Cham.

\bibitem[Ambrosio et~al., 2008]{ambrosio2008gradient}
Ambrosio, L., Gigli, N., and {Savar{\'{e}}}, G. (2008).
\newblock {\em Gradient Flows: in Metric Spaces and in the Space of Probability
  Measures}.
\newblock Birkh{\"{a}}user, Basel.

\bibitem[Ambrosio and Savar{\'e}, 2007]{ambrosio2007gradient}
Ambrosio, L. and Savar{\'e}, G. (2007).
\newblock Gradient flows of probability measures.
\newblock In {\em Handbook of differential equations: evolutionary equations},
  volume~3, pages 1--136. Elsevier.

\bibitem[Andrieu et~al., 2003]{andrieu2003introduction}
Andrieu, C., de~Freitas, N., Doucet, A., and Jordan, M.~I. (2003).
\newblock An introduction to {MCMC} for machine learning.
\newblock {\em Machine Learning}, 50(1):5--43.

\bibitem[Andrieu and Moulines, 2006]{andrieu2006}
Andrieu, C. and Moulines, {\'E}. (2006).
\newblock On the ergodicity properties of some adaptive {MCMC} algorithms.
\newblock {\em The Annals of Applied Probability}, 16(3):1462 -- 1505.

\bibitem[Andrieu and Thoms, 2008]{andrieu2008tutorial}
Andrieu, C. and Thoms, J. (2008).
\newblock A tutorial on adaptive {MCMC}.
\newblock {\em Statistics and Computing}, 18:343--373.

\bibitem[Ara{\'u}jo et~al., 2019]{araujo2019mean}
Ara{\'u}jo, D., Oliveira, R.~I., and Yukimura, D. (2019).
\newblock A mean-field limit for certain deep neural networks.
\newblock {\em arXiv preprint arXiv:1906.00193}.

\bibitem[Arbel et~al., 2019]{arbel2019maximum}
Arbel, M., Korba, A., Salim, A., and Gretton, A. (2019).
\newblock Maximum mean discrepancy gradient flow.
\newblock In {\em Proceedings of the 33rd Annual Conference on Neural
  Information Processing Systems (NeurIPS 2019)}, Vancouver, Canada.

\bibitem[Atchade et~al., 2011]{atchade2011adaptive}
Atchade, Y., Fort, G., Moulines, E., and Priouret, P. (2011).
\newblock Adaptive {Markov} chain {Monte Carlo}: theory and methods.
\newblock In {\em {Bayesian} Time Series Models}, volume~1. Cambridge
  University Press, Cambridge, UK.

\bibitem[Attia and Koren, 2023]{attia2023sgd}
Attia, A. and Koren, T. (2023).
\newblock {SGD} with adagrad stepsizes: Full adaptivity with high probability
  to unknown parameters, unbounded gradients and affine variance.
\newblock In {\em Proceedings of the 40th International Conference on Machine
  Learning (ICML 2023)}, Honolulu, HI.

\bibitem[Ay et~al., 2015]{nihat2015information}
Ay, N., Jost, J., L{\^e}, H.~V., and Schwachh{\"o}fer, L. (2015).
\newblock Information geometry and sufficient statistics.
\newblock {\em Probability Theory and Related Fields}, 162(1):327--364.

\bibitem[Baker et~al., 2019]{baker2019control}
Baker, J., Fearnhead, P., Fox, E.~B., and Nemeth, C. (2019).
\newblock Control variates for stochastic gradient {MCMC}.
\newblock {\em Statistics and Computing}, 29(3):599--615.

\bibitem[Bakry et~al., 2014]{bakry2014analysis}
Bakry, D., Gentil, I., and Ledoux, M. (2014).
\newblock {\em Analysis and geometry of {Markov} diffusion operators}.
\newblock Springer Cham.

\bibitem[Balasubramanian et~al., 2025]{balasubramanian2024improved}
Balasubramanian, K., Banerjee, S., and Ghosal, P. (2025).
\newblock Improved finite-particle convergence rates for {Stein} variational
  gradient descent.
\newblock In {\em Proceedings of the 13th International Conference on Learning
  Representations (ICLR 2025)}, Singapore.

\bibitem[Bardenet et~al., 2017]{bardenet2017markov}
Bardenet, R., Doucet, A., and Holmes, C. (2017).
\newblock On {Markov} chain {Monte Carlo} methods for tall data.
\newblock {\em Journal of Machine Learning Research}, 18(47):1--43.

\bibitem[Bauschke and Combettes, 2017]{bauschke2017convex}
Bauschke, H.~H. and Combettes, P.~L. (2017).
\newblock {\em Convex Analysis and Monotone Operator Theory in Hilbert Spaces}.
\newblock CMS Books in Mathematics. Springer, New York, NY, 2nd edition.

\bibitem[Beck and Teboulle, 2009]{beck2009fast}
Beck, A. and Teboulle, M. (2009).
\newblock A fast iterative shrinkage-thresholding algorithm for linear inverse
  problems.
\newblock {\em SIAM Journal on Imaging Sciences}, 2(1):183--202.

\bibitem[B{\'e}dard, 2008a]{bedard2008efficient}
B{\'e}dard, M. (2008a).
\newblock Efficient sampling using {Metropolis} algorithms: Applications of
  optimal scaling results.
\newblock {\em Journal of Computational and Graphical Statistics},
  17(2):312--332.

\bibitem[B{\'e}dard, 2008b]{bedard2008optimal}
B{\'e}dard, M. (2008b).
\newblock Optimal acceptance rates for {Metropolis} algorithms: moving beyond
  0.234.
\newblock {\em Stochastic Processes and their Applications},
  118(12):2198--2222.

\bibitem[Benamou and Brenier, 2000]{benamou2000computational}
Benamou, J.-D. and Brenier, Y. (2000).
\newblock A computational fluid mechanics solution to the {Monge-Kantorovich}
  mass transfer problem.
\newblock {\em Numerische Mathematik}, 84(3):375--393.

\bibitem[Benamou et~al., 2016]{benamou2016discretization}
Benamou, J.-D., Carlier, G., M{\'e}rigot, Q., and Oudet, E. (2016).
\newblock Discretization of functionals involving the {Monge--Amp{\`e}re}
  operator.
\newblock {\em Numerische mathematik}, 134(3):611--636.

\bibitem[Beskos et~al., 2013]{beskos2013optimal}
Beskos, A., Pillai, N., Roberts, G., Sanz-Serna, J.-M., and Stuart, A. (2013).
\newblock Optimal tuning of the hybrid {Monte Carlo} algorithm.
\newblock {\em Bernoulli}, 19(5A):1501--1534.

\bibitem[Blackard, 1998]{covertype1998}
Blackard, J. (1998).
\newblock Covertype.
\newblock UCI Machine Learning Repository.
\newblock {DOI}: 10.24432/C50K5N.

\bibitem[Boumal, 2023]{boumal2023introduction}
Boumal, N. (2023).
\newblock {\em An Introduction to Optimization on Smooth Manifolds}.
\newblock Cambridge University Press.

\bibitem[Boyd and Vandenberghe, 2004]{boyd2004convex}
Boyd, S. and Vandenberghe, L. (2004).
\newblock {\em Convex Optimization}.
\newblock Cambridge University Press, Cambridge, UK.

\bibitem[Brenier, 1991]{brenier1991polar}
Brenier, Y. (1991).
\newblock Polar factorization and monotone rearrangement of vector-valued
  functions.
\newblock {\em Communications on Pure and Applied Mathematics}, 44(4):375--417.

\bibitem[Bubeck et~al., 2015]{bubeck2015convex}
Bubeck, S. et~al. (2015).
\newblock Convex optimization: algorithms and complexity.
\newblock {\em Foundations and Trends{\textregistered} in Machine Learning},
  8(3-4):231--357.

\bibitem[Burger et~al., 2024]{burger2024covariance}
Burger, M., Erbar, M., Hoffmann, F., Matthes, D., and Schlichting, A. (2024).
\newblock Covariance-modulated optimal transport and gradient flows.
\newblock {\em Archive for Rational Mechanics and Analysis}, 249(1):7.

\bibitem[Cabezas et~al., 2024]{cabezas2024markovian}
Cabezas, A., Sharrock, L., and Nemeth, C. (2024).
\newblock Markovian flow matching: accelerating {MCMC} with continuous
  normalizing flows.
\newblock In {\em Proceedings of the 38th Annual Conference on Neural
  Information Processing Systems (NeurIPS 2024)}, Vancouver, Canada.

\bibitem[Campbell et~al., 2021]{campbell2021gradient}
Campbell, A., Chen, W., Stimper, V., Hernandez-Lobato, J.~M., and Zhang, Y.
  (2021).
\newblock A gradient based strategy for {Hamiltonian} {Monte Carlo}
  hyperparameter optimization.
\newblock In {\em Proceedings of the 38th International Conference on Machine
  Learning (ICML 2021)}, Virtual.

\bibitem[Cances et~al., 2020]{cances2020variational}
Cances, C., Gallou{\"e}t, T.~O., and Todeschi, G. (2020).
\newblock A variational finite volume scheme for {Wasserstein} gradient flows.
\newblock {\em Numerische Mathematik}, 146:437--480.

\bibitem[Cannella and Tarokh, 2022]{cannella2022semi}
Cannella, C. and Tarokh, V. (2022).
\newblock Semi-empirical objective functions for {MCMC} proposal optimization.
\newblock In {\em Proceedings of the 26th International Conference on Pattern
  Recognition (ICPR 2022)}, pages 4758--4764. IEEE.

\bibitem[Carmon and Hinder, 2022]{carmon2022making}
Carmon, Y. and Hinder, O. (2022).
\newblock Making {SGD} parameter-free.
\newblock In {\em Proceedings of the 35th Annual Conference on Learning Theory
  (COLT 2022)}, pages 2360--2389, London, UK.

\bibitem[Carrillo et~al., 2024a]{carrillo2024fisher}
Carrillo, J.~A., Chen, Y., Huang, D.~Z., Huang, J., and Wei, D. (2024a).
\newblock {Fisher-Rao} gradient flow: geodesic convexity and functional
  inequalities.
\newblock {\em arXiv preprint arXiv:2407.15693}.

\bibitem[Carrillo et~al., 2015]{carrillo2015finite}
Carrillo, J.~A., Chertock, A., and Huang, Y. (2015).
\newblock A finite-volume method for nonlinear nonlocal equations with a
  gradient flow structure.
\newblock {\em Communications in Computational Physics}, 17(1):233--258.

\bibitem[Carrillo et~al., 2019]{carrillo2019blob}
Carrillo, J.~A., Craig, K., and Patacchini, F.~S. (2019).
\newblock A blob method for diffusion.
\newblock {\em Calculus of Variations and Partial Differential Equations},
  58(2):53.

\bibitem[Carrillo and Skrzeczkowski, 2025]{carrillo2023convergence}
Carrillo, J.~A. and Skrzeczkowski, J. (2025).
\newblock Convergence and stability results for the particle system in the
  {Stein} gradient descent method.
\newblock {\em Mathematics of Computation}, 94:1793--1814.

\bibitem[Carrillo et~al., 2024b]{carrillo2024stein}
Carrillo, J.~A., Skrzeczkowski, J., and Warnett, J. (2024b).
\newblock The {Stein}--log--{Sobolev} inequality and the exponential rate of
  convergence for the continuous {Stein} variational gradient descent method.
\newblock {\em arXiv preprint arXiv:2412.10295}.

\bibitem[Carrillo and Vaes, 2021]{carrillo2021wasserstein}
Carrillo, J.~A. and Vaes, U. (2021).
\newblock Wasserstein stability estimates for covariance-preconditioned
  {Fokker--Planck} equations.
\newblock {\em Nonlinearity}, 34(4):2275.

\bibitem[Cesa-Bianchi et~al., 1997]{cesa1997use}
Cesa-Bianchi, N., Freund, Y., Haussler, D., Helmbold, D.~P., Schapire, R.~E.,
  and Warmuth, M.~K. (1997).
\newblock How to use expert advice.
\newblock {\em Journal of the ACM (JACM)}, 44(3):427--485.

\bibitem[Cesa-Bianchi and Lugosi, 2006]{cesa2006prediction}
Cesa-Bianchi, N. and Lugosi, G. (2006).
\newblock {\em Prediction, learning, and games}.
\newblock Cambridge University Press.

\bibitem[Chaudhuri et~al., 2009]{chaudhuri2009parameter}
Chaudhuri, K., Freund, Y., and Hsu, D.~J. (2009).
\newblock A parameter-free hedging algorithm.
\newblock In {\em Proceedings of the 23rd Annual Conference on Neural
  Information Processing Systems (NIPS 2009)}, Vancouver, Canada.

\bibitem[Chazal et~al., 2025]{chazal2025computable}
Chazal, C., Kanagawa, H., Shen, Z., Korba, A., Oates, C., et~al. (2025).
\newblock A computable measure of suboptimality for entropy-regularised
  variational objectives.
\newblock {\em arXiv preprint arXiv:2509.10393}.

\bibitem[Chen et~al., 2016]{chen2016bridging}
Chen, C., Carlson, D., Gan, Z., Li, C., and Carin, L. (2016).
\newblock Bridging the gap between stochastic gradient {MCMC} and stochastic
  optimization.
\newblock In {\em Proceedings of the 19th International Conference on
  Artificial Intelligence and Statistics (AISTATS 2016)}, Cadiz, Spain.

\bibitem[Chen et~al., 2018]{chen2018unified}
Chen, C., Zhang, R., Wang, W., Li, B., and Chen, L. (2018).
\newblock A unified particle-optimization framework for scalable {Bayesian}
  sampling.
\newblock In {\em Proceedings of the 34th Conference on Uncertainty in
  Artificial Intelligence (UAI 2018)}, Monterey, CA.

\bibitem[Chen et~al., 2024a]{chen2024uniform}
Chen, F., Lin, Y., Ren, Z., and Wang, S. (2024a).
\newblock Uniform-in-time propagation of chaos for kinetic mean field
  {Langevin} dynamics.
\newblock {\em Electronic Journal of Probability}, 29:1--43.

\bibitem[Chen et~al., 2022a]{chen2022uniform}
Chen, F., Ren, Z., and Wang, S. (2022a).
\newblock Uniform-in-time propagation of chaos for mean field {Langevin}
  dynamics.
\newblock {\em arXiv preprint arXiv:2212.03050}.

\bibitem[Chen et~al., 2022b]{chen2022better}
Chen, K., Langford, J., and Orabona, F. (2022b).
\newblock Better parameter-free stochastic optimization with {ODE} updates for
  coin-betting.
\newblock In {\em Proceedings of the 36th AAAI Conference on Artificial
  Intelligence (AAAI-22)}, pages 6239--6247, Virtual.

\bibitem[Chen et~al., 2025]{chen2025accelerating}
Chen, S., Li, Q., Tse, O., and Wright, S.~J. (2025).
\newblock Accelerating optimization over the space of probability measures.
\newblock {\em Journal of Machine Learning Research}, 26(31):1--40.

\bibitem[Chen et~al., 2014]{chen2014stochastic}
Chen, T., Fox, E., and Guestrin, C. (2014).
\newblock Stochastic gradient {Hamiltonian} {Monte Carlo}.
\newblock In {\em Proceedings of the 31st International Conference on Machine
  Learning (ICML 2014)}, pages 1683--1691, Beijing, China.

\bibitem[Chen et~al., 2019]{chen2018convergence}
Chen, X., Liu, S., Sun, R., and Hong, M. (2019).
\newblock On the convergence of a class of {Adam}-type algorithms for
  non-convex optimization.
\newblock In {\em Proceedings of the 7th International Conference on Learning
  Representations (ICLR 2019)}, New Orleans, LA.

\bibitem[Chen et~al., 2023a]{chen2023gradient}
Chen, Y., Huang, D.~Z., Huang, J., Reich, S., and Stuart, A.~M. (2023a).
\newblock Gradient flows for sampling: mean-field models, {Gaussian}
  approximations and affine invariance.
\newblock {\em arXiv preprint arXiv:2302.11024}.

\bibitem[Chen et~al., 2023b]{chen2023sampling}
Chen, Y., Huang, D.~Z., Huang, J., Reich, S., and Stuart, A.~M. (2023b).
\newblock Sampling via gradient flows in the space of probability measures.
\newblock {\em arXiv preprint arXiv:2310.03597}.

\bibitem[Chen et~al., 2024b]{chen2024efficient}
Chen, Y., Huang, D.~Z., Huang, J., Reich, S., and Stuart, A.~M. (2024b).
\newblock Efficient, multimodal, and derivative-free {Bayesian} inference with
  {Fisher-Rao} gradient flows.
\newblock {\em Inverse Problems}, 40(12):125001.

\bibitem[Cheng et~al., 2024]{cheng2024convergence}
Cheng, X., Lu, J., Tan, Y., and Xie, Y. (2024).
\newblock Convergence of flow-based generative models via proximal gradient
  descent in {Wasserstein} space.
\newblock {\em IEEE Transactions on Information Theory}.

\bibitem[Chewi, 2025]{chewi2024log}
Chewi, S. (2025).
\newblock {\em Log-Concave Sampling}.
\newblock Draft.
\newblock Accessed: October 2025.

\bibitem[Chewi et~al., 2024a]{chewi2024analysis}
Chewi, S., Erdogdu, M.~A., Li, M., Shen, R., and Zhang, M.~S. (2024a).
\newblock Analysis of {Langevin} {Monte Carlo} from poincare to log-{Sobolev}.
\newblock {\em Foundations of Computational Mathematics}, pages 1--51.

\bibitem[Chewi et~al., 2020a]{chewi2020exponential}
Chewi, S., Gouic, T.~L., Lu, C., Maunu, T., Rigollet, P., and Stromme, A.
  (2020a).
\newblock {Exponential ergodicity of mirror-{Langevin} diffusions}.
\newblock In {\em Proceedings of the 34th Annual Conference on Neural
  Information Processing Systems (NeurIPS 2020)}, Vancouver, Canada.

\bibitem[Chewi et~al., 2020b]{chewi2020svgd}
Chewi, S., {Le Gouic}, T., Lu, C., Maunu, T., and Rigollet, P. (2020b).
\newblock {SVGD} as a kernelized {Wasserstein} gradient flow of the chi-squared
  divergence.
\newblock In {\em Proceedings of the 34th Annual Conference on Neural
  Information Processing Systems (NeurIPS 2020)}, Vancouver, Canada.

\bibitem[Chewi et~al., 2020c]{chewi2020gradient}
Chewi, S., Maunu, T., Rigollet, P., and Stromme, A.~J. (2020c).
\newblock Gradient descent algorithms for {Bures--Wasserstein} barycenters.
\newblock In {\em Proceedings of the 33rd Annual Conference on Learning Theory
  (COLT 2020)}, Graz, Austria.

\bibitem[Chewi et~al., 2024b]{chewi2024statistical}
Chewi, S., Niles-Weed, J., and Rigollet, P. (2024b).
\newblock Statistical optimal transport.
\newblock {\em {Ecole d’Et\'{e} de Probabilit\'{e}s de Saint-Flour XLIX}}.

\bibitem[Chizat, 2022a]{chizat2022meanfield}
Chizat, L. (2022a).
\newblock Mean-field {Langevin} dynamics: Exponential convergence and
  annealing.
\newblock {\em Transactions on Machine Learning Research}.

\bibitem[Chizat, 2022b]{chizat2022sparse}
Chizat, L. (2022b).
\newblock Sparse optimization on measures with over-parameterized gradient
  descent.
\newblock {\em Mathematical Programming}, 194(1):487--532.

\bibitem[Chizat and Bach, 2018]{chizat2018global}
Chizat, L. and Bach, F. (2018).
\newblock On the global convergence of gradient descent for over-parameterized
  models using optimal transport.
\newblock In {\em Proceedings of the 32nd Annual Conference on Neural
  Information Processing Systems (NIPS 2018)}, Montreal, Canada.

\bibitem[Chizat et~al., 2018]{chizat2018interpolating}
Chizat, L., Peyr{\'e}, G., Schmitzer, B., and Vialard, F.-X. (2018).
\newblock An interpolating distance between optimal transport and {Fisher--Rao}
  metrics.
\newblock {\em Foundations of Computational Mathematics}, 18(1):1--44.

\bibitem[Chizat et~al., 2022]{chizat2022trajectory}
Chizat, L., Zhang, S., Heitz, M., and Schiebinger, G. (2022).
\newblock Trajectory inference via mean-field {Langevin} in path space.
\newblock In {\em Proceedings of the 36th Annual Conference on Neural
  Information Processing Systems (NeurIPS 2022)}, New Orleans, LA.

\bibitem[Chu et~al., 2020]{chu2020equivalence}
Chu, C., Minami, K., and Fukumizu, K. (2020).
\newblock The equivalence between {Stein} variational gradient descent and
  black-box variational inference.
\newblock In {\em Proceedings of the 8th International Conference on Learning
  Representations (ICLR 2020): Workshop on Integration of Deep Neural Models
  and Differential Equations}, Addis Ababa, Ethiopia.

\bibitem[Chwialkowski et~al., 2016]{chwialkowski2016kernel}
Chwialkowski, K., Strathmann, H., and Gretton, A. (2016).
\newblock A kernel test of goodness of fit.
\newblock In {\em Proceedings of the 33rd International Conference on Machine
  Learning (ICML 2016)}, New York City, NY.

\bibitem[Clement and Desch, 2008]{clement2008elementary}
Clement, P. and Desch, W. (2008).
\newblock An elementary proof of the triangle inequality for the {Wasserstein}
  metric.
\newblock {\em Proceedings of the American Mathematical Society},
  136(1):333--339.

\bibitem[Coullon et~al., 2023]{coullon2021efficient}
Coullon, J., South, L., and Nemeth, C. (2023).
\newblock Efficient and generalizable tuning strategies for stochastic gradient
  {MCMC}.
\newblock {\em Statistics and Computing}, 33(3):66.

\bibitem[Cover and Thomas, 2006]{cover2006elements}
Cover, T.~M. and Thomas, J.~A. (2006).
\newblock {\em Elements of Information Theory}.
\newblock John Wiley {\&} Sons, Inc, Hoboken, NJ.

\bibitem[Cutkosky and Orabona, 2018]{cutkosky2018black}
Cutkosky, A. and Orabona, F. (2018).
\newblock Black-box reductions for parameter-free online learning in banach
  spaces.
\newblock In {\em Proceedings of the 31st Annual Conference On Learning Theory
  (COLT 2018)}, pages 1493--1529, Stockholm, Sweden.

\bibitem[Dalalyan, 2017a]{dalalyan2017further}
Dalalyan, A. (2017a).
\newblock Further and stronger analogy between sampling and optimization:
  {Langevin} {Monte Carlo} and gradient descent.
\newblock In {\em Proceedings of the 30th Annual Conference on Learning Theory
  (COLT 2017)}, pages 678--689, Amsterdam, Netherlands.

\bibitem[Dalalyan, 2017b]{dalalyan2017theoretical}
Dalalyan, A.~S. (2017b).
\newblock Theoretical guarantees for approximate sampling from smooth and
  log-concave densities.
\newblock {\em Journal of the Royal Statistical Society Series B: Statistical
  Methodology}, 79(3):651--676.

\bibitem[Dalalyan and Karagulyan, 2019]{dalalyan2019user}
Dalalyan, A.~S. and Karagulyan, A. (2019).
\newblock User-friendly guarantees for the {Langevin} {Monte Carlo} with
  inaccurate gradient.
\newblock {\em Stochastic Processes and their Applications},
  129(12):5278--5311.

\bibitem[Danilova et~al., 2022]{danilova2022recent}
Danilova, M., Dvurechensky, P., Gasnikov, A., Gorbunov, E., Guminov, S.,
  Kamzolov, D., and Shibaev, I. (2022).
\newblock Recent theoretical advances in non-convex optimization.
\newblock In {\em High-Dimensional Optimization and Probability: With a View
  Towards Data Science}, pages 79--163. Springer.

\bibitem[Das and Nagaraj, 2023]{das2023provably}
Das, A. and Nagaraj, D. (2023).
\newblock Provably fast finite particle variants of {SVGD} via virtual particle
  stochastic approximation.
\newblock In {\em Proceedings of the 37th Annual Conference on Neural
  Information Processing Systems (NeurIPS 2023)}, New Orleans, LA.

\bibitem[De~Bortoli et~al., 2020]{de2020quantitative}
De~Bortoli, V., Durmus, A., Fontaine, X., and Simsekli, U. (2020).
\newblock Quantitative propagation of chaos for {SGD} in wide neural networks.
\newblock In {\em Proceedings of the 34th Annual Conference on Neural
  Information Processing Systems (NeurIPS 2020)}, Vancouver, Canada.

\bibitem[Defazio and Mishchenko, 2023]{defazio2023learning}
Defazio, A. and Mishchenko, K. (2023).
\newblock Learning-rate-free learning by d-adaptation.
\newblock In {\em Proceedings of the 40th International Conference on Machine
  Learning (ICML 2023)}, Honolulu, HI.

\bibitem[Dharamshi et~al., 2024]{dharamshi2023sampling}
Dharamshi, A., Ngo, V., and Rosenthal, J.~S. (2024).
\newblock Sampling by divergence minimization.
\newblock {\em Communications in Statistics - Simulation and Computation},
  53(12):6071--6095.

\bibitem[Diao et~al., 2023]{diao2023forward}
Diao, M.~Z., Balasubramanian, K., Chewi, S., and Salim, A. (2023).
\newblock Forward-backward {Gaussian} variational inference via {JKO} in the
  {Bures--Wasserstein} space.
\newblock In {\em Proceedings of the 40th International Conference on Machine
  Learning (ICML 2023)}, Honolulu, HI.

\bibitem[Dodd et~al., 2024]{dodd2024learning}
Dodd, D., Sharrock, L., and Nemeth, C. (2024).
\newblock Learning-rate-free stochastic optimization over {Riemannian}
  manifolds.
\newblock In {\em Proceedings of the 41st International Conference on Machine
  Learning (ICML 2024)}, Vienna, Austria.

\bibitem[Domingo-Enrich and Pooladian, 2023]{domingo2023explicit}
Domingo-Enrich, C. and Pooladian, A.-A. (2023).
\newblock An explicit expansion of the {Kullback-Leibler} divergence along its
  {Fisher-Rao} gradient flow.
\newblock {\em Transactions on Machine Learning Research}.

\bibitem[Dubey et~al., 2016]{dubey2016variance}
Dubey, K.~A., J~Reddi, S., Williamson, S.~A., Poczos, B., Smola, A.~J., and
  Xing, E.~P. (2016).
\newblock Variance reduction in stochastic gradient {Langevin} dynamics.
\newblock In {\em Proceedings of the 30th Annual Conference on Neural
  Information Processing Systems (NIPS 2016)}, Barcelona, Spain.

\bibitem[Duchi et~al., 2011]{duchi2011adaptive}
Duchi, J., Hazan, E., and Singer, Y. (2011).
\newblock Adaptive subgradient methods for online learning and stochastic
  optimization.
\newblock {\em Journal of Machine Learning Research}, 12(7).

\bibitem[Duncan et~al., 2023]{duncan2019geometry}
Duncan, A., N{\"{u}}sken, N., and Szpruch, L. (2023).
\newblock On the geometry of {Stein} variational gradient descent.
\newblock {\em Journal of Machine Learning Research}, 24(56):1--39.

\bibitem[Durmus et~al., 2019]{durmus2019analysis}
Durmus, A., Majewski, S., and Miasojedow, B. (2019).
\newblock Analysis of {Langevin} {Monte Carlo} via convex optimization.
\newblock {\em Journal of Machine Learning Research}, 20(73):1--46.

\bibitem[Durmus and Moulines, 2017]{durmus2017nonasymptotic}
Durmus, A. and Moulines, {\'E}. (2017).
\newblock Nonasymptotic convergence analysis for the unadjusted {Langevin}
  algorithm.
\newblock {\em The Annals of Applied Probability}, 27(3):1551 -- 1587.

\bibitem[Durmus and Moulines, 2019]{durmus2019high}
Durmus, A. and Moulines, {\'E}. (2019).
\newblock High-dimensional {Bayesian} inference via the unadjusted {Langevin}
  algorithm.
\newblock {\em Bernoulli}, 25(4A):2854 -- 2882.

\bibitem[Ene et~al., 2021]{ene2021adaptive}
Ene, A., Nguyen, H., and Vladu, A. (2021).
\newblock Adaptive gradient methods for constrained convex optimization and
  variational inequalities.
\newblock In {\em Proceedings of the 35th AAAI Conference on Artificial
  Intelligence (AAAI-21)}, Virtual.

\bibitem[Erbar, 2010]{erbar2010heat}
Erbar, M. (2010).
\newblock {The heat equation on manifolds as a gradient flow in the
  {Wasserstein} space}.
\newblock {\em Annales de l'Institut Henri Poincaré, Probabilités et
  Statistiques}, 46(1):1 -- 23.

\bibitem[Faw et~al., 2022]{faw2022power}
Faw, M., Tziotis, I., Caramanis, C., Mokhtari, A., Shakkottai, S., and Ward, R.
  (2022).
\newblock The power of adaptivity in {SGD}: Self-tuning step sizes with
  unbounded gradients and affine variance.
\newblock In {\em Proceedings of the 35th Annual Conference on Learning Theory
  (COLT 2022)}, London, UK.

\bibitem[Fearnhead et~al., 2025]{fearnhead2025scalable}
Fearnhead, P., Nemeth, C., Oates, C.~J., and Sherlock, C. (2025).
\newblock {\em Scalable {Monte Carlo} for {Bayesian} Learning}.
\newblock Institute of Mathematical Statistics Monographs. Cambridge University
  Press.

\bibitem[Fournier and Guillin, 2015]{fournier2015rate}
Fournier, N. and Guillin, A. (2015).
\newblock On the rate of convergence in {Wasserstein} distance of the empirical
  measure.
\newblock {\em Probability Theory and Related Fields}, 162(3):707--738.

\bibitem[Fruehwirth and Habring, 2024]{fruehwirth2024ergodicity}
Fruehwirth, L. and Habring, A. (2024).
\newblock Ergodicity of {Langevin} dynamics and its discretizations for
  non-smooth potentials.
\newblock {\em arXiv preprint arXiv:2411.12051}.

\bibitem[Fu et~al., 2023]{fu2023high}
Fu, G., Osher, S., and Li, W. (2023).
\newblock High order spatial discretization for variational time implicit
  schemes: {Wasserstein} gradient flows and reaction-diffusion systems.
\newblock {\em Journal of Computational Physics}, 491:112375.

\bibitem[Fu and Wilson, 2024]{fu2023mean}
Fu, Q. and Wilson, A. (2024).
\newblock Mean-field underdamped {Langevin} dynamics and its spacetime
  discretization.
\newblock In {\em Proceedings of the 41st International Conference on Machine
  Learning (ICML 2024)}, Vienna, Austria.

\bibitem[Fujisawa and Futami, 2024]{fujisawa2024convergence}
Fujisawa, M. and Futami, F. (2024).
\newblock Convergence of {SVGD} in {KL} divergence via approximate gradient
  flow.
\newblock {\em OpenReview ID:Va2IQ471GR}.

\bibitem[Futami et~al., 2021]{futami2021accelerated}
Futami, F., Iwata, T., Ueda, N., and Sato, I. (2021).
\newblock Accelerated diffusion-based sampling by the non-reversible dynamics
  with skew-symmetric matrices.
\newblock {\em Entropy}, 23(8):993.

\bibitem[Futami et~al., 2020]{futami2020accelerating}
Futami, F., Sato, I., and Sugiyama, M. (2020).
\newblock Accelerating the diffusion-based ensemble sampling by non-reversible
  dynamics.
\newblock In {\em Proceedings of the 37th International Conference on Machine
  Learning (ICML 2020)}, Virtual.

\bibitem[Gabri{\'e} et~al., 2022]{gabrie2022adaptive}
Gabri{\'e}, M., Rotskoff, G.~M., and Vanden-Eijnden, E. (2022).
\newblock Adaptive {Monte Carlo} augmented with normalizing flows.
\newblock {\em Proceedings of the National Academy of Sciences},
  119(10):e2109420119.

\bibitem[Galashov et~al., 2025]{galashov2024deep}
Galashov, A., de~Bortoli, V., and Gretton, A. (2025).
\newblock Deep {MMD} gradient flow without adversarial training.
\newblock In {\em Proceedings of the 13th International Conference on Learning
  Representations (ICLR 2025)}, Singapore.

\bibitem[Gallou{\"e}t and Monsaingeon, 2017]{gallouet2017jko}
Gallou{\"e}t, T.~O. and Monsaingeon, L. (2017).
\newblock A {JKO} splitting scheme for {Kantorovich--Fisher--Rao} gradient
  flows.
\newblock {\em SIAM Journal on Mathematical Analysis}, 49(2):1100--1130.

\bibitem[Gao and Kleywegt, 2023]{gao2023distributionally}
Gao, R. and Kleywegt, A. (2023).
\newblock Distributionally robust stochastic optimization with {Wasserstein}
  distance.
\newblock {\em Mathematics of Operations Research}, 48(2):603--655.

\bibitem[Garbuno-Inigo et~al., 2020a]{garbuno2020interacting}
Garbuno-Inigo, A., Hoffmann, F., Li, W., and Stuart, A.~M. (2020a).
\newblock Interacting {Langevin} diffusions: gradient structure and ensemble
  {Kalman} sampler.
\newblock {\em SIAM Journal on Applied Dynamical Systems}, 19(1):412--441.

\bibitem[Garbuno-Inigo et~al., 2020b]{garbuno2020affine}
Garbuno-Inigo, A., N{\"{u}}sken, N., and Reich, S. (2020b).
\newblock Affine invariant interacting {Langevin} dynamics for {Bayesian}
  inference.
\newblock {\em SIAM Journal on Applied Dynamical Systems}, 19(3):1633--1658.

\bibitem[Garrigos and Gower, 2023]{garrigos2023handbook}
Garrigos, G. and Gower, R.~M. (2023).
\newblock Handbook of convergence theorems for (stochastic) gradient methods.
\newblock {\em arXiv preprint arXiv:2301.11235}.

\bibitem[Gelman et~al., 1997]{gelman1997weak}
Gelman, A., Gilks, W.~R., and Roberts, G.~O. (1997).
\newblock Weak convergence and optimal scaling of random walk {Metropolis}
  algorithms.
\newblock {\em The Annals of Applied Probability}, 7(1):110--120.

\bibitem[Gershman et~al., 2012]{gershman2012nonparametric}
Gershman, S., Hoffman, M., and Blei, D. (2012).
\newblock Nonparametric variational inference.
\newblock In {\em Proceedings of the 29th International Conference on Machine
  Learning (ICML 2012)}, Edinburgh, UK.

\bibitem[Geshkovski et~al., 2025]{geshkovski2025mathematical}
Geshkovski, B., Letrouit, C., Polyanskiy, Y., and Rigollet, P. (2025).
\newblock A mathematical perspective on transformers.
\newblock {\em Bulletin of the American Mathematical Society}, 62(3):427--479.

\bibitem[Gigli, 2011]{gigli2011on}
Gigli, N. (2011).
\newblock On the inverse implication of {Brenier-McCann} theorems and the
  structure of {$(\mathcal{P}_2(M), W_2)$}.
\newblock {\em Methods and Applications of Analysis}, 18(2).

\bibitem[Gorham et~al., 2019]{gorham2019measuring}
Gorham, J., Duncan, A.~B., Vollmer, S.~J., and Mackey, L. (2019).
\newblock Measuring sample quality with diffusions.
\newblock {\em The Annals of Applied Probability}, 29(5):2884--2928.

\bibitem[Gorham and Mackey, 2017]{gorham2017measuring}
Gorham, J. and Mackey, L. (2017).
\newblock Measuring sample quality with kernels.
\newblock In {\em Proceedings of the 34th International Conference on Machine
  Learning (ICML 2017)}, Sydney, Australia.

\bibitem[Gorham et~al., 2020]{gorham2020stochastic}
Gorham, J., Raj, A., and Mackey, L. (2020).
\newblock Stochastic {Stein} discrepancies.
\newblock In {\em Proceedings of the 34th Annual Conference on Neural
  Information Processing Systems (NeurIPS 2020)}, Vancouver, Canada.

\bibitem[Grimmer, 2024]{grimmer2024optimal}
Grimmer, B. (2024).
\newblock On optimal universal first-order methods for minimizing heterogeneous
  sums.
\newblock {\em Optimization Letters}, 18(2):427--445.

\bibitem[Grumitt et~al., 2022]{grumitt2022deterministic}
Grumitt, R., Dai, B., and Seljak, U. (2022).
\newblock Deterministic {Langevin} {Monte Carlo} with normalizing flows for
  {Bayesian} inference.
\newblock In {\em Proceedings of the 36th Annual Conference on Neural
  Information Processing Systems (NeurIPS 2022)}, New Orleans, LA.

\bibitem[Guo et~al., 2022]{guo2022online}
Guo, W., Hur, Y., Liang, T., and Ryan, C. (2022).
\newblock Online learning to transport via the minimal selection principle.
\newblock In {\em Proceedings of the 35th Annual Conference on Learning Theory
  (COLT 2022)}, pages 4085--4109, London, UK.

\bibitem[Gupta et~al., 2018]{gupta2018preconditioned}
Gupta, V., Koren, T., and Singer, Y. (2018).
\newblock Shampoo: preconditioned stochastic tensor optimization.
\newblock In {\em Proceedings of the 35th International Conference on Machine
  Learning (ICML 2018)}, Stockholm, Sweden.

\bibitem[Haario et~al., 1999]{haario1999}
Haario, H., Saksman, E., and Tamminen, J. (1999).
\newblock Adaptive proposal distribution for random walk {Metropolis}
  algorithm.
\newblock {\em Computational Statistics}, 14(3):375--395.

\bibitem[Haario et~al., 2005]{haario2005componentwise}
Haario, H., Saksman, E., and Tamminen, J. (2005).
\newblock Componentwise adaptation for high dimensional {MCMC}.
\newblock {\em Computational Statistics}, 20(2):265--273.

\bibitem[Habring et~al., 2024]{habring2024subgradient}
Habring, A., Holler, M., and Pock, T. (2024).
\newblock Subgradient {Langevin} methods for sampling from nonsmooth
  potentials.
\newblock {\em SIAM Journal on Mathematics of Data Science}, 6(4):897--925.

\bibitem[Haviv et~al., 2024]{haviv2024wasserstein}
Haviv, D., Pooladian, A.-A., Pe'er, D., and Amos, B. (2024).
\newblock Wasserstein flow matching: generative modeling over families of
  distributions.
\newblock In {\em Proceedings of the 42nd International Conference on Machine
  Learning (ICML 2024)}, Vancouver, Canada.

\bibitem[Hazan et~al., 2016]{hazan2016introduction}
Hazan, E. et~al. (2016).
\newblock Introduction to online convex optimization.
\newblock {\em Foundations and Trends{\textregistered} in Optimization},
  2(3-4):157--325.

\bibitem[Hazan and Kakade, 2019]{hazan2019revisiting}
Hazan, E. and Kakade, S. (2019).
\newblock Revisiting the {Polyak} step size.
\newblock {\em arXiv preprint arXiv:1905.00313}.

\bibitem[Hazan and Kale, 2014]{hazan2014beyond}
Hazan, E. and Kale, S. (2014).
\newblock Beyond the regret minimization barrier: optimal algorithms for
  stochastic strongly-convex optimization.
\newblock {\em The Journal of Machine Learning Research}, 15(1):2489--2512.

\bibitem[Hazan and Megiddo, 2007]{hazan2007online}
Hazan, E. and Megiddo, N. (2007).
\newblock Online learning with prior knowledge.
\newblock In {\em Proceedings of the 20th Annual Conference on Learning Theory
  (COLT 2007)}, pages 499--513, San Diego, CA.

\bibitem[He et~al., 2024]{he2024regularized}
He, Y., Balasubramanian, K., Sriperumbudur, B.~K., and Lu, J. (2024).
\newblock Regularized {Stein} variational gradient flow.
\newblock {\em Foundations of Computational Mathematics}, 25:1199--1257.

\bibitem[Held and Holmes, 2006]{held2006bayesian}
Held, L. and Holmes, C.~C. (2006).
\newblock {Bayesian} auxiliary variable models for binary and multinomial
  regression.
\newblock {\em {Bayesian} Analysis}, 1(1):145--168.

\bibitem[Hellinger, 1909]{hellinger1909neue}
Hellinger, E. (1909).
\newblock Neue begründung der theorie quadratischer formen von unendlichvielen
  veränderlichen.
\newblock {\em Journal für die reine und angewandte Mathematik}, 136:210--271.

\bibitem[Hinton et~al., 2012]{hinton2012neural}
Hinton, G., Srivastava, N., and Swersky, K. (2012).
\newblock Neural networks for machine learning. lecture 6a: Overview of
  mini-batch gradient descent.

\bibitem[Hoffman et~al., 2021]{hoffman2021adaptive}
Hoffman, M., Radul, A., and Sountsov, P. (2021).
\newblock An adaptive {MCMC} scheme for setting trajectory lengths in
  {Hamiltonian} {Monte Carlo}.
\newblock In {\em Proceedings of the 24th International Conference on
  Artificial Intelligence and Statistics (AISTATS 2021)}, pages 3907--3915,
  Virtual. PMLR.

\bibitem[Horn and Johnson, 2012]{horn2012matrix}
Horn, R.~A. and Johnson, C.~R. (2012).
\newblock {\em Matrix Analysis}.
\newblock Cambridge University Press.

\bibitem[Hsieh et~al., 2018]{hsieh2018mirrored}
Hsieh, Y.-P., Kavis, A., Rolland, P., and Cevher, V. (2018).
\newblock Mirrored {Langevin} dynamics.
\newblock In {\em Proceedings of the 32nd Annual Conference on Neural
  Information Processing Systems (NeurIPS 2018)}, Montr\'{e}al, Canada.

\bibitem[Hu et~al., 2021]{hu2021meanfield}
Hu, K., Ren, Z., Šiška, D., and Łukasz Szpruch (2021).
\newblock {Mean-field {Langevin} dynamics and energy landscape of neural
  networks}.
\newblock {\em Annales de l'Institut Henri Poincaré, Probabilités et
  Statistiques}, 57(4):2043 -- 2065.

\bibitem[Ivgi et~al., 2023]{ivgi2023dog}
Ivgi, M., Hinder, O., and Carmon, Y. (2023).
\newblock {DoG} is {SGD}’s best friend: a parameter-free dynamic step size
  schedule.
\newblock In {\em Proceedings of the 40th International Conference on Machine
  Learning (ICML 2023)}, Honolulu, HI.

\bibitem[Javanmard et~al., 2020]{javanmard2020analysis}
Javanmard, A., Mondelli, M., and Montanari, A. (2020).
\newblock Analysis of a two-layer neural network via displacement convexity.
\newblock {\em The Annals of Statistics}, 48(6):3619 -- 3642.

\bibitem[Jiang et~al., 2024]{jiang2023algorithms}
Jiang, Y., Chewi, S., and Pooladian, A.-A. (2024).
\newblock Algorithms for mean-field variational inference via polyhedral
  optimization in the {Wasserstein} space.
\newblock In {\em Proceedings of 37th Conference on Learning Theory (COLT
  2024)}, Edmonton, Canada.

\bibitem[Jordan et~al., 1998]{jordan1998variational}
Jordan, R., Kinderlehrer, D., and Otto, F. (1998).
\newblock The variational formulation of the {Fokker--Planck} equation.
\newblock {\em SIAM Journal on Mathematical Analysis}, 29(1):1--17.

\bibitem[Jun and Orabona, 2019]{jun2019parameter}
Jun, K.-S. and Orabona, F. (2019).
\newblock Parameter-free online convex optimization with sub-exponential noise.
\newblock In {\em Proceedings of the 32nd Annual Conference on Learning Theory
  (COLT 2019)}, pages 1802--1823, Phoenix, AZ.

\bibitem[Kakutani, 1948]{kakutani1948equivalence}
Kakutani, S. (1948).
\newblock On equivalence of infinite product measures.
\newblock {\em Annals of Mathematics}, 49(1):214--224.

\bibitem[Karimi~Jaghargh et~al., 2023]{karimi2023stochastic}
Karimi~Jaghargh, M.~R., Hsieh, Y.-P., and Krause, A. (2023).
\newblock Stochastic approximation algorithms for systems of interacting
  particles.
\newblock In {\em Proceedings of the 37th Annual Conference on Neural
  Information Processing Systems (NeurIPS 2023)}, New Orleans, LA.

\bibitem[Karppinen and Vihola, 2021]{karppinen2021conditional}
Karppinen, S. and Vihola, M. (2021).
\newblock Conditional particle filters with diffuse initial distributions.
\newblock {\em Statistics and Computing}, 31(3):24.

\bibitem[Kavis et~al., 2019]{kavis2019unixgrad}
Kavis, A., Levy, K.~Y., Bach, F., and Cevher, V. (2019).
\newblock Unixgrad: a universal, adaptive algorithm with optimal guarantees for
  constrained optimization.
\newblock In {\em Proceedings of the 33rd Annual Conference on Neural
  Information Processing Systems (NeurIPS 2019)}, Vancouver, Canada.

\bibitem[Kavis et~al., 2022]{kavis2022high}
Kavis, A., Levy, K.~Y., and Cevher, V. (2022).
\newblock High probability bounds for a class of nonconvex algorithms with
  adagrad stepsize.
\newblock In {\em Proceedings of the 10th International Conference on Learning
  Representations (ICLR 2022)}, Virtual.

\bibitem[Kazeykina et~al., 2024]{kazeykina2020ergodicity}
Kazeykina, A., Ren, Z., Tan, X., and Yang, J. (2024).
\newblock Ergodicity of the underdamped mean-field {Langevin} dynamics.
\newblock {\em Annals of Applied Probability}, 34(3):3181--3226.

\bibitem[Khaled et~al., 2023]{khaled2023dowg}
Khaled, A., Mishchenko, K., and Jin, C. (2023).
\newblock {DoWG} unleashed: an efficient universal parameter-free gradient
  descent method.
\newblock In {\em Proceedings of the 37th Annual Conference on Neural
  Information Processing Systems (NeurIPS 2023)}, New Orleans, LA.

\bibitem[Kim et~al., 2022]{kim2022stochastic}
Kim, S., Song, Q., and Liang, F. (2022).
\newblock {Stochastic gradient {Langevin} dynamics with adaptive drifts}.
\newblock {\em Journal of Statistical Computation and Simulation},
  92(2):318--336.

\bibitem[Kingma and Ba, 2015]{kingma2014adam}
Kingma, D.~P. and Ba, J.~L. (2015).
\newblock Adam: a method for stochastic optimization.
\newblock In {\em Proceedings of the 3rd International Conference on Learning
  Representations (ICLR 2015)}, San Diego, CA.

\bibitem[Knoblauch et~al., 2022]{knoblauch2022optimization}
Knoblauch, J., Jewson, J., and Damoulas, T. (2022).
\newblock An optimization-centric view on {Bayes}' rule: reviewing and
  generalizing variational inference.
\newblock {\em Journal of Machine Learning Research}, 23(132):1--109.

\bibitem[Kondratyev et~al., 2016]{kondratyev2016new}
Kondratyev, S., Monsaingeon, L., and Vorotnikov, D. (2016).
\newblock A new optimal transport distance on the space of finite radon
  measures.
\newblock {\em Advances in Differential Equations}, 21(11/12):1117 -- 1164.

\bibitem[Kondratyev and Vorotnikov, 2019]{kondratyev2019spherical}
Kondratyev, S. and Vorotnikov, D. (2019).
\newblock Spherical {Hellinger--Kantorovich} gradient flows.
\newblock {\em SIAM Journal on Mathematical Analysis}, 51(3):2053--2084.

\bibitem[Kook et~al., 2024]{kook2024sampling}
Kook, Y., Zhang, M.~S., Chewi, S., Erdogdu, M.~A., and Li, M.~B. (2024).
\newblock Sampling from the mean-field stationary distribution.
\newblock In {\em Proceedings of the 37th Annual Conference on Learning Theory
  (COLT 2024)}, Edmonton, Canada.

\bibitem[Korba et~al., 2021]{korba2021kernel}
Korba, A., Pierre-Cyril, Aubin-Frankowski, Majewski, S., and Ablin, P. (2021).
\newblock {Kernel {Stein} Discrepancy Descent}.
\newblock In {\em Proceedings of the 38th International Conference on Machine
  Learning (ICML 2021)}, Virtual.

\bibitem[Korba et~al., 2020]{korba2020nonasymptotic}
Korba, A., Salim, A., Arbel, M., Luise, G., and Gretton, A. (2020).
\newblock A non-asymptotic analysis for {Stein} variational gradient descent.
\newblock In {\em Proceedings of the 34th Annual Conference on Neural
  Information Processing Systems (NeurIPS 2020)}, Vancouver, Canada.

\bibitem[Krauth, 2006]{krauth2006statistical}
Krauth, W. (2006).
\newblock {\em Statistical Mechanics: Algorithms and Computations}.
\newblock Oxford University Press.

\bibitem[Kreisler et~al., 2024]{kreisler2024accelerated}
Kreisler, I., Ivgi, M., Hinder, O., and Carmon, Y. (2024).
\newblock Accelerated parameter-free stochastic optimization.
\newblock In {\em Proceedings of the 37th Annual Conference on Learning Theory
  (COLT 2024)}, Edmonton, Canada.

\bibitem[Kuhn et~al., 2019]{kuhn2019wasserstein}
Kuhn, D., Esfahani, P.~M., Nguyen, V.~A., and Shafieezadeh-Abadeh, S. (2019).
\newblock Wasserstein distributionally robust optimization: theory and
  applications in machine learning.
\newblock In {\em Operations Research \& Management cience in the Age of
  Analytics}, pages 130--166. Informs.

\bibitem[Kuntz et~al., 2023]{kuntz2023particle}
Kuntz, J., Lim, J.~N., and Johansen, A.~M. (2023).
\newblock Particle algorithms for maximum likelihood training of latent
  variable models.
\newblock In {\em Proceedings of the 26th International Conference on
  Artificial Intelligence and Statistics (AISTATS 2023)}, Valencia, Spain.

\bibitem[Lacker, 2023]{lacker2023independent}
Lacker, D. (2023).
\newblock Independent projections of diffusions: Gradient flows for variational
  inference and optimal mean field approximations.
\newblock {\em arXiv preprint arXiv:2309.13332}.

\bibitem[Lambert et~al., 2022]{lambert2022variational}
Lambert, M., Chewi, S., Bach, F., Bonnabel, S., and Rigollet, P. (2022).
\newblock Variational inference via {Wasserstein} gradient flows.
\newblock In {\em Proceedings of the 36th Annual Conference on Neural
  Information Processing Systems (NeurIPS 2022)}, New Orleans, LA.

\bibitem[Lanzetti et~al., 2023]{lanzetti2023stochastic}
Lanzetti, N., Balta, E.~C., Liao-McPherson, D., and D{\"o}rfler, F. (2023).
\newblock Stochastic {Wasserstein} gradient flows using streaming data with an
  application in predictive maintenance.
\newblock {\em IFAC-PapersOnLine}, 56(2):3954--3959.

\bibitem[Lanzetti et~al., 2025]{lanzetti2022first}
Lanzetti, N., Bolognani, S., and D{\"o}rfler, F. (2025).
\newblock First-order conditions for optimization in the {Wasserstein} space.
\newblock {\em SIAM Journal on Mathematics of Data Science}, 7(1):274--300.

\bibitem[Lascu et~al., 2024]{lascu2024linear}
Lascu, R.-A., Majka, M.~B., {\v{S}}i{\v{s}}ka, D., and Szpruch, {\L}. (2024).
\newblock Linear convergence of proximal descent schemes on the {Wasserstein}
  space.
\newblock {\em arXiv preprint arXiv:2411.15067}.

\bibitem[{\L{}}atuszyński et~al., 2013]{latuszynski2013adaptive}
{\L{}}atuszyński, K., Roberts, G.~O., and Rosenthal, J.~S. (2013).
\newblock Adaptive {Gibbs} samplers and related {MCMC} methods.
\newblock {\em The Annals of Applied Probability}, 23(1):66 -- 98.

\bibitem[LeCun et~al., 2015]{lecun2015deep}
LeCun, Y., Bengio, Y., and Hinton, G. (2015).
\newblock Deep learning.
\newblock {\em Nature}, 521(7553):436--444.

\bibitem[Leimkuhler and Matthews, 2016]{leimkuhler2016efficient}
Leimkuhler, B. and Matthews, C. (2016).
\newblock Efficient molecular dynamics using geodesic integration and
  solvent–solute splitting.
\newblock {\em Proceedings of the Royal Society A: Mathematical, Physical and
  Engineering Sciences}, 472(2189):20160138.

\bibitem[Leimkuhler et~al., 2018]{leimkuhler2018ensemble}
Leimkuhler, B., Matthews, C., and Weare, J. (2018).
\newblock Ensemble preconditioning for {Markov} chain {Monte Carlo} simulation.
\newblock {\em Statistics and Computing}, 28(2):277--290.

\bibitem[Leli{\`{e}}vre and Stoltz, 2016]{lelievre2016partial}
Leli{\`{e}}vre, T. and Stoltz, G. (2016).
\newblock Partial differential equations and stochastic methods in molecular
  dynamics.
\newblock {\em Acta Numerica}, 25:681--880.

\bibitem[Levy et~al., 2017]{levy2017generalizing}
Levy, D., Hoffman, M.~D., and Sohl-Dickstein, J. (2017).
\newblock Generalizing {Hamiltonian} {Monte Carlo} with neural networks.
\newblock In {\em Proceedings of the 6th International Conference on Learning
  Representations (ICLR 2018)}, Vancouver, Canada.

\bibitem[Levy, 2017]{levy2017online}
Levy, K. (2017).
\newblock Online to offline conversions, universality and adaptive minibatch
  sizes.
\newblock In {\em Proceedings of the 31st Annual Conference on Neural
  Information Processing Systems (NIPS 2017)}, Long Beach, CA.

\bibitem[Levy et~al., 2018]{levy2018online}
Levy, K.~Y., Yurtsever, A., and Cevher, V. (2018).
\newblock Online adaptive methods, universality and acceleration.
\newblock In {\em Proceedings of the 32nd International Conference on Neural
  Information Processing Systems (NeurIPS 2018)}, Montr\'{e}al, Canada.

\bibitem[Li et~al., 2016]{li2016preconditioned}
Li, C., Chen, C., Carlson, D., and Carin, L. (2016).
\newblock Preconditioned stochastic gradient {Langevin} dynamics for deep
  neural networks.
\newblock In {\em Proceedings of the 30th AAAI Conference on Artificial
  Intelligence (AAAI-16)}, Phoenix, AZ.

\bibitem[Li et~al., 2023]{li2023sampling}
Li, L., Liu, Q., Korba, A., Yurochkin, M., and Solomon, J. (2023).
\newblock Sampling with mollified interaction energy descent.
\newblock In {\em Proceedings of the 11th International Conference on Learning
  Representations (ICLR 2023)}, Kigali, Rwanda.

\bibitem[Li et~al., 2022]{li2022mirror}
Li, R., Tao, M., Vempala, S.~S., and Wibisono, A. (2022).
\newblock The mirror {Langevin} algorithm converges with vanishing bias.
\newblock In {\em Proceedings of the 33rd International Conference on
  Algorithmic Learning Theory (ALT 2022)}.

\bibitem[Li et~al., 2020]{li2020fisher}
Li, W., Lu, J., and Wang, L. (2020).
\newblock Fisher information regularization schemes for {Wasserstein} gradient
  flows.
\newblock {\em Journal of Computational Physics}, 416:109449.

\bibitem[Li and Orabona, 2019]{li2019convergence}
Li, X. and Orabona, F. (2019).
\newblock On the convergence of stochastic gradient descent with adaptive
  stepsizes.
\newblock In {\em Proceedings of the 22nd International Conference on
  Artificial Intelligence and Statistics (AISTATS 2019)}, Okinawa, Japan.

\bibitem[Li and Orabona, 2020]{li2020high}
Li, X. and Orabona, F. (2020).
\newblock A high probability analysis of adaptive {SGD} with momentum.
\newblock In {\em Proceedings of the 37th International Conference on Machine
  Learning (ICML 2020): Workshop on ``Beyond First-Order Methods in Machine
  Learning''}, Virtual.

\bibitem[Liero et~al., 2018]{liero2018optimal}
Liero, M., Mielke, A., and Savar{\'e}, G. (2018).
\newblock Optimal entropy-transport problems and a new {Hellinger--Kantorovich}
  distance between positive measures.
\newblock {\em Inventiones mathematicae}, 211(3):969--1117.

\bibitem[Liu et~al., 2019a]{liu2019understanding}
Liu, C., Zhuo, J., Cheng, P., Zhang, R., Zhu, J., and Carin, L. (2019a).
\newblock Understanding and accelerating particle-based variational inference.
\newblock In {\em Proceedings of the 36th International Conference on Machine
  Learning (ICML 2019)}, Long Beach, CA.

\bibitem[Liu et~al., 2019b]{liu2019understandingmcmc}
Liu, C., Zhuo, J., and Zhu, J. (2019b).
\newblock Understanding {MCMC} dynamics as flows on the {Wasserstein} space.
\newblock In {\em Proceedings of the 36th International Conference on Machine
  Learning (ICML 2019)}, Long Beach, CA.

\bibitem[Liu, 2009]{liu2009monte}
Liu, J.~S. (2009).
\newblock {\em {Monte Carlo} Strategies in Scientific Computing}.
\newblock Springer-Verlag, New York.

\bibitem[Liu, 2017]{liu2017stein}
Liu, Q. (2017).
\newblock {Stein variational gradient descent as gradient flow}.
\newblock In {\em Proceedings of the 31st Annual Conference on Neural
  Information Processing Systems (NIPS 2017)}, Long Beach, CA.

\bibitem[Liu et~al., 2016]{liu2016kernelized}
Liu, Q., Lee, J., and Jordan, M. (2016).
\newblock A kernelized {Stein} discrepancy for goodness-of-fit tests.
\newblock In {\em Proceedings of the 33rd International Conference on Machine
  Learning (ICML 2016)}, New York City, NY.

\bibitem[Liu and Wang, 2016]{liu2016stein}
Liu, Q. and Wang, D. (2016).
\newblock {Stein} variational gradient descent: A general purpose {Bayesian}
  inference algorithm.
\newblock In {\em Proceedings of the 30th Conference on Neural Information
  Processing Systems (NIPS 2016)}, Barcelona, Spain.

\bibitem[Liu and Wang, 2018]{liu2018stein}
Liu, Q. and Wang, D. (2018).
\newblock Stein variational gradient descent as moment matching.
\newblock In {\em Proceedings of the 32nd Annual Conference on Neural
  Information Processing Systems (NeurIPS 2018)}, Montr\'{e}al, Canada.

\bibitem[Liu et~al., 2023]{liu2024towards}
Liu, T., Ghosal, P., Balasubramanian, K., and Pillai, N. (2023).
\newblock Towards understanding the dynamics of gaussian-stein variational
  gradient descent.
\newblock In {\em Proceedings of the 37th Annual Conference on Neural
  Information Processing Systems (NeurIPS 2023)}, New Orleans, LA.

\bibitem[Loshchilov and Hutter, 2019]{loshchilov2019decoupled}
Loshchilov, I. and Hutter, F. (2019).
\newblock Decoupled weight decay regularization.
\newblock In {\em Proceedings of the 7th International Conference on Learning
  Representations (ICLR 2019)}, New Orleans, LA.

\bibitem[Lu et~al., 2019a]{lu2019scaling}
Lu, J., Lu, Y., and Nolen, J. (2019a).
\newblock Scaling limit of the {Stein} variational gradient descent: The mean
  field regime.
\newblock {\em SIAM Journal on Mathematical Analysis}, 51(2):648--671.

\bibitem[Lu et~al., 2019b]{lu2019accelerating}
Lu, Y., Lu, J., and Nolen, J. (2019b).
\newblock Accelerating {Langevin} sampling with birth-death.
\newblock {\em arXiv preprint arXiv:1905.09863}.

\bibitem[Lu et~al., 2023]{lu2023birth}
Lu, Y., Slepcev, D., and Wang, L. (2023).
\newblock Birth–death dynamics for sampling: global convergence,
  approximations and their asymptotics.
\newblock {\em Nonlinearity}, 36(11).

\bibitem[Luo et~al., 2019]{luo2019adaptive}
Luo, L., Xiong, Y., Liu, Y., and Sun, X. (2019).
\newblock Adaptive gradient methods with dynamic bound of learning rate.
\newblock In {\em Proceedings of the 7th International Conference on Learning
  Representations (ICLR 2019)}, New Orleans, LA.

\bibitem[Ma et~al., 2015]{ma2015complete}
Ma, Y.-A., Chen, T., and Fox, E. (2015).
\newblock A complete recipe for stochastic gradient {MCMC}.
\newblock In {\em Proceedings of the 29th Annual Conference on Neural
  Information Processing Systems (NIPS 2015)}, Montreal, Canada.

\bibitem[MacKay, 1995]{mackay1995probable}
MacKay, D.~J. (1995).
\newblock Probable networks and plausible predictions-a review of practical
  {Bayesian} methods for supervised neural networks.
\newblock {\em Network: Computation in Neural Systems}, 6(3):469.

\bibitem[MacKay, 2003]{mackay2003information}
MacKay, D.~J. (2003).
\newblock {\em Information Theory, Inference, and Learning Algorithms}.
\newblock Cambridge University Press.

\bibitem[Malitsky and Mishchenko, 2020]{malitsky2020adaptive}
Malitsky, Y. and Mishchenko, K. (2020).
\newblock Adaptive gradient descent without descent.
\newblock In {\em Proceedings of the 37th International Conference on Machine
  Learning (ICML 2020)}, Virtual.

\bibitem[Maoutsa et~al., 2020]{maoutsa2020interacting}
Maoutsa, D., Reich, S., and Opper, M. (2020).
\newblock Interacting particle solutions of {Fokker--Planck} equations through
  gradient--log--density estimation.
\newblock {\em Entropy}, 22(8):802.

\bibitem[Marshall and Roberts, 2012]{marshall2012adaptive}
Marshall, T. and Roberts, G. (2012).
\newblock An adaptive approach to {Langevin} {MCMC}.
\newblock {\em Statistics and Computing}, 22(5):1041--1057.

\bibitem[Maurais and Marzouk, 2024]{maurais2024sampling}
Maurais, A. and Marzouk, Y. (2024).
\newblock Sampling in unit time with kernel {Fisher-Rao} flow.
\newblock In {\em Proceedings of the 41st International Conference on Machine
  Learning (ICML 2024)}, Vienna, Austria.

\bibitem[McCann, 1997]{mccann1997convexity}
McCann, R.~J. (1997).
\newblock A convexity principle for interacting gases.
\newblock {\em Advances in Mathematics}, 128(1):153--179.

\bibitem[McLatchie et~al., 2025]{mclatchie2025predictively}
McLatchie, Y., Cherief-Abdellatif, B.-E., Frazier, D.~T., and Knoblauch, J.
  (2025).
\newblock Predictively oriented posteriors.
\newblock {\em arXiv preprint arXiv:2510.01915}.

\bibitem[McMahan and Orabona, 2014]{mcmahan2014unconstrained}
McMahan, H.~B. and Orabona, F. (2014).
\newblock Unconstrained online linear learning in hilbert spaces: minimax
  algorithms and normal approximations.
\newblock In {\em Proceedings of the 27th Annual Conference on Learning Theory
  (COLT 2014)}, pages 1020--1039, Barcelona, Spain.

\bibitem[McMahan and Streeter, 2010]{mcmahan2010adaptive}
McMahan, H.~B. and Streeter, M. (2010).
\newblock Adaptive bound optimization for online convex optimization.
\newblock In {\em Proceedings of the 23rd Annual Conference on Learning Theory
  (COLT 2010)}, Haifa, Israel.

\bibitem[Mei et~al., 2019]{mei2019mean}
Mei, S., Misiakiewicz, T., and Montanari, A. (2019).
\newblock Mean-field theory of two-layers neural networks: dimension-free
  bounds and kernel limit.
\newblock In {\em Proceedings of the 32nd Annual Conference on Learning Theory
  (COLT 2019)}, Phoenix, AZ.

\bibitem[Mei et~al., 2018]{mei2018mean}
Mei, S., Montanari, A., and Nguyen, P.-M. (2018).
\newblock A mean field view of the landscape of two-layer neural networks.
\newblock {\em Proceedings of the National Academy of Sciences},
  115(33):E7665--E7671.

\bibitem[Miasojedow et~al., 2013]{miasojedow2013adaptive}
Miasojedow, B., Moulines, E., and Vihola, M. (2013).
\newblock An adaptive parallel tempering algorithm.
\newblock {\em Journal of Computational and Graphical Statistics},
  22(3):649--664.

\bibitem[Mishchenko and Defazio, 2024]{mishchenko2024prodigy}
Mishchenko, K. and Defazio, A. (2024).
\newblock Prodigy: An expeditiously adaptive parameter-free learner.
\newblock In {\em Proceedings of the 41st International Conference on Machine
  Learning (ICML 2024)}, Vienna, Austria.

\bibitem[Mohajerin~Esfahani and Kuhn, 2018]{mohajerin2018data}
Mohajerin~Esfahani, P. and Kuhn, D. (2018).
\newblock Data-driven distributionally robust optimization using the
  {Wasserstein} metric: Performance guarantees and tractable reformulations.
\newblock {\em Mathematical Programming}, 171(1):115--166.

\bibitem[Mokrov et~al., 2021]{mokrov2021large}
Mokrov, P., Korotin, A., Li, L., Genevay, A., Solomon, J., and Burnaev, E.
  (2021).
\newblock Large-scale {Wasserstein} gradient flows.
\newblock In {\em Proceedings of the 35th Conference on Neural Information
  Processing Systems (NeurIPS 2021)}, Virtual.

\bibitem[Moulines and Bach, 2011]{moulines2011non}
Moulines, E. and Bach, F. (2011).
\newblock Non-asymptotic analysis of stochastic approximation algorithms for
  machine learning.
\newblock In {\em Proceedings of the 25th Annual Conference on Neural
  Information Processing Systems (NIPS 2011)}.

\bibitem[Neal, 1992]{neal1992bayesian}
Neal, R. (1992).
\newblock {Bayesian} learning via stochastic dynamics.
\newblock In {\em Proceedings of the 5th International Conference on Neural
  Information Processing Systems (NIPS 1992)}, volume~5.

\bibitem[Neal, 1996]{neal1996bayesian}
Neal, R.~M. (1996).
\newblock {\em {Bayesian} Learning for Neural Networks}.
\newblock Springer, New York.

\bibitem[Nemeth and Fearnhead, 2021]{nemeth2021stochastic}
Nemeth, C. and Fearnhead, P. (2021).
\newblock Stochastic gradient {Markov} chain {Monte Carlo}.
\newblock {\em Journal of the American Statistical Association},
  116(533):433--450.

\bibitem[Nesterov, 2015]{nesterov2015universal}
Nesterov, Y. (2015).
\newblock Universal gradient methods for convex optimization problems.
\newblock {\em Mathematical Programming}, 152(1):381--404.

\bibitem[Nesterov, 2018]{nesterov2018lectures}
Nesterov, Y. (2018).
\newblock {\em Lectures on Convex Optimization}.
\newblock Springer Cham, 2nd edition.

\bibitem[Nitanda et~al., 2022]{nitanda2022convex}
Nitanda, A., Wu, D., and Suzuki, T. (2022).
\newblock Convex analysis of the mean field {Langevin} dynamics.
\newblock In {\em Proceedings of The 25th International Conference on
  Artificial Intelligence and Statistics (AISTATS 2022)}, pages 9741--9757,
  Virtual.

\bibitem[N{\"u}sken and Renger, 2023]{nusken2021stein}
N{\"u}sken, N. and Renger, D. R.~M. (2023).
\newblock Stein variational gradient descent: many-particle and long-time
  asymptotics.
\newblock {\em Foundations of Data Science}, 5(3):286--320.

\bibitem[Olkin and Pukelsheim, 1982]{olkin1982distance}
Olkin, I. and Pukelsheim, F. (1982).
\newblock The distance between two random vectors with given dispersion
  matrices.
\newblock {\em Linear Algebra and its Applications}, 48:257--263.

\bibitem[Orabona, 2013]{orabona2013dimension}
Orabona, F. (2013).
\newblock Dimension-free exponentiated gradient.
\newblock In {\em Proceedings of the 27th Annual Conference on Neural
  Information Processing Systems (NIPS 2013)}, Lake Tahoe, NV.

\bibitem[Orabona, 2014]{orabona2014simultaneous}
Orabona, F. (2014).
\newblock Simultaneous model selection and optimization through parameter-free
  stochastic learning.
\newblock In {\em Proceedings of the 28th Annual Conference on Neural
  Information Processing Systems (NIPS 2014)}, Montreal, Canada.

\bibitem[Orabona, 2023]{orabona2023normalized}
Orabona, F. (2023).
\newblock Normalized gradients for all.
\newblock {\em arXiv preprint arXiv:2308.05621}.

\bibitem[Orabona, 2025]{orabona2019modern}
Orabona, F. (2025).
\newblock A modern introduction to online learning.
\newblock {\em arXiv preprint arXiv:1912.13213}.

\bibitem[Orabona and P{\'a}l, 2016]{orabona2016coin}
Orabona, F. and P{\'a}l, D. (2016).
\newblock Coin betting and parameter-free online learning.
\newblock In {\em Proceedings of the 30th Annual Conference on Neural
  Information Processing Systems (NIPS 2016)}, Barcelona, Spain.

\bibitem[Orabona and P{\'a}l, 2021]{orabona2021parameter}
Orabona, F. and P{\'a}l, D. (2021).
\newblock Parameter-free stochastic optimization of variationally coherent
  functions.
\newblock {\em arXiv preprint arXiv:2102.00236}.

\bibitem[Orabona and Tommasi, 2017]{orabona2017training}
Orabona, F. and Tommasi, T. (2017).
\newblock Training deep networks without learning rates through coin betting.
\newblock In {\em Proceedings of the 31st Annual Conference on Neural
  Information Processing Systems (NIPS 2017)}, Long Beach, CA.

\bibitem[Otto, 2001]{otto2001geometry}
Otto, F. (2001).
\newblock The geometry of dissipative evolution equations: The porous medium
  equation.
\newblock {\em Communications in Partial Differential Equations},
  26(1-2):101--174.

\bibitem[Otto and Westdickenberg, 2005]{otto2005eulerian}
Otto, F. and Westdickenberg, M. (2005).
\newblock Eulerian calculus for the contraction in the {Wasserstein} distance.
\newblock {\em SIAM Journal on Mathematical Analysis}, 37(4):1227--1255.

\bibitem[Pardo, 2006]{pardo2018statistical}
Pardo, L. (2006).
\newblock {\em Statistical Inference Based on Divergence Measures}.
\newblock Chapman and Hall/CRC, 1st edition.

\bibitem[Pasarica and Gelman, 2010]{pasarica2010adaptively}
Pasarica, C. and Gelman, A. (2010).
\newblock Adaptively scaling the {Metropolis} algorithm using expected squared
  jumped distance.
\newblock {\em Statistica Sinica}, 20:343--364.

\bibitem[Polyak, 1987]{polyak1987introduction}
Polyak, B.~T. (1987).
\newblock {\em Introduction to Optimization}.
\newblock Optimization Software, Inc., New York.

\bibitem[Pompe et~al., 2020]{pompe2020framework}
Pompe, E., Holmes, C., and Łatuszyński, K. (2020).
\newblock A framework for adaptive {MCMC} targeting multimodal distributions.
\newblock {\em The Annals of Statistics}, 48(5):2930 -- 2952.

\bibitem[Prado et~al., 2024]{prado2024metropolis}
Prado, E., Nemeth, C., and Sherlock, C. (2024).
\newblock {Metropolis--Hastings} with scalable subsampling.
\newblock {\em arXiv preprint arXiv:2407.19602}.

\bibitem[Putcha et~al., 2023]{putcha2023preferential}
Putcha, S., Nemeth, C., and Fearnhead, P. (2023).
\newblock Preferential subsampling for stochastic gradient {Langevin} dynamics.
\newblock In {\em Proceedings of the 26th International Conference on
  Artificial Intelligence and Statistics (AISTATS 2023)}, pages 8837--8856,
  Valencia, Spain. PMLR.

\bibitem[Rahimian and Mehrotra, 2022]{rahimian2019distributionally}
Rahimian, H. and Mehrotra, S. (2022).
\newblock Frameworks and results in distributionally robust optimization.
\newblock {\em Open Journal of Mathematical Optimization}, 3(4).

\bibitem[Reddi et~al., 2018]{reddi2018convergence}
Reddi, S.~J., Kale, S., and Kumar, S. (2018).
\newblock On the convergence of adam and beyond.
\newblock In {\em Proceedings of the 6th International Conference on Learning
  Representations (ICLR 2018)}, Vancouver, Canada.

\bibitem[Ren et~al., 2024]{ren2024multi}
Ren, Y., Xiao, T., Gangwani, T., Rangi, A., Rahmanian, H., Ying, L., and
  Sanyal, S. (2024).
\newblock Multi-objective optimization via {Wasserstein-Fisher-Rao} gradient
  flow.
\newblock In {\em Proceedings of the 27th International Conference on
  Artificial Intelligence and Statistics (AISTATS 2024)}, Valencia, Spain.

\bibitem[Robert and Casella, 2004]{robert2004monte}
Robert, C.~P. and Casella, G. (2004).
\newblock {\em {Monte Carlo Statistical Methods}}.
\newblock Springer-Verlag, New York, 2nd edition.

\bibitem[Roberts and Rosenthal, 1998]{roberts1998optimal}
Roberts, G.~O. and Rosenthal, J.~S. (1998).
\newblock Optimal scaling of discrete approximations to {Langevin} diffusions.
\newblock {\em Journal of the Royal Statistical Society: Series B (Statistical
  Methodology)}, 60(1):255--268.

\bibitem[Roberts and Rosenthal, 2001]{roberts2001optimal}
Roberts, G.~O. and Rosenthal, J.~S. (2001).
\newblock Optimal scaling for various {Metropolis-Hastings} algorithms.
\newblock {\em Statistical science}, 16(4):351--367.

\bibitem[Roberts and Rosenthal, 2007]{roberts2007coupling}
Roberts, G.~O. and Rosenthal, J.~S. (2007).
\newblock Coupling and ergodicity of adaptive {Markov} chain {Monte Carlo}
  algorithms.
\newblock {\em Journal of Applied Probability}, 44(2):458--475.

\bibitem[Roberts and Tweedie, 1996]{roberts1996exponential}
Roberts, G.~O. and Tweedie, R.~L. (1996).
\newblock Exponential convergence of {Langevin} distributions and their
  discrete approximations.
\newblock {\em Bernoulli}, 2(4):341 -- 363.

\bibitem[Rockafellar, 1970]{rockafellar1970convex}
Rockafellar, R.~T. (1970).
\newblock {\em Convex Analysis}.
\newblock Princeton University Press, Princeton, NJ.

\bibitem[Rotskoff and Vanden-Eijnden, 2018]{rotskoff2018parameters}
Rotskoff, G. and Vanden-Eijnden, E. (2018).
\newblock Parameters as interacting particles: long time convergence and
  asymptotic error scaling of neural networks.
\newblock In {\em Proceedings of the 32nd Annual Conference on Neural
  Information Processing Systems (NIPS 2018)}, Montreal, Canada.

\bibitem[Rotskoff and Vanden-Eijnden, 2022]{rotskoff2022trainability}
Rotskoff, G. and Vanden-Eijnden, E. (2022).
\newblock Trainability and accuracy of artificial neural networks: an
  interacting particle system approach.
\newblock {\em Communications on Pure and Applied Mathematics},
  75(9):1889--1935.

\bibitem[Salim et~al., 2020]{salim2020wasserstein}
Salim, A., Korba, A., and Luise, G. (2020).
\newblock The {Wasserstein} proximal gradient algorithm.
\newblock In {\em Proceedings of the 34th Annual Conference on Neural
  Information Processing Systems (NeurIPS 2020)}, Vancouver, Canada.

\bibitem[Salim et~al., 2022]{salim2022convergence}
Salim, A., Sun, L., and Richtarik, P. (2022).
\newblock A convergence theory for {SVGD} in the population limit under
  {Talagrand’s} inequality {T1}.
\newblock In {\em Proceedings of the 39th International Conference on Machine
  Learning (ICML 2022)}, Baltimore, MD.

\bibitem[Santambrogio, 2015]{santambrogio2015optimal}
Santambrogio, F. (2015).
\newblock {\em Optimal Transport for Applied Mathematicians}.
\newblock Birkh{\"a}user Cham, 1st edition.

\bibitem[Santambrogio, 2017]{santambrogio2017euclidean}
Santambrogio, F. (2017).
\newblock Euclidean, metric, and {Wasserstein} gradient flows: an overview.
\newblock {\em Bulletin of Mathematical Sciences}, 7(1):87--154.

\bibitem[Shalev-Shwartz et~al., 2012]{shalev2012online}
Shalev-Shwartz, S. et~al. (2012).
\newblock Online learning and online convex optimization.
\newblock {\em Foundations and Trends{\textregistered} in Machine Learning},
  4(2):107--194.

\bibitem[Shamir and Zhang, 2013]{shamir2013stochastic}
Shamir, O. and Zhang, T. (2013).
\newblock Stochastic gradient descent for non-smooth optimization: Convergence
  results and optimal averaging schemes.
\newblock In {\em Proceedings of the 30th International Conference on Machine
  Learning (ICML 2013)}, pages 71--79, Atlanta, Georgia.

\bibitem[Sharrock et~al., 2024]{sharrock2023tuning}
Sharrock, L., Dodd, D., and Nemeth, C. (2024).
\newblock Tuning-free maximum likelihood training of latent variable models via
  coin betting.
\newblock In {\em Proceedings of the 27th International Conference on
  Artificial Intelligence and Statistics (AISTATS 2024)}, Valencia, Spain.

\bibitem[Sharrock et~al., 2023]{sharrock2023learning}
Sharrock, L., Mackey, L., and Nemeth, C. (2023).
\newblock Learning rate free sampling in constrained domains.
\newblock In {\em Proceedings of the 37th Annual Conference on Neural
  Information Processing Systems (NeurIPS 2023)}, New Orleans, LA.

\bibitem[Sharrock and Nemeth, 2023]{sharrock2023coin}
Sharrock, L. and Nemeth, C. (2023).
\newblock Coin sampling: gradient-based {Bayesian} inference without learning
  rates.
\newblock In {\em Proceedings of the 40th International Conference on Machine
  Learning (ICML 2023)}, Honolulu, HI.

\bibitem[Sharrock and Nemeth, 2025]{sharrock2025sampling}
Sharrock, L. and Nemeth, C. (2025).
\newblock Sampling as optimization on the space of probability measures.
\newblock {\em In preparation}.

\bibitem[Shi and Mackey, 2023]{shi2022finite}
Shi, J. and Mackey, L. (2023).
\newblock A finite-particle convergence rate for {Stein} variational gradient
  descent.
\newblock In {\em Proceedings of the 37th Annual Conference on Neural
  Information Processing Systems (NeurIPS 2023)}, New Orleans, LA.

\bibitem[Sirignano and Spiliopoulos, 2020a]{sirignano2020meanfield}
Sirignano, J. and Spiliopoulos, K. (2020a).
\newblock Mean field analysis of neural networks: a central limit theorem.
\newblock {\em Stochastic Processes and their Applications}, 130(3):1820--1852.

\bibitem[Sirignano and Spiliopoulos, 2020b]{sirignano2020mean}
Sirignano, J. and Spiliopoulos, K. (2020b).
\newblock Mean field analysis of neural networks: a law of large numbers.
\newblock {\em SIAM Journal on Applied Mathematics}, 80(2):725--752.

\bibitem[Stein and Li, 2025]{stein2025accelerated}
Stein, V. and Li, W. (2025).
\newblock Accelerated {Stein} variational gradient flow.
\newblock {\em arXiv preprint arXiv:2503.23462}.

\bibitem[Streeter and McMahan, 2010]{streeter2010less}
Streeter, M. and McMahan, H.~B. (2010).
\newblock Less regret via online conditioning.
\newblock {\em arXiv preprint arXiv:1002.4862}.

\bibitem[Streeter and McMahan, 2012]{streeter2012no}
Streeter, M. and McMahan, H.~B. (2012).
\newblock No-regret algorithms for unconstrained online convex optimization.
\newblock In {\em Proceedings of the 26th Annual Conference on Neural
  Information Processing Systems (NIPS 2012)}, Lake Tahoe, NV.

\bibitem[Stuart, 2010]{stuart2010inverse}
Stuart, A.~M. (2010).
\newblock Inverse problems: a {Bayesian} perspective.
\newblock {\em Acta Numerica}, 19:451--559.

\bibitem[Sun et~al., 2023]{sun2023convergence}
Sun, L., Karagulyan, A., and Richtarik, P. (2023).
\newblock Convergence of {Stein} variational gradient descent under a weaker
  smoothness condition.
\newblock In {\em Proceedings of the 26th International Conference on
  Artificial Intelligence and Statistics (AISTATS 2023)}, Valencia, Spain.

\bibitem[Suzuki et~al., 2023a]{suzuki2023uniform}
Suzuki, T., Nitanda, A., and Wu, D. (2023a).
\newblock Uniform-in-time propagation of chaos for the mean-field gradient
  {Langevin} dynamics.
\newblock In {\em Proceedings of the 11th International Conference on Learning
  Representations (ICLR 2023)}, Kigali, Rwanda.

\bibitem[Suzuki et~al., 2023b]{suzuki2023convergence}
Suzuki, T., Wu, D., and Nitanda, A. (2023b).
\newblock {Convergence of mean-field {Langevin} dynamics: time-space
  discretization, stochastic gradient, and variance reduction}.
\newblock In {\em Proceedings of the 37th Annual Conference on Neural
  Information Processing Systems (NeurIPS 2023)}, New Orleans, LA.

\bibitem[Tan and Lu, 2025]{tan2024accelerate}
Tan, L. and Lu, J. (2025).
\newblock Accelerate {Langevin} sampling with birth-death process and
  exploration component.
\newblock {\em SIAM/ASA Journal on Uncertainty Quantification},
  13(3):1265--1293.

\bibitem[Tian et~al., 2024]{tian2024adaptive}
Tian, Z., Lee, A., and Zhou, S. (2024).
\newblock Adaptive tempered reversible jump algorithm for {Bayesian} curve
  fitting.
\newblock {\em Inverse Problems}, 40(4):045024.

\bibitem[Titsias, 2023]{titsias2023optimal}
Titsias, M. (2023).
\newblock Optimal preconditioning and {Fisher} adaptive {Langevin} sampling.
\newblock In {\em Proceedings of the 37th Annual Conference on Neural
  Information Processing Systems (NeurIPS 2023)}, New Orleans, LA.

\bibitem[Titsias and Dellaportas, 2019]{titsias2019gradient}
Titsias, M. and Dellaportas, P. (2019).
\newblock Gradient-based adaptive {Markov chain Monte Carlo}.
\newblock In {\em Proceedings of the 33rd Annual Conference on Neural
  Information Processing Systems (NeurIPS 2019)}, Vancouver, Canada.

\bibitem[Trillos et~al., 2023]{trillos2023optimization}
Trillos, N.~G., Hosseini, B., and Sanz-Alonso, D. (2023).
\newblock From optimization to sampling through gradient flows.
\newblock {\em Notices of the American Mathematical Society}, 70(6).

\bibitem[Van~Erven and Harremos, 2014]{van2014renyi}
Van~Erven, T. and Harremos, P. (2014).
\newblock R{\'e}nyi divergence and {Kullback-Leibler} divergence.
\newblock {\em IEEE Transactions on Information Theory}, 60(7):3797--3820.

\bibitem[Vihola, 2012]{vihola2012robust}
Vihola, M. (2012).
\newblock Robust adaptive {Metropolis} algorithm with coerced acceptance rate.
\newblock {\em Statistics and Computing}, 22(5):997--1008.

\bibitem[Villani, 2001]{villani2001limites}
Villani, C. (2001).
\newblock Limites hydrodynamiques de l’{\'e}quation de {Boltzmann}.
\newblock {\em S{\'e}minaire Bourbaki}, 2000:365--405.

\bibitem[Villani, 2003]{villani2003topics}
Villani, C. (2003).
\newblock {\em Topics in Optimal Transportation}.
\newblock American Mathematical Society, Providence, Rhode Island.

\bibitem[Villani, 2008]{villani2008optimal}
Villani, C. (2008).
\newblock {\em Optimal Transport: Old and New}.
\newblock Springer-Verlag, Berlin.

\bibitem[Wang et~al., 2024]{wang2024reinforcement}
Wang, C., Chen, W., Kanagawa, H., Oates, C., et~al. (2024).
\newblock Reinforcement learning for adaptive {MCMC}.
\newblock {\em arXiv preprint arXiv:2405.13574}.

\bibitem[Wang et~al., 2025]{wang2025harnessing}
Wang, C., Fisher, M.~A., Kanagawa, H., Chen, W., Oates, C., et~al. (2025).
\newblock Harnessing the power of reinforcement learning for adaptive {MCMC}.
\newblock {\em arXiv preprint arXiv:2507.00671}.

\bibitem[Wang et~al., 2022]{wang2022projected}
Wang, Y., Chen, P., and Li, W. (2022).
\newblock Projected {Wasserstein} gradient descent for high-dimensional
  {Bayesian} inference.
\newblock {\em SIAM/ASA Journal on Uncertainty Quantification},
  10(4):1513--1532.

\bibitem[Wang and Li, 2022]{wang2022accelerated}
Wang, Y. and Li, W. (2022).
\newblock Accelerated information gradient flow.
\newblock {\em Journal of Scientific Computing}, 90:11.

\bibitem[Wang et~al., 2013]{mohamed2013adaptive}
Wang, Z., Mohamed, S., and de~Freitas, N. (2013).
\newblock Adaptive {Hamiltonian} and {Riemann} manifold {Monte Carlo} samplers.
\newblock In {\em Proceedings of the 30th International Conference on Machine
  Learning (ICML 2013)}, Atlanta, Georgia.

\bibitem[Ward et~al., 2020]{ward2020adagrad}
Ward, R., Wu, X., and Bottou, L. (2020).
\newblock Adagrad stepsizes: sharp convergence over nonconvex landscapes.
\newblock {\em Journal of Machine Learning Research}, 21(219):1--30.

\bibitem[Welling and Teh, 2011]{welling2011bayesian}
Welling, M. and Teh, Y.~W. (2011).
\newblock {Bayesian} learning via stochastic gradient {Langevin} dynamics.
\newblock In {\em Proceedings of the 28th International Conference on Machine
  Learning (ICML 2011)}, Bellevue, WA.

\bibitem[Wibisono, 2018]{wibisono2018sampling}
Wibisono, A. (2018).
\newblock Sampling as optimization in the space of measures: the {Langevin}
  dynamics as a composite optimization problem.
\newblock In {\em Proceedings of the 31st Annual Conference on Learning Theory
  (COLT 2018)}, Stockholm, Sweden.

\bibitem[Wilson and Izmailov, 2020]{wilson2020bayesian}
Wilson, A.~G. and Izmailov, P. (2020).
\newblock {Bayesian} deep learning and a probabilistic perspective of
  generalization.
\newblock In {\em Proceedings of the 34th Annual Conference on Neural
  Information Processing Systems (NeurIPS 2020)}, Vancouver, Canada.

\bibitem[Xu and Li, 2024]{xu2024forward}
Xu, Y. and Li, Q. (2024).
\newblock {Forward-Euler} time-discretization for {Wasserstein} gradient flows
  can be wrong.
\newblock {\em arXiv preprint arXiv:2406.08209}.

\bibitem[Xu and Dai, 2012]{xu2012new}
Xu, Z. and Dai, Y.-H. (2012).
\newblock New stochastic approximation algorithms with adaptive step sizes.
\newblock {\em Optimization Letters}, 6(8):1831--1846.

\bibitem[Yan et~al., 2024]{yan2024learning}
Yan, Y., Wang, K., and Rigollet, P. (2024).
\newblock Learning {Gaussian} mixtures using the {Wasserstein--Fisher--Rao}
  gradient flow.
\newblock {\em The Annals of Statistics}, 52(4):1774--1795.

\bibitem[Yang et~al., 2023a]{yang2024two}
Yang, J., Li, X., Fatkhullin, I., and He, N. (2023a).
\newblock Two sides of one coin: the limits of untuned {SGD} and the power of
  adaptive methods.
\newblock In {\em Proceedings of the 37th Annual Conference on Advances in
  Neural Information Processing Systems (NeurIPS 2023)}, New Orleans, LA.

\bibitem[Yang et~al., 2023b]{yang2023estimating}
Yang, Y., Eckstein, S., Nutz, M., and Mandt, S. (2023b).
\newblock Estimating the rate-distortion function by {Wasserstein} gradient
  descent.
\newblock In {\em Proceedings of the 37th Annual Conference on Neural
  Information Processing Systems (NeurIPS 2023)}, New Orleans, LA.

\bibitem[Yao and Yang, 2022]{yao2022mean}
Yao, R. and Yang, Y. (2022).
\newblock Mean field variational inference via {Wasserstein} gradient flow.
\newblock {\em arXiv preprint arXiv:2207.08074}.

\bibitem[Zamani and Glineur, 2025]{zamani2023exact}
Zamani, M. and Glineur, F. (2025).
\newblock Exact convergence rate of the last iterate in subgradient methods.
\newblock {\em SIAM Journal on Optimization}, 35(3):2182--2201.

\bibitem[Zeiler, 2012]{zeiler2012adadelta}
Zeiler, M.~D. (2012).
\newblock {ADADELTA}: an adaptive learning rate method.
\newblock {\em arXiv preprint arXiv:1212.5701}.

\bibitem[Zhang and Sra, 2016]{zhang2016first}
Zhang, H. and Sra, S. (2016).
\newblock First-order methods for geodesically convex optimization.
\newblock In {\em Proceedings of the 29th Annual Conference on Learning Theory
  (COLT 2016)}, New York, NY.

\bibitem[Zhang et~al., 2020]{zhang2020wasserstein}
Zhang, K.~S., Peyr{\'e}, G., Fadili, J., and Pereyra, M. (2020).
\newblock Wasserstein control of mirror {Langevin} {Monte Carlo}.
\newblock In {\em Proceedings of the 33rd Annual Conference on Learning Theory
  (COLT 2020)}, Graz, Austria.

\bibitem[Zhang et~al., 2018]{zhang2018policy}
Zhang, R., Chen, C., Li, C., and Carin, L. (2018).
\newblock Policy optimization as {Wasserstein} gradient flows.
\newblock In {\em Proceedings of the 35th International Conference on Machine
  Learning (ICML 2018)}, Stockholm, Sweden.

\bibitem[Zhu and Chen, 2025]{zhu2025convergence}
Zhu, S. and Chen, X. (2025).
\newblock Convergence analysis of the {Wasserstein} proximal algorithm beyond
  geodesic convexity.
\newblock {\em arXiv preprint arXiv:2501.14993}.

\bibitem[Zinkevich, 2003]{zinkevich2003online}
Zinkevich, M. (2003).
\newblock Online convex programming and generalized infinitesimal gradient
  ascent.
\newblock In {\em Proceedings of the 20th International Conference on Machine
  Learning (ICML 2003)}, pages 928--936, Washington, D.C.

\end{thebibliography}

\appendix 

\section{Additional Results: Forward Euler Discretization}
\label{app:additional-results}

In this appendix, we analyze the convergence of \eqref{eq:wasserstein-sub-grad-descent}, the forward Euler discretization of the WGF in \eqref{eq:wasserstein-grad-flow}. Our results mirror those given in the main text for the forward-flow discretization. 

\subsection{Forward Euler Discretization, Nonsmooth Setting}
\label{sec:forward-euler-nonsmooth}
We begin, as before, by considering the nonsmooth case. 

\subsubsection{Deterministic Case: Constant Step Size}
We first consider the case where the step size in \eqref{eq:wasserstein-sub-grad-descent} is constant, that is, $\eta_t=\eta$ for all $t\geq 0$. We will analyze the convergence of the average iterate defined according to the following recursion:
\begin{align}
    \bar{\mu}_1 &= \mu_0 \label{eq:bar-mu-t-0} \\
    \bar{\mu}_{t+1} &= \left[\left(1-\frac{1}{t+1}\right) \mathrm{id} + \frac{1}{t+1} \boldsymbol{t}_{\bar{\mu}_{t}}^{\mu_{t}}\right]_{\#}\bar{\mu}_{t},\quad t\geq 1. \label{eq:bar-mu-t}
\end{align}

\begin{lemma}
\label{lemma:average-bound}
    Suppose that Assumption \ref{assumption:lsc} and \ref{assumption:geo-convex} hold. Let $\bar{\mu}_T$ be the average iterate defined by \eqref{eq:bar-mu-t-0} - \eqref{eq:bar-mu-t}. Then, for all $T\geq 1$, it holds that 
    \begin{equation}
        \mathcal{F}(\bar{\mu}_{T}) - \mathcal{F}(\pi) \leq \frac{1}{T}\left[\sum_{t=0}^{T-1}\mathcal{F}(\mu_t) - \sum_{t=0}^{T-1}\mathcal{F}(\pi)\right]. \label{eq:average-bound-lemma}
    \end{equation}
\end{lemma}

\begin{proof}
    It is sufficient to show that $\mathcal{F}(\bar{\mu}_{T})\leq \frac{1}{T}\sum_{t=0}^{T-1} \mathcal{F}(\mu_t)$. We will prove this result by induction over $T\in\mathbb{N}$. First consider $T=1$. By direct computation, we have $\bar{\mu}_1 = \mu_0$, and so $\smash{\mathcal{F}(\bar{\mu}_1) = \mathcal{F}(\mu_0)  = }$ $\smash{\frac{1}{T}\sum_{t=0}^{0}\mathcal{F}(\mu_t)}$. This completes the base case. Suppose, now, that the result holds for $T$. Let $\mu\in\mathcal{P}_2(\mathbb{R}^d)$ and $\nu\in\mathcal{P}_2(\mathbb{R}^d)$. In addition, let $\xi_{\nu} = \boldsymbol{t}_{\mu}^{\nu} -\mathrm{id}\in\mathcal{T}_{\mu}\mathcal{P}_2(\mathbb{R}^d)$. Define $\smash{\lambda_{\mu}^{\nu}(\cdot):[0,1]\rightarrow\mathcal{P}_2(\mathbb{R}^d)}$ as a constant speed geodesic with $\lambda_{\mu}^{\nu}(0)= \mu$ and $\smash{\dot{\lambda_{\mu}^{\nu}}(0)= \xi_{\nu}}$. That is,
    \begin{equation}
        \lambda_{\mu}^{\nu}(s) = \left[\mathrm{id} + s \xi_{\nu} \right]_{\#}\mu = \left[\mathrm{id} + s\left(\boldsymbol{t}_{\mu}^{\nu} - \mathrm{id}\right)\right]_{\#}\mu.
    \end{equation} 
     It follows, working from the definition in  \eqref{eq:bar-mu-t-0} - \eqref{eq:bar-mu-t}, that
    \begin{align}
        \bar{\mu}_{T+1} &= \left[\left(1-\frac{1}{T+1}\right)\mathrm{id} + \frac{1}{T+1}\boldsymbol{t}_{\bar{\mu}_{T}}^{\mu_{T}}\right]\bar{\mu}_{T} = \left[\mathrm{id} + \frac{1}{T+1}\left(\boldsymbol{t}_{\bar{\mu}_{T}}^{\mu_{T}} - \mathrm{id}\right)\right]_{\#}\bar{\mu}_{T} =\lambda_{\bar{\mu}_{T}}^{\mu_{T}}\left(\frac{1}{T+1}\right).
    \end{align}
    Using the geodesic convexity of $\mathcal{F}$, we then have
    \begin{align}
        \mathcal{F}(\bar{\mu}_{T+1})&=\mathcal{F}\left(\lambda_{\bar{\mu}_{T}}^{\mu_{T}}\left(\frac{1}{T+1}\right)\right) \\
        &\leq \left(1-\frac{1}{T+1}\right) \mathcal{F}\left(\lambda_{\bar{\mu}_{T}}^{\mu_{T}}\left(0\right)\right) + \frac{1}{T+1}\mathcal{F}\left(\lambda_{\bar{\mu}_{T}}^{\mu_{T}}\left(1\right)\right) \\
        &=\left(1-\frac{1}{T+1}\right) \mathcal{F}\left(\bar{\mu}_{T}\right) + \frac{1}{T+1}\mathcal{F}\left(\mu_{T}\right).
    \end{align}
    Finally, using also the inductive assumption, we have that
    \begin{align}
        \mathcal{F}(\bar{\mu}_{T+1}) &\leq \left(1-\frac{1}{T+1}\right)\mathcal{F}(\bar{\mu}_{T}) + \frac{1}{T+1}\mathcal{F}(\mu_{T}) \\
        &\leq \left(1-\frac{1}{T+1}\right) \frac{1}{T}\displaystyle\sum_{t=0}^{T-1}\mathcal{F}(\mu_t) +\frac{1}{T+1}\mathcal{F}(\mu_{T}) 
        \\
        &= \frac{1}{T+1}\displaystyle\sum_{t=0}^{T}\mathcal{F}(\mu_t).
    \end{align}
\end{proof}

\begin{lemma}
\label{lemma:evi}
    Suppose that Assumption \ref{assumption:lsc} and \ref{assumption:geo-convex} hold. Let $(\mu_t)_{t\geq 0}$ denote the sequence of measures defined by \eqref{eq:wasserstein-sub-grad-descent}, the forward Euler discretization of the WGF in \eqref{eq:wasserstein-grad-flow}. Suppose that $\eta_t = \eta>0$ for all $t\geq 0$. Then, for any $t\geq 0$, and for any $\pi\in\mathcal{P}_2(\mathbb{R}^d)$, we have 
    \begin{equation}
        \mathcal{F}(\mu_t) - \mathcal{F}(\pi) \leq \frac{W^2_2(\mu_t,\pi) - W_2^2(\mu_{t+1},\pi)}{2\eta} + \frac{\eta}{2}\int_{\mathbb{R}^d}\|\xi_t(x)\|^2\,\mathrm{d}\mu_t(x).
        \label{eq:evi}
    \end{equation}
\end{lemma}

\begin{proof}
    From the definition of the Wasserstein-2 distance, we have that 
    \begin{align}
        W_2^2(\mu_{t},\pi) &=\inf_{\gamma \in \Gamma(\mu_{t},\pi)} \int_{\mathbb{R}^d\times\mathbb{R}^d}\|x - y\|^2\mathrm{d}\gamma(x,y)= \int_{\mathbb{R}^d} \| x - \boldsymbol{t}_{\mu_t}^{\pi}(x)\|^2 \,\mathrm{d}\mu_t(x) \\
        W_2^2(\mu_{t+1},\pi) &= \inf_{\gamma \in \Gamma(\mu_{t+1},\pi)} \int_{\mathbb{R}^d\times\mathbb{R}^d}\|x - y\|^2\mathrm{d}\gamma(x,y)\leq \int_{\mathbb{R}^d}\|(x-\eta \xi_t(x)) - \boldsymbol{t}_{\mu_{t}}^{\pi}(x)\|^2\mathrm{d}\mu_{t}(x),
    \end{align}
    Putting these two together, and expanding, it follows straightforwardly that
    \begin{align}
        \frac{W^2_2(\mu_t,\pi) - W_2^2(\mu_{t+1},\pi)}{2\eta} &\geq \frac{1}{2\eta} \int_{\mathbb{R}^d} \left[\| x - \boldsymbol{t}_{\mu_t}^{\pi}(x)\|^2  - \|(x-\eta \xi_t(x)) - \boldsymbol{t}_{\mu_{t}}^{\pi}(x)\|^2\right]\mathrm{d}\mu_{t}(x) \\
        & =  \int_{\mathbb{R}^d} \left\langle \xi_{t}(x), x-\boldsymbol{t}_{\mu_t}^{\pi}(x) \right\rangle \,\mathrm{d}\mu_t(x) - \frac{\eta}{2}\int_{\mathbb{R}^d} \|\xi_t(x)\|^2\,\mathrm{d}\mu_t(x) \\
        &\geq \mathcal{F}(\mu_t) - \mathcal{F}(\pi) - \frac{\eta}{2}\int_{\mathbb{R}^d} \|\xi_t(x)\|^2\,\mathrm{d}\mu_t(x)
    \end{align}
    where in the final line we have used the definition of geodesic convexity. Rearranging gives precisely the required result.
\end{proof}

\begin{proposition}
\label{prop:regret-bound}
    Suppose that Assumption \ref{assumption:lsc} and \ref{assumption:geo-convex} hold. Let $(\mu_t)_{t\geq 0}$ denote the sequence of measures defined by \eqref{eq:wasserstein-sub-grad-descent}, the forward Euler discretization of the WGF in \eqref{eq:wasserstein-grad-flow}. Suppose that $\eta_t=\eta>0$ for all $t\geq 0$. Then, for any $\pi\in\mathcal{P}_2(\mathbb{R}^d)$, 
    \begin{equation}
        \sum_{t=0}^{T-1} \mathcal{F}(\mu_t) - \sum_{t=0}^{T-1} \mathcal{F}(\pi) \leq \frac{W_2^2(\mu_0,\pi) - W_2^2(\mu_{T},\pi)}{2\eta} + \frac{\eta}{2}\sum_{t=0}^{T-1} \int_{\mathbb{R}^d}\|\xi_t(x)\|^2\,\mathrm{d}\mu_t(x).
    \end{equation}
\end{proposition}

\begin{proof}
    The result is an immediate consequence of Lemma \ref{lemma:evi}. In particular, summing \eqref{eq:evi} over $t\in 0,\dots,T-1$, we have 
    \begin{align}
        \sum_{t=0}^{T-1} \mathcal{F}(\mu_t) - \sum_{t=0}^{T-1} \mathcal{F}(\pi) &\leq \sum_{t=0}^{T-1}\left[\frac{W_2^2(\mu_{t},\pi) - W_2^2(\mu_{t+1},\pi)}{2\eta} + \frac{\eta}{2} \int_{\mathbb{R}^d} \|\xi_t(x)\|^2\,\mathrm{d}\mu_t(x)\right] \\
        &= \frac{\sum_{t=0}^{T-1} [W_2^2(\mu_t,\pi) - W_2^2(\mu_{t+1},\pi)]}{2\eta} + \frac{\eta}{2}\sum_{t=0}^{T-1} \int_{\mathbb{R}^d} \|\xi_t(x)\|^2\,\mathrm{d}\mu_t(x)\\
        &=\frac{W_2^2(\mu_0,\pi) - W_2^2(\mu_{T},\pi)}{2\eta} + \frac{\eta}{2}\sum_{t=0}^{T-1} \int_{\mathbb{R}^d}\|\xi_t(x)\|^2\,\mathrm{d}\mu_t(x).
    \end{align}
\end{proof}

\begin{corollary}
\label{corollary:average-bound}
    Suppose that Assumption \ref{assumption:lsc} and \ref{assumption:geo-convex} hold. Let $(\mu_t)_{t\geq 0}$ denote the sequence of measures defined by \eqref{eq:wasserstein-sub-grad-descent}, the forward Euler discretization of the WGF in \eqref{eq:wasserstein-grad-flow}. Suppose that $\eta_t=\eta>0$ for all $t\geq 0$. Then, for all $\pi\in\mathcal{P}_2(\mathbb{R}^d)$, 
    \begin{equation}
        \mathcal{F}(\bar{\mu}_{T}) - \mathcal{F}(\pi) \leq \frac{1}{T}\left[\frac{W_2^2(\mu_0,\pi)}{2\eta} + \frac{\eta}{2}\sum_{t=0}^{T-1} \int_{\mathbb{R}^d}\|\xi_t(x)\|^2\,\mathrm{d}\mu_t(x)\right], \label{eq:average-bound}
    \end{equation}
\end{corollary}

\begin{proof}
    The result follows immediately from Lemma \ref{lemma:average-bound} and Proposition \ref{prop:regret-bound}. 
\end{proof}


\begin{theorem}
    Suppose that Assumption \ref{assumption:lsc} and \ref{assumption:geo-convex} hold. Let $(\mu_t)_{t\geq 0}$ denote the sequence of measures defined by \eqref{eq:wasserstein-sub-grad-descent}, the forward Euler discretization of the WGF in \eqref{eq:wasserstein-grad-flow}. Then, for all $\pi\in\mathcal{P}_2(\mathbb{R}^d)$, the upper bound in \eqref{eq:average-bound} is minimized when
    \begin{equation}
        \eta = \frac{W_2(\mu_0,\pi)}{\sqrt{\sum_{t=0}^{T-1} \int_{\mathbb{R}^d}\|\xi_t(x)\|^2\,\mathrm{d}\mu_t(x)}}, \label{eq:optimal-lr}
    \end{equation}
    Moreover, this choice of $\eta$ guarantees
    \begin{equation}
        \mathcal{F}(\bar{\mu}_T) - \mathcal{F}(\pi) \leq \frac{1}{T} W_2(\mu_0,\pi)\sqrt{\sum_{t=0}^{T-1} \int_{\mathbb{R}^d}\|\xi_t(x)\|^2\,\mathrm{d}\mu_t(x)}. \label{eq:optimal-average-bound}
    \end{equation}
\end{theorem}

\begin{proof}
   This result is an immediate consequence of Corollary \ref{corollary:average-bound}. In particular, differentiating the RHS of \eqref{eq:average-bound} with respect to $\eta$, and setting equal to zero, we have 
   \begin{equation}
       0 = -\frac{W_2^2(\mu_0,\pi)}{\eta^2} + \sum_{t=0}^{T-1} \int_{\mathbb{R}^d} \|\xi_t(x)\|^2\,\mathrm{d}\mu_t(x).
   \end{equation}
   Solving for $\eta$ gives \eqref{eq:optimal-lr}, and substituting back into \eqref{eq:average-bound} gives \eqref{eq:optimal-average-bound}.
\end{proof}

\subsubsection{Deterministic Case: Adaptive Step Size}
\label{sec:adaptive-forward-euler}
Similar to before, motivated by the intractability of the ideal constant step size in \eqref{eq:optimal-lr}, we now introduce a sequence of step sizes $(\eta_t)_{t\geq 0}$ defined according to
\begin{equation}
\tcboxmath[colback=black!5,colframe=black!15,boxrule=0.4pt, arc=2pt, left=8pt, right=8pt, top=6pt, bottom=6pt]{
    \eta_t = \frac{\max_{0\leq s\leq t} \left[r_{\varepsilon},W_2(\mu_0,\mu_s)\right]}{\sqrt{\sum_{s=0}^{t} \int_{\mathbb{R}^d}\|\xi_s(x)\|^2\mathrm{d}\mu_s(x)}}
}, \label{eq:dog-lr}
\end{equation}
where $r_\varepsilon>0$ is some small initial value. Thus, in particular, the step size in the $t^{\text{th}}$ iteration is equal to the maximum of the Wasserstein distances between the initial law $\mu_0$ and the subsequent laws generated by the Wasserstein subgradient descent algorithm $(\mu_t)_{t\in\mathbb{N}}$, scaled by the square-root of the cumulative sum of the squared $L_2$ Wasserstein gradient norms. This defines the \textsc{Fuse} step size schedule for the forward Euler discretization of the WGF.

\paragraph{Convergence Analysis Assuming Bounded Iterates}
We now study the convergence of the Wasserstein subgradient descent algorithm in \eqref{eq:wasserstein-sub-grad-descent} when using the adaptive step size schedule in \eqref{eq:dog-lr}. Similar to before, it will be convenient to define the quantities
\begin{alignat}{3}
    &r_t = W_2(\mu_0,\mu_t),\quad &&\bar{r}_t = \max_{0\leq s \leq t}\left[r_s,r_\varepsilon\right],\quad &&G_t = \sum_{s=0}^t \int_{\mathbb{R}^d} \|\xi_s(x)\|^2\mathrm{d}\mu_s(x) \label{eq:defs1} \\
    &d_t = W_2(\mu_t,\pi),\quad &&\bar{d}_t = \max_{0\leq s \leq t} d_s
    ,\quad &&\bar{g}_t = \max_{0\leq s \leq t} \|\xi_s\|_{L^2(\mu_s)}. \label{eq:defs2}
\end{alignat}
We will analyze the convergence of the weighted average iterate defined according to the following recursion:
\begin{align}
    \tilde{\mu}_1 &= \mu_0 \label{eq:tilde-mu-t-1} \\
     \tilde{\mu}_{t+1} &= \left[\left(1-\frac{\bar{r}_t}{\sum_{s=0}^{t} \bar{r}_s }\right) \mathrm{id} + \frac{\bar{r}_t}{\sum_{s=0}^{t} \bar{r}_s } \boldsymbol{t}_{\tilde{\mu}_{t}}^{\mu_t}\right]_{\#}\tilde{\mu}_{t},\quad t\geq 1. \label{eq:tilde-mu-t-2}
\end{align}

\begin{lemma}
\label{lemma:tilde-mu-t-bound-1}
    Suppose that Assumption 
    \ref{assumption:lsc} and \ref{assumption:geo-convex} hold. Let $\tilde{\mu}_T$ be the weighted average iterate defined in \eqref{eq:tilde-mu-t-1} - \eqref{eq:tilde-mu-t-2}. Then, for all $T\geq 1$, it holds that 
    \begin{align}
        \mathcal{F}(\tilde{\mu}_{T})  - \mathcal{F}(\pi) &\leq \frac{1}{\sum_{t=0}^{T-1}\bar{r}_t} \sum_{t=0}^{T-1}\bar{r}_t \left[\mathcal{F}(\mu_t) - \mathcal{F}(\pi) \right]. \label{eq:average-ineq} 
    \end{align}
\end{lemma}

\begin{proof}
    Similar to the proof of Corollary \ref{corollary:average-bound}, we will prove this result by induction over $T\in\mathbb{N}$. First consider $T=1$. By direct computation, we have $\tilde{\mu}_1 = \mu_0$, and so $\smash{\mathcal{F}(\tilde{\mu}_1) = \mathcal{F}(\mu_0)  = }$ $\smash{\frac{1}{T}\sum_{t=0}^{0}\mathcal{F}(\mu_t)}$. This proves the base case. Suppose, now, that the result holds for all $t\in\{0,\dots,T\}$. Let $\mu\in\mathcal{P}_2(\mathbb{R}^d)$ and $\nu\in\mathcal{P}_2(\mathbb{R}^d)$. In addition, let $\xi_{\nu} = \boldsymbol{t}_{\mu}^{\nu} -\mathrm{id}\in\mathcal{T}_{\mu}\mathcal{P}_2(\mathbb{R}^d)$. Define $\lambda_{\mu}^{\nu}(\cdot):[0,1]\rightarrow\mathcal{P}_2(\mathbb{R}^d)$ as a constant speed geodesic with $\lambda_{\mu}^{\nu}(0) = \mu$ and $\dot{\lambda_{\mu}^{\nu}}(0) = \xi_{\nu}$. That is,
    \begin{equation}
        \lambda_{\mu}^{\nu}(s) = \left[\mathrm{id} + s \xi_{\nu} \right]_{\#}\mu = \left[\mathrm{id} + s\left(\boldsymbol{t}_{\mu}^{\nu} - \mathrm{id}\right)\right]_{\#}\mu.
    \end{equation} 
     It follows, working from the definition in \eqref{eq:tilde-mu-t-1} - \eqref{eq:tilde-mu-t-2}, that
    \begin{align}
        \tilde{\mu}_{T+1} &= \left[\left(1-\frac{\bar{r}_T}{\sum_{t=0}^{T} \bar{r}_t}\right)\mathrm{id} + \frac{\bar{r}_T}{\sum_{t=0}^{T} \bar{r}_t}\boldsymbol{t}_{\tilde{\mu}_{T}}^{\mu_{T}}\right]\tilde{\mu}_{T} \\
        &= \left[\mathrm{id} + \frac{\bar{r}_T}{\sum_{t=0}^{T} \bar{r}_t}\left(\boldsymbol{t}_{\tilde{\mu}_{T}}^{\mu_{T}} - \mathrm{id}\right)\right]_{\#}\tilde{\mu}_{T} =\lambda_{\tilde{\mu}_{T}}^{\mu_{T}}\left(\frac{\bar{r}_T}{\sum_{t=0}^{T} \bar{r}_t}\right).
    \end{align}
    Using the geodesic convexity of $\mathcal{F}$, we then have
    \begin{align}
        \mathcal{F}(\tilde{\mu}_{T+1})&=\mathcal{F}\left(\lambda_{\tilde{\mu}_{T}}^{\mu_{T}}\left(\frac{\bar{r}_T}{\sum_{t=0}^{T} \bar{r}_t}\right)\right) \\
        &\leq \left(1-\frac{\bar{r}_T}{\sum_{t=0}^{T} \bar{r}_t}\right) \mathcal{F}\left(\lambda_{\tilde{\mu}_{T}}^{\mu_{T}}\left(0\right)\right) + \frac{\bar{r}_T}{\sum_{t=0}^{T} \bar{r}_t}\mathcal{F}\left(\lambda_{\tilde{\mu}_{T}}^{\mu_{T}}\left(1\right)\right) \\
        &=\left(1-\frac{\bar{r}_T}{\sum_{t=0}^{T} \bar{r}_t}\right) \mathcal{F}\left(\tilde{\mu}_{T}\right) + \frac{\bar{r}_T}{\sum_{t=0}^{T} \bar{r}_t}\mathcal{F}\left(\mu_{T}\right).
    \intertext{Finally, using also the inductive assumption, it follows that}
        \mathcal{F}(\tilde{\mu}_{T+1}) &\leq \left(1-\frac{\bar{r}_T}{\sum_{t=0}^{T} \bar{r}_t}\right)\mathcal{F}(\tilde{\mu}_{T}) + \frac{\bar{r}_T}{\sum_{t=0}^{T} \bar{r}_t}\mathcal{F}(\mu_{T}) \\
        &\leq \left(1-\frac{\bar{r}_T}{\sum_{t=0}^{T} \bar{r}_t}\right) \frac{1}{\sum_{t=0}^{T-1}\bar{r}_t} \sum_{t=0}^{T-1}\bar{r}_t\mathcal{F}(\mu_t) +\frac{\bar{r}_T}{\sum_{t=0}^{T} \bar{r}_t}\mathcal{F}(\mu_{T}) \\
        &= \frac{1}{\sum_{t=0}^{T}\bar{r}_t}\displaystyle\sum_{t=0}^{T-1}\bar{r}_{t}\mathcal{F}(\mu_t) + \frac{1}{\sum_{t=0}^{T} \bar{r}_t}\bar{r}_T\mathcal{F}(\mu_{T}) \\
        &=\frac{1}{\sum_{t=0}^{T}\bar{r}_t}\displaystyle\sum_{t=0}^{T}\bar{r}_{t}\mathcal{F}(\mu_t).
    \end{align}
    which completes the inductive step, and hence the proof.
\end{proof}

\begin{lemma}
\label{lemma:tilde-mu-t-bound-1-1}
    Suppose that Assumption 
    \ref{assumption:lsc} and \ref{assumption:geo-convex} hold. Let $(\tilde{\mu}_t)_{t\geq 0}$ be the sequence of measures defined according to \eqref{eq:tilde-mu-t-1} - \eqref{eq:tilde-mu-t-2}, where $(\mu_t)_{t\geq 0}$ denote the sequence of measures defined by \eqref{eq:wasserstein-sub-grad-descent}, the forward Euler discretization of the WGF in \eqref{eq:wasserstein-grad-flow}. Then, for all $T\geq 1$, it holds that 
    \begin{align}
        \sum_{t=0}^{T-1}\bar{r}_t \left[\mathcal{F}(\mu_t) - \mathcal{F}(\pi) \right]
        &\leq \sum_{t=0}^{T-1}\bar{r}_t\int_{\mathbb{R}^d}\langle \xi_t(x), x- \boldsymbol{t}_{\mu_t}^{\pi}(x) \rangle \,\mathrm{d}\mu_t(x). \label{eq:convex-ineq}
    \end{align}
\end{lemma}

\begin{proof}
    The result follows immediately from the geodesic convexity of $\mathcal{F}$.
\end{proof}

\begin{lemma}
\label{lemma:tilde-mu-t-bound-2}
    Suppose that Assumption \ref{assumption:lsc} and \ref{assumption:geo-convex} hold. Let $(\mu_t)_{t\geq 0}$ denote the sequence of measures defined by \eqref{eq:wasserstein-sub-grad-descent}, the forward Euler discretization of the WGF in \eqref{eq:wasserstein-grad-flow}, with $(\eta_t)_{t\geq 0}$ defined as in \eqref{eq:dog-lr}. Then, for all $T\geq 1$,
    \begin{equation}
        \sum_{t=0}^{T-1}\bar{r}_t\int_{\mathbb{R}^d}\langle \xi_t(x), x- \boldsymbol{t}_{\mu_t}^{\pi}(x) \rangle \,\mathrm{d}\mu_t(x) \leq \bar{r}_T (2\bar{d}_T+ \bar{r}_T) \sqrt{G_{T-1}}.
    \end{equation}
\end{lemma}

\begin{proof}
    Arguing as in the proof of Lemma \ref{lemma:evi}, we have 
    \begin{align}
        \frac{W^2_2(\mu_t,\pi) - W_2^2(\mu_{t+1},\pi)}{2\eta} &\geq \frac{1}{2\eta_t} \int_{\mathbb{R}^d} \left[\| x - \boldsymbol{t}_{\mu_t}^{\pi}(x)\|^2  - \|(x-\eta \xi_t(x)) - \boldsymbol{t}_{\mu_{t}}^{\pi}(x)\|^2\right]\mathrm{d}\mu_{t}(x) \\
        & =  \int_{\mathbb{R}^d} \left\langle \xi_{t}(x), x-\boldsymbol{t}_{\mu_t}^{\pi}(x) \right\rangle \,\mathrm{d}\mu_t(x) - \frac{\eta_t}{2}\int_{\mathbb{R}^d} \|\xi_t(x)\|^2\,\mathrm{d}\mu_t(x).
    \end{align}
    Rearranging and taking a weighted sum, we obtain
    \begin{align}
        \sum_{t=0}^{T-1} \bar{r}_t \int_{\mathbb{R}^d} \left\langle \xi_{t}(x), x-\boldsymbol{t}_{\mu_t}^{\pi}(x) \right\rangle \,\mathrm{d}\mu_t(x) &\leq \frac{1}{2}\underbrace{\sum_{t=0}^{T-1} \frac{\bar{r}_t}{\eta_t}\left[W^2_2(\mu_t,\pi) - W_2^2(\mu_{t+1},\pi)\right]}_{\mathrm{(I)}} \\[-2.5mm]
        &\hspace{10mm}+ \frac{1}{2} \underbrace{\sum_{t=0}^{T-1} \bar{r}_t\eta_t\int_{\mathbb{R}^d} \|\xi_t(x)\|^2\,\mathrm{d}\mu_t(x)}_{\mathrm{(II)}}.
    \end{align}
    It remains to bound (I) and (II), which we will do in turn. For (I), arguing as in \cite[][Lemma 1]{ivgi2023dog}, we have
    \begin{align}
        \mathrm{(I)}
        &=\sum_{t=0}^{T-1} \sqrt{G_t}\left[d_t^2 - d_{t+1}^2\right] \\
        &=d_0^2\sqrt{G_0} - d_{T}^2\sqrt{G_{T-1}} + \sum_{t=1}^{T-1} d_t^2\left[\sqrt{G_t} - \sqrt{G_{t-1}}\right] \\
        &\leq\bar{d}_{T}^2\sqrt{G_0} - d_{T}^2 \sqrt{G_{T-1}} + \bar{d}_{T}^2\sum_{t=1}^{T-1} \left[\sqrt{G_t} - \sqrt{G_{t-1}}\right] 
        \\[1mm]
        &= \sqrt{G_{T-1}}\left[\bar{d}_T^2 - d_{T}^2\right] \\[2mm]
        &\leq 4\bar{r}_T\bar{d}_T\sqrt{G_{T-1}}.
    \end{align}
    In the first inequality, we have used (a) $\bar{d}_t\leq \bar{d}_T$ for all $t\leq T$, and (b) $G_t$ is non-decreasing, which follows from \eqref{eq:defs1} - \eqref{eq:defs2}. In the second inequality, writing $S\in\argmax_{0\leq t\leq T}d_t$, we have used that $\smash{\bar{d}_T^2 - d_T^2= ({d}_S - d_T)({d}_S + d_T) \leq W_2(\mu_S,\mu_T) (d_S + d_T) \leq (\bar{r}_S+\bar{r}_T)(d_s+d_T)\leq  4\bar{r}_T\bar{d}_T}$. We now turn our attention to (II). In this case, we have that 
    \begin{align}
        \mathrm{(II)} 
        =\sum_{t=0}^{T-1}\frac{\bar{r}_t^2 \int_{\mathbb{R}^d} \|\xi_t(x)\|^2\,\mathrm{d}\mu_t(x)}{\sqrt{G}_t} \leq \bar{r}_T^2 \sum_{t=0}^{T-1}\frac{\int_{\mathbb{R}^d} \|\xi_t(x)\|^2\,\mathrm{d}\mu_t(x)}{\sqrt{G}_t} \leq 2\bar{r}_{T}^2 \sqrt{G_{T-1}},
    \end{align}
    where the first inequality follows from the fact that $\bar{r}_t$ is non-decreasing, and the second inequality is a consequence of \cite[][Lemma 4]{ivgi2023dog}, with $\smash{a_t = \int_{\mathbb{R}^d}\|\xi_t(x)\|^2\,\mathrm{d}\mu_t(x)}$, $\smash{G_t=\sum_{s=0}^{t} \int_{\mathbb{R}^d}\|\xi_s(x)\|^2\,\mathrm{d}\mu_s(x)}$. Combining our bounds for (I) and (II), we have the desired result.
\end{proof}

\begin{proposition}
\label{prop:dog-wgd-bound-1}
    Suppose that Assumption 
    \ref{assumption:lsc} and \ref{assumption:geo-convex} hold. Let $(\mu_t)_{t\geq 0}$ be the sequence of measures defined according to \eqref{eq:wasserstein-sub-grad-descent}, the forward Euler discretization of the WGF in \eqref{eq:wasserstein-grad-flow}, with $(\eta_t)_{t\geq 0}$ defined as in \eqref{eq:dog-lr}. Let $(\tilde{\mu}_t)_{t\geq 0}$ be the sequence of measures defined according to \eqref{eq:tilde-mu-t-1} - \eqref{eq:tilde-mu-t-2}. Then, for all $t\leq T$, we have 
    \begin{equation}
    \label{eq:dog-wgd-bound-1}
        \mathcal{F}(\tilde{\mu}_t) - \mathcal{F}(\pi) =  \mathcal{O}\left(\frac{(d_0 + \bar{r}_t)\sqrt{G_{t-1}}}{\sum_{s=0}^{t-1}\bar{r}_s/\bar{r}_t}\right).
    \end{equation}
\end{proposition}

\begin{proof}
    The result follows as an immediate consequence of Lemma \ref{lemma:tilde-mu-t-bound-1}, Lemma \ref{lemma:tilde-mu-t-bound-1-1}, and Lemma \ref{lemma:tilde-mu-t-bound-2}, using also the fact that $\bar{d}_t \leq d_0 + \bar{r}_t$.
\end{proof}

\begin{corollary}
\label{corr:dog-wgd-bound-2}
    Suppose that Assumption 
    \ref{assumption:lsc} and \ref{assumption:geo-convex} hold. Let $(\mu_t)_{t\geq 0}$ be the sequence of measures defined according to \eqref{eq:wasserstein-sub-grad-descent}, the forward Euler discretization of the WGF in \eqref{eq:wasserstein-grad-flow}, with $(\eta_t)_{t\geq 0}$ defined as in \eqref{eq:dog-lr}. Let $(\tilde{\mu}_t)_{t\geq 0}$ be the sequence of measures defined according to \eqref{eq:tilde-mu-t-1} - \eqref{eq:tilde-mu-t-2}. Let $D\geq d_0$, and define $G = \sup_{0\leq t\leq T} \|\xi_t\|_{L^2(\mu_t)}$. In addition, let $\smash{\tau\in\argmax_{0\leq t\leq T} \sum_{s=0}^{t-1} \bar{r}_s /\bar{r}_t}$. Then, on the event $\{\bar{r}_T \leq D\}$, it holds that
    \begin{equation}
        \mathcal{F}(\tilde{\mu}_{\tau}) - \mathcal{F}(\pi) = \mathcal{O}\left(\frac{DG}{\sqrt{T}}\log_{+}\left(\frac{D}{r_{\varepsilon}}\right)\right).
    \end{equation}
\end{corollary}

\begin{proof}
    First, using \citep[][Lemma 3]{ivgi2023dog}, we have that
    \begin{equation}
    \label{eq:ivgi-lemma-3}
        \max_{t\leq T} \sum_{s=0}^{t-1} \frac{\bar{r}_s}{\bar{r}_t} \geq \frac{1}{e}\left(\frac{T - \log_{+}(\bar{r}_T/\bar{r}_0)}{ \log_{+}(\bar{r}_T/\bar{r}_0)}\right)  \geq \frac{1}{e} \left( \frac{T}{\log_{+}(\bar{r}_T/\bar{r}_0)}\right) \geq \frac{1}{e} \left( \frac{T}{\log_{+}(D/r_\varepsilon)}\right),
    \end{equation}
    where in the final inequality we just consider the event $\{\bar{r}_T \leq D\}$. Combining \eqref{eq:ivgi-lemma-3} with \eqref{eq:dog-wgd-bound-1} in Proposition \ref{prop:dog-wgd-bound-1}, using the fact that $\smash{d_0\leq D}$ and that $\smash{\bar{r}_t \leq \bar{r}_T \leq D}$ on the event $\smash{\{\bar{r}_T \leq D\}}$, we have
    \begin{equation}
        \mathcal{F}(\tilde{\mu}_t) - \mathcal{F}(\pi) =\mathcal{O}\left(\frac{(d_0 + \bar{r}_t)\sqrt{G_{t-1}}}{\sum_{s=0}^{t-1}\bar{r}_s/\bar{r}_t}\right) =  \mathcal{O}\left(\frac{D \sqrt{G_{t-1}}}{T}\log_{+}\left(\frac{D}{r_{\varepsilon}}\right)\right).
    \end{equation}
    Finally, using that $\smash{\bar{g}_t \leq G}$ for all $t\in[T]$, which in turn implies that $G_{t-1} \leq G_{T-1} \leq T{G}^2$, we arrive at
    \begin{equation}
        \mathcal{F}(\tilde{\mu}_t) - \mathcal{F}(\pi)= \mathcal{O}\left(\frac{DG}{\sqrt{T}}\log_{+}\left(\frac{D}{r_{\varepsilon}}\right)\right).
    \end{equation}
\end{proof}

\paragraph{Uniform Averaging} As before, it is also possible to derive guarantees similar to those obtained in Proposition \ref{prop:dog-wgd-bound-1} and Corollary \ref{corr:dog-wgd-bound-2} for the unweighted average $\bar{\mu}_T$ defined in \eqref{eq:bar-mu-t-0} - \eqref{eq:bar-mu-t}. 

\begin{proposition}
\label{prop:dog-wgd-bound-1-uniform}
    Suppose that Assumption 
    \ref{assumption:lsc} and \ref{assumption:geo-convex} hold. Let $(\mu_t)_{t\geq 0}$ be the sequence of measures defined according to \eqref{eq:wasserstein-sub-grad-descent}, the forward Euler discretization of the WGF in \eqref{eq:wasserstein-grad-flow},  with $(\eta_t)_{t\geq 0}$ defined as in \eqref{eq:dog-lr}. Let $(\bar{\mu}_t)_{t\geq 0}$ be the sequence of measures defined according to \eqref{eq:bar-mu-t-0} - \eqref{eq:bar-mu-t}.  Then, for all $T\geq 1$, we have
    \begin{equation}
    \label{eq:dog-wgd-bound-1-uniform}
        \mathcal{F}(\bar{\mu}_T) - \mathcal{F}(\pi) =  \mathcal{O}\left(\frac{(d_0\log_{+}\frac{\bar{r}_{T}}{r_{\varepsilon}} + \bar{r}_{T})\sqrt{G_{T-1}}}{T}\right).
    \end{equation}
\end{proposition}

\begin{proof}
    Our proof follows the proof of \citet[][Proposition 3]{ivgi2023dog}. Define the times
    \begin{align}
        \tau_0&=0 \\
        \tau_{k+1} &= \min\left\{\left\{t\in\{\tau_{k}+1,\dots,T\}\,:\,\bar{r}_t\geq 2\bar{r}_{\tau_{k}}\right\}\cup\{T\}\right\}, \quad \quad k\geq 0.
    \end{align}
    Let $\smash{K= \min\{k:\tau_k=T\}}$. Then, for all $k\leq K-1$, it holds that $\bar{r}_{\tau_{k}}\geq 2\bar{r}_{\tau_{k-1}} \implies \bar{r}_{\tau_{k}} \geq 2^{k} \bar{r}_{\tau_0} = 2^{k}\bar{r}_{\varepsilon} \implies$ $\smash{\log_{2} \frac{\bar{r}_{\tau_k}}{\bar{r}_{\varepsilon}}\geq k}$. In particular, setting $k=K-1$, we have $\smash{K\leq 1+ \log_{2}\frac{\bar{r}_{\tau_{K-1}}}{\bar{r}_{\varepsilon}}}$. By construction, we also have that $\bar{r}_{\tau_{K-1}}\leq \bar{r}_{\tau_K}$ and $\bar{r}_{\tau_{K}} = \bar{r}_{T}$. Combining this with the previous inequality, we thus have $\smash{K \leq 1+ \log_{2}\frac{\bar{r}_{T}}{\bar{r}_{\varepsilon}}}$. Arguing as in the proof of Lemma \ref{lemma:tilde-mu-t-bound-2}, we have that
    \begin{align}
    \sum_{t={\tau_{k-1}}}^{\tau_{k}-1}\bar{r}_t \int_{\mathbb{R}^d}\langle \xi_t(x), x- \boldsymbol{t}_{\mu_t}^{\pi}(x) \rangle \,\mathrm{d}\mu_t(x) &\leq \bar{r}_T (2\bar{d}_{\tau_k}+ \bar{r}_{\tau_k}) \sqrt{G_{{\tau_k}-1}} =\mathcal{O}\left(\bar{r}_{\tau_k}(d_0 + \bar{r}_{\tau_k})\sqrt{G_{T-1}}\right) \label{eq:tilde-mu-t-bound-2-recall}
    \end{align}
    for any $k\in\{0,\dots,K\}$, where in the second line we have used the fact that $G_{\tau_k-1} \leq G_{(T)-1} = G_{T-1}$ as in the proof of Proposition \ref{prop:dog-wgd-bound-1}, and the fact that $\bar{d}_{\tau_k}\leq d_0 + \bar{r}_{\tau_k}$. It follows that 
    \begin{align}
        \sum_{t=\tau_{k-1}}^{\tau_k-1} \mathcal{F}(\mu_t) - \mathcal{F}(\pi) 
        &=\sum_{t=\tau_{k-1}}^{\tau_k-1} \frac{\bar{r}_t}{\bar{r}_t}\left[\mathcal{F}(\mu_t) - \mathcal{F}(\pi)\right] \leq \frac{1}{\bar{r}_{\tau_{k-1}}} \sum_{t=\tau_{k-1}}^{\tau_k-1} \bar{r}_t \left[\mathcal{F}(\mu_t) - \mathcal{F}(\pi)\right] \hspace{-5mm}  \\
        &\leq \frac{1}{\bar{r}_{\tau_{k-1}}}  \sum_{t={\tau_{k-1}}}^{\tau_{k}-1}\bar{r}_t \int_{\mathbb{R}^d}\langle \xi_t(x), x- \boldsymbol{t}_{\mu_t}^{\pi}(x) \rangle \,\mathrm{d}\mu_t(x) \hspace{-5mm}  \\
        &\leq \mathcal{O}\left(\frac{\bar{r}_{\tau_k}}{\bar{r}_{\tau_{k}-1}}(d_0 + \bar{r}_{\tau_k})\sqrt{G_{T-1}}\right) \hspace{-5mm} \label{eq:non-weight-bound-forward-euler}
    \end{align}
    where in the final line we have used the fact that $\smash{\bar{r}_{\tau_k-1} \leq 2\bar{r}_{\tau_{k-1}} \iff \frac{1}{\bar{r}_{\tau_{k-1}}}\leq \frac{2}{\bar{r}_{\tau_k-1}}}$. We require a bound for the ratio $\smash{\frac{\bar{r}_{\tau_{k}}}{\bar{r}_{\tau_k-1}}}$. For any $t\geq 1$, observe that
    \begin{align}
        \bar{r}_{t+1} & \leq \bar{r}_t + W_2(\mu_{t},\mu_{t+1}) 
        \leq \bar{r}_t + \eta_t\|\xi_t\|_{L^2(\mu_t)} \\
        &= \bar{r}_t + \frac{\bar{r}_t}{\sqrt{G_t}}\sqrt{G_t - G_{t-1}} = \bar{r}_t\Big(1+\sqrt{1-\frac{G_{t-1}}{G_t}}\Big) \leq 2\bar{r}_t. \label{eq:r_t1_bound-forward-euler}
    \end{align}Substituting \eqref{eq:r_t1_bound-forward-euler} into \eqref{eq:non-weight-bound-forward-euler}, we thus have that 
    \begin{equation}
    \sum_{t=\tau_{k-1}}^{\tau_k-1} \mathcal{F}(\mu_t) - \mathcal{F}(\pi)  = \mathcal{O}\left((d_0 + \bar{r}_{\tau_k})\sqrt{G_{T-1}}\right).
    \end{equation}
    Summing this display over $k=1,\dots,K$, it follows that 
    \begin{equation}
    \sum_{t=0}^{T-1} \Big[\mathcal{F}(\mu_t) - \mathcal{F}(\pi)\Big] 
    = \sum_{k=1}^{K}\Big[\sum_{t=\tau_{k-1}}^{\tau_k-1} \mathcal{F}(\mu_t) - \mathcal{F}(\pi)\Big]  = \mathcal{O}\Big(\Big(d_0K + \sum_{k=1}^{K}\bar{r}_{\tau_k}\Big)\sqrt{G_{T-1}}\Big). \label{eq:unweighted-penultimate-forward-euler}
    \end{equation}
    From above, we have that $K=\mathcal{O}(\log_{+}\frac{\bar{r}_{T}}{r_{\varepsilon}})$. In addition, from the definition, we have that ${\bar{r}_{\tau_{K-1}}}\geq 2^{K-1-k} \bar{r}_{\tau_{K-(K-k)}} \implies \bar{r}_{\tau_k} \leq 2^{-K+k+1}\bar{r}_{\tau_{K-1}} \leq 2^{-K+k+1} \bar{r}_{\tau_{K}}= 2^{-K+k+1} \bar{r}_{T}$ for all $0\leq k\leq K-1$. Thus, in particular,  it holds that $\sum_{k=1}^{K} \bar{r}_{\tau_k}\leq \bar{r}_{T} \sum_{k=1}^{K} 2^{-K+k+1} = \mathcal{O}(\bar{r}_{T})$. Substituting back into \eqref{eq:unweighted-penultimate-forward-euler}, and using Lemma \ref{lemma:average-bound}, we have that
    \begin{equation}
    \mathcal{F}(\bar{\mu}_T) - \mathcal{F}(\pi) \leq \frac{1}{T} \sum_{t=1}^{T} \left[\mathcal{F}(\mu_t) - \mathcal{F}(\pi)\right] 
    = \mathcal{O}\left(\frac{\left(d_0\log_{+}\frac{\bar{r}_{T}}{r_{\varepsilon}} + \bar{r}_{T}\right)\sqrt{G_{T-1}}}{T}\right). 
    \end{equation}
\end{proof}

\begin{corollary}
\label{corr:dog-wgd-bound-2-uniform}
    Suppose that Assumption 
    \ref{assumption:lsc} and \ref{assumption:geo-convex} hold. Let $(\mu_t)_{t\geq 0}$ be the sequence of measures defined according to \eqref{eq:wasserstein-sub-grad-descent}, the forward Euler discretization of the WGF in \eqref{eq:wasserstein-grad-flow},  with $(\eta_t)_{t\geq 0}$ defined as in \eqref{eq:dog-lr}. Let $(\bar{\mu}_t)_{t\geq 0}$ be the sequence of measures defined according to \eqref{eq:bar-mu-t-0} - \eqref{eq:bar-mu-t}. 
    Let $D\geq d_0$, and define $G = \sup_{0\leq t\leq T} \|\xi_t\|_{L^2(\mu_t)}$. Then, on the event $\{\bar{r}_{T} \leq D\}$, it holds that
    \begin{equation}
        \mathcal{F}(\bar{\mu}_{T}) - \mathcal{F}(\pi) = \mathcal{O}\left(\frac{DG}{\sqrt{T}}\log_{+}\left(\frac{D}{r_{\varepsilon}}\right)\right)
    \end{equation}
\end{corollary}

\begin{proof}
    On the event $\smash{\{\bar{r}_{T} \leq D\}}$, we have $\smash{\bar{r}_{T}\leq D}$, $\smash{d_0\leq D}$. We also have $\smash{\bar{g}_t \leq G}$ for all $1\leq t\leq T$, which implies that $G_{T} \leq \sum_{t=1}^{T}\bar{g}_t^2 \leq T{G}^2$. It follows that 
    \begin{equation}
        \mathcal{F}(\bar{\mu}_T) - \mathcal{F}(\pi) =\mathcal{O}\left(\frac{\left(D\log_{+}\frac{D}{r_{\varepsilon}} + D\right)\sqrt{TG^2}}{T}\right) = \mathcal{O}\left(\frac{DG}{\sqrt{T}}\log_{+}\left(\frac{D}{r_{\varepsilon}}\right)\right).
    \end{equation}
\end{proof}

\paragraph{Iterate Stability Bound.} Once again, following \citet{ivgi2023dog}, we can also define a variant of the dynamic step size formula $(\eta_t)_{t\geq 0}$ in \eqref{eq:dog-lr} which ensures that the iterates $(\mu_t)_{t\geq 0}$ remain bounded with high probability. In particular, we now define
\begin{equation}
    \eta_t = \frac{\max\left[r_\varepsilon,\max_{0\leq s\leq t} W_2(\mu_0,\mu_s)\right]}{\sqrt{8^4\log_{+}^2 (1+ t\bar{g}_t^2/\bar{g}_0^2)(\sum_{s=0}^{t-1} \int_{\mathbb{R}^d}\|\xi_s(x)\|^2\mathrm{d}\mu_s(x) + 16\bar{g}_t^2)}}:= \frac{\bar{r}_t}{\sqrt{G'_{t}}} \label{eq:tamed-dog-lr}
\end{equation}
where, similar to before, we have adopted the convention that $\sum_{s=0}^{-1}[\cdot] = 0$; we recall that $\bar{g}_t = \max_{0\leq s \leq t} \|\xi_t\|_{L^2(\mu_t)}$, and we have defined 
\begin{equation}
    G'_t = 8^4\log_{+}^2 \left(1+ \frac{t\bar{g}_t^2}{\bar{g}_0^2}\right)\left(G_{t-1} + 16\bar{g}_t^2\right).
\end{equation}
In this case, we have the following iterate stability guarantee.

\begin{proposition}
    Suppose that Assumption 
    \ref{assumption:lsc} and \ref{assumption:geo-convex} hold. Let $(\mu_t)_{t\geq 0}$ be the sequence of measures defined according to \eqref{eq:wasserstein-sub-grad-descent}, the forward Euler discretization of the WGF in \eqref{eq:wasserstein-grad-flow},  with $(\eta_t)_{t\geq 0}$ defined as in \eqref{eq:tamed-dog-lr}. Let $(\tilde{\mu}_t)_{t\geq 0}$ be the sequence of measures defined according to \eqref{eq:tilde-mu-t-1} - \eqref{eq:tilde-mu-t-2}. In addition, let $r_{\varepsilon}\leq 3d_0$. Then, for all $T\geq 1$, $\bar{r}_T\leq 3d_0$.
\end{proposition}

\begin{proof}
    From Lemma \ref{lemma:evi}, we have that
    \begin{align}
         {d_{s+1}^2 - d_s^2} &\leq 2\eta_s \left(\mathcal{F}(\pi) - \mathcal{F}(\mu_s)\right)  + \eta_s^2\int_{\mathbb{R}^d}\|\xi_s(x)\|^2\mathrm{d}\mu_s(x) \leq  \eta_s^2\int_{\mathbb{R}^d}\|\xi_s(x)\|^2\mathrm{d}\mu_s(x).
        \label{eq:evi-time-dep}
    \end{align}
    We now proceed by induction. The base case holds by assumption: $\bar{r}_0 = r_{\varepsilon} \leq 3d_0$. Suppose now that $\bar{r}_t\leq 3d_0$. We then have, via a telescoping sum, that 
    \begin{align}
        d_{t+1}^2  - d_0^2 = \sum_{s=0}^{t} \left(d_{s+1}^2 - d_{s}^2 \right) &\leq \sum_{s=0}^t \eta_s^2 \int_{\mathbb{R}^d}\|\xi_s(x)\|^2\mathrm{d}\mu_s(x) \\
         & =  \sum_{s=0}^t\frac{\bar{r}_s^2\int_{\mathbb{R}^d}\|\xi_s(x)\|^2\mathrm{d}\mu_s(x)}{G'_s} \\
         &= \sum_{s=0}^t\frac{\bar{r}_s^2\left(\sum_{i=0}^{s}\int_{\mathbb{R}^d}\|\xi_i(x)\|^2\mathrm{d}\mu_i(x) - \sum_{i=0}^{s-1}\int_{\mathbb{R}^d}\|\xi_i(x)\|^2\mathrm{d}\mu_i(x)\right)}{8^4\log_{+}^2 \left(1+ \frac{s\bar{g}_s^2}{\bar{g}_0^2}\right)\left(G_{s-1} + 16\bar{g}_s^2\right)} \\
         &\leq \frac{\bar{r}_t^2}{8^4}\sum_{s=0}^t\frac{G_{s} - G_{s-1}}{\log_{+}^2 \left(1+ \frac{s\bar{g}_s^2}{\bar{g}_0^2}\right)\left(G_{s-1} + 16\bar{g}_s^2\right)} \\
         &\leq \frac{\bar{r}_t^2}{8^4}\sum_{s=0}^t\frac{G_{s} - G_{s-1}}{\log_{+}^2 \left(\frac{G_s + \bar{g}_s^2}{\bar{g}_0^2}\right)\left(G_{s} + \bar{g}_s^2\right)} \label{eq:penultimate} \\
         &\leq \frac{\bar{r}_t^2}{8^4}\leq \frac{9d_0^2}{8^4} \label{eq:final}
    \end{align}
    where in the penultimate line \eqref{eq:penultimate} we have used the fact that, since $\bar{g}_s\geq \|\xi_s\|_{L^2(\mu_s)}$ for all $s$, it holds that
    \begin{align}
        \log_{+}^2 \left(1+ \frac{s\bar{g}_s^2}{\bar{g}_0^2}\right)\left(G_{s-1} + 16\bar{g}_s^2\right) &\geq \log_{+}^2 \left(\frac{\bar{g}_0^2 + \sum_{i=1}^s \bar{g}_s^2}{\bar{g}_0^2}\right)\left(G_{s-1} + \bar{g}_s^2 + \bar{g}_s^2\right) \\
        &\geq \log_{+}^2 \left(\frac{\sum_{i=0}^s \bar{g}_i^2}{\bar{g}_0^2}\right)\left(G_{s-1} + \|\xi_s\|_{L^2(\mu_k)}^2 + \bar{g}_s^2\right) \\
        &\geq \log_{+}^2 \left( \frac{\sum_{i=0}^s \|\xi_i\|_{L^2(\mu_i)}^2}{\bar{g}_0^2}\right)\left(G_{s} + \bar{g}_s^2\right) \\
        &= \log_{+}^2 \left( \frac{G_s}{\bar{g}_0^2}\right)\left(G_{s} + \bar{g}_s^2\right) 
    \end{align}
    and in the final line \eqref{eq:final} we have used \cite[][Lemma 6]{ivgi2023dog} for the first inequality; and the inductive assumption for the second inequality. Rearranging \eqref{eq:final}, we have that 
    \begin{equation}
        d_{t+1}^2 \leq \left(1 + \frac{9}{8^4}\right)d_0^2 \leq 2^2 d_0^2 
    \end{equation}
    which implies, in particular, that $d_{t+1}\leq 2d_0$. Thus, using the triangle inequality for the Wasserstein distance \citep[e.g.,][]{clement2008elementary}, we conclude, as required, that 
    \begin{equation}
        r_{t+1} := W_2(\mu_0,\mu_{t+1}) \leq W_2(\mu_0,\pi) + W_2(\mu_{t+1},\pi) := d_0 + d_{t+1} \leq 3d_0.
    \end{equation}
\end{proof}

\subsubsection{Stochastic Case: Adaptive Step Size}
We now turn our attention to the stochastic Wasserstein subgradient descent algorithm in \eqref{eq:wasserstein-stoc-sub-grad-descent}, \eqref{eq:WGF-stoc}, with an adaptive step size now given by
\begin{equation}
    \eta_t = \frac{\max\left[r_\varepsilon,\max_{0\leq s\leq t} W_2(\mu_0,\mu_s)\right]}{\sqrt{\sum_{s=0}^{t} \int_{\mathbb{R}^d}\|\hat{\xi}_s(x)\|^2\,\mathrm{d}\mu_s(x)}}. \label{eq:dog-lr-stoc}
\end{equation} 
Similar to before, cf. \eqref{eq:defs1} - \eqref{eq:defs2}, we will frequently make use of
\begin{alignat}{3}
    &r_t = W_2(\mu_0,\mu_t),\quad &&\bar{r}_t = \max_{0\leq s \leq t}\left[r_s,r_\varepsilon\right],\quad&&G_t= \sum_{s=0}^t \int_{\mathbb{R}^d} \|\hat{\xi}_s(x)\|^2\mathrm{d}\mu_s(x),\label{eq:defs1-stoc} \\
    &d_t = W_2(\mu_t,\pi),\quad &&\bar{d}_t = \max_{0\leq s \leq t} d_s
    ,\quad &&\bar{g}_t = \max_{0\leq s \leq t} \|\hat{\xi}_s\|_{L^2(\mu_s)}. \label{eq:defs2-stoc}
\end{alignat}
In addition, for our stochastic analysis we will require an additional definition:
\begin{equation}
    \theta_{t,\delta}:= \log \left(\frac{60\log (6t)}{\delta}\right).
\end{equation}
Once again, we will analyze the convergence of a weighted average $(\tilde{\mu}_t)_{t\geq 1}$ of the measures $(\mu_t)_{t\geq 0}$ output via the stochastic Wasserstein subgradient descent algorithm, as defined in \eqref{eq:tilde-mu-t-1} - \eqref{eq:tilde-mu-t-2}. 

\begin{lemma}
\label{lemma:tilde-mu-t-bound-1-stoch}
    Suppose that Assumption 
    \ref{assumption:lsc} and \ref{assumption:geo-convex} hold. Let $(\mu_t)_{t\geq 0}$ be the sequence of measures defined according to \eqref{eq:wasserstein-stoc-sub-grad-descent}, the stochastic-gradient forward Euler discretization of the WGF in \eqref{eq:wasserstein-grad-flow}, with $(\eta_t)_{t\geq 0}$ defined as in \eqref{eq:dog-lr-stoc}. Let $(\tilde{\mu}_t)_{t\geq 0}$ be the sequence of measures defined according to \eqref{eq:tilde-mu-t-1} - \eqref{eq:tilde-mu-t-2}. In addition, let  $\delta_t(x):=\hat{\xi}_t(x) - {\xi}_t(x)$. Then, for all $T\geq 1$, 
    \begin{align}
        \mathcal{F}(\tilde{\mu}_{T}) - \mathcal{F}(\pi) 
        &\leq \frac{1}{\sum_{t=0}^{T-1}\bar{r}_t} \bigg[\underbrace{\sum_{t=0}^{T-1}\bar{r}_t\int_{\mathbb{R}^d}\langle \hat{\xi}_t(x), x- \boldsymbol{t}_{\mu_t}^{\pi}(x) \rangle \,\mathrm{d}\mu_t(x)}_{\text{weighted regret}} 
 - \underbrace{\sum_{t=0}^{T-1}\bar{r}_t\int_{\mathbb{R}^d}\langle {\delta}_t(x), x- \boldsymbol{t}_{\mu_t}^{\pi}(x) \rangle \,\mathrm{d}\mu_t(x)}_{\text{noise}}\bigg]. \label{eq:decomp-stoc} \vspace{-8mm}
    \end{align}
\end{lemma}

\begin{proof}
    The result follows immediately from Lemma \ref{lemma:tilde-mu-t-bound-1} and Lemma \ref{lemma:tilde-mu-t-bound-1-1}.
\end{proof}

\begin{lemma}
\label{lemma:tilde-mu-t-bound-2-stoch}
    Suppose that Assumption 
    \ref{assumption:lsc} and \ref{assumption:geo-convex} hold. Let $(\mu_t)_{t\geq 0}$ be the sequence of measures defined according to \eqref{eq:wasserstein-stoc-sub-grad-descent}, the stochastic-gradient forward Euler discretization of the WGF in \eqref{eq:wasserstein-grad-flow}, with $(\eta_t)_{t\geq 0}$ defined as in \eqref{eq:dog-lr-stoc}. Let $(\tilde{\mu}_t)_{t\geq 0}$ be the sequence of measures defined according to \eqref{eq:tilde-mu-t-1} - \eqref{eq:tilde-mu-t-2}. Then, for all $T\geq 1$,
    \begin{equation}
        \sum_{t=0}^{T-1}\bar{r}_t\int_{\mathbb{R}^d}\langle \hat{\xi}_t(x), x- \boldsymbol{t}_{\mu_t}^{\pi}(x) \rangle \,\mathrm{d}\mu_t(x) \leq \bar{r}_T (2\bar{d}_T+ \bar{r}_T) \sqrt{G_{T-1}}\quad \text{a.s.}
    \end{equation}
\end{lemma}

\begin{proof}
    The proof is identical to the proof of Lemma \ref{lemma:tilde-mu-t-bound-2}, using the definitions in \eqref{eq:defs1-stoc} - \eqref{eq:defs2-stoc} in place of the definitions in \eqref{eq:defs1} - \eqref{eq:defs2}.
\end{proof}

\begin{lemma}
\label{lemma:tilde-mu-t-bound-3-stoch}
    Suppose that Assumption \ref{assumption:bounded-grads} holds. Let $(\mu_t)_{t\geq 0}$ be the sequence of measures defined according to \eqref{eq:wasserstein-stoc-sub-grad-descent}, the stochastic-gradient forward Euler discretization of the WGF in \eqref{eq:wasserstein-grad-flow}, with $(\eta_t)_{t\geq 0}$ defined as in \eqref{eq:dog-lr-stoc}. Let $(\tilde{\mu}_t)_{t\geq 0}$ be the sequence of measures defined according to \eqref{eq:tilde-mu-t-1} - \eqref{eq:tilde-mu-t-2}. Then, for all $0<\delta<1$, $G>0$, and $T\geq 1$, 
    \begin{align}
        &\mathbb{P}\left(\exists t \leq T : \left|\sum_{k=0}^{t-1} \bar{r}_k \int_{\mathbb{R}^d} \langle \delta_k(x),x - \boldsymbol{t}_{\mu_k}^{\pi}(x)\rangle\mathrm{d}\mu_k(x) \right| \geq 8 \bar{r}_{t-1}\bar{d}_{t-1}\sqrt{\theta_{t,\delta} G_{t-1} + \theta_{t,\delta}^2G^2}\right) \leq \delta + \mathbb{P}\big(\bar{g}_{T}> G\big).
    \end{align}
\end{lemma}

\begin{proof}
    For $0\leq t\leq T$, define the random variables
    \begin{align}
        Z_t &= \bar{r}_t \bar{d}_t, \quad X_t = \frac{1}{\bar{d}_t} \int_{\mathbb{R}^d} \langle \delta_t(x), x - \boldsymbol{t}_{\mu_t}^{\pi}(x) \rangle \,\mathrm{d}\mu_t(x), \quad \hat{X}_t = -\frac{1}{\bar{d}_t} \int_{\mathbb{R}^d} \langle \xi_t(x), x - \boldsymbol{t}_{\mu_t}^{\pi}(x) \rangle \,\mathrm{d}\mu_t(x).
    \end{align}
    Using these definitions, we have that
    \begin{align}
        \sum_{k=0}^{t-1} Z_k X_k &= \sum_{t=0}^{k-1}\bar{r}_k\int_{\mathbb{R}^d}\langle {\delta}_k(x), x- \boldsymbol{t}_{\mu_k}^{\pi}(x) \rangle \mathrm{d}\mu_k(x)
    \intertext{
    and also that 
    }
        \sum_{k=0}^{t-1} (X_k - \hat{X}_k)^2 &= \sum_{t=0}^{t-1} \left[\frac{1}{\bar{d}_k}\int_{\mathbb{R}^d}\langle \hat{\xi}_k(x), x- \boldsymbol{t}_{\mu_k}^{\pi}(x) \rangle \mathrm{d}\mu_k(x)\right]^2 \\
        &\leq \sum_{k=0}^{t-1} \frac{1}{\bar{d}_k^2}\int_{\mathbb{R}^d} \|\hat{\xi}_k(x)\|^2\mathrm{d}\mu_k(x) \int_{\mathbb{R}^d} \|x- \boldsymbol{t}_{\mu_k}^{\pi}(x)\|^2 \mathrm{d}\mu_k(x) \\
        &=\sum_{k=0}^{t-1} \frac{d_k^2}{\bar{d}_k^2}\int_{\mathbb{R}^d} \|\hat{\xi}_k(x)\|^2\mathrm{d}\mu_k(x)  \leq \sum_{k=0}^{t-1} \int_{\mathbb{R}^d} \|\hat{\xi}_k(x)\|^2\mathrm{d}\mu_k(x)= G_{t-1}
    \end{align}
    where the second line follows from the Cauchy-Schwarz inequality, and the third line from the definition of $\bar{d}_k$. It follows, substituting these (in)equalities, that
    \begin{align}
        &\mathbb{P}\left(\exists t \leq T : \left|\sum_{k=0}^{t-1} \bar{r}_k \int_{\mathbb{R}^d} \langle \delta_k(x),x - \boldsymbol{t}_{\mu_k}^{\pi}(x)\rangle\,\mathrm{d}\mu_k(x) \right| \geq 8 \bar{r}_{t-1}\bar{d}_{t-1}\sqrt{\theta_{t,\delta} G_{t-1} + \theta_{t,\delta}^2G^2}\right) \\
        =\hspace{.5mm}&\mathbb{P}\left(\exists t \leq T : \left|\sum_{k=0}^{t-1} Z_kX_k \right| \geq 8 Z_{t-1}\sqrt{\theta_{t,\delta} G_{t-1} + \theta_{t,\delta}^2G^2}\right) \\
        \leq\hspace{.5mm}&\mathbb{P}\left(\exists t \leq T : \left|\sum_{k=0}^{t-1} Z_kX_k \right| \geq 8 Z_{t-1}\sqrt{\theta_{t,\delta} \sum_{k=0}^{t-1}(X_k - \hat{X}_k)^2 + \theta_{t,\delta}^2G^2}\right) \\[1mm]
        \leq &\hspace{.5mm}\delta + \mathbb{P}(\exists t\leq T: g_t>G) \\[3mm]
        = \hspace{.5mm}&\delta + \mathbb{P}(\bar{g}_T > G)
    \end{align}
    where in penultimate line we have used \citep[][Lemma 7]{ivgi2023dog}, which we can apply since, with probability one, we can bound $|X_t|$ and $|\hat{X}_t|$, e.g.,
    \begin{align}
        |\hat{X}_t| 
        \leq \frac{1}{\bar{d}_t} \left[\int_{\mathbb{R}^d} \|\xi_t(x)\|^2 \,\mathrm{d}\mu_t(x)\right]^{\frac{1}{2}} \left[\int_{\mathbb{R}^d} \|x-t_{\mu_t}^{\pi}(x)\|^2 \,\mathrm{d}\mu_t(x)\right]^{\frac{1}{2}} 
        \leq \frac{d_t}{\bar{d}_t} \bar{g}_t \leq \bar{g}_t, 
    \end{align}
    where in the first inequality line we have once again made use of the Cauchy-Schwarz inequality, in the second inequality the definition of ${d}_t$, and in the final inequality the definition of $\bar{d}_t$ and $\bar{g}_t$.
\end{proof}

\begin{proposition}
\label{prop:dog-wgd-bound-1-stoc}
    Suppose that Assumptions 
    \ref{assumption:lsc}, \ref{assumption:geo-convex}, and \ref{assumption:bounded-grads} hold. Let $(\mu_t)_{t\geq 0}$ be the sequence of measures defined in \eqref{eq:wasserstein-stoc-sub-grad-descent}, the stochastic-gradient forward Euler discretization of the WGF in \eqref{eq:wasserstein-grad-flow}, with $(\eta_t)_{t\geq 0}$ defined as in \eqref{eq:dog-lr-stoc}. Let $(\tilde{\mu}_t)_{t\geq 0}$ be the sequence of measures defined according to \eqref{eq:tilde-mu-t-1} - \eqref{eq:tilde-mu-t-2}. Then, for all $0<\delta <1$, $G>0$, and for all $t\leq T$, we have with probability at least $1 - \delta - \mathbb{P}(\bar{g}_T>G)$, 
    \begin{equation}
    \label{eq:dog-wgd-bound-1-stoc}
        \mathcal{F}(\tilde{\mu}_t) - \mathcal{F}(\pi) =  \mathcal{O}\left(\frac{(d_0 + \bar{r}_t)\sqrt{G_{t-1} + G_{t-1}\theta_{t,\delta} + \theta_{t,\delta}^2 G^2}}{\sum_{s=0}^{t-1}\bar{r}_s/\bar{r}_t}\right).
    \end{equation}
\end{proposition}

\begin{proof}
    The result follows as an immediate consequence of Lemma \ref{lemma:tilde-mu-t-bound-1-stoch}, Lemma \ref{lemma:tilde-mu-t-bound-2-stoch}, and Lemma \ref{lemma:tilde-mu-t-bound-3-stoch}, using also the fact that $\bar{d}_t \leq d_0 + \bar{r}_t$.
\end{proof}

\begin{corollary}
    Suppose that Assumption 
    \ref{assumption:lsc}, \ref{assumption:geo-convex} and \ref{assumption:bounded-grads} hold. Let $(\mu_t)_{t\geq 0}$ be the sequence of measures defined in \eqref{eq:WGF-stoc}, with $(\eta_t)_{t\geq 0}$ defined as in \eqref{eq:dog-lr-stoc}. Let $(\tilde{\mu}_t)_{t\geq 0}$ be the sequence of measures defined according to \eqref{eq:tilde-mu-t-1} - \eqref{eq:tilde-mu-t-2}. Let $D\geq d_0$, and define $\smash{G = \sup_{0\leq t\leq T} \|\hat{\xi}_t\|}$. In addition, let $\smash{\tau\in\argmax_{t\leq T} \sum_{i=0}^{t-1} \bar{r}_i /\bar{r}_t}$. Then, with probability  at least $1 - \delta - \mathbb{P}(\bar{r}_T > D)$, it holds that
    \begin{equation}
        \mathcal{F}(\tilde{\mu}_{\tau}) - \mathcal{F}(\pi) = \mathcal{O}\left(\frac{D\sqrt{G_{t-1}\theta_{t,\delta} + G^2 \theta_{t,\delta}^2}}{T}\log_{+}\left(\frac{D}{r_{\varepsilon}}\right)\right)=\mathcal{O}\left(\frac{DG}{\sqrt{T}}\theta_{t,\delta}\log_{+}\left(\frac{D}{r_{\varepsilon}}\right)\right)
    \end{equation}
\end{corollary}

\begin{proof}
    The result follows from Proposition \ref{prop:dog-wgd-bound-1-stoc} in the same way that Corollary \ref{corr:dog-wgd-bound-2} follows from Proposition \ref{prop:dog-wgd-bound-1}.
\end{proof}

\subsection{Forward Euler Discretization, Smooth Setting}
\label{sec:forward-euler-smooth}
We now turn our attention to the smooth setting. Similar to Euclidean gradient descent -- i.e., the forward Euler discretization of the Euclidean gradient flow -- in this case we can improve the convergence rate from $\smash{\mathcal{O}(T^{-\frac{1}{2}})}$ to $\smash{\mathcal{O}(T^{-1})}$.

\subsubsection{Deterministic Case: Constant Step Size}
Similar to before, we begin our analysis in the case where the step size in \eqref{eq:wasserstein-sub-grad-descent} is constant, that is, $\eta_t=\eta$ for all $t\geq 0$. 

\begin{lemma}
\label{lemma:f-descent-smooth}
Suppose that Assumption \ref{assumption:lsc}
and \ref{assumption:smooth} hold. Let $(\mu_t)_{t\geq 0}$ denote the sequence of measures defined by \eqref{eq:wasserstein-sub-grad-descent}, the forward Euler discretization of the WGF in \eqref{eq:wasserstein-grad-flow}. Suppose that $\eta_t=\eta$ for all $t\geq 0$. Then, for any $t\geq 0$, we have 
\begin{equation}
    \mathcal{F}(\mu_{t+1}) \leq \mathcal{F}(\mu_t) - \eta\left(1-\frac{L\eta}{2}\right) \int_{\mathbb{R}^d}\|\xi_t(x)\|^2\,\mathrm{d}\mu_t(x).  \label{eq:f-descent-smooth}
\end{equation}
In particular, for $\eta\in(0,\frac{1}{L}]$, we have that 
\begin{equation}
    \mathcal{F}(\mu_{t+1}) \leq \mathcal{F}(\mu_t) - \frac{\eta}{2}\int_{\mathbb{R}^d}\|\xi_t(x)\|^2\,\mathrm{d}\mu_t(x).  \label{eq:f-descent-smooth-2}
\end{equation}
\end{lemma}

\begin{proof}
    The proof proceeds in a similar fashion to the proof of Lemma \ref{lemma:g-descent-smooth}. For each $s\in[0,1]$, define the family of transport maps $\Phi_{s}:\mathbb{R}^d\rightarrow\mathbb{R}^d$ according to 
    \begin{equation}
        \Phi_s(x) = x + s\left(T_{\eta}(x) - x\right), \quad \quad T_{\eta}(x) = x-\eta \xi_t(x).
    \end{equation}
    We clearly then have $\Phi_{0}(x) = x$ and $\Phi_1(x) = x - \eta \xi_t(x)$. In addition, define the curve $(\nu_s)_{s\in[0,1]}$ as the pushforward of $\mu_t$ under this transport map. That is, $\nu_s = (\Phi_s)_{\#}\mu_t$, with $\nu_0 = \mu_t$ and $\nu_1 = \mu_{t+1}$. This curve solves the continuity equation $\partial_s \nu_s + \nabla \cdot (\nu_s v_s) = 0$, with the velocity field $v_s:\mathbb{R}^d\rightarrow\mathbb{R}^d$ given by
    \begin{equation}
        v_s(x) = \partial_{s}\Phi_s\left(\Phi_s^{-1}(x)\right) = \left(T_{\eta} - \mathrm{id}\right) \circ \Phi_{s}^{-1}(x)
    \end{equation}
    Using the chain rule on the Wasserstein space \citep[e.g.,][Section 10.1.2]{ambrosio2008gradient}, we can then compute 
    \begin{equation}
        \frac{\mathrm{d}}{\mathrm{d}s}\mathcal{F}(\nu_s) = \int_{\mathbb{R}^d} \left\langle \nabla_{W_2}\mathcal{F}(\nu_s)(x), v_s(x) \right\rangle \nu_s(\mathrm{d}x)
    \end{equation}
    Thus, integrating over $s\in[0,1]$, before pulling back to the reference measure $\nu_0 = \mu_t$, we have that 
    \begin{align}
        \mathcal{F}(\nu_1) - \mathcal{F}(\nu_0)&=\int_{0}^{1} \left[ \int_{\mathbb{R}^d} \left\langle \nabla_{W_2}\mathcal{F}(\nu_s)(x), v_s(x)\right\rangle \nu_s(\mathrm{d}x)\right]\mathrm{d}s \\
        &=\int_0^1 \left[\int_{\mathbb{R}^d}\left\langle \nabla_{W_2}\mathcal{F}(\nu_s) \circ \Phi_s(x), v_s\circ \Phi_s(x) \right\rangle \, (\Phi_s^{-1})_{\#} \nu_s(\mathrm{d}x)\right]\mathrm{d}s \\
        &= \int_0^1 \left[\int_{\mathbb{R}^d} \left\langle \nabla_{W_2}\mathcal{F}(\nu_s) \circ \Phi_s(x), T_{\eta}(x) -x\right\rangle \nu_0(\mathrm{d}x) \right]\mathrm{d}s.
    \end{align}
    Adding and subtracting the Wasserstein gradient of $\mathcal{F}$ evaluated at the reference measure, viz $\nabla_{W_2}\mathcal{F}(\nu_0):=\nabla_{W_2}\mathcal{F}(\mu_t)$, inside the inner product, the previous display rewrites as 
    \begin{align}
        \mathcal{F}(\nu_1) - \mathcal{F}(\nu_0)&=\int_0^1 \left[\int_{\mathbb{R}^d} \left\langle \nabla_{W_2}\mathcal{F}(\nu_0) , T_{\eta}(x) - x\right\rangle \nu_0(\mathrm{d}x) \right]\mathrm{d}s \\
        &+\int_0^1 \left[\int_{\mathbb{R}^d} \left\langle \left(\nabla_{W_2}\mathcal{F}(\nu_s) \circ \Phi_s - \nabla_{W_2}\mathcal{F}(\nu_0)\right)(x), T_{\eta}(x) - x\right\rangle \nu_0(\mathrm{d}x) \right]\mathrm{d}s \nonumber \\
        &=-\eta \int_{\mathbb{R}^d} \| \xi_t(x)\|^2 \nu_0(\mathrm{d}x)  \\
        &+\int_0^1 \left[\int_{\mathbb{R}^d} \left\langle \left(\nabla_{W_2}\mathcal{F}(\nu_s) \circ \Phi_s - \nabla_{W_2}\mathcal{F}(\nu_0)\right)(x), T_{\eta}(x) - x\right\rangle \nu_0(\mathrm{d}x) \right]\mathrm{d}s \nonumber
    \end{align}
    where in the second line we have used the fact that $T_{\eta} - \mathrm{id} = -\eta \xi_t$, and $\nabla_{W_2}\mathcal{F}(\nu_0)=\nabla_{W_2}\mathcal{F}(\mu_t) = \xi_t$. For the second term, using Cauchy-Schwarz and the smoothness assumption (Assumption \ref{assumption:smooth}), we have 
    \begin{align}
        \mathrm{(II)} &\leq \int_0^1 \|\nabla_{W_2}\mathcal{F}(\nu_s) \circ \Phi_s - \nabla_{W_2}\mathcal{F}(\nu_0)\|_{L^2(\nu_0)} \|T_{\eta} - \mathrm{id}\|_{L^2(\nu_0)} \mathrm{d}s \\
        &\leq \int_0^1 L\|\Phi_s - \mathrm{id}\|_{L^2(\nu_0)} \|T_{\eta} - \mathrm{id}\|_{L^2(\nu_0)}\mathrm{d}s = L\int_0^1 s\|T_{\eta} - \mathrm{id}\|^2_{L^2(\nu_0)}\mathrm{d}s \\
        &=\frac{L}{2}\|T_{\eta} - \mathrm{id}\|^2_{L^2(\nu_0)} = \frac{L\eta^2}{2} \int_{\mathbb{R}^d}\|\xi_t(x)\|^2\,\mathrm{d}\mu_t(x),
    \end{align}
    where in the second line we have used the fact that $\|\Phi_{s} - \mathrm{id}\|_{L^2(\nu_0)} = s\|T_{\eta} - \mathrm{id}\|_{L^2(\nu_0)}$, and in the final line the fact that $\|T_{\eta} - \mathrm{id}\|^2_{L^2(\nu_0)} = \|\eta \xi_t\|_{L^2(\mu_t)}^2 = \eta^2 \|\xi_t\|_{L^2(\mu_t)}^2$. 
    Substituting this into the previous inequality, and using the fact that $\nu_0:=\mu_{t}$, $\nu_1:=\mu_{t+1}$, we arrive at precisely the required bound:
    \begin{align}
        \mathcal{F}(\mu_{t+1}) - \mathcal{F}(\mu_t) &\leq -\eta \int_{\mathbb{R}^d} \| \xi_t(x)\|^2 \,\mathrm{d}\mu_t(x) + \frac{L\eta^2}{2} \int \|\xi_t(x)\|^2\,\mathrm{d}\mu_t(x) \\
        &= -\eta\left(1-\frac{L\eta}{2}\right)\int_{\mathbb{R}^d}\|\xi_t(x)\|^2\,\mathrm{d}\mu_t(x).
    \end{align}
    The second bound follows straightforwardly, noting that $-(1-\frac{L\eta}{2})\leq-(1-\frac{L\frac{1}{L}}{2})=-(1-\frac{1}{2}) = -\frac{1}{2}$ for all $\eta\in(0,\frac{1}{L}]$.
\end{proof}

\begin{proposition}
\label{prop:f-descent-smooth}
   Suppose that Assumptions \ref{assumption:lsc}, \ref{assumption:geo-convex}, and \ref{assumption:smooth} hold. Let $(\mu_t)_{t\geq 0}$ denote the sequence of measures defined by \eqref{eq:wasserstein-sub-grad-descent}, the forward Euler discretization of the WGF in \eqref{eq:wasserstein-grad-flow}. Suppose that $\eta_t=\eta$ for all $t\geq 0$, with $\eta\in(0,\frac{1}{L}]$. Then, for all $\pi\in\mathcal{P}_2(\mathbb{R}^d)$, we have that 
    \begin{align}
    \label{eq:f-descent-smooth-prop}
        \mathcal{F}(\mu_{t+1}) - \mathcal{F}(\pi) \leq \frac{W^2_2(\mu_t,\pi) - W_2^2(\mu_{t+1},\pi)}{2\eta}. 
    \end{align}
\end{proposition}

\begin{proof}
From Lemma \ref{lemma:evi}, cf. \eqref{eq:evi}, we have that 
    \begin{equation}
        \mathcal{F}(\mu_t) - \mathcal{F}(\pi) \leq \frac{W^2_2(\mu_t,\pi) - W_2^2(\mu_{t+1},\pi)}{2\eta} + \frac{\eta}{2}\int_{\mathbb{R}^d}\|\xi_t(x)\|^2\,\mathrm{d}\mu_t(x).
        \label{eq:evi-recall}
    \end{equation}
This implies, in particular, that 
    \begin{equation}
        -\frac{\eta}{2} \int_{\mathbb{R}^d}\|\xi_t(x)\|^2\,\mathrm{d}\mu_t(x)\leq \frac{W^2_2(\mu_t,\pi) - W_2^2(\mu_{t+1},\pi)}{2\eta} - \left(\mathcal{F}(\mu_t) - \mathcal{F}(\pi)\right). \label{eq:squared-grad-bound}
    \end{equation}
Meanwhile, from the first inequality from Lemma \ref{lemma:f-descent-smooth}, we have that
   \begin{align}
       \mathcal{F}(\mu_{t+1}) \leq \mathcal{F}(\mu_t) - \frac{\eta}{2}\int_{\mathbb{R}^d}\|\xi_t(x)\|^2\,\mathrm{d}\mu_t(x).  \label{eq:f-descent-smooth-2-recall}
   \end{align}
Substituting \eqref{eq:squared-grad-bound} into \eqref{eq:f-descent-smooth-2-recall}, it follows that 
    \begin{align}
        \mathcal{F}(\mu_{t+1}) 
        &\leq\mathcal{F}(\mu_t) + \left[\frac{W^2_2(\mu_t,\pi) - W_2^2(\mu_{t+1},\pi)}{2\eta} - \left(\mathcal{F}(\mu_t) - \mathcal{F}(\pi)\right)\right] \\
        &=\frac{W^2_2(\mu_t,\pi) - W_2^2(\mu_{t+1},\pi)}{2\eta} + \mathcal{F}(\pi).
    \end{align}
    That is, after rearrangement, precisely the required inequality.
    \begin{equation}
       \mathcal{F}(\mu_{t+1}) - \mathcal{F}(\pi) \leq \frac{W^2_2(\mu_t,\pi) - W_2^2(\mu_{t+1},\pi)}{2\eta}.
    \end{equation}
\end{proof}

\begin{corollary}
\label{corollary:f-sum-bound-smooth}
    Suppose that Assumptions \ref{assumption:lsc}, \ref{assumption:geo-convex}, and \ref{assumption:smooth} hold. Let $(\mu_t)_{t\geq 0}$ denote the sequence of measures defined by \eqref{eq:wasserstein-sub-grad-descent}, the forward Euler discretization of the WGF in \eqref{eq:wasserstein-grad-flow}. Suppose that $\eta_t=\eta$ for all $t\geq 0$, with $\eta\in(0,\frac{1}{L}]$. Then, for all $\pi\in\mathcal{P}_2(\mathbb{R}^d)$, we have that
    \begin{equation}
        \sum_{t=1}^T \mathcal{F}(\mu_t) - \sum_{t=1}^T \mathcal{F}(\pi) \leq \frac{W_2^2(\mu_0,\pi) - W_2^2(\mu_{T},\pi)}{2\eta}
    \end{equation}
\end{corollary}

\begin{proof}
    The result is a straightforward consequence of the inequality obtained in Proposition \ref{prop:f-descent-smooth}, cf. \eqref{eq:f-descent-smooth}. In particular, summing \eqref{eq:f-descent-smooth}, and cancelling like terms in the telescoping sum, we have that 
    \begin{align}
        \sum_{t=1}^T \mathcal{F}(\mu_t) - \sum_{t=1}^T \mathcal{F}(\pi) &= \sum_{t=0}^{T-1} \mathcal{F}(\mu_{t+1}) - \sum_{t=0}^{T-1} \mathcal{F}(\pi) \\
        &\leq \sum_{t=0}^{T-1}\left[\frac{W^2_2(\mu_t,\pi) - W_2^2(\mu_{t+1},\pi)}{2\eta}\right] =\frac{W_2^2(\mu_0,\pi) - W_2^2(\mu_{T},\pi)}{2\eta}. 
    \end{align}
\end{proof}

\begin{proposition}
\label{prop:smooth-convergence-rate}
    Suppose that Assumptions \ref{assumption:lsc}, \ref{assumption:geo-convex}, and \ref{assumption:smooth} hold. Let $(\mu_t)_{t\geq 0}$ denote the sequence of measures defined by \eqref{eq:wasserstein-sub-grad-descent}, the forward Euler discretization of the WGF in \eqref{eq:wasserstein-grad-flow}. Suppose that $\eta_t=\eta$ for all $t\geq 0$, with $\eta\in(0,\frac{1}{L}]$. Let $\bar{\mu}_T$ be the average iterate defined according to \eqref{eq:bar-mu-t-0} - \eqref{eq:bar-mu-t}. Then, for all $\pi\in\mathcal{P}_2(\mathbb{R}^d)$, we have that
    \begin{equation}
        \mathcal{F}(\mu_T) - \mathcal{F}(\pi)\leq \frac{W_2^2(\mu_0,\pi)}{2\eta T}, \quad \quad \mathcal{F}(\bar{\mu}_T) - \mathcal{F}(\pi)\leq \frac{W_2^2(\mu_0,\pi)}{2\eta T}. \label{eq:average-bound-smooth}
    \end{equation}
\end{proposition}

\begin{proof}
    Using Corollary \ref{corollary:f-sum-bound-smooth}, and the fact that $\mathcal{F}(\mu_t)$ is decreasing, we have that 
    \begin{equation}
    \mathcal{F}(\mu_T) - \mathcal{F}(\pi) \leq \frac{1}{T}\left[\sum_{t=1}^T \mathcal{F}(\mu_t) - \sum_{t=1}^T\mathcal{F}(\pi)\right] \leq \frac{W_2^2(\mu_0,\pi)}{2\eta T}.
    \end{equation}
    Meanwhile, using Corollary \ref{corollary:f-sum-bound-smooth}, and Lemma \ref{lemma:average-bound}, we have that
    \begin{equation}
    \mathcal{F}(\bar{\mu}_T) - \mathcal{F}(\pi) \leq \frac{1}{T}\left[\sum_{t=1}^T \mathcal{F}(\mu_t) - \sum_{t=1}^T\mathcal{F}(\pi)\right] \leq \frac{W_2^2(\mu_0,\pi)}{2\eta T}.
    \end{equation}
\end{proof}

\begin{corollary}
\label{corollary:constant-step-smooth}
    Suppose that Assumptions \ref{assumption:lsc}, \ref{assumption:geo-convex}, and \ref{assumption:smooth} hold. Let $(\mu_t)_{t\geq 0}$ denote the sequence of measures defined by \eqref{eq:wasserstein-sub-grad-descent}, the forward Euler discretization of the WGF in \eqref{eq:wasserstein-grad-flow}. Suppose that $\eta_t=\eta$ for all $t\geq 0$, with $\eta\in(0,\frac{1}{L}]$. Let $\bar{\mu}_T$ be the average iterate defined in \eqref{eq:bar-mu-t-forward-flow-1-first-def} - \eqref{eq:bar-mu-t-forward-flow-2-first-def}. Then, for all $\smash{\pi\in\mathcal{P}_2(\mathbb{R}^d)}$, the upper bound in \eqref{eq:average-bound-smooth} is minimized when $\smash{\eta = \frac{1}{L}}$. Moreover, this choice of $\eta$ guarantees
    \begin{equation}
        \mathcal{F}(\mu_T) - \mathcal{F}(\pi)\leq \frac{LW_2^2(\mu_{0},\pi)}{2T} ,\quad\quad \mathcal{F}(\bar{\mu}_T) - \mathcal{F}(\pi) \leq  \frac{LW_2^2(\mu_{0},\pi)}{2T}. 
        \label{eq:optimal-average-bound-smooth}
    \end{equation}
\end{corollary}

\begin{proof}
The result follows immediately from Proposition \ref{prop:smooth-convergence-rate}. Clearly, the upper bound in \eqref{eq:average-bound-smooth} is decreasing as a function of the step size $\eta$. Thus, given the constraint that $\eta\in(0,\frac{1}{L}]$, it is minimized when $\eta=\frac{1}{L}$. Substituting back into \eqref{eq:average-bound-smooth} yields the bounds in \eqref{eq:optimal-average-bound-smooth}. 
\end{proof}

In fact, we can derive a more general bound which holds for all $\eta\in(0,\frac{2}{L})$. The following result provides the Wasserstein analogue of \citet[Proposition 2.1.14]{nesterov2018lectures}.

\begin{proposition}
\label{prop:f-descent-smooth-nesterov}
   Suppose that Assumptions \ref{assumption:lsc}, \ref{assumption:geo-convex}, and \ref{assumption:smooth} hold. Let $(\mu_t)_{t\geq 0}$ denote the sequence of measures defined by \eqref{eq:wasserstein-sub-grad-descent}, the forward Euler discretization of the WGF in \eqref{eq:wasserstein-grad-flow}. Suppose that $\eta_t=\eta$ for all $t\geq 0$, with $\eta\in(0,\frac{2}{L})$. Then, for all $\pi\in\mathcal{P}_2(\mathbb{R}^d)$, we have that 
    \begin{align}
    \label{eq:convergence-rate-smooth-nesterov}
    \mathcal{F}(\mu_T) - \mathcal{F}(\pi) \leq \frac{2W_2^2(\mu_0,\pi)\left(\mathcal{F}(\mu_0) - \mathcal{F}(\pi)\right)}{2W_2^2(\mu_0,\pi) + T\eta(2-L\eta) \left(\mathcal{F}(\mu_0)-\mathcal{F}(\pi)\right)}
\end{align}
\end{proposition}

\begin{proof}
First, using the geodesic convexity of $\mathcal{F}$, and the Cauchy-Schwarz inequality, we have that
\begin{align}
    \mathcal{F}(\mu_t) - \mathcal{F}(\pi) &\leq \int \left\langle \xi_t(x), x- \boldsymbol{t}_{\mu_t}^{\pi}(x)\right\rangle \,\mathrm{d}\mu_t(x) 
    \leq \left[\int \|\xi_t(x)\|^2\,\mathrm{d}\mu_t(x)\right]^{\frac{1}{2}}W_2(\mu_0,\pi).
\end{align}
Rearranging, and then substituting into the first inequality from Lemma \ref{lemma:f-descent-smooth}, cf. \eqref{eq:f-descent-smooth}, it follows that
\begin{align}
    \mathcal{F}(\mu_{t}) - \mathcal{F}(\mu_{t+1}) \geq \frac{\eta}{2}(2-L\eta)\int_{\mathbb{R}^d}\|\xi_t(x)\|^2\,\mathrm{d}\mu_t(x) \geq \frac{\eta}{2}(2-L\eta)\frac{\left(\mathcal{F}(\mu_t) - \mathcal{F}(\pi)\right)^2}{W_2^2(\mu_0,\pi)}
\end{align}
That is, introducing the notation $B_t:= \mathcal{F}(\mu_t) - \mathcal{F}(\pi)$, $\omega := \frac{\eta}{2}(2-L\eta)$, $W_k^2:= W_2^2(\mu_k,\pi)$, and $\alpha:= \frac{\omega}{W_0^2}$, the recursion
\begin{equation}
    B_{t} - B_{t+1} \geq \alpha B_t^2  \implies B_{t+1} \leq B_t(1-\alpha B_t).
\end{equation}
Taking the reciprocal of the latter inequality, it follows that 
\begin{equation}
    \frac{1}{B_{t+1}}\geq \frac{1}{B_t(1-\alpha B_t)} = \frac{1}{B_t} \, \frac{1}{1-\alpha B_t} \geq \frac{1}{B_t} + \alpha
\end{equation}
where in the final inequality we have used the fact that $\frac{1}{1-u}\geq 1+u$ for $u\geq 0$. Summing this inequality over $t\in[0,T-1]$, we arrive at
\begin{equation}
    \sum_{t=0}^{T-1} \left[\frac{1}{B_{t+1}} - \frac{1}{B_t}\right]  = \frac{1}{B_T} - \frac{1}{B_0} \geq \alpha T  \implies B_T\leq \frac{B_0}{1+\alpha T B_0} \
\end{equation}
That is, substituting back in the definitions
\begin{align}
    \mathcal{F}(\mu_T) - \mathcal{F}(\pi) \leq \frac{\mathcal{F}(\mu_0) - \mathcal{F}(\pi)}{1+\frac{\frac{\eta}{2}(2-L\eta)}{W_2^2(\mu_0,\pi)}T \left(\mathcal{F}(\mu_0) - \mathcal{F}(\pi)\right)}. 
\end{align}
which, after simplifying, is precisely the required bound.
\end{proof}

Once again, we find that the optimal step size is given by $\eta = \frac{1}{L}$. This is the subject of the next result.

\begin{corollary}
\label{corollary:constant-step-smooth-nesterov}
    Suppose that Assumptions \ref{assumption:lsc}, \ref{assumption:geo-convex}, and \ref{assumption:smooth} hold. Let $(\mu_t)_{t\geq 0}$ denote the sequence of measures defined by \eqref{eq:wasserstein-sub-grad-descent}, the forward Euler discretization of the WGF in \eqref{eq:wasserstein-grad-flow}. Suppose that $\eta_t=\eta$ for all $t\geq 0$, with $\eta\in(0,\frac{1}{L}]$. Let $\bar{\mu}_T$ be the average iterate defined in \eqref{eq:bar-mu-t-forward-flow-1-first-def} - \eqref{eq:bar-mu-t-forward-flow-2-first-def}. Then, for all $\smash{\pi\in\mathcal{P}_2(\mathbb{R}^d)}$, the upper bound in \eqref{eq:average-bound-smooth} is minimized when $\smash{\eta = \frac{1}{L}}$. Moreover, this choice of $\eta$ guarantees
    \begin{equation}
        \mathcal{F}(\mu_T) - \mathcal{F}(\pi)\leq \frac{2LW_2^2(\mu_{0},\pi)}{T+4}.
        \label{eq:optimal-average-bound-smooth-nesterov}
    \end{equation}
\end{corollary}

\begin{proof}
To maximise the RHS of \eqref{eq:convergence-rate-smooth-nesterov}, it is sufficient to maximise the function $\phi(\eta) = \eta(2-L\eta)$ with respect to $\eta$. Taking the derivative, and setting equal to zero, we have $\phi'(\eta) = 2-2L\eta = 0\implies \eta = \frac{1}{L}$. Substituting back into \eqref{eq:convergence-rate-smooth-nesterov}, we have that 
\begin{equation}
\mathcal{F}(\mu_T) - \mathcal{F}(\pi) \leq \frac{2LW_2^2(\mu_0,\pi)\left(\mathcal{F}(\mu_0) - \mathcal{F}(\pi)\right)}{2LW_2^2(\mu_0,\pi) + T \left(\mathcal{F}(\mu_0)-\mathcal{F}(\pi)\right)}. \label{eq:intermediate-nesterov-bound}
\end{equation}
Next, using the smoothness assumption (Assumption \ref{assumption:smooth}), and the fact that $\xi_{\pi}(x)=0$ $\pi$-a.e., we have that 
\begin{align}
    \mathcal{F}(\mu_0) - \mathcal{F}(\pi)&\leq \int\langle \xi_\pi(x), \boldsymbol{t}_{\mu_0}^{\pi}(x) - x\rangle \mathrm{d}\pi(x) + \frac{L}{2}W_2^2(\mu_0,\pi)\leq \frac{L}{2}W_2^2(\mu_0,\pi). \label{eq:init-pi-estimate}
\end{align}
Finally, substituting this into \eqref{eq:intermediate-nesterov-bound}, which is justified since the RHS is an increasing function of $\mathcal{F}(\mu_0) - \mathcal{F}(\pi)$, yields the required result. 
\end{proof}

\subsubsection{Deterministic Case: Adaptive Step Size}
Finally, we turn our attention to the case where the step size $(\eta_t)_{t\geq 0}$ is defined adaptively according to the formula in \eqref{eq:dog-lr}. In particular, we will show that -- up to logarithmic factors -- it is possible to recover the optimal convergence rate of Wasserstein gradient descent, cf. \eqref{eq:optimal-average-bound-smooth}.

\begin{lemma}
\label{lemma:smooth-grad-bound-forward-euler}
Suppose that Assumption \ref{assumption:lsc}
and \ref{assumption:smooth} hold. Let $(\mu_t)_{t\geq 0}$ denote the sequence of measures defined by \eqref{eq:wasserstein-sub-grad-descent}, the forward Euler discretization of the WGF in \eqref{eq:wasserstein-grad-flow}. 
Then 
\begin{align}
    \sum_{t=0}^{T-1} \int_{\mathbb{R}^d}\|\xi_t(x)\|^2\,\mathrm{d}\mu_t(x) &\leq 2L\sum_{t=0}^{T-1} \left(\mathcal{F}(\mu_t) - \mathcal{F}(\pi)\right)  \label{eq:smooth-grad-bound-1-forward-euler}
\end{align}
\end{lemma}

\begin{proof}
For $\eta>0$, let $\mu_{\eta}^{+} = (\mathrm{id} - \eta \xi(\mu))_{\#}\mu$. Then, for any $\mu\in\mathcal{P}_2(\mathcal{X})$, we have, arguing as in the proof of Lemma \ref{lemma:f-descent-smooth}, that
\begin{align}
    \mathcal{F}(\mu_{\eta}^{+}) \leq \mathcal{F}(\mu) 
     -\eta\int_{\mathbb{R}^d}\|\xi(\mu)(x)\|^2\mathrm{d}\mu(x) + L\frac{\eta^2}{2}\int_{\mathbb{R}^d}\|\xi(\mu)(x)\|^2\mathrm{d}\mu(x) 
\end{align}
Setting $\eta = \frac{1}{L}$, it follows that 
\begin{align}
    \mathcal{F}(\mu_{\eta}^{+}) &\leq \mathcal{F}(\mu) -\frac{1}{L}\int_{\mathbb{R}^d}\|\xi(\mu)(x)\|^2\mathrm{d}\mu(x) + L\frac{1}{2L^2}\int_{\mathbb{R}^d}\|\xi(\mu)(x)\|^2\mathrm{d}\mu(x) \\
    &= \mathcal{F}(\mu)-\frac{1}{2L}\int_{\mathbb{R}^d}\|\xi(\mu)(x)\|^2\mathrm{d}\mu(x). 
\end{align}
By definition, we have $\mathcal{F}(\pi) :=\inf_{\mu} \mathcal{F}(\mu) \leq \mathcal{F}(\mu_{\eta}^{+})$. Combining this with the previous display it follows that, for each $\mu\in\mathcal{P}_2(\mathcal{X})$, 
\begin{equation}
    \mathcal{F}(\pi)\leq \mathcal{F}(\mu) -\frac{1}{2L}\int_{\mathbb{R}^d}\|\xi(\mu)(x)\|^2\mathrm{d}\mu(x).
\end{equation}
Rearranging, summing over $t\in[1,T]$, and identifying $\xi(\mu_t):=\xi_t$, it follows straightforwardly that
\begin{align}
    \sum_{t=1}^T \int_{\mathbb{R}^d}\|\xi_t(x)\|^2\,\mathrm{d}\mu_t(x) &\leq 2
    L\sum_{t=1}^T \left(\mathcal{F}(\mu_t) - \mathcal{F}(\pi)\right).
\end{align}
\end{proof}

\begin{proposition}
\label{prop:dog-wgd-bound-1-smooth}
    Suppose that Assumption 
    \ref{assumption:lsc} and \ref{assumption:geo-convex}, and \ref{assumption:smooth} hold. Let $(\mu_t)_{t\geq 0}$ be the sequence of measures defined according to \eqref{eq:wasserstein-sub-grad-descent}, the forward Euler discretization of the WGF in \eqref{eq:wasserstein-grad-flow}, with $(\eta_t)_{t\geq 0}$ defined as in \eqref{eq:dog-lr}. Let $(\bar{\mu}_t)_{t\geq 0}$ be the sequence of measures defined according to \eqref{eq:bar-mu-t-0} - \eqref{eq:bar-mu-t}. Then, for all $t\leq T$, we have 
    \begin{equation}
    \label{eq:dog-wgd-bound-1-smooth}
        \mathcal{F}(\bar{\mu}_T) - \mathcal{F}(\pi) =  \mathcal{O}\left(\frac{L(d_0\log_{+}\frac{\bar{r}_T}{r_{\varepsilon}} + \bar{r}_t)^2}{T}\right).
    \end{equation}
\end{proposition}

\begin{proof}
    From Proposition \ref{prop:dog-wgd-bound-1-uniform} and Lemma \ref{lemma:smooth-grad-bound-forward-euler}, we have that
    \begin{equation}
    \label{eq:dog-wgd-bound-1-smooth-intermediate}
        \frac{1}{T}\sum_{t=0}^{T-1}\left(\mathcal{F}(\mu_t) - \mathcal{F}(\pi)\right) =  \mathcal{O}\left(\frac{(d_0\log_{+}\frac{\bar{r}_T}{r_{\varepsilon}} + \bar{r}_T)\sqrt{\frac{L}{T}\sum_{t=0}^{T-1}\left(\mathcal{F}(\mu_t) - \mathcal{F}(\pi)\right)}}{\sqrt{T}}\right).
    \end{equation}
    Dividing through by $\sqrt{\frac{1}{T}\sum_{t=0}^{T-1}\left(\mathcal{F}(\mu_t)- \mathcal{F}(\pi)\right)}$, and squaring, it follows that 
    \begin{equation}
    \label{eq:dog-wgd-bound-1-smooth-intermediate-2}
        \frac{1}{T}\sum_{t=0}^{T-1}\left(\mathcal{F}(\mu_t) - \mathcal{F}(\pi)\right) =  \mathcal{O}\left(\frac{{L}(d_0\log_{+}\frac{\bar{r}_T}{r_{\varepsilon}} + \bar{r}_T)^2}{T}\right).
    \end{equation}
    The result then follows from Lemma \ref{lemma:average-bound}.
\end{proof}

\begin{corollary}
\label{corr:dog-wgd-bound-2-smooth}
    Suppose that Assumption 
    \ref{assumption:lsc} and \ref{assumption:geo-convex}, and \ref{assumption:smooth} hold. Let $(\mu_t)_{t\geq 0}$ be the sequence of measures defined according to \eqref{eq:wasserstein-sub-grad-descent}, the forward Euler discretization of the WGF in \eqref{eq:wasserstein-grad-flow}, with $(\eta_t)_{t\geq 0}$ defined as in \eqref{eq:dog-lr}. Let $(\bar{\mu}_t)_{t\geq 0}$ be the sequence of measures defined according to \eqref{eq:bar-mu-t-0} - \eqref{eq:bar-mu-t}. Then, with probability at least $1 - \mathbb{P}(\bar{r}_T > D)$, it holds that
    \begin{equation}
        \mathcal{F}(\bar{\mu}_{T}) - \mathcal{F}(\pi) = \mathcal{O}\left(\frac{LD^2}{T}\log_{+}^2\left(\frac{D}{r_{\varepsilon}}\right)\right)
    \end{equation}
\end{corollary}

\begin{proof}
Using the fact that $\smash{d_0\leq D}$ and that $\smash{\bar{r}_t \leq \bar{r}_T \leq D}$ on the event $\smash{\{\bar{r}_T \leq D\}}$, it follows from Proposition \ref{prop:dog-wgd-bound-1-smooth} that
    \begin{equation}
        \mathcal{F}(\bar{\mu}_T) - \mathcal{F}(\pi) =\mathcal{O}\left(\frac{L(d_0\log_{+}\frac{\bar{r}_T}{r_{\varepsilon}} + \bar{r}_t)^2}{T}\right) =  \mathcal{O}\left(\frac{LD^2}{T}\log_{+}^2\left(\frac{D}{r_{\varepsilon}}\right)\right).
    \end{equation}
\end{proof}

\section{Additional Proofs: Forward-Flow Discretization}
\label{app:additional-proofs}
In this appendix, we provide the proofs of our main results for the forward-flow discretization of the WGF (see Section \ref{sec:main-results}).

\subsection{Forward-Flow Discretization, Nonsmooth Setting}
\label{app:additional-proofs-forward-flow-non-smooth}

\subsubsection{Deterministic Case: Constant Step Size}
\label{app:additional-proofs-forward-flow-non-smooth-deterministic-constant}
\begin{proof}[Proof of Lemma \ref{lemma:average}]
    It is sufficient to show that $\mathcal{F}(\bar{\mu}_{T})\leq \frac{1}{T}\sum_{t=1}^{T} \mathcal{F}(\mu_t)$. We will prove this result by induction over $T\in\mathbb{N}$. First consider $T=1$. By direct computation, we have $\bar{\mu}_1 = \mu_1$, and so $\smash{\mathcal{F}(\bar{\mu}_1) = \mathcal{F}(\mu_1)  = }$ $\smash{\frac{1}{T}\sum_{t=1}^{1}\mathcal{F}(\mu_t)}$. This proves the base case. Suppose, now, that the result holds for all $t\in\{1,\dots,T\}$. Let $\mu\in\mathcal{P}_2(\mathbb{R}^d)$ and $\nu\in\mathcal{P}_2(\mathbb{R}^d)$. In addition, let $\smash{\xi_{\nu} = \boldsymbol{t}_{\mu}^{\nu} -\mathrm{id}\in\mathcal{T}_{\mu}\mathcal{P}_2(\mathbb{R}^d)}$. Define $\smash{\lambda_{\eta}^{\mu\rightarrow \nu}:[0,1]\rightarrow\mathcal{P}_2(\mathbb{R}^d)}$ as a constant speed geodesic with $\lambda_{0}^{\mu\rightarrow\nu} = \mu$ and $\dot{\lambda_{\mu}^{\nu}}(0) = \xi_{\nu}$. That is,
    \begin{equation}
        \lambda_{\mu}^{\nu}(s) = \left[\mathrm{id} + s \xi_{\nu} \right]_{\#}\mu = \left[\mathrm{id} + s\left(\boldsymbol{t}_{\mu}^{\nu} - \mathrm{id}\right)\right]_{\#}\mu.
    \end{equation} 
     It follows, working from the definition in 
     \eqref{eq:bar-mu-t-forward-flow-1-first-def} - \eqref{eq:bar-mu-t-forward-flow-2-first-def}, that
    \begin{align}
        \bar{\mu}_{T+1} &= \left[\left(1-\frac{1}{T+1}\right)\mathrm{id} + \frac{1}{T+1}\boldsymbol{t}_{\bar{\mu}_{T}}^{\mu_{T+1}}\right]\bar{\mu}_{T}  \\
        &= \left[\mathrm{id} + \frac{1}{T+1}\left(\boldsymbol{t}_{\bar{\mu}_{T}}^{\mu_{T+1}} - \mathrm{id}\right)\right]_{\#}\bar{\mu}_{T} =\lambda_{\bar{\mu}_{T}}^{\mu_{T+1}}\left(\frac{1}{T+1}\right).
    \end{align}
    Using the geodesic convexity of $\mathcal{F}$, we then have
    \begin{align}
        \mathcal{F}(\bar{\mu}_{T+1})&=\mathcal{F}\left(\lambda_{\bar{\mu}_{T}}^{\mu_{T+1}}\left(\frac{1}{T+1}\right)\right)  \\
        &\leq \left(1-\frac{1}{T+1}\right) \mathcal{F}\left(\lambda_{\bar{\mu}_{T}}^{\mu_{T+1}}\left(0\right)\right) + \frac{1}{T+1}\mathcal{F}\left(\lambda_{\bar{\mu}_{T}}^{\mu_{T+1}}\left(1\right)\right) \\
        &=\left(1-\frac{1}{T+1}\right) \mathcal{F}\left(\bar{\mu}_{T}\right) + \frac{1}{T+1}\mathcal{F}\left(\mu_{T+1}\right).
    \end{align}
    Finally, using also the inductive assumption, 
    \begin{align}
        \mathcal{F}(\bar{\mu}_{T+1}) &\leq \left(1-\frac{1}{T+1}\right)\mathcal{F}(\bar{\mu}_{T}) + \frac{1}{T+1}\mathcal{F}(\mu_{T+1}) \\
        &\leq \left(1-\frac{1}{T+1}\right) \frac{1}{T}\displaystyle\sum_{t=1}^{T}\mathcal{F}(\mu_t) +\frac{1}{T+1}\mathcal{F}(\mu_{T+1}) \\
        &= \frac{1}{T}\displaystyle\sum_{t=1}^{T-1}\mathcal{F}(\mu_t) + \frac{1}{T}\mathcal{F}(\mu_T) = \frac{1}{T+1}\displaystyle\sum_{t=1}^{T+1}\mathcal{F}(\mu_t).
    \end{align}
    This completes the inductive step, and thus the proof.
\end{proof}

\begin{proof}[Proof of Lemma \ref{lemma:evi-forward-flow-transport}]
Using the definition of the Wasserstein-2 distance, we have that
    \begin{align}
        W_2^2(\mu_{t},\pi) &=\inf_{\gamma \in \Gamma(\mu_{t},\pi)} \int_{\mathbb{R}^d\times\mathbb{R}^d}\|x - y\|^2\mathrm{d}\gamma(x,y)= \int_{\mathbb{R}^d} \| x - \boldsymbol{t}_{\mu_t}^{\pi}(x)\|^2 \,\mathrm{d}\mu_t(x) \\
        W_2^2(\mu_{t+\frac{1}{2}},\pi) &= \inf_{\gamma \in \Gamma(\mu_{t+\frac{1}{2}},\pi)} \int_{\mathbb{R}^d\times\mathbb{R}^d}\|x - y\|^2\mathrm{d}\gamma(x,y)\leq \int_{\mathbb{R}^d}\|(x-\eta \zeta_t(x)) - \boldsymbol{t}_{\mu_{t}}^{\pi}(x)\|^2\mathrm{d}\mu_{t}(x),
    \end{align}
    Putting these two together, and expanding, it follows straightforwardly that
    \begin{align}
        \frac{W^2_2(\mu_t,\pi) - W_2^2(\mu_{t+\frac{1}{2}},\pi)}{2\eta} &\geq \frac{1}{2\eta} \int_{\mathbb{R}^d} \left[\| x - \boldsymbol{t}_{\mu_t}^{\pi}(x)\|^2  - \|(x-\eta \zeta_t(x)) - \boldsymbol{t}_{\mu_{t}}^{\pi}(x)\|^2\right]\mathrm{d}\mu_{t}(x) \label{eq:intermidate-result-g-1} \\
        & =  \int_{\mathbb{R}^d} \left\langle \zeta_{t}(x), x-\boldsymbol{t}_{\mu_t}^{\pi}(x) \right\rangle \,\mathrm{d}\mu_t(x) - \frac{\eta}{2}\int_{\mathbb{R}^d} \|\zeta_t(x)\|^2\,\mathrm{d}\mu_t(x) \\
        &\geq \mathcal{E}(\mu_t) - \mathcal{E}(\pi) - \frac{\eta}{2}\int_{\mathbb{R}^d} \|\zeta_t(x)\|^2\,\mathrm{d}\mu_t(x) \label{eq:intermidate-result-g-3} 
    \end{align}
    where in the final line we have used the fact that $\mathcal{E}$ is geodesically convex. Rearranging, we thus have that 
    \begin{equation}
        \mathcal{E}(\mu_t) - \mathcal{E}(\pi) \leq \frac{W^2_2(\mu_t,\pi) - W_2^2(\mu_{t+\frac{1}{2}},\pi)}{2\eta} + \frac{\eta}{2}\int_{\mathbb{R}^d} \|\zeta_t(x)\|^2\,\mathrm{d}\mu_t(x). \label{eq:g-bound}
    \end{equation}
\end{proof}

\subsubsection{Deterministic Case: Adaptive Step Size}
\label{app:additional-proofs-forward-flow-non-smooth-deterministic-adaptive}
\begin{proof}[Proof of Lemma \ref{lemma:tilde-mu-t-bound-1-forward-flow}].
    The proof is essentially identical to the proof of Lemma \ref{lemma:average}. We prove the result by induction over $T\in\mathbb{N}$. First consider $T=1$. By direct computation, we have $\tilde{\mu}_1 = \mu_1$, and so $\smash{\mathcal{F}(\tilde{\mu}_1) = \mathcal{F}(\mu_1)  = }$ $\smash{\frac{1}{T}\sum_{t=1}^{1}\mathcal{F}(\mu_t)}$. This proves the base case. Suppose, now, that \eqref{eq:average-ineq-forward-flow} holds for all $t\in\{1,\dots,T\}$. Let $\mu\in\mathcal{P}_2(\mathbb{R}^d)$ and $\nu\in\mathcal{P}_2(\mathbb{R}^d)$. In addition, let $\xi_{\nu} = \boldsymbol{t}_{\mu}^{\nu} -\mathrm{id}\in\mathcal{T}_{\mu}\mathcal{P}_2(\mathbb{R}^d)$. Define $\lambda_{\mu}^{\nu}(\cdot):[0,1]\rightarrow\mathcal{P}_2(\mathbb{R}^d)$ as a constant speed geodesic with $\lambda_{\mu}^{\nu}(0) = \mu$ and $\dot{\lambda}_{\mu}^{\nu}(0) = \xi_{\nu}$. That is,
    \begin{equation}
        \lambda_{\mu}^{\nu}(s) = \left[\mathrm{id} + s \xi_{\nu} \right]_{\#}\mu = \left[\mathrm{id} + s\left(\boldsymbol{t}_{\mu}^{\nu} - \mathrm{id}\right)\right]_{\#}\mu.
    \end{equation} 
     It follows, working from the definition in \eqref{eq:tilde-mu-t-1-forward-flow} - \eqref{eq:tilde-mu-t-2-forward-flow}, that
    \begin{align}
        \tilde{\mu}_{T+1} &= \left[\left(1-\frac{\bar{r}_T}{\sum_{t=1}^{T+1} \bar{r}_t}\right)\mathrm{id} + \frac{\bar{r}_T}{\sum_{t=1}^{T+1} \bar{r}_t}\boldsymbol{t}_{\tilde{\mu}_{T}}^{\mu_{T+1}}\right]_{\#}\tilde{\mu}_{T} \\
        &= \left[\mathrm{id} + \frac{\bar{r}_T}{\sum_{t=1}^{T+1} \bar{r}_t}\left(\boldsymbol{t}_{\tilde{\mu}_{T}}^{\mu_{T+1}} - \mathrm{id}\right)\right]_{\#}\tilde{\mu}_{T} =\lambda_{\tilde{\mu}_{T}}^{\mu_{T+1}}\left(\frac{\bar{r}_T}{\sum_{t=1}^{T+1} \bar{r}_t}\right).
    \end{align}
    Using the geodesic convexity of $\mathcal{F}$, we then have
    \begin{align}
        \mathcal{F}(\tilde{\mu}_{T+1})&=\mathcal{F}\left(\lambda_{\tilde{\mu}_{T}}^{\mu_{T+1}}\left(\frac{\bar{r}_T}{\sum_{t=1}^{T+1} \bar{r}_t}\right)\right) \\
        &\leq \left(1-\frac{\bar{r}_T}{\sum_{t=1}^{T+1} \bar{r}_t}\right) \mathcal{F}\left(\lambda_{\tilde{\mu}_{T}}^{\mu_{T+1}}\left(0\right)\right) + \frac{\bar{r}_T}{\sum_{t=1}^{T+1} \bar{r}_t}\mathcal{F}\left(\lambda_{\tilde{\mu}_{T}}^{\mu_{T+1}}\left(1\right)\right) \\
        &=\left(1-\frac{\bar{r}_T}{\sum_{t=1}^{T+1} \bar{r}_t}\right) \mathcal{F}\left(\tilde{\mu}_{T}\right) + \frac{\bar{r}_T}{\sum_{t=1}^{T+1} \bar{r}_t}\mathcal{F}\left(\mu_{T+1}\right).
    \end{align}
    Finally, using also the inductive hypothesis,
    \begin{align}
        \mathcal{F}(\tilde{\mu}_{T+1}) &\leq \left(1-\frac{\bar{r}_T}{\sum_{t=1}^{T+1} \bar{r}_t}\right)\mathcal{F}(\tilde{\mu}_{T}) + \frac{\bar{r}_T}{\sum_{t=1}^{T+1} \bar{r}_t}\mathcal{F}(\mu_{T+1}) \\
        &\leq \left(1-\frac{\bar{r}_T}{\sum_{t=1}^{T+1} \bar{r}_t}\right) \frac{1}{\sum_{t=1}^{T}\bar{r}_t} \sum_{t=1}^{T}\bar{r}_t\mathcal{F}(\mu_t) +\frac{\bar{r}_T}{\sum_{t=1}^{T+1} \bar{r}_t}\mathcal{F}(\mu_{T+1}) \\
        &= \frac{1}{\sum_{t=1}^{T+1}\bar{r}_t}\displaystyle\sum_{t=1}^{T}\bar{r}_{t}\mathcal{F}(\mu_t) + \frac{1}{\sum_{t=1}^{T+1} \bar{r}_t}\bar{r}_T\mathcal{F}(\mu_{T+1}) \\
        &=\frac{1}{\sum_{t=1}^{T+1}\bar{r}_t}\displaystyle\sum_{t=1}^{T+1}\bar{r}_{t}\mathcal{F}(\mu_t).
    \end{align}
    This completes the inductive step, and hence the proof.
\end{proof}

\begin{proof}[Proof of Lemma \ref{lemma:tilde-mu-t-bound-2-forward-flow-v2}]
    We will bound the two terms on the LHS separately, for now taking out the factor $\frac{1}{2}$. For the first term, 
    arguing similarly to \citet[][Lemma 1]{ivgi2023dog}, we have
    \begin{align}
        &\sum_{t=1}^{T} \frac{\bar{r}_t}{\eta_t}\left[d_t^2 - d_{t+1}^2\right] 
        =\sum_{t=1}^{T} \sqrt{G_t}\left[d_t^2 - d_{t+1}^2\right] 
        =d_1^2\sqrt{G_1} - d_{T+1}^2\sqrt{G_{T}} + \sum_{t=2}^{T} d_t^2\left[\sqrt{G_t} - \sqrt{G_{t-1}}\right] \\[-1mm]
        &\leq\bar{d}_{T+1}^2\sqrt{G_1} - d_{T+1}^2 \sqrt{G_{T}} + \bar{d}_{T+1}^2\sum_{t=2}^{T} \left[\sqrt{G_t} - \sqrt{G_{t-1}}\right] 
        = \sqrt{G_{T}}\left[\bar{d}_{T+1}^2 - d_{T+1}^2\right] \leq 4\bar{r}_{T+1}\bar{d}_{T+1}\sqrt{G_{T}}.
    \end{align}
    In the first inequality, we have used (a) ${d}_{t}\leq d_{T+1}$ for all $t\leq T+1$, and (b) $G_t$ is non-decreasing, which follows from the definitions in \eqref{eq:defs1-forward-flow} - \eqref{eq:defs2-forward-flow}. In the second inequality, writing $S\in\argmax_{1\leq t\leq T+1}d_t$, we have used the fact that 
    \begin{align}
        \bar{d}_{T+1}^2 - d_{T+1}^2 &= (\bar{d}_{T+1} - d_{T+1})(\bar{d}_{T+1} + d_{T+1}) = ({d}_S - d_{T+1})({d}_S + d_{T+1}) \\
        &\leq W_2(\mu_{S-\frac{1}{2}},\mu_{T+\frac{1}{2}}) (d_S + d_{T+1})
        \leq (\bar{r}_S+\bar{r}_{T+1})(d_S+d_{T+1})\leq  4\bar{r}_{T+1}\bar{d}_{T+1}.
    \end{align}
    We now turn our attention to the second term. In this case, we have that 
    \begin{align}
        \sum_{t=1}^T \bar{r}_t\eta_t \hspace{-1mm} \int_{\mathbb{R}^d} \|\zeta_t(x)\|^2\,\mathrm{d}\mu_t(x) =\sum_{t=1}^{T}\frac{\bar{r}_t^2 \int_{\mathbb{R}^d} \|\zeta_t(x)\|^2\,\mathrm{d}\mu_t(x)}{\sqrt{G_t}} \leq \bar{r}_{T+1}^2 \sum_{t=1}^{T}\frac{\int_{\mathbb{R}^d} \|\zeta_t(x)\|^2\,\mathrm{d}\mu_t(x)}{\sqrt{G_t}} \leq 2\bar{r}_{T+1}^2 \sqrt{G_{T}},
    \end{align}
    where the first inequality follows from the fact that $\bar{r}_t$ is non-decreasing, and the second inequality is a consequence of a standard summation result (e.g., \citealp{shalev2012online}, Lemma 2.8; \citealp{ivgi2023dog}, Lemma 4), with $\smash{a_t = \int_{\mathbb{R}^d}\|\zeta_t(x)\|^2\,\mathrm{d}\mu_t(x),~G_t = \sum_{i=1}^{t} \int_{\mathbb{R}^d} \|\zeta_i(x)\|^2\mathrm{d}\mu_i(x)}$. Combining our bounds for (I) and (II), we have the desired result.
\end{proof}

\begin{proof}[Proof of Corollary \ref{corr:dog-wgd-bound-2-forward-flow}]
    First, using \citep[][Lemma 3]{ivgi2023dog}, we have that
    \begin{equation}
    \label{eq:ivgi-lemma-3-forward-flow}
        \max_{1\leq t\leq T} \sum_{s=1}^{t} \frac{\bar{r}_s}{\bar{r}_{t+1}} \geq \frac{1}{e}\left(\frac{T - \log_{+}(\bar{r}_{T+1}/\bar{r}_1)}{ \log_{+}(\bar{r}_{T+1}/\bar{r}_1)}\right)  = \frac{1}{e} \left( \frac{T}{\log_{+}(\bar{r}_{T+1}/\bar{r}_1)} - 1\right) \geq \frac{1}{e} \left( \frac{T}{\log_{+}(D/r_\varepsilon)} - 1\right),
    \end{equation}
    where in the final inequality we just consider the event $\{\bar{r}_{T+1} \leq D\}$. Combining \eqref{eq:ivgi-lemma-3-forward-flow} with \eqref{eq:dog-wgd-bound-1-forward-flow} in Proposition \ref{prop:dog-wgd-bound-1-forward-flow}, using the fact that $\smash{d_1\leq D}$ and that $\smash{\bar{r}_{\tau+1} \leq \bar{r}_{T+1} \leq D}$ on the event $\smash{\{\bar{r}_{T+1} \leq D\}}$, we have
    \begin{equation}
        \mathcal{F}(\tilde{\mu}_\tau) - \mathcal{F}(\pi) =\mathcal{O}\left(\frac{(d_1 + \bar{r}_{\tau+1})\sqrt{G_{\tau}}}{\sum_{s=1}^{\tau}\bar{r}_s/\bar{r}_{\tau+1}}\right) =  \mathcal{O}\left(\frac{D \sqrt{G_{\tau}}}{T}\log_{+}\left(\frac{D}{r_{\varepsilon}}\right)\right).
    \end{equation}
    Finally, using that $\smash{\|\zeta_t\|_{L^2(\mu_t)} \leq G_{D}}$ for all $1\leq t\leq T$ on the event $\{\bar{r}'_{T+1}\leq D\}$, which in turn implies that $G_{\tau} \leq G_{T}= \sum_{t=1}^{T}\|\zeta_t\|_{L^2(\mu)}^2 \leq T{G}_D^2$, we arrive at
    \begin{equation}
        \mathcal{F}(\tilde{\mu}_t) - \mathcal{F}(\pi)= \mathcal{O}\left(\frac{DG_D}{\sqrt{T}}\log_{+}\left(\frac{D}{r_{\varepsilon}}\right)\right).
    \end{equation}
\end{proof}

\begin{proof}[Proof of Proposition \ref{prop:dog-wgd-bound-1-forward-flow-uniform}]
    Our proof follows the proof of \citet[][Proposition 3]{ivgi2023dog}. Define the times
    \begin{align}
        \tau_1&=1 \\
        \tau_{k+1} &= \min\left\{\left\{t\in\{\tau_{k}+1,\dots,T+1\}\,:\,\bar{r}_t\geq 2\bar{r}_{\tau_{k}}\right\}\cup\{T+1\}\right\}, \quad \quad k\geq 1.
    \end{align}
    Let $K+1= \min\{k:\tau_k=T+1\}$. Then, for all $k\leq K$, it holds that $\bar{r}_{\tau_{k}}\geq 2\bar{r}_{\tau_{k-1}} \implies \bar{r}_{\tau_{k}} \geq 2^{k-1} \bar{r}_{\tau_1} = 2^{k-1}\bar{r}_{\varepsilon} \implies 1+\log_{2} \frac{\bar{r}_{\tau_k}}{\bar{r}_{\varepsilon}}\geq k$. In particular, setting $k=K$, we have $K\leq 1+ \log_{2}\frac{\bar{r}_{\tau_{K}}}{\bar{r}_{\varepsilon}}$. By construction, we also have that $\bar{r}_{\tau_K}\leq \bar{r}_{\tau_{K+1}}$ and $\bar{r}_{\tau_{K+1}} = \bar{r}_{T+1}$. Combining this with the previous inequality, we thus have
\begin{equation}
    K \leq 1+ \log_{2}\frac{\bar{r}_{T+1}}{\bar{r}_{\varepsilon}}. \label{eq:k-ineq}
\end{equation}
    Arguing as in the proof of Lemma \ref{lemma:tilde-mu-t-bound-2-forward-flow-v2}, we have that
    \begin{align}
    \sum_{t={\tau_{k-1}}}^{\tau_{k}-1}\bar{r}_t \left[\frac{W_2^2(\mu_{t-\frac{1}{2}},\pi) - W_2^2(\mu_{t+\frac{1}{2}}, \pi)}{2\eta_t} + \frac{\eta_t}{2} \int_{\mathbb{R}^d} \|\zeta_t(x)\|^2\,\mathrm{d}\mu_t(x)\right] &\leq \bar{r}_{\tau_k} (2\bar{d}_{\tau_k}+ \bar{r}_{\tau_k}) \sqrt{G_{\tau_k-1}} \label{eq:tilde-mu-t-bound-2-forward-flow-recall} \\
    &=\mathcal{O}\left(\bar{r}_{\tau_k}(d_1 + \bar{r}_{\tau_k})\sqrt{G_{T}}\right)
    \end{align}
    for any $k\in\{2,\dots,K+1\}$, where in the second line we have used the fact that $G_{\tau_k-1} \leq G_{(T+1)-1} = G_{T}$ as in the proof of Proposition \ref{prop:dog-wgd-bound-1-forward-flow}, and the fact that $\bar{d}_{\tau_k}\leq d_1 + \bar{r}_{\tau_k}$. It follows that 
    \begin{align}
        \sum_{t=\tau_{k-1}}^{\tau_k-1} \mathcal{F}(\mu_t) - \mathcal{F}(\pi) 
        &=\sum_{t=\tau_{k-1}}^{\tau_k-1} \frac{\bar{r}_t}{\bar{r}_t}\left[\mathcal{F}(\mu_t) - \mathcal{F}(\pi)\right] \leq \frac{1}{\bar{r}_{\tau_{k-1}}} \sum_{t=\tau_{k-1}}^{\tau_k-1} \bar{r}_t \left[\mathcal{F}(\mu_t) - \mathcal{F}(\pi)\right] \hspace{-5mm}  \\
        &\leq \frac{1}{\bar{r}_{\tau_{k-1}}}  \sum_{t={\tau_{k-1}}}^{\tau_{k}-1}\bar{r}_t \left[\frac{W_2^2(\mu_{t-\frac{1}{2}},\pi) - W_2^2(\mu_{t+\frac{1}{2}}, \pi)}{2\eta_t} + \frac{\eta_t}{2} \int_{\mathbb{R}^d} \|\zeta_t(x)\|^2\,\mathrm{d}\mu_t(x)\right] \hspace{-5mm}  \\
        &\leq \mathcal{O}\left(\frac{\bar{r}_{\tau_k}}{\bar{r}_{\tau_{k}-1}}(d_1 + \bar{r}_{\tau_k})\sqrt{G_{T}}\right) \hspace{-5mm} \label{eq:non-weight-bound}
    \end{align}
    where in the final line we have used the fact that $\bar{r}_{\tau_k-1} \leq 2\bar{r}_{\tau_{k-1}} \iff \frac{1}{\bar{r}_{\tau_{k-1}}}\leq \frac{2}{\bar{r}_{\tau_k-1}}$. We require a bound for the ratio $\frac{\bar{r}_{\tau_{k}}}{\bar{r}_{\tau_k-1}}$. For any $t\geq 1$, observe that
    \begin{align}
        \bar{r}_{t+1} & \leq \bar{r}_t + W_2(\mu_{t-\frac{1}{2}},\mu_{t+\frac{1}{2}}) 
        \leq \bar{r}_t + \underbrace{W_2(\mu_{t-\frac{1}{2}},\mu_{t})}_{\mathrm{(I)}} + \underbrace{W_2(\mu_{t},\mu_{t+\frac{1}{2}})}_{\mathrm{(II)}}. \label{eq:r_t1_bound}
    \end{align}
    First consider (I). Let $x_{t-\frac{1}{2}}\sim \mu_{t-\frac{1}{2}}$, and define $x_t := x_{t-\frac{1}{2}} + \sqrt{2\eta_{t-1}} z_t$ with $z_t\sim \mathcal{N}(0,\mathbf{I}_d)$, so that $x_t \sim \mu_t$. We then have
    \begin{equation}
        W_2^2(\mu_{t-\frac{1}{2}},\mu_{t}) \leq \mathbb{E}\left[\|x_{t} - x_{t-\frac{1}{2}}\|^2\right] = \mathbb{E}\left[\|\sqrt{2\eta_{t-1}}z_t\|^2\right] = 2\eta_{t-1}d.  \label{eq:1step-bound}
    \end{equation}
    Now consider (II). Define the transport map $T_t:\mathbb{R}^d\rightarrow\mathbb{R}^d$ according to $T_t(x) = x - \eta_t\zeta_t(x)$, so that $(T_t)_{\#}\mu_t = \mu_{t+\frac{1}{2}}$. It follows that 
    \begin{equation}
    W_2^2(\mu_{t},\mu_{t+\frac{1}{2}})\leq \int_{\mathbb{R}^d} \|x - T_t(x)\|^2\,\mathrm{d}\mu_t(x) 
    = \eta_t^2 \int_{\mathbb{R}^d} \|\zeta_t(x)\|^2\,\mathrm{d}\mu_t(x) = \eta_t^2\|\zeta_t\|^2_{L^2(\mu_t)}. \label{eq:2step-bound}
    \end{equation}
    Substituting \eqref{eq:1step-bound} - \eqref{eq:2step-bound} into \eqref{eq:r_t1_bound}, and simplifying, it follows that
    \begin{align}
    \bar{r}_{t+1} &\leq \bar{r}_t + \sqrt{2\eta_{t-1} d} + \eta_t \|\zeta_t\|_{L^2(\mu_t)} = \bar{r}_t + \sqrt{2\eta_{t-1} d} + \frac{\bar{r_t}}{\sqrt{G_t}} \sqrt{G_t - G_{t-1}} \\
    &=\bar{r}_t\left(1 + \frac{\sqrt{2\eta_{t-1}d}}{\bar{r}_t} + \sqrt{1-\frac{G_{t-1}}{G_{t}}} \right)  \leq \bar{r}_t\left(2 + \frac{\sqrt{2\eta_{\mathrm{max}} d}}{{r}_{\varepsilon}} \right) \leq \bigg(2+c_{\varepsilon}\bigg) \bar{r}_t \label{eq:r-t-1-bound}
    \end{align}
    where in the final inequality we have defined $c_{\varepsilon}:=\frac{\sqrt{2\eta_{\mathrm{max}}d}}{r_{\varepsilon}}$. Substituting \eqref{eq:r-t-1-bound} into \eqref{eq:non-weight-bound}, we thus have that 
    \begin{equation}
    \sum_{t=\tau_{k-1}}^{\tau_k-1} \mathcal{F}(\mu_t) - \mathcal{F}(\pi)  = \mathcal{O}\left((d_1 + \bar{r}_{\tau_k})\sqrt{G_{T}}\right).
    \end{equation}
    Summing this display over $k=2,\dots,K+1$, it follows that 
    \begin{equation}
    \sum_{t=1}^{T} \left[\mathcal{F}(\mu_t) - \mathcal{F}(\pi)\right] 
    = \sum_{k=2}^{K+1}\left[\sum_{t=\tau_{k-1}}^{\tau_k-1} \mathcal{F}(\mu_t) - \mathcal{F}(\pi)\right]  = \mathcal{O}\left(\left(d_1K + \sum_{k=2}^{K+1}\bar{r}_{\tau_k}\right)\sqrt{G_{T}}\right). \label{eq:unweighted-penultimate}
    \end{equation}
    From \eqref{eq:k-ineq}, we have that $K=\mathcal{O}(\log_{+}\frac{\bar{r}_{T+1}}{r_{\varepsilon}})$. In addition, from the definition, we have that ${\bar{r}_{\tau_K}} \geq 2^{K-k} \bar{r}_{\tau_k} \implies \bar{r}_{\tau_k} \leq 2^{-K+k}\bar{r}_{\tau_K} \leq 2^{-K+k} \bar{r}_{\tau_{K+1}}= 2^{-K+k} \bar{r}_{T+1}$ for all $2\leq k\leq K$. Thus, in particular,  it holds that $\sum_{k=2}^{K+1} \bar{r}_{\tau_k}\leq \bar{r}_{T+1} \sum_{k=2}^{K+1} 2^{-K+k} = \mathcal{O}(\bar{r}_{T+1})$. Substituting back into \eqref{eq:unweighted-penultimate}, and using Lemma \ref{lemma:average}, we have that
    \begin{equation}
    \mathcal{F}(\bar{\mu}_T) - \mathcal{F}(\pi) \leq \frac{1}{T} \sum_{t=1}^{T} \left[\mathcal{F}(\mu_t) - \mathcal{F}(\pi)\right] 
    = \mathcal{O}\left(\frac{\left(d_1\log_{+}\frac{\bar{r}_{T+1}}{r_{\varepsilon}} + \bar{r}_{T+1}\right)\sqrt{G_{T}}}{T}\right). 
    \end{equation}
\end{proof}

\begin{proof}[Proof of \ref{corr:dog-wgd-bound-2-forward-flow-uniform}]
    On the event $\smash{\{\bar{r}_{T+1} \leq D\}\cap\{\bar{r}'_{T+1}\leq D\}}$, we have $\smash{\bar{r}_{T+1}\leq D}$, $\smash{d_1\leq D}$ and $\smash{\bar{g}_t \leq G_{D}}$ for all $1\leq t\leq T$. The latter implies that $G_{T} \leq \sum_{t=1}^{T}\bar{g}_t^2 \leq T{G}_D^2$. It follows that 
    \begin{equation}
        \mathcal{F}(\bar{\mu}_T) - \mathcal{F}(\pi) =\mathcal{O}\left(\frac{\left(D\log_{+}\frac{D}{r_{\varepsilon}} + D\right)\sqrt{TG_D^2}}{T}\right) = \mathcal{O}\left(\frac{DG_D}{\sqrt{T}}\log_{+}\left(\frac{D}{r_{\varepsilon}}\right)\right).
    \end{equation}
\end{proof}

\begin{proof}[Proof of Proposition \ref{prop:iterate-stability-forward-flow}]
    From Lemma \ref{corollary:evi-forward-flow}, we have that
    \begin{align}
         {d_{k+1}^2 - d_k^2} &\leq 2\eta_k \left(\mathcal{F}(\pi) - \mathcal{F}(\mu_k)\right)  + \eta_k^2\int_{\mathbb{R}^d}\|\xi_k(x)\|^2\mathrm{d}\mu_k(x) \\
         &\leq  \eta_k^2\int_{\mathbb{R}^d}\|\zeta_k(x)\|^2\mathrm{d}\mu_k(x).
        \label{eq:evi-time-dep-forward-flow}
    \end{align}
    We now proceed by induction. The base case holds by assumption: $\bar{r}_1 = r_{\varepsilon} \leq 3d_1$. Suppose now that $\bar{r}_t\leq 3d_1$. We then have, via a telescoping sum, that 
    \begin{align}
        d_{t+1}^2  - d_1^2 = \sum_{k=1}^{t} \left(d_{k+1}^2 - d_{k}^2 \right) &\leq \sum_{k=1}^t \eta_k^2 \int_{\mathbb{R}^d}\|\zeta_k(x)\|^2\mathrm{d}\mu_k(x) \\
         & =  \sum_{k=1}^t\frac{\bar{r}_k^2\int_{\mathbb{R}^d}\|\zeta_k(x)\|^2\mathrm{d}\mu_k(x)}{G'_k} \\
         &= \sum_{k=1}^t\frac{\bar{r}_k^2\left(\sum_{i=1}^{k}\int_{\mathbb{R}^d}\|\zeta_i(x)\|^2\mathrm{d}\mu_i(x) - \sum_{i=1}^{k-1}\int_{\mathbb{R}^d}\|\zeta_i(x)\|^2\mathrm{d}\mu_i(x)\right)}{8^4\log_{+}^2 \left(1+ \frac{k\bar{g}_k^2}{\bar{g}_1^2}\right)\left(G_{k-1} + 16\bar{g}_k^2\right)} \\
         &\leq \frac{\bar{r}_t^2}{8^4}\sum_{k=1}^t\frac{G_{k} - G_{k-1}}{\log_{+}^2 \left(1+ \frac{k\bar{g}_k^2}{\bar{g}_1^2}\right)\left(G_{k-1} + 16\bar{g}_k^2\right)} \\
         &\leq \frac{\bar{r}_t^2}{8^4}\sum_{k=1}^t\frac{G_{k} - G_{k-1}}{\log_{+}^2 \left(\frac{G_k + \bar{g}_k^2}{\bar{g}_1^2}\right)\left(G_{k} + \bar{g}_k^2\right)} \label{eq:penultimate-forward-flow} \\
         &\leq \frac{\bar{r}_t^2}{8^4}\leq \frac{9d_1^2}{8^4} \label{eq:final-forward-flow}
    \end{align}
    where in the penultimate line \eqref{eq:penultimate-forward-flow} we have used the fact that, since $\bar{g}_k\geq \|\zeta_k\|_{L^2(\mu_k)}$ for all $k$, 
    \begin{align}
        \log_{+}^2 \left(1+ \frac{k\bar{g}_k^2}{\bar{g}_1^2}\right)\left(G_{k-1} + 16\bar{g}_k^2\right) &\geq \log_{+}^2 \left(\frac{\bar{g}_1^2 + \bar{g}_k^2 + \sum_{i=2}^{k} \bar{g}_i^2}{\bar{g}_1^2}\right)\left(G_{k-1} + \bar{g}_k^2 + \bar{g}_k^2\right) \\
        &\geq \log_{+}^2 \left(\frac{\sum_{i=1}^{k} \bar{g}_i^2 + \bar{g}_k^2}{\bar{g}_1^2}\right)\left(G_{k-1} + \|\zeta_k\|_{L^2(\mu_k)}^2 + \bar{g}_k^2\right) \\
        &\geq \log_{+}^2 \left( \frac{\sum_{i=1}^k \|\zeta_k\|_{L^2(\mu_k)}^2+\bar{g}_k^2}{\bar{g}_1^2}\right)\left(G_{k} + \bar{g}_k^2\right) \\
        &= \log_{+}^2 \left( \frac{G_k+ \bar{g}_k^2}{\bar{g}_1^2}\right)\left(G_{k} + \bar{g}_k^2\right) 
    \end{align}
    and in the final line \eqref{eq:final-forward-flow} we have used \cite[][Lemma 6]{ivgi2023dog} for the first inequality; and the inductive assumption for the second inequality. Rearranging \eqref{eq:final-forward-flow}, we have that 
    \begin{equation}
        d_{t+1}^2 \leq \left(1 + \frac{9}{8^4}\right)d_1^2 \leq 2^2 d_1^2 
    \end{equation}
    which implies, in particular, that $d_{t+1}\leq 2d_1$. Thus, using the triangle inequality for the Wasserstein distance \citep[e.g.,][]{clement2008elementary}, we conclude, as required, that 
    \begin{equation}
        r_{t+1} := W_2(\mu_{\frac{1}{2}},\mu_{t+\frac{1}{2}}) \leq W_2(\mu_\frac{1}{2},\pi) + W_2(\mu_{t+\frac{1}{2}},\pi) := d_1 + d_{t+1} \leq 3d_1.
    \end{equation}
\end{proof}

\subsubsection{Stochastic Case: Adaptive Step Size}
\label{app:additional-proofs-forward-flow-non-smooth-stochastic-adaptive}
\begin{proof}[Proof of Lemma \ref{lemma:tilde-mu-t-bound-1-stoch-forward-flow}]
    Recall, using the definition of $\mathcal{F}$, that
    \begin{align}
        \mathcal{F}(\mu_t) - \mathcal{F}(\pi) = \left[\mathcal{E}(\mu_t) - \mathcal{E}(\pi)\right] + \left[\mathcal{H}(\mu_t) - \mathcal{H}(\pi)\right] \label{eq:f-def}
    \end{align}
     We will bound each term on the RHS separately. For $\mathcal{E}$, using geodesic convexity; and the definition of $\delta_t(x)$; we have that
     \begin{align}
     \hspace{-4mm} \mathcal{E}({\mu}_t) - \mathcal{E}(\pi) &\leq \int_{\mathbb{R}^d} \langle {\zeta}_t(x), x - t_{\mu_t}^{\pi}(x)\rangle \,\mathrm{d}\mu_t(x) \nonumber \\[1mm]
     &\leq\int_{\mathbb{R}^d} \langle \hat{\zeta}_t(x), x - t_{\mu_t}^{\pi}(x)\rangle \,\mathrm{d}\mu_t(x) - \int_{\mathbb{R}^d} \langle \underbrace{\hat{\zeta}_t(x) - {\zeta}_t(x)}_{\delta_t(x)}, x-\boldsymbol{t}_{\mu_t}^{\pi}(x) \rangle \,\mathrm{d}\mu_t(x) \\[-2mm]
     &= \int_{\mathbb{R}^d} \langle \hat{\zeta}_t(x), x - t_{\mu_t}^{\pi}(x)\rangle \,\mathrm{d}\mu_t(x) - \int_{\mathbb{R}^d} \langle \delta_t(x), x-\boldsymbol{t}_{\mu_t}^{\pi}(x) \rangle \,\mathrm{d}\mu_t(x) \\[3mm]
     &\leq \frac{W_2^2(\mu_t,\pi) - W_2^2(\mu_{t+\frac{1}{2}},\pi)}{2\eta_t}+\frac{\eta_t}{2}\int_{\mathbb{R}^d} \|\hat{\zeta}_t(x)\|^2\,\mathrm{d}\mu_t(x) - \int \langle \delta_t(x), x-\boldsymbol{t}_{\mu_t}^{\pi}(x)\rangle \,\mathrm{d}\mu_t(x) \label{ineq:g}
     \end{align}
     where in the final line we have upper bounded the first term using the same argument as in \eqref{eq:intermidate-result-g-1} - \eqref{eq:intermidate-result-g-3}. Meanwhile, for $\mathcal{H}$, using \citet[Lemma 5]{durmus2019analysis}, we have that
     \begin{equation}
         \mathcal{H}(\mu_t) - \mathcal{H}(\pi) \leq \frac{W_2^2(\mu_{t-\frac{1}{2}},\pi) - W_2^2(\mu_t,\pi)}{2\eta_t} \label{ineq:h}
     \end{equation}
     Substituting the inequalities in \eqref{ineq:g} and \eqref{ineq:h} into \eqref{eq:f-def}; cancelling like terms; and taking a weighted sum, the result follows. 
\end{proof}

\begin{proof}[Proof of Lemma \ref{lemma:tilde-mu-t-bound-2-stoch-forward-flow}]
    The proof is identical to the proof of Lemma 
    \ref{lemma:tilde-mu-t-bound-2-forward-flow-v2}, using the definitions in \eqref{eq:defs1-stoc-forward-flow} - \eqref{eq:defs2-stoc-forward-flow} in place of the definitions in \eqref{eq:defs1-forward-flow} - \eqref{eq:defs2-forward-flow}.
\end{proof}

\begin{proof}[Proof of Lemma \ref{lemma:tilde-mu-t-bound-3-stoch-forward-flow}]
    For $1\leq t\leq T$, define the random variables
    \begin{align}
        Z_t = \bar{r}_{t} \bar{d}_{t}, \quad X_t = \frac{1}{\bar{d}_{t}} \int_{\mathbb{R}^d} \langle \delta_t(x), x - \boldsymbol{t}_{\mu_t}^{\pi}(x) \rangle \,\mathrm{d}\mu_t(x), \quad \hat{X}_t = -\frac{1}{\bar{d}_{t}} \int_{\mathbb{R}^d} \langle \zeta_t(x), x - t_{\mu_t}^{\pi}(x) \rangle \,\mathrm{d}\mu_t(x). 
    \end{align}
    Using these definitions, we have that
    \begin{align}
        \sum_{k=1}^{t} Z_k X_k &= \sum_{k=1}^{t}\bar{r}_{k}\int_{\mathbb{R}^d}\langle {\delta}_k(x), x- \boldsymbol{t}_{\mu_k}^{\pi}(x) \rangle \mathrm{d}\mu_k(x)
    \intertext{
    and also that 
    }
        \sum_{k=1}^{t} (X_k - \hat{X}_k)^2 &= \sum_{k=1}^{t} \left[\frac{1}{\bar{d}_{k}}\int_{\mathbb{R}^d}\langle \hat{\zeta}_k(x), x- \boldsymbol{t}_{\mu_k}^{\pi}(x) \rangle \mathrm{d}\mu_k(x)\right]^2 \\
        &\leq \sum_{k=1}^{t} \frac{1}{\bar{d}_{k}^2}\int_{\mathbb{R}^d} \|\hat{\zeta}_k(x)\|^2\mathrm{d}\mu_k(x) \int_{\mathbb{R}^d} \|x- \boldsymbol{t}_{\mu_k}^{\pi}(x)\|^2 \mathrm{d}\mu_k(x) \\
        &=\sum_{k=1}^{t} \frac{d_{k}^2}{\bar{d}_{k}^2}\int_{\mathbb{R}^d} \|\hat{\zeta}_k(x)\|^2\mathrm{d}\mu_k(x) \\
        & \leq \sum_{k=1}^{t} \int_{\mathbb{R}^d} \|\hat{\zeta}_k(x)\|^2\mathrm{d}\mu_k(x)= G_{t}
    \end{align}
    where the second line follows from the Cauchy-Schwarz inequality, and the third line from the definition of $\bar{d}_k$. It follows, substituting these (in)equalities, that
    \begin{align}
        &\mathbb{P}\left(\exists t \leq T : \left|\sum_{k=1}^{t} \bar{r}_k \int_{\mathbb{R}^d} \langle \delta_k(x),x - \boldsymbol{t}_{\mu_k}^{\pi}(x)\rangle\mathrm{d}\mu_k(x) \right| \geq 8 \bar{r}_{t}\bar{d}_{t}\sqrt{\theta_{t,\delta} G_{t} + \theta_{t,\delta}^2G^2}\right) \\
        =\hspace{.5mm}&\mathbb{P}\left(\exists t \leq T : \left|\sum_{k=1}^{t} Z_kX_k \right| \geq 8 Z_{t}\sqrt{\theta_{t,\delta} G_{t} + \theta_{t,\delta}^2G^2}\right) \\
        \leq\hspace{.5mm}&\mathbb{P}\left(\exists t \leq T : \left|\sum_{k=1}^{t} Z_kX_k \right| \geq 8 Z_{t}\sqrt{\theta_{t,\delta} \sum_{k=1}^{t}(X_k - \hat{X}_k)^2 + \theta_{t,\delta}^2G^2}\right) \\[1mm]
        \leq &\hspace{.5mm}\delta + \mathbb{P}(\exists t\leq T: g_t>G) \\[3mm]
        = \hspace{.5mm}&\delta + \mathbb{P}(\bar{g}_T > G)
    \end{align}
    where in penultimate line we have used \citep[][Lemma 7]{ivgi2023dog}, which we can apply since, with probability one, we can bound $|X_t|$ and $|\hat{X}_t|$, e.g., 
    \begin{align}
        |\hat{X}_t| 
        \leq \frac{1}{\bar{d}_t} \left[\int_{\mathbb{R}^d} \|\zeta_t(x)\|^2 \,\mathrm{d}\mu_t(x)\right]^{\frac{1}{2}} \left[\int_{\mathbb{R}^d} \|x-t_{\mu_t}^{\pi}(x)\|^2 \,\mathrm{d}\mu_t(x)\right]^{\frac{1}{2}} 
        \leq \frac{d_t}{\bar{d}_t} g(\mu_t) \leq \bar{g}_t, 
    \end{align}
    where in the first inequality line we have once again made use of the Cauchy-Schwarz inequality, in the second inequality we use the definition of ${d}_t$, and in the final inequality the definition of $\bar{d}_t$ and $\bar{g}_t$.
\end{proof}

\subsection{Forward-Flow Discretization, Smooth Setting}
\label{app:additional-proofs-forward-flow-smooth}

\subsubsection{Deterministic Case: Constant Step Size}
\label{app:additional-proofs-forward-flow-smooth-deterministic-constant}

\begin{proof}[Proof of Lemma \ref{lemma:g-descent-smooth}]
For each $s\in[0,1]$, let us define the family of transport maps $\Phi_s:\mathbb{R}^d\rightarrow\mathbb{R}^d$ according to
\begin{align}
    \Phi_{s}(x) &= x+s(T_{\eta}(x) - x), \quad \quad \partial_{s}\Phi_s(x) = T_\eta(x) - x
\end{align}
with $T_{\eta}(x) - x = -\eta\zeta_t(x)$. Thus, in particular, $\Phi_{0}(x) = x$ and $\Phi_{1}(x) = x-\eta\zeta_t(x)$. In addition, let us define the curve $(\nu_s)_{s\in[0,1]}$ as the pushforward of $\mu_t$ under this transport map, viz $\nu_s = (\Phi_s)_{\#}\mu_t$, with $\nu_0 = \mu_t$ and $\smash{\nu_1 = \mu_{t+\frac{1}{2}}}$. It is straightforward to verify that this curve solves the continuity equation $\smash{\partial_{s}\nu_s + \nabla \cdot(
\nu_sv_s) = 0}$, with velocity field
\begin{equation}
    v_s(x) = \partial_{s}\Phi_s\left(\Phi_s^{-1}(x)\right) = (T_{\eta} - \mathrm{id})\circ \Phi_{s}^{-1}(x).
\end{equation}
We can verify this as follows. Let $\varphi\in C_c^\infty(\mathbb R^d)$. We then have
\begin{align}
&\frac{\partial}{\partial s} \left[\int_{\mathbb R^d}\varphi(x) \mathrm d\nu_s(x)\right] =
\frac{\partial}{\partial s}\left[\int_{\mathbb R^d}\varphi\big(\Phi_s(x)\big)\mathrm d\nu(x)\right] \\
&= \int \big\langle \nabla\varphi\big(\Phi_s(x)\big),\partial_s\Phi_s(x)\big\rangle\mathrm d\nu(x) =\int \big\langle \nabla\varphi\!\big(\Phi_s(x)\big),T_\eta(x)-x\big\rangle\mathrm d\nu(x) \\
&=\int \big\langle \nabla\varphi(x),\big(T_\eta-\mathrm{id}\big)\!\circ\Phi_s^{-1}(x)\big\rangle\mathrm d\nu_s(x)
=\int \langle \nabla\varphi(x),v_s(x)\rangle\mathrm d\nu_s(x).
\end{align}
which is precisely the weak form of the continuity equation $ \partial_{s}\nu_s + \nabla \cdot(
\nu_sv_s) = 0$. Now, using the chain rule \citep[e.g.,][Section 10.1.2]{ambrosio2008gradient}, we can compute
\begin{equation}
    \frac{\mathrm{d}}{\mathrm{d}s}\mathcal{E}(\nu_s) = \int_{\mathbb{R}^d} \langle \nabla_{W_2}\mathcal{E}(\nu_s)(x), v_s(x)\rangle \nu_s(\mathrm{d}x).
\end{equation}
It follows, integrating over $s\in[0,1]$, and pulling back to the reference measure $\nu_0 = \mu_t$, that 
\begin{align}
    \mathcal{E}(\mu_{t+\frac{1}{2}}) - \mathcal{E}(\mu_t) &= \int_0^1 \left[\int_{\mathbb{R}^d} \left\langle \nabla_{W_2}\mathcal{E}(\nu_s)(x), v_s(x) \right\rangle \nu_s(\mathrm{d}x)\right] \mathrm{d}s \\
    &=\int_0^1 \left[\int_{\mathbb{R}^d} \left\langle \nabla_{W_2}\mathcal{E}(\nu_s)\circ \Phi_s(x), v_s\circ \Phi_s (x) \right\rangle *(\Phi_s^{-1})_{\#}\nu_s(\mathrm{d}x)\right] \mathrm{d}s \\
    &=\int_0^1 \left[\int_{\mathbb{R}^d} \left\langle \nabla_{W_2}\mathcal{E}(\nu_s)\circ \Phi_s(x), T_{\eta}(x) - x \right\rangle \nu_0(\mathrm{d}x)\right] \mathrm{d}s \\
    &=\int_0^1 \left[\int_{\mathbb{R}^d} \left\langle \nabla_{W_2}\mathcal{E}(\nu_s)\circ \Phi_s(x), T_{\eta}(x) - x \right\rangle \,\mathrm{d}\mu_t(x)\right] \mathrm{d}s
\end{align}
Next, adding and subtracting the Wasserstein gradient $\nabla_{W_2}\mathcal{E}(\nu_0) :=\nabla_{W_2}\mathcal{E}(\mu_t)$ inside the inner product, we obtain
\begin{align}
    \mathcal{E}(\mu_{t+\frac{1}{2}}) - \mathcal{E}(\mu_t) 
    &= 
    \int_0^1 \left[\int_{\mathbb{R}^d} \left\langle \nabla_{W_2}\mathcal{E}(\mu_t)(x), T_{\eta}(x) - x \right\rangle \,\mathrm{d}\mu_t(x)\right] \mathrm{d}s \\
    &+ \int_0^1 \left[\int_{\mathbb{R}^d} \left\langle \left[\nabla_{W_2}\mathcal{E}(\nu_s)\circ \Phi_s - \nabla_{W_2}\mathcal{E}(\mu_t)\right](x), T_{\eta}(x) - x \right\rangle \,\mathrm{d}\mu_t(x)\right] \mathrm{d}s \nonumber \\
    &=- \eta \int_0^1 \left[\int_{\mathbb{R}^d} \|\zeta_t(x)\|^2\,\mathrm{d}\mu_t(x)\right] \mathrm{d}s \label{eq:bound-1} \\
    &\hspace{3mm}+\int_0^1 \left[\int_{\mathbb{R}^d} \left\langle \left[\nabla_{W_2}\mathcal{E}(\nu_s)\circ \Phi_s - \nabla_{W_2}\mathcal{E}(\nu_0)\right](x), T_{\eta}(x) - x \right\rangle \,\mathrm{d}\mu_t(x)\right] \mathrm{d}s \nonumber 
\end{align}
where in the second line we have used the facts that $T_{\eta} - \mathrm{id} = -\eta \zeta_t$ and $\nabla_{W_2}\mathcal{E}(\mu_t) = \zeta_t$. For the second term, applying Cauchy-Schwarz, and using smoothness (Assumption \ref{assumption:smooth-forward-flow}),\footnote{We use an equivalent form of the smoothness assumption in terms of the Wasserstein gradients, described in the commentary after Assumption \ref{assumption:smooth-forward-flow} (see Section \ref{sec:assumptions}). In particular, for every a.c.\ $\mu$ and every measurable
$S$ with $S_{\#}\mu=\nu$, it holds that
\begin{equation}
\label{eq:Lip-grad}
\big\|(\nabla_{\mathsf W}\mathcal G(\nu))\circ S-\nabla_{\mathsf W}\mathcal G(\mu)\big\|_{L^2(\mu)}
\ \le\ L\,\|S-\mathrm{id}\|_{L^2(\mu)}.
\end{equation}} we have
\begin{align}
\Big|\mathrm{(II)}\Big|
&\leq \int_0^1
\big\|\nabla_{W_2}\mathcal{E}(\nu_s)\circ\Phi_s-\nabla_{W_2}\mathcal{E}(\nu_0)\big\|_{L^2(\mu_t)}
\|T_\eta-\mathrm{id}\|_{L^2(\mu_t)}\,\mathrm ds\\
&\leq \int_0^1 L\,\|\Phi_s-\mathrm{id}\|_{L^2(\mu_t)}\|T_\eta-\mathrm{id}\|_{L^2(\mu_t)}\,\mathrm ds
\ =\ \int_0^1 L\,s\,\|T_\eta-\mathrm{id}\|_{L^2(\mu_t)}^2\,\mathrm ds\\
&=\ \frac{L}{2}\,\|T_\eta-\mathrm{id}\|_{L^2(\mu_t)}^2
\ =\ \frac{L\eta^2}{2}\int_{\mathbb{R}^d}\|\zeta_t\|^2\,\mathrm d\mu_t, \label{eq:bound-2}
\end{align}
where in the second line we have used the fact that $\|\Phi_s-\mathrm{id}\|_{L^2(\mu_t)}=s\,\|T_\eta-\mathrm{id}\|_{L^2(\mu_t)}$; and in the final line the fact that $\smash{\|T_{\eta} - \mathrm{id}\|_{L^2(\mu_t)}^2 = \|\eta\zeta_t\|^2_{L^2(\mu_t)} = \eta^2 \|\zeta_t\|_{L^2(\mu_t)}^2}$. Combining \eqref{eq:bound-1} and \eqref{eq:bound-2}, and simplifying, we have
\begin{equation}
    \mathcal{E}(\mu_{t+\frac{1}{2}}) \leq \mathcal{E}(\mu_t) -\eta \int_{\mathbb{R}^d}\|\zeta_t(x)\|^2\,\mathrm{d}\mu_t(x) + \frac{L\eta^2}{2}\int_{\mathbb{R}^d}\|\zeta_t(x)\|^2\,\mathrm{d}\mu_t(x)
\end{equation}
which is precisely the desired bound.
\end{proof}

\begin{proof}[Proof of Lemma \ref{lem:heat-step-tight}]
Let $\smash{X \sim \mu_{t+\frac{1}{2}}}$, and set $Y = X + Z$, where $\smash{Z\sim \mathcal{N}(0,2\eta\mathbf{I}_d)}$, independent of $X$, so that $\smash{Y\sim\mu_{t+1}}$. Define the coupling
\begin{equation}
\gamma := \mathrm{Law}\big(X,Y\big), 
\end{equation}
so that $\gamma\in \Gamma(\mu_{t+\frac{1}{2}},\mu_{t+1})$. For $s\in[0,1]$, define
\begin{equation}
\Psi_s(x,y) \coloneqq (1-s)x + s y,\qquad \nu_s \coloneqq (\Psi_s)_{\#}\gamma,
\end{equation}
so that $\nu_0=\mu_{t+\frac12}$ and $\nu_1=\mu_{t+1}$. Using standard results \citep[e.g.,][Section 8.3]{ambrosio2008gradient}, $(\nu_s)_{s\in[0,1]}$ solves the continuity equation $\partial_s \nu_s + \nabla\!\cdot(\nu_s v_s) = 0$ with the velocity field
\begin{equation}
v_s(z)=\int_{\mathbb{R}^d\times\mathbb{R}^d} (y-x) \mathrm d\gamma^z_s(x,y),
\end{equation}
where $\gamma^z_s$ denotes the disintegration of $\gamma$ along $\Psi_s$, i.e., $\int \varphi(z)v_s(z)\nu_s(\mathrm dz) = \iint \varphi(\Psi_s(x,y))(y-x)\mathrm d\gamma(x,y)$ for all test functions $\varphi$. Using the chain rule in Wasserstein space \citep[][Section 10.1.2]{ambrosio2008gradient}, we obtain
\begin{align}
\frac{\mathrm d}{\mathrm ds}\mathcal{E}(\nu_s)
&= \int_{\mathbb{R}^d} \Big\langle \nabla_{W_2}\mathcal{E}(\nu_s)(z),v_s(z)\Big\rangle \nu_s(\mathrm dz) \\
&= \int_{\mathbb{R}^d\times\mathbb{R}^d} \Big\langle \nabla_{W_2}\mathcal{E}(\nu_s)\big(\Psi_s(x,y)\big),y-x\Big\rangle \mathrm d\gamma(x,y). \label{eq:heat-tight-deriv}
\end{align}
Similar to before, integrating over $s\in[0,1]$, and adding and subtracting $\nabla_{W_2}\mathcal{E}(\nu_0)$, we arrive at
\begin{align}
\mathcal{E}(\nu_1)-\mathcal{E}(\nu_0)
&= \underbrace{\int_{\mathbb{R}^d\times\mathbb{R}^d} \Big\langle \nabla_{W_2}\mathcal{E}(\nu_0)(x),y-x\Big\rangle \mathrm d\gamma(x,y)}_{\mathrm{(I)}} \label{eq:heat-tight-I+II}\\
&\quad + \underbrace{\int_0^1 \int_{\mathbb{R}^d\times\mathbb{R}^d} \Big\langle \big(\nabla_{W_2}\mathcal{E}(\nu_s)\circ \Psi_s - \nabla_{W_2}\mathcal{E}(\nu_0)\circ \Psi_0\big)(x,y),y-x\Big\rangle \mathrm d\gamma(x,y)\mathrm ds}_{\mathrm{(II)}}. \nonumber
\end{align}
For the first term, by the definition of $\gamma$, $z=y-x$ is independent of $x$, and satisfies $\mathbb{E}[z]=0$. It follows that 
\begin{equation}
    \mathrm{(I)} = \mathbb{E}\left[\big\langle \nabla_{W_2}\mathcal{E}(\eta_0)(X), Z\big\rangle\right] = \left\langle \mathbb{E}\left[\nabla_{W_2}\mathcal{E}(\eta_0)(X)\right], \mathbb{E}\left[Z\right]\right\rangle = 0. \label{eq:heat-tight-I}
\end{equation}
For the second term, applying Cauchy–Schwarz, and using the smoothness assumption (Assumption~\ref{assumption:smooth-forward-flow}), we have
\begin{align}
|\mathrm{(II)}|
&\leq \int_0^1 \left\| \nabla_{W_2}\mathcal{E}(\nu_s)\circ \Psi_s - \nabla_{W_2}\mathcal{E}(\nu_0)\circ \Psi_0 \right\|_{L^2(\gamma)} \ \left\| y-x \right\|_{L^2(\gamma)} \mathrm{d}s \\
&\leq \int_0^1 L\left\| \Psi_s - \Psi_0 \right\|_{L^2(\gamma)} \ \left\| y-x \right\|_{L^2(\gamma)} \mathrm ds
= \int_0^1 Ls \left\| y-x \right\|_{L^2(\gamma)}^2 \mathrm ds \\
&= \frac{L}{2}\int_{\mathbb{R}^d\times\mathbb{R}^d} \|y-x\|^2 \mathrm d\gamma(x,y)
= \frac{L}{2}\mathbb{E}\|Z\|^2. \label{eq:heat-tight-II}
\end{align}
Substituting \eqref{eq:heat-tight-I}, \eqref{eq:heat-tight-II} into \eqref{eq:heat-tight-I+II}, and identifying $\nu_0=\mu_{t+\frac12}$ and $\nu_1=\mu_{t+1}$, we finally arrive at
\begin{equation}
\mathcal{E}(\mu_{t+1})\leq \mathcal{E}(\mu_{t+\frac12})  + \frac{L}{2}\mathbb{E}\|Z\|^2
= \frac{L}{2}\cdot 2\eta d  = L\eta d.
\end{equation}
\end{proof}

\subsubsection{Deterministic Case: Adaptive Step Size}
\label{app:additional-proofs-forward-flow-smooth-deterministic-adaptive}

\begin{proof}[Proof of Lemma \ref{lemma:smooth-grad-bound-improved}]
For $\eta\in (0,\frac{1}{L})$, and for any $\mu\in\mathcal{P}_2(\mathcal{X})$, let us define $\mu_{\eta}^{+}=(T_{\eta})_{\#}\mu$, where $T_{\eta}:\mathbb{R}^d\rightarrow\mathbb{R}^d$ is the transport map given by $T_{\eta}(x)= x - \eta\zeta(\mu)(x)$. We then have that (see, e.g., Lemma \ref{lemma:g-descent-smooth})
\begin{align}
    \mathcal{E}(\mu_{\eta}^{+}) &\leq \mathcal{E}(\mu) - \eta\left(1-\frac{\eta L}{2}\right) \int \|\zeta(\mu)(x)\|^2\mathrm{d}\mu(x).
     \label{eq:g-half-step}
\end{align} 
We will seek a similar bound for $\mathcal{H}$. Under our smoothness assumption, we have that $\smash{\|\nabla \zeta(\mu)(\cdot)\|_{\mathrm{op}}\leq L}$ for all $\smash{\mu\in\mathcal{P}_2(\mathcal{X})}$. It follows, using also our assumption on the step size, that $\smash{T_{\eta}'(x) = \mathbf{I} - \eta\nabla \zeta(\mu)(x)}$ is invertible, which in turn implies that $T_{\eta}$ is a diffeomorphism. We can thus apply the change-of-variables formula to obtain the density of the pushforward measure $\mu_{\eta}^{+}:=(T_{\eta}^{+})_{\#}\mu$. In particular, we have that 
    \begin{equation}
        \mu_{\eta}^{+}(y) = \mu(x)|\mathrm{det}T'_{\eta}(x)|^{-1} \implies \log \mu_{\eta}^{+}(y) = \log \mu(x) - \log |\mathrm{det}T'_{\eta}(x)|
    \end{equation}
    It follows, integrating w.r.t. $\mu_{\eta}^{+}$, that
    \begin{align}
        \mathcal{H}(\mu_{\eta}^{+}) &:= \int \log \mu_{\eta}^{+}(y) \mu_{\eta}^{+}(y)\mathrm{d}y = \int \left(\log \mu(x) - \log |\mathrm{det}T'_{\eta}(x)| \right) \mu(x)\mathrm{d}x \nonumber \\
        &= \int \log \mu(x) \mu(x)\mathrm{d}x - \int \log |\det T_{\eta}'(x)| \mu(x)\mathrm{d}x = \mathcal{H}(\mu) - \int \log |\det(
        \mathbf{I} - \eta \nabla \zeta(x))|\mu(x)\mathrm{d}x.
    \end{align}
    Using standard results on matrices \citep[e.g.,][]{horn2012matrix}, we have that
    \begin{align}
        -\log \det(\mathbf{I} - \mathbf{A}) &= - \mathrm{Tr}\left[\log (\mathbf{I} - \mathbf{A})\right] = -\mathrm{Tr}\left[\sum_{k=1}^{\infty} \frac{(-1)^{k+1}}{k} (-\mathbf{A})^k\right] \\
        &=\mathrm{Tr}\left[\mathbf{A}\right] + \sum_{k=2}^{\infty}\frac{\mathbf{A}^k}{k}  \\
        &\leq \mathrm{Tr}\left[\mathbf{A}\right] + \sum_{k=2}^{\infty} \frac{\|\mathbf{A}\|_{\mathrm{F}} \|\mathbf{A}^{k-1}\|_{\mathrm{F}}}{k} \\
        &\leq \mathrm{Tr}\left[\mathbf{A}\right] + \sum_{k=2}^{\infty} \frac{\|\mathbf{A}\|_{\mathrm{F}} \|\mathbf{A}\|_{\mathrm{F}}\|\mathbf{A}\|^{k-2}_{\mathrm{op}}}{k} \\
        &\leq \mathrm{Tr}\left[\mathbf{A}\right] + \|\mathbf{A}\|_{\mathrm{F}}^2 \sum_{k=2}^{\infty} \frac{ \|\mathbf{A}\|^{k-2}_{\mathrm{op}}}{k} \\
        &\leq \mathrm{Tr}\left[\mathbf{A}\right] + \frac{\|\mathbf{A}\|_{F}^2}{2(1-\|\mathbf{A}\|_{\mathrm{op}})}.
    \end{align}
    Applying this inequality with $\mathbf{A} = \eta\nabla \zeta(\mu)(x)$, and substituting into the previous display, we arrive at 
    \begin{align}
        \mathcal{H}(\mu_{\eta}^{+}) &\leq \mathcal{H}(\mu) + \eta \int \mathrm{tr}(\nabla \zeta(\mu)(x))\mu(x)\mathrm{d}x + \int \frac{\|\nabla \zeta(\mu)(x)\|_{F}^2}{2(1-\|\nabla \zeta(\mu)(x)\|_{\mathrm{op}})} \mu(x)\mathrm{d}x \\
        &\leq \mathcal{H}(\mu) + \eta d L + \frac{\eta^2 d L}{2(1- \eta L)} \label{eq:h-half-step}
    \end{align}    
    where in the second line we have used the fact that $\|\nabla \zeta(\mu)(x)\|_{\mathrm{op}}\leq L \implies \|\nabla \zeta(\mu)(x)\|_{F}^2 \leq d\|\nabla \zeta(\mu)(x)\|_{\mathrm{op}}^2 \leq d L^2$. Summing the inequalities in \eqref{eq:g-half-step} and \eqref{eq:h-half-step}, we arrive at 
    \begin{equation}
        \mathcal{F}(\mu_{\eta}^{+})\leq \mathcal{F}(\mu) - \eta\left(1-\frac{\eta L}{2}\right) \int \|\zeta(\mu)(x)\|^2 \mathrm{d}\mu(x) + \eta d L + \frac{\eta^2 d L}{2(1-\eta L)}.
    \end{equation}
    Let $\eta=\frac{1}{2L}$. Other choices of $\eta\in(0,\frac{1}{L})$ are possible, and may result in slightly tighter bounds, but this choice will simplify the presentation. We then have 
    \begin{equation}
        \eta-\tfrac{L\eta^2}{2}=\frac{3}{8L}, 
        \quad
        \eta dL=\frac d2,
        \quad
        \frac{\eta^2 d L^2}{2(1-\eta L)}=\frac{d}{4}.
    \end{equation}
    It follows that
    \begin{equation}
        \mathcal{F}(\mu_{\eta}^{+})\leq \mathcal{F}(\mu) - \frac{3}{8L} \int \|\zeta(\mu)(x)\|^2 \mathrm{d}\mu(x) + \frac{3d}{4}
    \end{equation}
    By definition, we have that $\mathcal{F}(\pi):= \inf\mathcal{F}(\mu) \leq \mathcal{F}(\mu_{\eta}^{+})$. Combining this with the previous display, it follows that, for each $\mu\in\mathcal{P}_2(\mathcal{X})$,  
    \begin{equation}
    \mathcal{F}(\pi) \leq \mathcal{F}(\mu) - \frac{3}{8L} \int \|\zeta(\mu)(x)\|^2 \mathrm{d}\mu(x) + \frac{3d}{4}
    \end{equation}
    Rearranging, we have that
    \begin{equation}
        \int \|\zeta(\mu)(x)\|^2 \mathrm{d}\mu(x) \leq \frac{8L}{3}\left(\mathcal{F}(\mu) - \mathcal{F}(\pi)\right) + 2Ld
    \end{equation}
    Finally, summing over $t\in[1,T]$, and identifying $\zeta(\mu_t)(x):=\zeta_t(x)$, we have the desired result:
    \begin{equation}
    \sum_{t=1}^T \int \|\zeta_t(x)\|^2 \,\mathrm{d}\mu_t(x) \leq \frac{8L}{3} \sum_{t=1}^T \left(\mathcal{F}(\mu_t) - \mathcal{F}(\pi)\right) + 2LdT.
    \end{equation}
    \end{proof}

\begin{proof}[Proof of Proposition \ref{prop:smooth-adaptive-convergence-rate-improved}]
    For ease of notation, we begin by defining 
    \begin{equation}
        S_T:=\sum_{t=1}^{T} \left(\mathcal{F}(\mu_t) - \mathcal{F}(\pi)\right), \quad K_T=\left(d_1\log_{+}\frac{\bar{r}_{T+1}}{r_{\varepsilon}} + \bar{r}_{T+1}\right), \quad C_T = dT. \label{eq:quantities}
    \end{equation}
    Now, using Proposition \ref{prop:dog-wgd-bound-1-forward-flow-uniform} and Lemma \ref{lemma:smooth-grad-bound-improved}, we have that
    \begin{align}
        S_T = \mathcal{O}\left(K_T{\sqrt{G_T}}\right) = \mathcal{O}\left(K_T\sqrt{\frac{8L}{3}S_T + 2LC_T}\right)
    \implies
        S_T^2 = \mathcal{O}\left(K_T^2\left(\frac{8L}{3}S_T + 2LC_T\right)\right).
    \end{align}
    Solving this quadratic inequality, it follows that
    \begin{equation}
        S_T= \mathcal{O}\left(\frac{8}{3}LK_T^2 + K_T\sqrt{2LC_T}\right)
    \end{equation}
    Thus, dividing through by $T$, we have 
    \begin{align}
        \frac{1}{T}S_T &= \frac{1}{T}\sum_{t=1}^{T} \left(\mathcal{F}(\mu_t) - \mathcal{F}(\pi)\right) =\mathcal{O}\left(\frac{\frac{8}{3}LK_T^2}{T} + \frac{K_T\sqrt{2LC_T}}{T}\right)
    \end{align}
    Finally, using Lemma \ref{lemma:average} (i.e., the Jensen-type inequality for the average), and re-substituting the three quantities in \eqref{eq:quantities}, the result follows. 
\end{proof}

\begin{proof}[Proof of Corollary \ref{corollary:smooth-dog-convergence}]
First, note that we can neglect the first term in Proposition \ref{prop:smooth-adaptive-convergence-rate-improved}, since asymptotically the second term dominates. Then, on the event $\{\bar{r}_{T+1}\leq D\}$, we have $\bar{r}_{T+1}\leq D$, $d_1\leq D$. It follows that 
    \begin{equation}
        \mathcal{F}(\bar{\mu}_T) - \mathcal{F}(\pi) =\mathcal{O}\left(\frac{\left(D\log_{+}\frac{D}{r_{\varepsilon}} + D\right)\sqrt{2Ld}}{\sqrt{T}}\right) = \mathcal{O}\left(\frac{D\sqrt{2Ld}}{\sqrt{T}}\log_{+}\left(\frac{D}{r_{\varepsilon}}\right)\right).
    \end{equation}
\end{proof}

\end{document}